\documentclass[a4paper,12pt]{article}
\usepackage{graphicx}% Include figure files
\usepackage{dcolumn}% Align table columns on decimal point
\usepackage{bm}%
\usepackage{color}

\usepackage{textcomp}
\usepackage{latexsym}
\usepackage{amssymb}
\usepackage[all]{xy}
\usepackage{graphicx}
\usepackage{lscape}
\usepackage{fancyhdr}
\usepackage[a4paper]{geometry}
\usepackage{pspicture}
\usepackage{indentfirst}
\usepackage{delarray}
\usepackage{epsfig}

\newtheorem{theorem}{Theorem}[section]
\newtheorem{proposition}[theorem]{Proposition}
\newtheorem{lemma}[theorem]{Lemma}

\newtheorem{definition}[theorem]{Definition}

\newtheorem{remark}[theorem]{Remark}

\newtheorem{proof}{Proof}[section]
\newtheorem{prooft}{Proof of theorem}[section]
\newcommand{\cqfd}{\hfill $\square$}

\title{ Quantum deformed algebras : coherent states and special functions}
\author{J.D. Bukweli-Kyemba and M.N. Hounkonnou}

\begin{document}

\maketitle

\begin{abstract}
The Heisenberg algebra is first deformed with the set of parameters $\{q,\; l,\;\lambda\}$ to generate a new family of generalized coherent states. The matrix elements of relevant operators are exactly computed. A proof on sub-Poissonian character of the statistics of the main deformed states is provided. This property is used to determine a generalized metric.
A unified method of calculating structure functions  from commutation relations of deformed single-mode oscillator algebras is then presented. A natural approach to building coherent states associated to  deformed algebras is deduced.
Known deformed algebras are  given as illustration.

Futhermore, we generalize a class of  two-parameter deformed Heisenberg algebras related to meromorphic functions, % defined on $\mathbb{C}\times\mathbb{C}$, 
called {\it ${\cal R}(p,q)$-deformed algebra.} 
Relevant families of coherent states maps  are probed and their corresponding hypergeometric
series are computed. The latter  generalizes known hypergeometric series and gives  a  generalization of the binomial theorem.
 We provide new noncommutative algebras and show that the involved notions of  differentiation and integration  
generalize the usual $q$- and $(p,q)$-differentiation and integration.
A Hopf algebra structure compatible with the ${\cal R}(p,q)$-algebra is deduced.

Besides,  we succeed in giving a new characterization of Rogers- Szeg\"o
polynomials, called {\it ${\cal R}(p,q)$-deformed Rogers- Szeg\"o
polynomials}.
{\it Continuous ${\cal R}(p,q)$-deformed  Hermite polynomials} and their recursion relation are also deduced. Novel algebraic relations are provided and discussed.

The whole formalism is performed in  a unified way, generalizing known relevant results which are straightforwardly derived as particular cases. 
\end{abstract}

\section{Introduction}

Much attention has been paid to the study of quantum algebras (or groups) in their both physical and mathematical aspects, motivated by the passage from classical physical systems to quantum systems.
Recall that the classical and quantum mechanics share two basic concepts of {\it states} and {\it observables}.
Indeed, in classical mechanics, states are generally points of a symplectic manifold $M$, and observables are real-valued functions on it. % , i.e., ordinary
Of course they may be regarded as belonging to  the space $\mathcal{F}(M)$ of differentiable complex-valued functions on $M$, but the outcomes of classical measurements are real. 
In quantum mechanics, states are one-dimensional subspaces of a complex Hilbert space $\mathcal{H}$, and observables are
self-adjoint linear
 operators on $\mathcal{H}$, with real spectrum. 
The connection between classical and quantum mechanics is well made in the language of observables. In both classical and quantum mechanics, the observables form an associative algebra, which is commutative under the pointwise multiplication in the classical case and noncommutative under the composition of linear operator in the quantum case.
 Moreover, the time evolution for the  classical system is expressed by the equation
\begin{eqnarray}
 \frac{df(x(t))}{dt}= \left\{H_c,\;f\right\}(x(t)),\quad f\in\mathcal{F}(M),
\end{eqnarray}
where $H_c$ is the classical Hamiltonian and $x(t)\in M$ is the state of the system at time $t$, while for the quantum system the time evolution of an operator $A$ is given by
\begin{eqnarray}
 \frac{d\;A}{dt}= \big[H_q,\;A\big],
\end{eqnarray}
where $H_q$ is some operator called quantum Hamiltonian.

Hence, {\it the quantization} or the passage from classical to quantum mechanics is somewhat like replacing a Poisson algebra by a  Lie algebra.

The illustrative example is given by a particle moving along the real line. In the classical case, the manifold $M$ is the cotangent bundle $T^*(\mathbb{R})$, and if $q$ is the coordinate on $\mathbb{R}$ (position) and $p$ the coordinate in the fiber direction (momentum), the Poisson bracket is 
\begin{eqnarray}
 \left\{f_1,\;f_2\right\}=\frac{\partial f_1}{\partial q}\frac{\partial f_2}{\partial p}-
\frac{\partial f_2}{\partial q}\frac{\partial f_1}{\partial p}\qquad f_1,\;f_2\in\mathcal{C}^\infty(\mathbb{R}).
\end{eqnarray}
In particular, the Poisson bracket of coordinate functions is
\begin{eqnarray}
 \left\{q,\;p\right\}= 1.
\end{eqnarray}
In quantum case, the Hilbert space $\mathcal{H}$ is the space $L^2(\mathbb{R})$ of square integrable function of $q$, and operators are combinations of the operators $\hat{q}$ and $\hat{p}$ whose actions correspond to the multiplication by $q$ and the derivative $-i\hbar\frac{\partial}{\partial_q}$, respectively. So, 
\begin{eqnarray}\label{Heisenbergal}
 \big[\hat{q},\;\hat{p}\big]=i\hbar\mathbf{1};\qquad [\hat{q},\;\mathbf{1}]=0;\qquad[\hat{p},\;\mathbf{1}]=0
\end{eqnarray}
where $[A, \;B]= AB - BA$,  $\hbar$ is the Planck's constant and $\mathbf{1}$ the identity operator. 
Given physical parameters $\mathbf{m}$ and $\omega$ carrying physical dimensions, it is possible to present, in an equivalent way, the same algebra in terms of the following combinations:
\begin{eqnarray}
 \hat{q}=\sqrt{\hbar/2\mathbf{m}\omega}\left(b+b^\dagger\right),\quad
\hat{p}= -i\sqrt{\left({\mathbf{m}\hbar\omega}/{2}\right)}(b-b^\dagger)
\end{eqnarray}
where $b$ and $b^\dagger$ are called {\it annihilation and creation operators} of the harmonic oscillator obeying
\begin{eqnarray}\label{Harmonic1}
 \big[b,\;b^\dagger\big]=\mathbf{1};\qquad [b,\;\mathbf{1}]=0;\qquad [b^\dagger,\;\mathbf{1}]=0.
\end{eqnarray}
The algebra generated by $\{\hat{q},\;\hat{p},\;\mathbf{1}\}$ or $\{b,\; b^\dagger,\mathbf{1}\}$ satisfying (\ref{Heisenbergal}) or (\ref{Harmonic1}) is  called {\it Weyl-Heisenberg algebra}.

From (\ref{Harmonic1}), one defines the operator $N:=b^\dagger b$, also called {\it number operator}, with the properties:
\begin{eqnarray}\label{Harmonic2}
 [N,\;b]=-b;\quad [N,\;b^\dagger]=b^\dagger;\quad [N,\;\mathbf{1}]=0.
\end{eqnarray}

Let  $\mathcal{F}$  be a Fock space and  $\{|n\rangle|n\in\mathbb{N}\cup\{0\}\}$ be its orthonormal basis. The actions of $b,$ $b^\dagger$ and $N$ on $\mathcal{F}$ are given by
\begin{eqnarray}\label{action}
 b|n\rangle=\sqrt{n}|n-1\rangle,\quad b^\dagger|n\rangle= \sqrt{n+1}|n+1\rangle\quad\mbox{and}\quad N|n\rangle= n|n\rangle
\end{eqnarray}
where $\{|0\rangle\}$ is a normalized vacuum:  $ b|0\rangle=0,\quad \langle0|0\rangle= 1$.  

From (\ref{action}) the states $\{|n\rangle\}$ for $n\ge1$  are built as follows:
\begin{eqnarray}
 |n\rangle= \frac{1}{\sqrt{n!}}(b^\dagger)^n|0\rangle
\end{eqnarray}
satisfying the orthogonality and completeness conditions:
\begin{eqnarray}
\langle n|m\rangle=\delta_{n,m}\quad\mbox{and}\quad \sum_{n=0}^\infty|n\rangle\langle n|=\mathbf{1}.
\end{eqnarray}
Thus, the Hilbert $\mathcal{H}$ space is the Fock space $\mathcal{F}$ and the operators are elements of the algebra generated by  $\{b,\;b^\dagger,\;\mathbf{1}\}$ with the relations (\ref{Harmonic1})-(\ref{Harmonic2}). This algebra, also called the {\it Fock algebra} or {\it quantum oscillator algebra}, successfully describes physical phenomena (%the most 
widely spread in nonrelativistic quantum mechanics) but unfortunately fails in solving some important problems (like divergences in field theories, physics under the Planck scale, symmetry breaking, etc.)

Generalization of {\it canonical commutation relations} (\ref{Heisenbergal}) or (\ref{Harmonic1}) was suggested long before the discovery of quantum groups, by Heisenberg to achieve the regularization for (nonrenormalizable) nonlinear spinor field theory. The issue was consi-dered as small additions to the canonical commutations relations \cite{Dancoff,Tamm}. 
Snyder, investigating the infrared catastrophe of soft photons in the Compton scattering, raised this issue   and built a noncommutative Lorentz invariant space-time where the non-commutativity of space operators is proportional to non-linear combinations of phase space operators \cite{Snyder}. 
Further modifications of the oscillator algebra and their possible physical interpretations as spectrum generating algebras of non standard statistics have been made since the earlier work of Snyder.  As matter of citation, let us mention the $q$-oscillator algebras by  Coon and coworkers  \cite{Arik&Coon,Yu&al72}, Kuryshkin \cite{Kuryshkin80}, Jannussis and collaborators
\cite{Jannussis,Jannussis&al81,Jannussis&al83a,Jannussis&al83b}.

With the development  of quantum groups, new aspects of $q$-oscillators have been identified \cite{Biedenharn,Wess&al99, Macfarlane,Sun&Fu,Wess97} as a tool for providing a boson realization of  quantum algebra $su_q(2)$ using a $q$-analogue of the harmonic oscillator and the Jordan-Schwinger mapping, and then generalizations in view of unifying  or extending different existing $q$-deformed algebras were elaborated \cite{Borzov,Burban1,Meljanac,Odzijewicz98}.

Quantum groups made their first appearance  in the mathematical physics literature in connection  with the quantum inverse scattering method, a technique for studying integrable quantum systems \cite{Faddeev&Takhtajan79,Faddeev&Yu,Pasquier,Vega,Witten}.
It was shown in Ref. \cite{Kulish&Reshetikhin83} that the quantum linear problem of the quantum sine-Gordon equation is associated with the deformation of the Lie algebra $sl_2$ unlike the classical problem which is associated with $sl_2$ itself. Later Sklyanin \cite{Sklyanin83} showed that deformations of Lie algebraic structures were not bound to this particular equation but they were a part of a more general theory.

In the second half of 1980's, Drinfel'd realized that the algebraic structure associated to quantum inverse scattering method can be reproduced by suitable algebraic quantization of the Poisson Lie algebras \cite{Drinfeld85, Drinfeld}. Jimbo obtained independently the same result using the representation theory of the corresponding algebra \cite{Jimbo85,Jimbo}.
In fact, they discovered that {\it quantum groups  are dual category of Hopf algebras which are neither commutative nor co-commutative} \cite{Drinfeld, Jimbo}. 
Notice that most of the well studied concrete examples of quantum groups are deformations of the universal enveloping algebra of the semi-simple Lie algebras \cite{Chari&Pressley,Drinfeld,Faddeev&al90,Jimbo85,Kassel,Klimyk&Schmudgen,Kulish&Reshetikhin83,Sahai97,Sklyanin83}.

Despite all useful properties and applications motivated by various one-pa-rameter deformed algebras,  the multi-parameter deformations aroused much interest because of their  flexibility  when  dealing with concrete physical models \cite{Baloitcha,Burban93,Burban2007,Burban&Klimyk,Chakrabarti&Jagan,Floreanini&al93a,Floreanini&al93b,Gelfand&al94,Hounkonnou&Bukweli10,Hounkonnou&Ngompe07a,Hounkonnou&Ngompe07b}.

Coherent states have practically followed the same trend as the quantum algebras. They were invented by Schr\"odinger in 1926 in  the context of the quantum harmonic oscillator. They were defined as minimum-uncertainty states that exhibit classical behavior  \cite{Schrodinger}.

In 1963, coherent states % have been
were simultaneous rediscovered by Glauber \cite{Glauber1,Glauber2}, Klauder \cite{Klauder63a,Klauder63b} and Sudarshan \cite{Sudarshan} in   quantum optics of coherent light beams emitted by lasers.
The following definition summarizes the concept of Schr\"odinger-Glauber-Klauder-Sudarshan coherent states:
\begin{definition}\label{CSdefintro}
Coherent states (CS) are normalized states $|z\rangle\in\mathcal{H}$ satisfying one of the  following three equivalent conditions:
\begin{itemize}
\item[(i)] they saturate the Heisenberg inequality
\begin{eqnarray}
 (\Delta \hat{q})(\Delta \hat{p})\geq \frac{\hbar}{2},
\end{eqnarray}
where { $(\Delta X)^2:=\langle z|X^2-\langle X\rangle^2|z\rangle$}  with { $\langle X\rangle:= \langle z|X|z\rangle$};
 \item [(ii)] they are eigenvectors of the annihilation operator, with eigenvalue $z\in \mathbb{C}$:
\begin{eqnarray}
 b|z\rangle= z|z\rangle;
\end{eqnarray}
\item[(iii)] they are obtained from the ground state $|0\rangle$ of the harmonic oscillator by a unitary action of the Weyl-Heisenberg group:
\begin{eqnarray}%\label{bch2}
 |z\rangle=e^{zb^\dagger-\bar z b}|0\rangle.
\end{eqnarray} 
\end{itemize}
\end{definition}
The important feature of these coherent states resides in the partition (resolution) of identity:
\begin{eqnarray}
\int_\mathbb{C}{{[d^2z]}\over\pi} |z\rangle\langle z|= \mathbf{1},
\end{eqnarray}
where we have put $[d^2z]= d(Rez)d(Imz)$ for simplicity.

Since there, CS became very popular objects in mathematics (specially in functional analysis, group theory and representations, geometric quantization, etc.), and in nearly all branches of quantum physics (nuclear,  atomic and solid state physics, statistical mechanics, quantum electrodynamics, path integral, quantum field theory, etc.). For more information
we refer the reader to the references \cite{Ali&al,Klauder&Skagerstam,Perelomov72,Zhang&al90}. 

The vast field covered  by coherent states motivated their generalizations to other families of states deducible from noncanonical operators and satisfying not necessarily all  above mentioned  properties.

The first  class of generalizations, based on the  equivalent conditions given in definition \ref{CSdefintro}, include:\\
a) The approach  by Barut and Girardello \cite{Barut}  considering coherent states as eigenstates of the annihilation operator. This approach was unsuccessful because of its drawbacks from both mathematical and physics point of view  as
detailed in  \cite{Gilmore74a,Perelomov72}.\\
b) The approach based on the  minimum-uncertainty states, i.e.  essentially on the original motivation of Schr\"odinger in his construction of wavepackets which follow the motion of a classical particle while retaining their shapes. This was the basis for building the intelligent coherent states for various dynamical systems \cite{Aragone&al76, Aragone&al74, Nieto&Simmons78,Nieto&Simmons79, Nieto&al81}. Nevertheless, as has been emphasized by Zhang {\it et al} \cite{Zhang&al90}, such a  generalization has several limitations.\\
c) The approach related to the unitary representation of the group generated by the creation and annihilation operators. In two papers by Klauder \cite{Klauder63a,Klauder63b} devoted to a set of continuous states, one finds the basic ideas of  coherent states construction for arbitrary Lie groups, which have been exploited by  Gilmore \cite{Gilmore} and  Perelomov \cite{Perelomov72, Perelomov}  to formulate a general and complete formalism of building coherent states for Heisenberg groups with various properties similar to those of the harmonic oscillator. The key result of  this development was the  intimate connection of  the coherent states  with the dynamical group of  a given  physical system.

 Two other generalizations complete this first class of generalizations:  
$(i)$  the covariant coherent states introduced in Ref. \cite{Ali&al},  considered as a generalization of Gilmore-Perelomov formalism  in the sense that the CS are built from more general groups (homogeneous spaces), and 
$(ii)$  the nonlinear coherent states  related to nonlinear algebras. 
Even though nonlinear coherent states have been used to analyze some quantum mechanical systems as  the motion of a trapped ion \cite{Junker&Roy,deMatos&Vogel}, they are not merely mathematical objects. They were defined as right eigenstates of a generalized annihilation operator \cite{Manko&al,deMatos&Vogel}.

The second class of generalizations is essentially based on the overcompleteness property of coherent states. This property was the {\it raison d'\^etre} of the mathematically oriented construction of generalized coherent states by Ali {\it et al} \cite{Ali&al99,Ali&al} or of the ones with physical orientations \cite{Gazeau09,Gazeau&Klauder,Klauder&al}. 
The construction of generalized CS corresponds to the problem of finding a map from a row set $X$, equiped with a measure $\mu(dx)$, to a (projective) Hilbert space of quantum states ${\cal H}$ (a closed subspace of $L^2(X,\mu)$), $x\mapsto |x\rangle$, defining a family of states $\{|x\rangle\}_{x\in X}$ obeying the following two conditions:
\begin{itemize}
\item[$(1)$] {\it Normalization}: $\langle x|x\rangle =1$;
\item[$(2)$] {\it Resolution of the unit in ${\cal H}$}:
\begin{eqnarray}
\int_X |x\rangle\langle x|\nu(dx)= \mathbf{1}_{\cal H},
\end{eqnarray}
where $\nu(dx)$ is another measure on $X$, usually absolutely continuous with respect to $\mu(dx)$. This means that there exists a positive measurable function $h(x)$ such that $\nu(dx)=h(x)\mu(dx)$ \cite{Gazeau09}.
\end{itemize}
Numerous publications continue to appear  using this property, see for example \cite{Daoud,Hounkonnou&Ngompe07b,Hounkonnou&Sodoga,Popov} and references therein. The overcompleteness property is the most important criteria to be satisfied by CS  as required by  Klauder's criteria \cite{Klauder&al}. 

Let us also   mention  the generalization performed through the so-called {\it coherent state map,} elaborated by Odzijewicz \cite{Odzijewicz98} in 1998.  It is now known that  the coherent state
map may be used as a tool for  the geometric quantization {\it \`a la Kostant-Souriau}  \cite{Maximov&Odzijewicz,Odzijewicz98}. See the works by Kirillov \cite{Kirillov} and Kostant \cite{Kostant} for details on geometric quantization. Further generalization of the latter approach is performed in the framework of this thesis.

More recently, Horzela and Szafraniec \cite{Horzela&Szafraniec12} have introduced the measure-free approach of building CS, requiring two main objects:
\begin{itemize}
\item[$(1')$] a separable Hilbert space ${\cal H}$ with a fixed orthonormal basis $\{e_n\}_{n=0}^d$, $\;d+1= dim{\cal H}$;
\item[$(2')$] a sequence $\{\phi_n\}_{n=0}^d$ of complex valued functions on $X$ satisfying
\begin{eqnarray}
\sum_{n=0}^d|\phi(x)|^2<\infty\quad x\in X
\end{eqnarray}
and
\begin{eqnarray}
\{\alpha_n\}_{n=0}^d\;\;\mbox{and}\;\; \sum_{n=0}^d\alpha_n\phi_n(x)=0\;\;\mbox{for all}\; x \; \mbox{imply that all}\;\; \alpha_n\;\mbox{'s}\;\;\mbox{are}\;0.
\end{eqnarray}
\end{itemize}
The formula
\begin{eqnarray}
K(x,y):= \sum_{n=0}^d\phi_n(y)\overline{\phi_n(x)}
\end{eqnarray}
is regarded as the definition of the reproducing kernel. Therefore, if $K(x,x)\neq0$, $x\in X$, then the prospective CS may be defined  at $x$ by
\begin{eqnarray}
|x\rangle:=\sum_{n=0}^dK(x,x)^{-\frac{1}{2}}\overline{\phi_n(x)}e_n.
\end{eqnarray}
Hence,
\begin{eqnarray}
\langle x|y\rangle= (K(x,x)K(y,y))^{-\frac{1}{2}}K(x,y)
\end{eqnarray}
which means that the states $|x\rangle$ are not normalized. To evoid any further renormalization, the authors of \cite{Horzela&Szafraniec12} assumed that $K(x,x)=1$, $x\in X$ that leads to
the normalized CS
\begin{eqnarray}
C_x:=\sum_{n=0}^d\overline{\phi_n(x)}e_n.
\end{eqnarray}
 
This work is organized as follows. In Section \ref{chap1}, the Heisenberg algebra is deformed with the set of parameters
 $\{q,\; l,\;\lambda\}$  and  the   structure function is deduced. The spectrum of  the associated deformed oscillator is computed. Then, we built the associated deformed coherent states using the Klauder approach and investigate the  quantum statistics and geometry  of the deformed coherent states.

Section \ref{chap2} is devoted to the unification of deformed single-mode oscillator algebras. We give a method of computing the so-called {\it structure function} which is the basis of coherent states construction related to a given algebra. 
We analyse known deformed oscillator algebras and compute corresponding structure functions and give coherent states satisfying the Klauder criteria.

Section \ref{chap3} generalizes a class of  two-parameter
deformed Heisenberg algebras related to meromorphic functions%
% defined on $\mathbb{C}\times\mathbb{C}$. 
Relevant families of coherent states maps are probed and their corresponding hypergeometric
series are computed. Moreover, an ${\cal R}(p, q)$-binomial theorem, generalizing the
$(p,q)$-binomial theorem given in ~\cite{Jagannathan&Rao} is deduced.
This chapter provides the  definitions of the 
${\cal R}(p, q)$-trigonometric, hyperbolic and $(p,q)$-Bessel functions, including
their main relevant properties. The framework developed in this chapter can be considered as a reverse
of the previous one.

In Section \ref{chap4} we build a framework for ${\cal R}(p,q)$-deformed calculus, which provides a method of the computation for a deformed ${\cal R}(p,q)$-derivative,  generalizing known deformed derivatives of analytic functions
defined on a complex disc as particular cases corresponding to  conveniently chosen  meromorphic functions.
We introduce a new  algebra generated by four quantities provided some conditions are satisfied and define the ${\cal R}(p,q)-$differential calculus  yielding the ${\cal R}(p,q)-$integration.   
Also a  construction of Hopf algebra structure is given in Subsection \ref{2section4}, while in
Subsection \ref{2section5} we show that some particular cases can be deduced from the constructed general formalism.

Section \ref{chap5} addresses a new characterization of  ${\cal R}(p,q)$-deformed Rogers-Szeg\"o \,\, polynomials by providing their three-term recursion relations and the associated quantum algebra built with  corresponding 
creation and annihilation operators. The whole construction is performed in  
 a unified way, generalizing all known relevant results which are straightforwardly derived as particular cases. 
Continuous ${\cal R}(p,q)-$deformed  Hermite polynomials and their reccursion
relations are also deduced. Novel relations are provided and discussed.
Finally, there follow the concluding remarks.

%%%%%%%%%%%%%%%%%%%%%%%%%%%%%%%%%%%%%%%%%%%%%%%%%%%%%%%%%%%%%%%%%%%%%%%%%%%%%

\section{The $(q;l,\lambda)$-deformed Heisenberg algebra: coherent states, their statistics and geometry}\label{chap1}
% \author{J.D. Bukweli-Kyemba and M.N. Hounkonnou,}

% \begin{document}
$\;$
% \maketitle

% \begin{abstract}
The Heisenberg algebra is deformed with the set of parameters
 $\{q, l,\lambda\}$  to generate  a new family of generalized 
coherent states respecting  the Klauder criteria. In this framework, 
 the matrix elements of  relevant operators are exactly computed. Then, 
 a proof on the sub-Poissonian character of the statistics of  the main deformed states 
 is provided. This property   is used to   determine   the induced generalized metric.

% \end{abstract}
% 
% The chapter is  organized as follows. In Section \ref{Sec1.2}, the deformed Heisenberg algebra is described and  the   structure function is deduced. The spectrum of  the associated deformed oscillator is computed. The Section~\ref{Sec1.3}  is devoted to the construction of the deformed coherent states using the Klauder approach.  In section \ref{Sec1.4},  quantum statistics and geometry  of the deformed coherent states are investigated. Concluding remarks end the chapter in Section \ref{Sec1.5}.

\subsection{$(q;l,\lambda)$-deformed Heisenberg algebra} \label{Sec1.2}
 Consider now  the following $(q;l,\lambda)$-deformed Heisenberg algebra  \cite{Bukweli&Hounkonnou12c}  generated by operators $N$, $a$, $a^\dagger$
satisfying 
\begin{eqnarray}
[N,\;a]=-a,\qquad [N,\;a^\dagger]= a^\dagger,\label{Kalnins1}
\end{eqnarray}
with the operator products 
\begin{eqnarray}
&& aa^\dagger-a^\dagger a =l^2q^{-N+\lambda-1},\label{Kalnins2}
\end{eqnarray}
where $l$ and $\lambda$ are real numbers with $l\neq0$.\\
One can readily check that the commutator $[.,\;.]$ of operators is antisymmetric and satisfies the Jacobi identity conferring  a Lie algebra structure to the $(q;l,\lambda)$-deformed Heisenberg algebra. This algebra plays an important role in  mathematical sciences in general, and, in  particular, in mathematical physics. In a notable work \cite{Kalnins},
similar associative algebra has been investigated  by Kalnins {\it et al}
 under the form:
\begin{eqnarray}
&&[H,\;E_+]= E_+\qquad [H,\;E_-]=-E_-\cr
&&[E_+,\;E_-]=-q^{-H}\mathcal{E}\qquad [\mathcal{E},\;E_{\pm}]=0=[\mathcal{E},\;H],
\end{eqnarray}
where $q$ is a  real number such that $0<q<1$. These authors showed that
the elements $\mathcal{C}= qq^{-H}\mathcal{E}+(q-1)E_+E_-$ and $\mathcal{E}$
lie in the center of this algebra.  It admits a class of irreducible representations for $\mathcal{C}=l^2I$ and $\mathcal{E}= l^2q^{\lambda-1}I$.

The $(q;l,\lambda)$-deformed Heisenberg algebra (\ref{Kalnins1}) is a generalized algebra in the sense that it can generate a series of existing algebras as particular cases. For instance, even
 the generalization of the Quesne-algebra performed in \cite{Hounkonnou&Ngompe07a,Quesne} can  be deduced from (\ref{Kalnins1}) by setting   $l=1$ and $\lambda=0$.

In the sequel, we consider the Fock space of the Bose oscillator constructed as follows. From the vacuum vector $|0\rangle$ defined by $a|0\rangle=0$, the normalized vectors $|n\rangle$ for $n\ge 1,$ i.e. eigenvectors of the operator $N,$ are obtained as $|n\rangle=C_n(a^\dagger)^n|0\rangle$, where $C_n$ stands for some  normalization constant to be determined. 

\begin{proposition}
The structure function of the $(q; l, \lambda)$-deformed Heisenberg algebra $(\ref{Kalnins1})-(\ref{Kalnins2})$ is given by
\begin{eqnarray}\label{strucure}
\varphi(n)= l^2q^\lambda\frac{1-q^{-n}}{q-1}=l^2q^{\lambda-n}[n]_q,\;\; q>0,
\end{eqnarray}
where $[n]_q= \frac{1-q^n}{1-q},$ with $0<q<1$ or $1<q,$ is the $q_n-$ factors (also known as $q$-deformed numbers in Physics literature  \cite{Gasper}).
\end{proposition}
{\bf Proof:} From the definition (\ref{uq02}),  $a^\dagger a=\varphi(N)$ and $aa^\dagger= \varphi(N+1).$ Thus, (\ref{Kalnins2}) is written as
\begin{eqnarray*}
 \varphi(N+\mathbf{1})-\varphi(N)=l^2q^{-N+\lambda-1}.
\end{eqnarray*}
Applying this relation to the vectors $|n\rangle,$  we obtain the recursion relation 
\begin{eqnarray*}
\varphi(n+1)-\varphi(n)=l^2q^{\lambda-n-1},\quad \forall n\in\mathbb{N}
\end{eqnarray*}
from which we deduce
\begin{eqnarray*}
 \varphi(n)= \varphi(0)+l^2q^\lambda\frac{1-q^{-n}}{q-1}.
\end{eqnarray*}
Since, in particular, $\varphi(N)|0\rangle=a^\dagger a|0\rangle=0$ implies $\varphi(0)|0\rangle=0$, we have $\varphi(0)=0.$ Then (\ref{strucure}) follows.
The structure  function  is also a strictly increasing function for $x\in\mathbb{R}$ since  
\begin{eqnarray*}
 \frac{d\varphi(x)}{dx}= l^2q^{\lambda-x}\frac{\ln q}{q-1}>0,\;\mbox{ for } q>0.
\end{eqnarray*}
Since $\varphi(0)=0$, it follows that $\varphi(x)\geq0$ for any real $x>0$ and in particular $\varphi(n)\geq 0$, $\forall n\geq0$. \hfill$\square$

\begin{proposition}\label{corola1}
The orthonormalized basis of the Fock space $\mathcal{F}$ is given by
\begin{eqnarray}\label{fockstate}
 |n\rangle= \frac{q^{n(n+1)/4}}{\sqrt{(l^2q^\lambda)^n[n]_q!}}(a^\dagger)^n|0\rangle,\qquad n=0,\;1,\;2,\; ...
\end{eqnarray}
where $[0]_q!:=1$ and $[n]_q!:=[n]_q[n-1]_q...[1]_q$.\\
Moreover, the action of the operators $a$, $a^\dagger$, $N$, $a^\dagger a$ and $aa^\dagger$ on the  vectors $|n\rangle$ for $n\ge 1$ are given by 
\begin{eqnarray}
&&a|n\rangle=\sqrt{l^2q^{\lambda-n}[n]_q}|n-1\rangle,\\&& 
a^\dagger|n\rangle=\sqrt{l^2q^{\lambda-n-1}[n+1]_q}|n+1\rangle,\\&& 
N|n\rangle=n|n\rangle, \\&&  a^\dagger a|n\rangle=l^2q^{\lambda-n}[n]_q|n\rangle,\\&& 
aa^\dagger|n\rangle=l^2q^{\lambda-n-1}[n+1]_q|n\rangle.
\end{eqnarray}
\end{proposition}
{\bf Proof:} To determine the constant of normalization $C_n,$ we set
\begin{eqnarray*}
 1=:\langle n|n\rangle=|C_n|^2\langle 0|a^n(a^\dagger)^n|0\rangle= |C_n|^2\varphi(n)\varphi(n-1)...\varphi(1)\langle 0|0\rangle
\end{eqnarray*}
leading to $C_n=\frac{q^{n(n+1)/4}}{\sqrt{(l^2q^\lambda)^n[n]_q!}}$. Replacing $C_n$ by its value in the definition of $|n\rangle$ given above yields (\ref{fockstate}). The orthogonality of the vectors $|n\rangle$ is a direct consequence of $a|0\rangle=0$. The rest of the proof is obtained from (\ref{fockstate}) using (\ref{Kalnins1}), (\ref{Kalnins2}) and (\ref{strucure}).\hfill$\square$

\begin{theorem}\label{thm}
 The operators $(a+a^\dagger)$ and $i(a-a^\dagger),$ defined on the Fock space $\mathcal{F},$ are bounded and, consequently, self-adjoint  if $q>1$. If $q<1$, they are not self-adjoint. 
\end{theorem}
{\bf Proof:} The matrix elements of  the operator $(a+a^\dagger)$ on the basis $|n\rangle$ are given by
\begin{eqnarray}\label{jacobi01}
\langle m|(a+a^\dagger)|n\rangle=x_n\delta_{m, n-1}+x_{n+1}\delta_{m, n+1}, n,\;m=0,\;1,\;2,\;\cdots 
\end{eqnarray}
while the matrix elements of the operator  $i(a-a^\dagger)$  are given by
\begin{eqnarray}\label{jacobi02}
\langle m|i(a-a^\dagger)|n\rangle=ix_n\delta_{m, n-1}-ix_{n+1}\delta_{m, n+1}, n,\;m=0,\;1,\;2,\;\cdots
\end{eqnarray}
where $x_n= \left(l^2q^{\lambda-n}[n]q\right)^{1/2}$.
Besides, the operators $(a+a^\dagger)$ and $i(a-a^\dagger)$ can be represented by the two following symmetric Jacobi matrices, respectively:
\begin{eqnarray}\label{jacobir}
 \left(\begin{array}{cccccc}0&x_1&0&0&0&\cdots\\x_1&0&x_2&0&0&\cdots\\0&x_2&0&x_3&0&\cdots\\\vdots&\ddots&\ddots&\ddots&\ddots&\ddots
       \end{array}\right)
\end{eqnarray}
and
\begin{eqnarray}\label{jacobic}
 \left(\begin{array}{cccccc}0&-ix_1&0&0&0&\cdots\\ix_1&0&-ix_2&0&0&\cdots\\0&ix_2&0&-ix_3&0&\cdots\\\vdots&\ddots&\ddots&\ddots&\ddots&\ddots
       \end{array}\right)
\end{eqnarray}
Two situations deserve investigation:
\newline
$\bullet$ Suppose    $q>1$. Then,
\begin{eqnarray*}
 \left|x_n\right|=\left(\frac{l^2q^\lambda}{q-1}\frac{q^n-1}{q^n}\right)^{1/2}<\left(\frac{l^2q^\lambda}{q-1}\right)^{1/2},\; \forall n\geq1.
\end{eqnarray*}
Therefore, the Jacobi matrices  in (\ref{jacobir}) and (\ref{jacobic}) are bounded and self-adjoint (Theorem 1.2., Chapter VII in Ref. \cite{Berezanskii}). Thus, $(a+a^\dagger)$ and $i(a-a^\dagger)$ are bounded and, consequently, self-adjoint.\\
$\bullet$ Contrarily, if $q<1,$ then
\begin{eqnarray}
 \lim_{n\to\infty}x_n=\lim_{n\to\infty}\left(l^2q^\lambda\frac{1-q^{-n}}{q-1}\right)^{1/2}=\infty.
\end{eqnarray}
Considering the series $\sum_{n=1}^\infty 1/x_n$, we obtain 
\begin{eqnarray*}
\overline{\lim_{n\to\infty}}\left(\frac{1/x_{n+1}}{1/x_n}\right)=\overline{\lim_{n\to\infty}}\left(\frac{1-q^{-n}}{1-q^{-n-1}}\right)^{1/2}= q^{1/2}<1.
\end{eqnarray*}
This ratio test  leads to the conclusion  that the series $\sum_{n=1}^\infty 1/x_n$ converges. Moreover, $1-2q+q^2=(1-q)^2\geq0\Longrightarrow q^{-1}+q\geq2$. Hence,
\begin{eqnarray*}
&& 0\leq\left(\frac{l^2q^\lambda}{q-1}\right)^2\left(1-q^{-n}(q+q^{-1})+q^{-2n}\right)\leq\left(1-2q^{-n}+q^{-2n}\right)\left(\frac{l^2q^\lambda}{q-1}\right)^2 \cr
\Leftrightarrow&&0\leq\left(l^2q^\lambda\frac{1-q^{-n+1}}{q-1}\right)\left(l^2q^\lambda\frac{1-q^{-n-1}}{q-1}\right)\leq \left(l^2q^\lambda\frac{1-q^{-n}}{q-1}\right)^2 \cr
\Leftrightarrow&&0\leq\left(l^2q^\lambda\frac{1-q^{-n+1}}{q-1}\right)^{1/2}\left(l^2q^\lambda\frac{1-q^{-n-1}}{q-1}\right)^{1/2}\leq \left(l^2q^\lambda\frac{1-q^{-n}}{q-1}\right) \cr
\Leftrightarrow &&
0\leq x_{n-1}x_{n+1}\leq x_n^2.
\end{eqnarray*}
Therefore, the Jacobi matrices  in (\ref{jacobir}) and (\ref{jacobic}) are not self-adjoint (Theorem 1.5., Chapter VII in Ref. \cite{Berezanskii}).\hfill$\square$

\begin{definition}
The $(q;l,\lambda)$-deformed position, momentum  
and Hamiltonian  operators denoted by $X_{l,\lambda,q}$, $P_{l,\lambda,q}$ and $H_{l,\lambda,q},$ respectively, are defined as follows:
\begin{eqnarray}
 X_{l,\lambda,q}&:=& \left({\hbar}/{2\mathbf{ m}\omega}\right)^{1/2}(a+a^\dagger),\cr
%\end{eqnarray}
%\begin{eqnarray}
P_{l,\lambda,q}&:=& -i\left({\mathbf{ m}\hbar\omega}/{2}\right)^{1/2}(a-a^\dagger)\cr
%\end{eqnarray}
%and
%\begin{eqnarray}
 H_{l,\lambda,q}&:=& \frac{1}{2\mathbf{ m}}(P_{l,\lambda,q})^2+\frac{1}{2}\mathbf{ m}\omega^2(X_{l,\lambda,q})^2
\cr&=& \frac{\hbar\omega}{2}(a^\dagger a+aa^\dagger).
\end{eqnarray}
%respectively.  
\end{definition}

\begin{proposition} The following system characterization holds:
\begin{itemize}
\item The vectors $|n\rangle$ are eigenvectors of the $(q;l,\lambda)$-deformed Hamiltonian with respect to the eigenvalues
\begin{eqnarray}\label{mecaprop1}
 E_{l,\lambda,q}(n)= \frac{\hbar\omega}{2}l^2q^{\lambda-n-1}\big(q[n]_q+ [n+1]_q\big).
\end{eqnarray}
\item The mean values of $X_{l,\lambda,q}$ and $P_{l,\lambda,q}$ in the states $|n\rangle$ are zero while their variances are given by
\begin{eqnarray}
 (\Delta X_{l,\lambda,q})_n^2 &=& \frac{\mathbf{ m}\hbar\omega}{2}l^2q^{\lambda-n-1}\big(q[n]_q+ [n+1]_q\big),\label{mecaprop2}\\
(\Delta P_{l,\lambda,q})_n^2 &=& \frac{\hbar}{2\mathbf{ m}\omega}l^2q^{\lambda-n-1}\big(q[n]_q+ [n+1]_q\big),\label{mecaprop3}
\end{eqnarray}
where $(\Delta A)_n^2=\langle A^2\rangle_n-\langle A\rangle_n^2$  with $\langle A\rangle_n=\langle n|A|n\rangle$.

\item The position-momentum uncertainty relation is given by
\begin{eqnarray}\label{mecaprop4}
 (\Delta X_{l,\lambda,q})_n(\Delta P_{l,\lambda,q})_n= \frac{h}{2}l^2q^{\lambda-n-1}\big(q[n]_q+ [n+1]_q\big)
\end{eqnarray}
which is reduced, for the vacuum state, to the expression
\begin{eqnarray}\label{mecaprop5}
 (\Delta X_{l,\lambda,q})_0(\Delta P_{l,\lambda,q})_0= \frac{h}{2}l^2q^{\lambda-1}.
\end{eqnarray}
\end{itemize}
\end{proposition}
{\bf Proof:} Indeed, using the result of the Proposition \ref{corola1}, we get
\begin{eqnarray*}
 H_{l,\lambda,q}|n\rangle= \frac{\hbar\omega}{2}(a^\dagger a+aa^\dagger)\rangle= \frac{\hbar\omega}{2}l^2q^{\lambda-n-1}\big(q[n]_q+ [n+1]_q\big)|n\rangle.
\end{eqnarray*}
The  first two relations (\ref{jacobi01}) and (\ref{jacobi02}) in the proof of the previous Theorem \ref{thm} yield
$\langle n|(a+a^\dagger)|n\rangle=0=\langle n|i(a-a^\dagger)|n\rangle$ and 
$\langle n|(a+a^\dagger)^2|n\rangle=x_n^2+x_{n+1}^2=\langle n|i^2(a-a^\dagger)^2|n\rangle$. Therefore, $\langle n|X_{l,\lambda}|n\rangle=0=\langle n|P_{l,\lambda}|n\rangle$,
$\langle n|X^2_{l,\lambda}|n\rangle=\frac{\mathbf{ m}\hbar\omega}{2}(x_n^2+x_{n+1}^2)$ and $\langle n|P^2_{l,\lambda}|n\rangle=\frac{\hbar}{2\mathbf{ m}\omega}(x_n^2+x_{n+1}^2)$. The rest of the proof is obtained replacing $x_n$ and  $x_{n+1}$ by their expressions.\hfill$\square$

\subsection{Coherent states $|z\rangle_{l,\lambda}$}\label{Sec1.3}

\begin{definition}
The coherent states associated with the algebra (\ref{Kalnins1})-(\ref{Kalnins2}) are defined as
\begin{eqnarray}\label{KalninCS}
 |z\rangle_{l,\lambda}:= \mathcal{N}_{l,\lambda}^{-1/2}(|z|^2)\sum_{n=0}^\infty\frac{q^{n(n+1)/4}z^n}{\sqrt{(l^2q^\lambda)^n[n]_q!}}|n\rangle,\; z\in\mathbf{D}_{l,\lambda},
\end{eqnarray}
where
\begin{eqnarray}
 \mathcal{N}_{l,\lambda}(x)&=&\sum_{n=0}^\infty\frac{q^{n(n+1)/2}x^n}{(l^2q^\lambda)^n[n]_q!}
=\sum_{n=0}^\infty\frac{q^{n(n-1)/2}}{(q;q)_n}\left(\frac{(1-q)qx}{ l^2q^\lambda}\right)^n
\end{eqnarray}
and
\begin{eqnarray}
 \mathbf{D}_{l,\lambda}=\left\{z\in\mathbb{C}:\; |z|^2<R_{l,\lambda}\right\},\;
\mbox{with }\; R_{l,\lambda}=\left\{\begin{array}{lcl} \infty&\mbox{if }& 0<q<1\\\frac{l^2q^\lambda}{q-1}&\mbox{if }& q>1.\end{array}\right.
\end{eqnarray}
\end{definition}
 $R_{l,\lambda}$ is the   convergence radius of the series $\mathcal{N}_{l,\lambda}(x)$.

Remark  that the $q$-deformed coherent states introduced in \cite{Quesne} are recovered as a particular case  corresponding to $l=1$ and $\lambda=0$.

We now aim  at showing that  the coherent states (\ref{KalninCS}) satisfy the Klauder's criteria \cite{Klauder&Skagerstam,Klauder&al}. To this end let us first prove the following lemma:

\begin{lemma}\label{Buklemma} If $q>1,$ then
\begin{eqnarray}
&& \frac{\mathcal{N}_{l,\lambda}(x)}{\mathcal{N}_{l,\lambda}(q^{-1}x)}=\frac{1}{1-(q-1)x/(l^2q^\lambda)}\;,\label{Kaprop1}\\
&&\mathcal{N}_{l,\lambda}(x)= \frac{1}{\left((q-1)x/(l^2q^\lambda);q^{-1}\right)_\infty},\label{Kaprop2}\\
&&\int_0^{R_{l,\lambda}}x^n\left(\mathcal{N}_{l,\lambda}(q^{-1}x)\right)^{-1}d_q^{l,\lambda}x= (l^2q^\lambda)^nq^{-n(n+1)/2}[n]_q!.\label{Kaprop3}
\end{eqnarray}
\end{lemma}
{\bf Proof:} 
We use  the $(q;l,\lambda)$-derivative defined by
\begin{eqnarray}\label{Kaderiva}
 \partial_q^{l,\lambda}f(x)= l^2q^\lambda\frac{f(x)-f(q^{-1}x)}{(q-1)x}
\end{eqnarray}
to obtain
\begin{eqnarray*}
\mathcal{N}_{l,\lambda}(x)=\partial_q^{l,\lambda}\mathcal{N}_{l,\lambda}(x)=
l^2q^\lambda\frac{\mathcal{N}_{l,\lambda}(x)-\mathcal{N}_{l,\lambda}(q^{-1}x)}{(q-1)x}
\end{eqnarray*}
which leads to (\ref{Kaprop1}) and 
\begin{eqnarray}
\mathcal{N}_{l,\lambda}(x)=\frac{\mathcal{N}_{l,\lambda}(q^{-n}x)}{\prod_{k=0}^{n-1}\big(1-(q-1)q^{-k}x/(l^2q^\lambda)\big)}, \; n=1,\;2,\;...
\end{eqnarray}
Letting $n$ to $+\infty$ and taking into account the fact that $\mathcal{N}_{l,\lambda}(0)=1$ lead to (\ref{Kaprop2}).

Next, we use  the $(q;l,\lambda)$-integration given by
\begin{eqnarray}\label{Kaintegra}
\int_0^af(x)d_q^{l,\lambda}x=\frac{q-1}{l^2q^\lambda}a\sum_{k=0}^\infty q^{-k}f(aq^{-k})
\end{eqnarray}
to get
\begin{eqnarray*}
\int_0^{R_{l,\lambda}}x^n\left(\mathcal{N}_{l,\lambda}(q^{-1}x)\right)^{-1}d_q^{l,\lambda}x &=& \sum_{k=0}^\infty q^{-(n+1)k}\frac{(l^2q^\lambda)^n}{(q-1)^n}\big(q^{-(k+1)};q^{-1}\big)_\infty
\cr&=&\frac{(l^2q^\lambda)^n}{(q-1)^n}\big(q^{-1};q^{-1}\big)_\infty
\sum_{k=0}^\infty \frac{q^{-(n+1)k}}{\big(q^{-1};q^{-1}\big)_k}
\cr&=&\frac{(l^2q^\lambda)^n}{(q-1)^n}\frac{\big(q^{-1};q^{-1}\big)_\infty}{\big(q^{-(n+1)};q^{-1}\big)_\infty}
\cr&=&\frac{(l^2q^\lambda)^n}{(q-1)^n}\big(q^{-1};q^{-1}\big)_n
=(l^2q^\lambda)^nq^{-n(n+1)/2}[n]_q!.
\end{eqnarray*}
\hfill$\Box$

\begin{proposition}\label{prop123}
The coherent states defined in (\ref{KalninCS})
\begin{itemize}
 \item [(i)] are normalized eigenvectors of the operator $a$ with eigenvalue $z$, i.e.
\begin{eqnarray}
a|z\rangle_{l,\lambda}=z|z\rangle_{l,\lambda},\qquad {}_{l,\lambda}\langle z|z\rangle_{l,\lambda}=1;
\end{eqnarray}
\item[(ii)] are not orthogonal to each other, i.e.
\begin{eqnarray}
 {}_{l,\lambda}\langle z_1|z_2\rangle_{l,\lambda} \neq0,\;\;\mbox{when}\;\;z_1\neq z_2;% \delta(z_1-z_2);
\end{eqnarray}
\item[(iii)]are continuous in their labels $z$;
\item[(iv)] resolve the identity, i.e.
\begin{eqnarray}\label{unitf}
 \mathbf{1}=\int_{\mathbf{D}_{l,\lambda}} d\mu_{l,\lambda}(\bar z,z)|z\rangle_{l,\lambda} {}_{l,\lambda}\langle z|,
\end{eqnarray}
where
\begin{eqnarray}
 d\mu_{l,\lambda}(\bar z,z)=\frac{1-q}{l^2q^\lambda\ln{q^{-1}}}\frac{\mathcal{N}_{l,\lambda}(\bar zz)}{\mathcal{N}_{l,\lambda}(\bar zz/q)}\frac{d^2z}{\pi},\;\mbox{ if }\; 0<q<1,
\end{eqnarray}
and
\begin{eqnarray}
 d\mu(\bar z,z)=\frac{1}{2\pi}\frac{d_q^{l,\lambda}x\;d\theta}{1-(q-1)x/(l^2q^\lambda)}, \quad x=|z|^2,\;\;\theta=\arg(z),
\end{eqnarray}
with $0<x<\frac{l^2q^\lambda}{q-1}$ and $0\leq\theta\leq2\pi$ for $q>1$.
\end{itemize}
\end{proposition}
{\bf Proof:}
\\
{$\bullet$\it Non orthogonality and normalizability}
\begin{eqnarray}
 {}_{l,\lambda}\langle z_1|z_2\rangle_{l,\lambda}= \frac{\mathcal{N}_{l,\lambda}(\bar z_1z_2)}{\left(\mathcal{N}_{l,\lambda}(|z_1|^2)\mathcal{N}_{l,\lambda}(|z_2|^2)\right)^{1/2}}\neq 0
\end{eqnarray}
imply that the coherent states are not orthogonal.
\\
{$\bullet$\it Normalizability}
\\
From the above relation taking $z_1=z_2=z$ we obtain ${}_{l,\lambda}\langle z|z\rangle_{l,\lambda}=1$. Also,
\begin{eqnarray*}
a|z\rangle_{l,\lambda}&=&\mathcal{N}_{l,\lambda}^{-1/2}(|z|^2)\sum_{n=0}^\infty\frac{q^{n(n+1)/4}z^n}{\sqrt{(l^2q^\lambda)^n[n]_q!}}a|n\rangle\cr
&=&\mathcal{N}_{l,\lambda}^{-1/2}(|z|^2)\sum_{n=1}^\infty\frac{q^{n(n-1)/4}z^n}{\sqrt{(l^2q^\lambda)^{n-1}[n-1]_q!}}|n-1\rangle\cr
&=&z\mathcal{N}_{l,\lambda}^{-1/2}(|z|^2)\sum_{n=0}^\infty\frac{q^{n(n+1)/4}z^{n}}{\sqrt{(l^2q^\lambda)^n[n]_q!}}|n\rangle.
\end{eqnarray*}
\\
{$\bullet$\it Continuity in  the labels $z$}
\begin{eqnarray*}
 |||z_1\rangle_{l,\lambda}-|z_2\rangle_{l,\lambda}||^2= 2\left(1-\mathcal{R}e{}_{l,\lambda}\langle z_1|z_2\rangle_{l,\lambda}\right).
\end{eqnarray*}
So, $|||z_1\rangle_{l,\lambda}-|z_2\rangle_{l,\lambda}||^2\to 0$ as $|z_1-z_2|\to 0$, since ${}_{l,\lambda}\langle z_1|z_2\rangle_{l,\lambda}\to 1$ as $|z_1-z_2|\to 0$.
\\
{$\bullet$\it Resolution of the identity}

The computation of the RHS of (\ref{unitf}) gives
\begin{eqnarray*}
 \int_{\mathbf{D}_{l,\lambda}} d\mu_{l,\lambda}(\bar z,z)|z\rangle_{l,\lambda}{}_{l,\lambda}\langle z|= \sum_{n,m}|n\rangle\langle m|\frac{q^{[n(n+1)+m(m+1)]/4}}{\sqrt{(l^2q^\lambda)^{n+m}[n]_q![m]_q!}}
\int_{\mathbf{D}_{l,\lambda}} \bar z^nz^m\frac{d\mu_{l,\lambda}(\bar z,z)}{\mathcal{N}_{l,\lambda}(|z|^2)}.
\end{eqnarray*}
So, in order to satisfy (\ref{unitf}) it is required 
\begin{eqnarray}
\int_{\mathbf{D}_{l,\lambda}} \bar z^nz^m\frac{d\mu_{l,\lambda}(\bar z,z)}{\mathcal{N}_{l,\lambda}(|z|^2)}=\delta_{m n}(l^2q^\lambda)^n q^{-n(n+1)/2}[n]_q!,\quad n,\;m= 0,\; 1,\; 2,\; ...
\end{eqnarray}
Upon passing to polar coordinates,  $z=\sqrt x\;e^{i\theta}$, $d\mu_{l,\lambda}(\bar z,z)= d\omega_{l,\lambda}(x)d\theta$
 where $0\leq\theta\leq2\pi$, $0<x< R_{l,\lambda}$  and $\omega_{l,\lambda}$ is a positive valued function, this is equivalent to 
the classical Stieltjes power moment problem when $0<q<1$ or the Hausdorff power moment problem when $q>1$ \cite{Akhiezer, Tarmakin}:
\begin{eqnarray}\label{KaMoment}
 \int_0^{R_{l,\lambda}}x^n\;\frac{2\pi\;d\omega_{l,\lambda}(x)}{\mathcal{N}_{l,\lambda}(x)}=(l^2q^\lambda)^n q^{-n(n+1)/2}[n]_q!,\quad n= 0,\; 1,\; 2,\; ...
\end{eqnarray}

If $0<q<1,$ then we have the following Stieltjes power moment problem:
\begin{eqnarray}
 \int_0^{+\infty}x^n\frac{2\pi\;d\omega_{l,\lambda}(x)}{\mathcal{N}_{l,\lambda}(x)}= (l^2q^\lambda)^n q^{-n(n+1)/2}[n]_q!,
\end{eqnarray}
or, equivalently, 
\begin{eqnarray}
 \int_0^{+\infty}y^n\frac{2\pi\;d\omega_{l,\lambda}(l^2q^\lambda y)}{E_q\big((1-q)qy\big)}
=q^{-n(n+1)/2}[n]_q!, 
\end{eqnarray}
where the change of variable $y=\frac{x}{l^2q^\lambda}$ has been made.
Atakishiyev and Atakishiyeva \cite{Atakishiyeva} have proved that
\begin{eqnarray}
 g_q(n)= \int_0^{+\infty}\frac{y^{n-1}dy}{E_q((1-q)y)}= \frac{\ln{q^{-1}}}{1-q}q^{-n(n-1)/2}[n-1]!_q.
\end{eqnarray}
Therefore we deduce
\begin{eqnarray*}
 d\omega_{l,\lambda}(l^2q^\lambda y)=\frac{1}{2\pi}\frac{1-q}{\ln{q^{-1}}}\frac{E_q((1-q)qy)dy}{E_q((1-q)y)}
\end{eqnarray*}
or
\begin{eqnarray}
 d\omega_{l,\lambda}(x)&=&\frac{1}{2\pi}\frac{1-q}{l^2q^\lambda\ln{q^{-1}}}\frac{E_q((1-q)qx/(l^2q^\lambda))dx}{E_q((1-q)x/(l^2q^\lambda ))}
\cr&=&\frac{1}{2\pi}\frac{1-q}{l^2q^\lambda\ln{q^{-1}}}\frac{\mathcal{N}_{l,\lambda}(x)dx}{\mathcal{N}_{l,\lambda}(x/q)}.
\end{eqnarray}
Hence
\begin{eqnarray}
 d\mu_{l,\lambda}(\bar z,z)=\frac{1-q}{l^2q^\lambda\ln{q^{-1}}}\frac{\mathcal{N}_{l,\lambda}(\bar zz)}{\mathcal{N}_{l,\lambda}(\bar zz/q)}\frac{d^2z}{\pi}.
\end{eqnarray}

On the other hand, if $q>1,$ then combining (\ref{KaMoment}), (\ref{Kaprop1}) and (\ref{Kaprop2}) of the Lemma \ref{Buklemma} we get
\begin{eqnarray}
 d\mu(\bar z,z)=\frac{1}{2\pi}\frac{d_q^{l,\lambda}x\;d\theta}{1-(q-1)x/(l^2q^\lambda)}, \quad x=|z|^2,\;\;\theta=\arg(z),
\end{eqnarray}
where $0<x<\frac{l^2q^\lambda}{q-1}$ and $0\leq\theta\leq2\pi$.\hfill$\square$

\subsection{Statistics and geometry of coherent states $|z\rangle_{l,\lambda}$}\label{Sec1.4}

The conventional boson operators $b$ and $b^\dagger$ may be expressed in terms of the deformed operators $a$ and $a^\dagger$ as
\begin{eqnarray}
 b= a\;\sqrt{\frac{N}{\varphi(N)}}\quad\mbox{and}\quad b^\dagger=\sqrt{\frac{N}{\varphi(N)}}\;a^\dagger,\quad {\varphi(N)}\neq {\varphi(0)}
\end{eqnarray}
and their actions  on the states $|n\rangle$ are given by
\begin{eqnarray}
 b|n\rangle=\sqrt{n}|n-1\rangle,\quad \mbox{and} \quad 
b^\dagger|n\rangle=\sqrt{n+1}|n+1\rangle.
\end{eqnarray}
Besides,
\begin{eqnarray}
 b^r|n\rangle=\sqrt{\frac{n!}{(n-r)!}}|n-r\rangle, \qquad 0\leq r\leq n
\end{eqnarray}
and
\begin{eqnarray}
 (b^\dagger)^s|n\rangle=\sqrt{\frac{(n+s)!}{n!}}|n+s\rangle.
\end{eqnarray}
\subsubsection{\it  Quantum statistics of the coherent states $|z\rangle_{l,\lambda}$}
\begin{proposition}
 The expectation value of monomials  of boson creation and annihilation operators $b^\dagger$, $b$ in the coherent states $|z\rangle_{l,\lambda}$ are given by
\begin{eqnarray}
 \langle(b^\dagger)^sb^r\rangle
=\frac{\bar z^sz^r}{\mathcal{N}_{l,\lambda}(|z|^2)}\sum_{n=0}^\infty \sqrt{\frac{q^{[(n+s)(n+s+1)+(n+r)(n+r+1)]/2}(n+r)!(n+s)!}{(l^2q^\lambda)^{(n+s)+(n+r)}[n+s]_q![n+r]_q!}}\frac{|z|^{2n}}{n!},
\end{eqnarray}
where $s=0,\;1,\;2,\cdots$ and $r=0,\;1,\;2,\cdots$.\\
In particular,
\begin{eqnarray}
 \langle(b^\dagger)^rb^r\rangle= \frac{x^{r}}{\mathcal{N}_{l,\lambda}(x)}\left(\frac{d}{dx}\right)^r\mathcal{N}_{l,\lambda}(x),\quad x=|z|^2,\quad r=0,\;1,\;2,\cdots,
\end{eqnarray}
and
\begin{eqnarray}
 \langle N\rangle= x\frac{\mathcal{N}_{l,\lambda}'(x)}{\mathcal{N}_{l,\lambda}(x)}\;,
\end{eqnarray}
where $\mathcal{N}_{l,\lambda}'(x)$ denotes the derivative with respect to $x$.
\end{proposition}
{\bf Proof:} Indeed, for $s=0,\;1,\;2,\cdots$ and $r=0,\;1,\;2,\cdots$, we have 
\begin{eqnarray*}
&& \langle(b^\dagger)^sb^r\rangle:={}_{l,\lambda}\langle z|(b^\dagger)^sb^r|z\rangle_{l,\lambda}
\cr&&\quad=\frac{1}{\mathcal{N}_{l,\lambda}(|z|^2)}\sum_{m=0}^\infty\sum_{n=r}^\infty \sqrt{\frac{q^{[m(m+1)+n(n+1)]/2}n!(n-r+s)!}{(l^2q^\lambda)^{m+n}[m]_q![n]_q!(n-r)!(n-r)!}}\bar z^mz^n\langle m|n+s-r\rangle
\cr&&\quad=\frac{1}{\mathcal{N}_{l,\lambda}(|z|^2)}\sum_{n=r}^\infty \sqrt{\frac{q^{[(n+s-r)(n+s-r+1)+n(n+1)]/2}n!(n-r+s)!}{(l^2q^\lambda)^{n+s-r+n}[n+s-r]_q![n]_q!(n-r)!(n-r)!}}\bar z^{n+s-r}z^n
\cr&&\quad=\frac{\bar z^sz^r}{\mathcal{N}_{l,\lambda}(|z|^2)}\sum_{n=0}^\infty \sqrt{\frac{q^{[(n+s)(n+s+1)+(n+r)(n+r+1)]/2}(n+r)!(n+s)!}{(l^2q^\lambda)^{(n+s)+(n+r)}[n+s]_q![n+r]_q!}}\frac{|z|^{2n}}{n!},
\end{eqnarray*}

In the special case $s=r$, we have 
\begin{eqnarray*}
 \langle(b^\dagger)^rb^r\rangle&=& \frac{x^{r}}{\mathcal{N}_{l,\lambda}(x)}\sum_{n=0}^\infty \frac{q^{(n+r)(n+r+1)/2}(n+r)!}{(l^2q^\lambda)^{(n+r)}[n+r]_q!}\frac{x^{n}}{n!}
\cr&=& \frac{x^{r}}{\mathcal{N}_{l,\lambda}(x)}\sum_{n=r}^\infty \frac{q^{n(n+1)/2}(n)!}{(l^2q^\lambda)^{(n)}[n]_q!}\frac{x^{n-r}}{(n-r)!}
\cr&=& \frac{x^{r}}{\mathcal{N}_{l,\lambda}(x)}\left(\frac{d}{dx}\right)^r\mathcal{N}_{l,\lambda}(x),\quad x=|z|^2.
\end{eqnarray*}
In particular
\begin{eqnarray*}
 \langle N\rangle\equiv\langle b^\dagger b\rangle= x\frac{\mathcal{N}_{l,\lambda}'(x)}{\mathcal{N}_{l,\lambda}(x)}.
\end{eqnarray*}
\hfill$\square$

The probability of finding $n$ quanta in the deformed state $|z\rangle_{l,\lambda}$ is given by
\begin{eqnarray}
 \mathcal{P}_{l,\lambda}(n):=|\langle n|z\rangle_{l,\lambda}|^2= \frac{q^{n(n+1)/2}x^n}{(l^2q^\lambda)^n[n]_q!\mathcal{N}_{l,\lambda}(x)}.
\end{eqnarray}

The Mendel parameter  measuring the deviation from the Poisson statistics is defined by the quantity
\begin{eqnarray}
 Q_{l,\lambda}:=\frac{\langle N^2\rangle-\langle N\rangle^2-\langle N\rangle}{\langle N\rangle}.
\end{eqnarray}
Let us evaluate it explicitly.
From the expectation value of the operator  $N^2=(b^\dagger)^2b^2+ N$ provided by
\begin{eqnarray}
 \langle N^2\rangle%= \langle(b^\dagger)^2b^2+ N\rangle
= x^2\frac{\mathcal{N}_{l,\lambda}''(x)}{\mathcal{N}_{l,\lambda}(x)}+
x\frac{\mathcal{N}_{l,\lambda}'(x)}{\mathcal{N}_{l,\lambda}(x)},
\end{eqnarray} 
we readily deduce
\begin{eqnarray}
 Q_{l,\lambda}=x\left(\frac{\mathcal{N}_{l,\lambda}''(x)}{\mathcal{N}_{l,\lambda}'(x)} -\frac{\mathcal{N}_{l,\lambda}'(x)}{\mathcal{N}_{l,\lambda}(x)}\right).
\end{eqnarray}
It is then worth noticing that for $x<<1$, 
\begin{eqnarray}
 Q_{l,\lambda}=-\frac{q(1-q)}{l^2q^\lambda(1+q)}x+ o(x^2)
\end{eqnarray}
 meaning that the $\mathcal{P}_{l,\lambda}(n)$ is a sub-Poissonian distribution \cite{Klauder&al}.

\subsubsection{\it  Geometry of the states $|z\rangle_{l,\lambda}$}
The geometry of a quantum state space can be described by the corresponding metric tensor. This real and positive definite metric is defined on the underlying manifold that the quantum states form, or belong to, by calculating the distance function (line element) between
two quantum states. So, it is also known as a Fubini-Study metric of the ray space. The knowledge of the quantum metric enables one to calculate quantum mechanical transition probability and uncertainties 

In the case $q<1$, the map from $z$ to $|z\rangle_{l,\lambda}$ defines a map from the space $\mathbb{C}$ of complex numbers onto
a continuous subset of unit vectors in Hilbert space and generates in the latter a two-dimensional surface with the following Fubini-Study metric:
\begin{eqnarray}
 d\sigma^2:= ||d|z\rangle_{l,\lambda}||^2-|_{l,\lambda}\langle z|d|z\rangle_{l,\lambda}|^2
\end{eqnarray}
\begin{proposition}
The above Fubini-Study metric  is reduced to
\begin{eqnarray}
 d\sigma^2= W_{l,\lambda}(x)d\bar z dz,
\end{eqnarray}
where $x=|z|^2$ and 
\begin{eqnarray}
 W_{l,\lambda}(x)=\left(x\frac{\mathcal{N}_{l,\lambda}'(x)}{\mathcal{N}_{l,\lambda}(x)}\right)'= \frac{d}{dx}\langle N\rangle.
\end{eqnarray}
In polar coordinates, $z= re^{i\theta}$,
\begin{eqnarray}
d\sigma^2= W_{l,\lambda}(r^2)(dr^2+r^2d\theta^2).
\end{eqnarray}
\end{proposition}
{\bf Proof:}
Computing $d|z\rangle_{l,\lambda}$ by taking into account the fact that any change of the form
$d|z\rangle_{l,\lambda}=\alpha|z\rangle_{l,\lambda}$, $\alpha\in\mathbb{C}$, has zero distance, we get
\begin{eqnarray*}
 d|z\rangle_{l,\lambda}= \mathcal{N}_{l,\lambda}(|z|^2)^{-1/2}\sum_{n=0}^\infty\frac{q^{n(n+1)/4}nz^{n-1}}{\sqrt{(l^2q^\lambda)^n[n]_q!}}|n\rangle\;dz.
\end{eqnarray*}
Then,
\begin{eqnarray*}
 ||d|z\rangle_{l,\lambda}||^2&=&\mathcal{N}_{l,\lambda}(|z|^2)^{-1}\sum_{n=0}^\infty\frac{q^{n(n+1)/2}n^2|z|^{2(n-1)}}{(l^2q^\lambda)^n[n]_q!}d\bar z dz
 \cr&=&\mathcal{N}_{l,\lambda}(|z|^2)^{-1}\left(\sum_{n=0}^\infty\frac{q^{n(n+1)/2}n|z|^{2(n-1)}}{(l^2q^\lambda)^n[n]_q!}\right.
\cr&&\quad\left.+|z|^2\sum_{n=0}^\infty\frac{q^{n(n+1)/2}n(n-1)|z|^{2(n-2)}}{(l^2q^\lambda)^n[n]_q!}\right)d\bar z dz
\cr&=&\mathcal{N}_{l,\lambda}(x)^{-1}\left(\mathcal{N}_{l,\lambda}'(x)+x\mathcal{N}_{l,\lambda}''(x)\right)d\bar z dz
\cr&=&\mathcal{N}_{l,\lambda}(x)^{-1}\left(x\mathcal{N}_{l,\lambda}'(x)\right)'d\bar z dz
\end{eqnarray*}
and
\begin{eqnarray*}
 |_{l,\lambda}\langle z|d|z\rangle_{l,\lambda}|^2&=&
\left|\mathcal{N}_{l,\lambda}(|z|^2)^{-1}\sum_{n=0}^\infty\frac{q^{n(n+1)/2}n|z|^{2(n-1)}}{(l^2q^\lambda)^n[n]_q!}\bar z dz\right|^2
\cr&=&x\mathcal{N}_{l,\lambda}(x)^{-2}\left(\mathcal{N}_{l,\lambda}'(x)\right)^2d\bar z dz.
\end{eqnarray*}
Therefore,
\begin{eqnarray*}
 d\sigma^2&=&\left(\mathcal{N}_{l,\lambda}(x)^{-1}\left(\mathcal{N}_{l,\lambda}'(x)+x\mathcal{N}_{l,\lambda}''(x)\right)-
x\mathcal{N}_{l,\lambda}(x)^{-2}\left(\mathcal{N}_{l,\lambda}'(x)\right)^2\right)d\bar{z}dz
\cr&=&\left(x\frac{\mathcal{N}_{l,\lambda}'(x)}{\mathcal{N}_{l,\lambda}(x)}\right)'d\bar{z}dz= \left(\frac{d}{dx}\langle N\rangle\right)d\bar{z}dz.
\end{eqnarray*}
\hfill$\square$

For $x<<1$, we have
\begin{eqnarray}
 W_{l,\lambda}(x)=\frac{q}{l^2q^\lambda}\left[1-\frac{2q(1-q)}{l^2q^\lambda(1+q)}x+o(x^2)\right].
\end{eqnarray}

%%%%%%%%%%%%%%%%%%%%%%%%%%%%%%%%%%%%%%%%%%%%%%%%%%%%%%%%%%%%%%%%%%%%%%%%%%%%%%%%%%%%%%%%

\section{ On generalized oscillator algebras and their associated coherent states}\label{chap2}
$\;$

% \begin{abstract}
A unified method of calculating structure functions  from commutation relations of deformed single-mode oscillator algebras 
is  presented. A natural approach to building coherent states associated to  deformed algebras is then deduced \cite{Bukweli&Hounkonnou12b}.
Known deformed algebras are  given as illustration and such mathematical properties  as the continuity in the label, normalizability and resolution of the identity of the corresponding coherent states are discussed.
% \end{abstract}

\subsection{Unified deformed single-mode oscillator algebras}\label{Sec2.2}

\begin{definition}
We  call deformed Heisenberg  algebra,  an associative algebra generated by the set of operators $\{\mathbf{1},\; a,\; a^\dagger,\; N\}$
satisfying the relations
\begin{eqnarray}\label{dal}
 [N,\;a^\dagger]= a^\dagger,\qquad [N,\;a]= -a,\label{uq01}
%&& a^\dagger a = \varphi(N),\qquad aa^\dagger= \varphi(N+\mathbf{1}),\label{uq02}
\end{eqnarray}
such that there exists a non-negative analytic function $f$, called the structure function,  defining the operator products $ a^\dagger a$ and $aa^\dagger$ in the following way:
\begin{eqnarray}
 a^\dagger a := f(N),\qquad aa^\dagger:= f(N+\mathbf{1}),\label{uq02}
\end{eqnarray}
where $N$ is a self-adjoint operator, $a$ and its Hermitian conjugate $a^\dagger$ denote the deformed annihilation and creation operators, respectively.
\end{definition}

Afore-mentioned deformed Heisenberg algebras have a common property characterized by the existence of
a self-adjoint  number operator $N$, a lowering operator $a$  and its formal adjoint, called raising operator, $a^\dagger$ and  differ by the expression of the  structure function $f$.

The associated Fock  space $\mathcal{F}$ is now spanned by the orthonormalized eigenstates of the number operator $N$ given by:
\begin{eqnarray}
 |n\rangle= \frac{1}{\sqrt{f(n)!}}(a^\dagger)^n|0\rangle,\quad n\in\mathbb{N}\cup\{0\},
\end{eqnarray}
where
\begin{eqnarray}
f(n)!=f(n)f(n-1)...f(1)\qquad \mbox{ with}\quad f(0)=0.
\end{eqnarray}
Moreover,
\begin{eqnarray}
 a|n\rangle= \sqrt{f(n)}|n-1\rangle,\; a^\dagger|n\rangle=\sqrt{f(n+1)}|n+1\rangle.
\end{eqnarray}

We emphasize that the structure function $f$ is a key unifying methods of coherent state construction corresponding to deformed algebras. To this end let us first recall the definition of the canonical coherent states.

\begin{definition}\label{CSdef1}
The canonical coherent states (CS) are normalized states $|z\rangle\in\mathcal{H}$ satisfying one of the  following three equivalent conditions \cite{Glauber1,Glauber2,Klauder63a,Klauder63b,Schrodinger,Sudarshan}:
\begin{itemize}
\item[(i)] they saturate the Heisenberg inequality:
\begin{eqnarray}
 (\Delta \hat{q})(\Delta \hat{p})= \frac{\hbar}{2},
\end{eqnarray}
where { $(\Delta A)^2:=\langle z|A^2-\langle A\rangle^2|z\rangle$}  with { $\langle A\rangle:= \langle z|A|z\rangle$};
 \item [(ii)] they are eigenvectors of the annihilation operator, with eigenvalue $z\in \mathbb{C}$:
\begin{eqnarray}
 b|z\rangle= z|z\rangle;
\end{eqnarray}
\item[(iii)] they are obtained from the ground state $|0\rangle$ of the harmonic oscillator by a unitary action of the Weyl-Heisenberg group:
\begin{eqnarray}\label{bch2}
 |z\rangle=e^{zb^\dagger-\bar z b}|0\rangle.
\end{eqnarray}

\end{itemize}
\end{definition}
From  (\ref{bch2}) and using the famous  Baker-Campbell-Hausdorff formula
\begin{eqnarray}
e^{A+B}=e^{-{{1}\over{2}}[A,B]}e^Ae^B
\end{eqnarray}
whenever $[A,[A,B]]=[B,[A, B]]= 0$, one obtains
\begin{equation}
  |z\rangle=e^{-|z|^2/2}\sum_{n=0}^\infty\frac{z^n}{\sqrt{n!}}|n\rangle= e^{-|z|^2/2}e^{zb^\dagger}|0\rangle,\; z\in\mathbb{C}\label{bch1}.
\end{equation}
The important feature of these coherent states resides in the partition (resolution) of identity:
\begin{eqnarray}
\int_\mathbb{C}{{[d^2z]}\over\pi} |z\rangle\langle z|= \mathbf{1},
\end{eqnarray}
where we have put $[d^2z]= d(Rez)d(Imz)$ for simplicity.
\\
Suppose that $f(0)=0$, $f(n)>0$ for all $n\in\mathbb{N}$ and denote $\mathbf{D}_f=\left\{z\in\mathbb{C}:|z|^2<R_f\right\}$, where $R_f$ is the radius of convergence of the series ( called {\it deformed exponential function}):
\begin{eqnarray}\label{ser1}
\mathcal{N}_f(x):=\sum_{n=0}^\infty \frac{x^n}{[f(n)]!}.
\end{eqnarray}
Then, the following holds:
\begin{proposition}
 The states
\begin{eqnarray}
 |z,f\rangle&:=& (\mathcal{N}_f(|z|^2))^{-1/2}\sum_{n=0}^\infty \frac{z^n}{[f(n)]!}(a^\dagger)^n|0\rangle
 \cr&=&(\mathcal{N}_f(|z|^2))^{-1/2}\sum_{n=0}^\infty \frac{z^n}{\sqrt{[f(n)]!}}|n\rangle,\;\;z\in\mathbf{D}_f,
\end{eqnarray}
are normalized eigenstates of  the raising operator $ a$ with eigenvalue $z$. They are not orthogonal to each other.
Moreover, the map $z\mapsto|z,f\rangle$ from $\mathbf{D}_f\subset\mathbb{C}$ to the Fock space $\mathcal{F}$ is continuous.
\end{proposition}
{\bf Proof:}
The first assertion is true by definition of states $|z,f\rangle$ and the action of the raising operator $a$.
To prove the non orthogonality, let $z_1,z_2\in\mathbf{D}_f$. Then,
\begin{eqnarray}
 \langle z_1,f|z_2,f\rangle=\frac{\mathcal{N}_f(\bar z_1z_2)}{\left(\mathcal{N}_f(|z_1|^2)\mathcal{N}_f(|z_2|^2)\right)^{1/2}}\neq 0,\quad\mbox{when}\quad z_1\neq z_2.
\end{eqnarray}
Furthermore, 
\begin{eqnarray}
 |||z_1,f\rangle-|z_2,f\rangle||^2= 2\left(1-\mathcal{R}e\langle z_1,f|z_2,f\rangle\right)\to 0\;
\mbox{ as }\; |z_1-z_2|\to 0
\end{eqnarray}
that means the map  $\mathbf{D}_f\ni z\mapsto|z,f\rangle\in\mathcal{F}$ is continuous. \hfill$\square$

The family $\left\{|z,f\rangle: z\in\mathbf{D}_f\right\}$ will be called {\it coherent states}  whether there exists a positive measure $\mu_f$ such that \cite{Klauder&Skagerstam}: 
\begin{eqnarray}\label{unitf2}
 \int_{\mathbf{D}_f} d\mu_f(\bar z,z)|z,f\rangle\langle z,f|= \sum_{n=0}^\infty|n\rangle\langle n|=\mathbf{1},
\end{eqnarray}
thus forming an overcomplete set of states,
or equivalently
\begin{eqnarray}
 \int_{\mathbf{D}_f} \bar z^nz^m\frac{d\mu_f(\bar z,z)}{\mathcal{N}_f(|z|^2)}= \delta_{n m}[f(n)]!,\quad n,\;m= 0,\; 1,\; 2,\; ...
\end{eqnarray}
Passing to polar coordinates,  $z=\sqrt{x} e^{i\theta}$, 
where $0\leq\theta\leq 2\pi$, $ 0<x< R_f$, and $d\mu(\bar z,z)= d\omega_f(x)d\theta$, the latter equation  leads to the following classical Stieltjes (for $R_f=\infty$)  or  Hausdorff ($R_f<\infty$)  power-moment problem \cite{Akhiezer,Tarmakin}:
\begin{eqnarray}
\label{PMoment}
 \int_0^{R_f}x^n\;\frac{2\pi\;d\omega_f(x)}{\mathcal{N}_f(x)}=[f(n)]!,\quad n= 0,\; 1,\; 2,\; ...
\end{eqnarray}
Note immediately that not  all deformed algebras lead to coherent states because the moment problem (\ref{PMoment}) does  not always have solution \cite{Akhiezer,Tarmakin}.
Nevertheless, it is remarkable that the structure function $f$ plays an important role in the construction of coherent states associated to an algebra. So, the question arises  is then how to determine the structure function corresponding to a given algebra.

Many techniques have been proposed in literature \cite{Baloitcha,Borzov,Burban1,Kosinski,Meljanac}. A more general answer to this question can be given starting from the Meljanac {\it et al} \cite{Meljanac} point of view. Indeed, these authors introduced the generalized $q$-deformed single-mode oscillator algebra through the identity operator $\mathbf{1}$,
a self-adjoint number operator $N$,  a lowering operator $a$ and  an  operator $\bar{a}$ which is not necessarily conjugate to $a$
 satisfying
\begin{eqnarray}
 && [N,\;a]= -a, \quad %\label{uq1}\\&& 
[N,\;\bar{a}]= \bar{a}, \label{uq2}\\
&& a\bar{a}-F(N)\bar{a}a=G(N)  \label{uq3}
\end{eqnarray}
where $F$ and $G$ are arbitrary complex analytic functions.

Such an algebra furnishes an appropriate approach for the  unification of  classes of  deformed algebras known in the literature.

For the purpose, let us  start from the relations (\ref{uq2}) to get
\begin{eqnarray}
[N,\;a\bar{a}]= 0=[N,\;\bar{a}a]
\end{eqnarray}
implying the existence of a complex analytic function $\varphi$ such that
\begin{eqnarray}
 \bar{a}a=\varphi(N)\quad\mbox{and}\qquad a\bar{a}=\varphi(N+1).\label{uq4}
\end{eqnarray}
Therefore, Eq.(\ref{uq3}) can be rewritten as follows 
\begin{eqnarray}\label{uq5}
\varphi(N+1)-F(N)\varphi(N)= G(N). 
\end{eqnarray}

Denote now $a^\dagger$ the Hermitian conjugate of the operator $a$. Then,
\begin{eqnarray}\label{uq7}
[N,\;a^\dagger]= a^\dagger,\qquad\mbox{and}\qquad \bar{a}=c(N)a^\dagger,
\end{eqnarray}
where $c(N)$ is a complex function. For convenience take $c(N)= e^{i\arg{\varphi(N)}}$. Therefore, from (\ref{uq4}) and the fact that $a^\dagger a$ and $aa^\dagger$ are Hermitian operators we necessarily have
\begin{eqnarray}\label{uq8}
 a^\dagger a= |\varphi(N)|\qquad\mbox{and}\qquad  aa^\dagger= |\varphi(N+1)|.
\end{eqnarray}
We now assume the existence of a "vacuum state" $|0\rangle$ such that 
\begin{eqnarray}\label{uq9}
 N|0\rangle=0,\quad a|0\rangle=0 \quad \mbox{and}\quad
\langle 0|0\rangle= 1,
\end{eqnarray}
and construct the non normalized eigenvectors $(a^\dagger)^n|0\rangle$ of the operator $N$. It follows that 
\begin{eqnarray}
 \langle0|a^m(a^\dagger)^n|0\rangle=\delta_{m n}\prod_{k=1}^n|\varphi(k)|=:(|\varphi(n)|!)\delta_{m n},\; m,\;n= 0,\;1,\;2,\;\cdots
\end{eqnarray}
and we have the following proposition

\begin{proposition}
Suppose that the initial condition $\varphi(0)=0$ is satisfied.  Then

\begin{eqnarray} \label{stctr}
\varphi(n)= [F(n-1)]!\sum_{k=0}^{n-1}\frac{G(k)}{[F(k)]!},\quad n\geq 1,
\end{eqnarray}
where
\begin{eqnarray}
 [F(k)]!=\left\{\begin{array}{lcr}F(k)F(k-1)\cdots F(1)&\mbox{ if }& k\geq 1\\1&\mbox{ if }& k=0 \end{array}\right..
\end{eqnarray}
\end{proposition}
{\bf Proof:}
Applying (\ref{uq5}) to the vector $(a^\dagger)^n|0\rangle$, we obtain
\begin{eqnarray}
\varphi(n+1)-F(n)\varphi(n)= G(n), \quad n=0,\;1,\;2\;\cdots
\end{eqnarray}
Then the result follows.\hfill$\square$

Notice that for all $n= 1,\; 2,\; ...$ each ''excited'' state $(a^\dagger)^n|0\rangle$ 
is a eigenstate of the operator $N$ corresponding to the eigenvalue $n$ with norm
\begin{eqnarray}
 ||(a^\dagger)^n|0\rangle||= \sqrt{|\varphi(n)|!}\;.
\end{eqnarray}
Of course, these states are orthogonal, i.e.
\begin{eqnarray}
 \langle0|a^m(a^\dagger)^n|0\rangle=0\quad\mbox{for}\quad m\neq n.
\end{eqnarray}

Now, if  $\varphi(n)\neq0$ for all $n\geq1$, one normalizes the eigenstates $(a^\dagger)^n|0\rangle$ of $N$ and gets the vectors $|n\rangle\in \mathcal{F}$ as 
\begin{eqnarray}
 |n\rangle=\frac{(a^\dagger)^n}{\sqrt{[|\varphi(n)|]!}}|0\rangle.
\end{eqnarray}
In the opposite, if $\varphi(n_0)=0$ for some $n_0$ , then the state $(a^\dagger)^{n_0}|0\rangle$  has zero norm and, consistently,
we can put $|n_0,\varphi\rangle\equiv0$. The corresponding Hilbert space is the finite-dimensional space $\mathbb{C}^{n_0}$.

Besides,  the following relations hold:
\begin{eqnarray}
&&a|n,\varphi\rangle=\sqrt{|\varphi(n)|}|n-1,\varphi\rangle,\quad%\\&& 
a^\dagger|n,\varphi\rangle=\sqrt{|\varphi(n+1)|}|n+1,\varphi\rangle,\\&& 
%N|n,\varphi\rangle=n|n,\varphi\rangle, \\&&  
a^\dagger a|n,\varphi\rangle = |\varphi(n)||n,\varphi\rangle\quad\mbox{and}\quad%\\&& 
aa^\dagger|n,\varphi\rangle= |\varphi(n+1)||n,\varphi\rangle
\end{eqnarray}
showing that the structure function characterizing a given deformation is defined as follows:
$f(n)= \varphi(n)$ if $\varphi(n)\geq0$, and $f(n)= |\varphi(n)|$, in general.

Moreover, the conventional boson operators $b$ and $b^\dagger$ may be expressed in terms of the deformed operators $a$ and $a^\dagger$ as
\begin{eqnarray}
 b= a\;\sqrt{\frac{N}{f(N)}}\quad\mbox{and}\quad b^\dagger=\sqrt{\frac{N}{f(N)}}\;a^\dagger.
\end{eqnarray}
Thus the actions of $b$ and $b^\dagger$ on the states are as usual
\begin{eqnarray}
 b|n,\varphi\rangle=\sqrt{n}|n-1,\varphi\rangle,\quad \mbox{and} \quad 
b^\dagger|n,\varphi\rangle=\sqrt{n+1}|n+1,\varphi\rangle.
\end{eqnarray}
So, for simplicity of notations we set $|n,\varphi\rangle=|n\rangle$. Note also that the Weyl-Heisenberg algebra oscillator corresponds to $F(N)=G(N)=1$.

In the next section, we analyse known algebras in the light of the above developed formalism.

\subsection{Application to known deformed  algebras}\label{Sec2.3}

\subsection{The Tamm-Dancoff deformed  algebra}
This algebra appeared in the frame of Tamm-Dancoff method \cite{Dancoff,Tamm}, in quantum field theory, and was defined by the commutation relations
\begin{eqnarray}
&& [N,\;a^\dagger]= a^\dagger,\qquad [N,\;a]= -a,\\
&& aa^\dagger - qa^\dagger a= q^N,\label{Tammalg}
\end{eqnarray}
where $q$ is an arbitrary complex non-zero number. This corresponds to the case $F(N)= q$ and $G(N)= q^N$, and yields
\begin{eqnarray}
 \varphi(n)= nq^{n-1}\qquad \mbox{and}\qquad f(n)= n|q|^{n-1}.
\end{eqnarray}
In this case, the exponential function  (\ref{ser1}) written as
\begin{eqnarray}
 \mathcal{N}_q(x)= \sum_{n=0}^\infty\frac{x^n}{n!|q|^{n(n-1)/2}}\label{TammExp}
\end{eqnarray}
converges on the whole complex plane $\mathbb{C}$ for $|q|\geq 1$.
\begin{proposition}
 The moment problem (\ref{PMoment}) with $R_f=+\infty$ has the following solution
\begin{eqnarray}
 d\omega_q(x)=\frac{\mathcal{N}_q(x)\; dx}{2\pi}\int_0^\infty \frac{\sqrt{|q|}\;\exp\left\{-\left(\sqrt{|q|}t+\frac{\ln^2(x/t)}{2\ln{|q|}}\right)\right\}}{x\sqrt{2\pi\ln{|q|}}}dt.
\end{eqnarray}
\end{proposition}
{\bf Proof:}
Setting  $\displaystyle \tilde W_q(x)dx=2\pi\frac{d\omega_q(x)}{\mathcal{N}_q(x)}$, the moment problem (\ref{PMoment}) is then written as follows:
\begin{eqnarray}
 \int_0^\infty x^n\tilde{W}_q(x)dx=n!|q|^{n(n-1)/2},\quad n=0,\;1,\;2,\;...
\end{eqnarray}
The inverse Mellin transforms of $\tilde f_1(s)= |q|^{s^2/2}$ and $\tilde f_2(s)=|q|^{-s/2}\Gamma(s)$ gives \cite{Prudnikov,Sneddon}
\begin{eqnarray}
 \mathcal{M}^{-1}\{\tilde f_1(s)\}= \frac{e^{-\frac{\ln^2{x}}{2\ln{|q|}}}}{\sqrt{2\pi\ln{|q|}}}=:f_1(x)
\quad\mbox{and}\quad\mathcal{M}^{-1}\{\tilde f_2(s)\}=e^{-\sqrt{|q|}x}=:f_2(x)
\end{eqnarray}
respectively. Thus, the solution of the integral equation
\begin{eqnarray}
 \int_0^\infty x^s\;\tilde{W}_q(x)dx= |q|^{s(s-1)/2}\;\Gamma(s+1),
\end{eqnarray}
can be obtained using the Mellin integral formula \cite{Polyanin}
\begin{eqnarray}
 \mathcal{M}\left\{\int_0^\infty f_1(x/t)f_2(t)dt \right\}=\tilde f_1(s)\tilde f_2(s+1)
\end{eqnarray}
which, in particular, gives 
\begin{eqnarray*}
 \int_0^\infty x^{s-1}\left(\int_0^\infty \frac{e^{-\frac{\ln^2(x/t)}{2\ln{|q|}}}}{\sqrt{2\pi\ln{|q|}}}e^{-\sqrt{|q|}t}dt\right)dx
=|q|^{-1/2}|q|^{s(s-1)/2}\Gamma(s+1)
\end{eqnarray*}
or equivalently
\begin{eqnarray}
 \int_0^\infty x^{s}\left(\int_0^\infty \frac{\sqrt{|q|}\;\exp\left\{-\left(\sqrt{|q|}t+\frac{\ln^2(x/t)}{2\ln{|q|}}\right)\right\}}
{x\sqrt{2\pi\ln{|q|}}}dt\right)dx=|q|^{s(s-1)/2}\Gamma(s+1).
\end{eqnarray}
Therefore,
\begin{eqnarray}
 \tilde W_q(x)=\int_0^\infty \frac{\sqrt{|q|}\;\exp\left\{-\left(\sqrt{|q|}t+\frac{\ln^2(x/t)}{2\ln{|q|}}\right)\right\}}{x\sqrt{2\pi\ln{|q|}}}dt.
\end{eqnarray}
Then follows the result.\hfill$\square$

Hence, the states defined in $\mathbb{C}$ by
\begin{eqnarray}
 |z,q\rangle&=&(\mathcal{N}_q(|z|^2)^{-1/2}\sum_{n=0}^\infty\frac{z^n}{n!|q|^{n(n-1)/2}}(a^\dagger)^n|0\rangle
\cr&=&(\mathcal{N}_q(|z|^2)^{-1/2}\sum_{n=0}^\infty\frac{z^n}{\sqrt{n!|q|^{n(n-1)/2}}}|n\rangle
\end{eqnarray}
constitute a family of coherent states. 

\subsubsection{The Arick-Coon-Kuryskin deformed algebra  (1976)}
Arick and Coon first introduced this algebra \cite{Arik&Coon} whose generators satisfy the following relations:
\begin{eqnarray}
&& [N,\;a^\dagger]= a^\dagger,\qquad [N,\;a]= -a,\\
&& aa^\dagger - qa^\dagger a= 1,\label{Arikalg}.
\end{eqnarray}
The same algebra was examined  independently by other authors like Kuryshkin\cite{Kuryshkin80}, Jannussis \cite{Jannussis}, etc.,
and has gained popularity because of its connection to the developed mathematical $q$-analysis theory.

One can check that $F(N)= q$ and $G(N)=1$ leading to
\begin{eqnarray}
 \varphi(n)= \frac{q^n-1}{q-1}=:[n]_q\; \mbox{and } f(n)=\varphi(n)\;\mbox{for } q\in[-1,\;1[\cup]1,\;\infty).
\end{eqnarray}
wihch is one of the forms of the so called $q$-numbers. This result is also obtained replacing (\ref{Arikalg}) by
\begin{eqnarray}
 [a,\;a^\dagger]= q^N.
\end{eqnarray}
Then follows the series
\begin{eqnarray}
 \mathcal{N}_q(x)=\sum_{n=0}^\infty\frac{x^n}{[n]_q!}.
\end{eqnarray}

Arick and Coon have considered the case $0 < q < 1$ for which $\mathcal{N}_q(x)= e_q((1-q)x)$ where $e_q(x)$ is one of the famous Jackson $q$-exponential  functions,  which converges for $|x|<1$ \cite{Gasper}. In this case the radius of convergence of$\mathcal{N}_q(x)$ is $R_q=\frac{1}{1-q}$.
Note also that,
\begin{eqnarray}\label{qderiva1}
 \partial_q\mathcal{N}_q(x)= \mathcal{N}_q(x),
\end{eqnarray}
where the $q$-derivative $\partial_q$ is defined as follows
\begin{eqnarray}\label{qderiva2}
 \partial_qf(x)=\frac{f(x)-f(qx)}{(1-q)x}.
\end{eqnarray}
The following statements hold:
\begin{lemma}
\begin{eqnarray}
&&\frac{\mathcal{N}_q(x)}{\mathcal{N}_q(qx)}= \frac{1}{1-(1-q)x},\qquad%\\ \label{rem1}&&
\mathcal{N}_q(x)= \frac{1}{((1-q)x;q)_\infty}, \label{rem2}\\
&&e_q(x)= \sum_{k=0}^{\infty}\frac{x^n}{(q;q)_n}=\mathcal{N}_q(x/(1-q))= \frac{1}{(x;q)_\infty},\label{rem3}
\end{eqnarray}
where $\displaystyle(a;q)_n=\prod_{k=0}^{n-1}(1-aq^k)$. Moreover,
\begin{eqnarray}
\int_0^{(1-q)^{-1}}x^n\left(\mathcal{N}_q(qx)\right)^{-1}d_qx= [n]_q!\quad n= 0,\; 1,\; 2,\; ...\label{rem4},
\end{eqnarray}
where 
\begin{equation}
 \int_0^af(x)d_qx=a(1-q)\sum_{k=0}^\infty f(aq^n)q^n
\end{equation}
defines the Jackson integral \cite{Gasper} of a function $f$.
\end{lemma}
{\bf Proof:} Equations (\ref{qderiva1}) and (\ref{qderiva2}) lead to (\ref{rem2})-(\ref{rem3}).
Using the Jackson integral we obtain
\begin{eqnarray*}
&& \int_0^{(1-q)^{-1}}x^n\left(\mathcal{N}_q(qx)\right)^{-1}d_qx
=\sum_{l=0}^\infty\frac{q^{(n+1)l}}{(1-q)^n}\left(\mathcal{N}_q(q^{l+1}/(1-q))\right)^{-1}
\cr&&\qquad\qquad=\sum_{l=0}^\infty\frac{q^{(n+1)l}}{(1-q)^n} (q^{l+1};q)_\infty
= \frac{(q;q)_\infty}{(1-q)^n}\sum_{l=0}^\infty\frac{q^{(n+1)l}}{(q;q)_l}
\cr&&\qquad\qquad= \frac{(q;q)_\infty}{(1-q)^n}\frac{1}{(q^{n+1};q)_\infty}=\frac{(q;q)_n}{(1-q)^n}.
\end{eqnarray*}
Then (\ref{rem4}).\hfill$\Box$
\begin{proposition}
 The solution of the moment problem (\ref{PMoment}) with the $R_q=\frac{1}{1-q}$ is given by
\begin{equation}\label{PMA}
d\omega_q(x)=\frac{1}{2\pi}\frac{d_qx}{1-(1-q)x}\mbox{  ,} \quad x=|z|^2
\end{equation}
\end{proposition}
{\bf Proof}: In this case the moment problem (\ref{PMoment}) becomes
\begin{eqnarray}\label{rem0}
 \int_0^{\frac{1}{1-q}}x^n\;\frac{2\pi\;d\omega_q(x)}{\mathcal{N}_q(x)}=[n]_q!,\quad n= 0,\; 1,\; 2,\; ...
\end{eqnarray}
The comparison of moment problem (\ref{rem0}) and (\ref{rem4}) with the  use of the first equality of previous relations (\ref{rem2}) leads to (\ref{PMA}).\hfill$\square$

Thus, the states 
\begin{eqnarray}
 |z\rangle_q&=& (\mathcal{N}_q(|z|^2))^{-1/2}\sum_{n=0}^\infty\frac{z^n}{[n]_q!}(a^\dagger)^n|0\rangle
\cr&=& (\mathcal{N}_q(|z|^2))^{-1/2}\sum_{n=0}^\infty\frac{z^n}{\sqrt{[n]_q!}}|n\rangle
\end{eqnarray}
define a family of coherent states on $\mathbf{D}_q=\left\{z\in\mathbb{C}:|z|<(1-q)^{-1/2}\right\}$.

\subsubsection{The Feinsilver deformed algebra (1987) }
The generators of the  algebra by Feinsilver algebra \cite{Feinsilver1,Feinsilver2} verify
\begin{eqnarray}
 [a,\; a^\dagger] = q^{-2N},\quad [N,\; a] = -a, \quad [N,\; a^\dagger ] = a^\dagger,
\end{eqnarray}
where $q$ is non-zero real number. It follows that $F(N)=1$ and $G(N)= q^{-2N}$ implying
\begin{eqnarray}
 f(n)=\varphi(n)=\frac{1-q^{-2n}}{1-q^{-2}}=[n]_{q^{-2}}.
\end{eqnarray}
The change of parameters $\tilde{q}=q^{-2}$, with $q>1$,  leads to the previous  Arick-Coon-Kuryshkin algebra.% This may be used when .

\subsubsection{The Biedenharn-Macfarlane oscillator algebra (1989)}

The generators of the deformed algebra introduced  by Biedenharn \cite{Biedenharn} and independently by Macfarlane \cite{Macfarlane}, in the context of oscillator realization of the quantum algebra $su_q(2)$, satisfy 
\begin{eqnarray}
&& [N,\;a^\dagger]= a^\dagger,\qquad [N,\;a]= -a,\\
&& aa^\dagger - qa^\dagger a= q^{-N} \;\mbox{ or }\; aa^\dagger - q^{-1}a^\dagger a= q^{N}\label{BM},\quad q^2\neq 1.
\end{eqnarray}
Here $F(N)=q$ and $G(N)=q^{-N}$. Therefore,
\begin{eqnarray}
 \varphi(n)= q^{n-1}\sum_{j=0}^{n-1}\frac{q^{-j}}{q^j}= \frac{q^n-q^{-n}}{q-q^{-1}},\; q^2\neq1
\end{eqnarray}
 Then, the structure function  is given by
\begin{eqnarray}
 f(n)=\frac{q^n-q^{-n}}{q-q^{-1}}=:[n]_q^{^{B}},\qquad q\in\mathbb{R}_+^*\setminus\{1\}.
\end{eqnarray}
Hence, considering the series
\begin{eqnarray}
 \mathcal{N}_{{}_B}(x)=\sum_{n=0}^\infty\frac{x^n}{[n]^{{}^B}_q!},
\end{eqnarray}
one notices that  its radius of convergence $R_{{}_B}=+\infty$.

Using the deformed derivative and integration defined by
\begin{eqnarray}
 \partial_q^Bf(x):=\frac{f(qx)-f(q^{-1}x)}{(q-q^{-1})}
\end{eqnarray}
and
\begin{eqnarray}
 \int_{x_i}^{x^j}f(x)d_q^Bx:= (q-q^{-1})x\sum_{l=i}^{f}q^{2l}f(q^{2l}x),\;\; x_i= xq^i,\; x_f=xq^f
\end{eqnarray}
respectively,
You-quan and Zheng-mao \cite{Li&Sheng} showed that
\begin{eqnarray}
 \int_{0}^{\infty}x^n\mathcal{N}_{{}_B}(-x)d_q^Bx= [n]_q^B!,\quad n=0,\;1,\;2,\;\cdots
\end{eqnarray}
Therefore, the power-moment problem (\ref{PMoment}) has a solution given by
\begin{eqnarray}
d\omega_{{}_B}(x)=\frac{1}{2\pi}\mathcal{N}_{{}_B}(x)\mathcal{N}_{{}_B}(-x)d_q^Bx
\end{eqnarray}
and the states
\begin{eqnarray}
 |z\rangle_{{}_B}= (\mathcal{N}_{{}_B}(|z|^2))^{-1/2}\sum_{n=0}^\infty\frac{z^n}{\sqrt{[n]^{{}^B}_q!}}|n\rangle
\end{eqnarray}
define a family of coherent states with $z\in\mathbb{C}$.

\begin{remark} \end{remark}
\begin{itemize}
\item El Baz and Hassouni \cite{ElBaz} demonstrated, for $|q|= 1$ and   using the Fourrier transforms, that
the power-moment problem (\ref{PMoment}) has the solution
\begin{eqnarray}
d\omega_{{}_B}(x)= \frac{1}{2\pi}\tilde\mathcal{N}_{{}_B}(x) \tilde W_B(x)dx
\end{eqnarray}
where $\tilde W_{{}_B}(x)$ is the Fourier transform of the series
\begin{eqnarray}
\overline{W}_{{}_B}(y)=\sum_{n=0}^\infty \frac{[n]_q^B!}{n!}(iy)^n
\end{eqnarray}
i.e.
\begin{eqnarray}
\tilde W_{{}_B}(x)= \frac{1}{2\pi}\int_{-\infty}^\infty e^{-ixy}\overline{W}_{{}_B}(y)dy.
\end{eqnarray}
Notice that in this case the function $\mathcal{N}_{{}_B}(x)$ is replaced by
$\displaystyle\tilde\mathcal{N}_{{}_B}(x)=\sum_{n=0}^\infty\frac{x^n}{|[n]^{{}^B}_q!|}$.
\item Yan \cite{Yan} has later examined this algebra with  relations
\begin{eqnarray}
a^\dagger a = [N],\quad aa^\dagger = [N + \mathbf{1}],\quad [a,\; a^\dagger ] = [N + \mathbf{1}] - [N].
\end{eqnarray}
\end{itemize}

\subsubsection{The Calogero-Vasiliev oscillator algebra (1991)}
In 1991 Vasiliev \cite{Vasiliev} introduced a deformed algebra whose generators satisfy
\begin{eqnarray}
 &&[N,\;a]=-a,\qquad [N,\;a^\dagger]= a^\dagger,\quad [N,\;K] = 0,\\
&& aK = -Ka,\quad a^\dagger K = -K a^\dagger,\quad K^2= \mathbf{1},\\
&& [a,\;a^\dagger]= \mathbf{1}+ \nu K,\label{calogero}
\end{eqnarray}
where $\nu$ is a real such that $\nu>-\frac{1}{2}$ and $K= (-)^N$ is the Klein operator interpreted as the generator of the symmetric group $S_2$.
From (\ref{calogero}), we have $F(N)=\mathbf{1}$ and $G(N)= \mathbf{1}+\nu(-)^N$. Therefore
\begin{eqnarray}
 \varphi(2n)= 2n\quad\mbox{and}\quad \varphi(2n+1)=2(n+\nu)+1, \; n=0,\;1,\;2,\; ... 
\end{eqnarray}
% Thus, the Fock boson space is generated by the orthonormalized states
% \begin{eqnarray}
%  |2n\rangle_\nu&=& \frac{(a^\dagger)^{2n}}{\sqrt{(2n)!}}|0\rangle,\cr
% |2n+1\rangle_\nu&=&\frac{(a^\dagger)^{2n+1}}{\sqrt{(2(n+\nu)+1)!}}|0\rangle.
% \end{eqnarray}
The exponential function  (\ref{ser1}) written as
\begin{eqnarray}
 \mathcal{N}_\nu(x)=\sum_{n=0}^\infty\frac{x^n}{\varphi(n)!}
\end{eqnarray}
converges everywhere $x$. However, the corresponding moment problem (\ref{PMoment})
\begin{eqnarray}
\int_0^\infty x^{n}\;\frac{2\pi\;d\omega_\nu(x)}{\mathcal{N}_\nu(x)}=\varphi(n)!,\quad n=0,\;1,\;2,\;\cdots
\end{eqnarray}
remains to solve.

\subsubsection{The $(p,q)$-Chakrabarti-Jagannathan oscillator algebra (1991)}
The two-parameter quantum algebra was first introduced by Chakrabarti and Jagannathan \cite{Chakrabarti&Jagan} in order to generalize or/and unify the Arick-Coon-Kurskin oscillator algebra ($p=1$) and Biedenharn-Macfarlane oscillator algebra ($p=q$). The generators satisfy 
\begin{eqnarray}
 &&[N,\;a]=-a,\qquad [N,\;a^\dagger]= a^\dagger,\\
&& aa^\dagger-q a^\dagger a= p^{-N},\;\mbox{ or }\; aa^\dagger-p^{-1} a^\dagger a= q^{N},\label{Jagannalg}
\end{eqnarray}
where $p,q\in\mathbb{R}^*_+$. 
%This algebra generalizes the Arick-Coon-Kurskin oscillator algebra ($p=1$) and Biedenharn-Macfarlane oscillator algebra ($p=q$).

From the first relation of (\ref{Jagannalg}) we deduce $F(N)=q$ and $G(N)= p^{-N}$. Therefore,
\begin{eqnarray}
 \varphi(n)=q^{n-1}\sum_{j=0}^{n-1}\frac{p^{-j}}{q^j}=q^{n-1}\frac{1-((pq)^{-1})^n}{1-(pq)^{-1}}
= \frac{p^{-n}-q^n}{p^{-1}-q}=:[n]_{p^{-1},q}.
\end{eqnarray}
Notice that the second relation of (\ref{Jagannalg}) gives the same result. Hence, it suffices  to consider only one of the two relations (\ref{Jagannalg}).

The Fock space of the Bose oscillator $\mathcal{F}_{CJ}$ associated to this deformation is generated by the orthonormalized states
\begin{eqnarray}
 |n\rangle= \frac{(a^\dagger)^n}{\sqrt{[n]_{p^{-1},q}!}}|0\rangle,\qquad n=0,\;1,\;2,\; ...
\end{eqnarray}
where
\begin{eqnarray}
[n]_{p^{-1},q}!= \left\{\begin{array}{lccl}\;1 & & \mbox{if} & n=0\\
\;[n]_{p^{-1},q}[n-1]_{p^{-1},q}...[1]_{p^{-1},q}& &\mbox{if}& n\geq1\end{array}.\right.
\end{eqnarray}

\subsubsection{The Kalnins-Mukherjee-Miller oscillator algebra (1993)}
This $q$-oscillator algebra generated by four elements $H$, $E_+$, $E_-$ and $\mathcal{E}$ satisfying
\begin{eqnarray}
&&[H,\;E_+]= E_+\qquad [H,\;E_-]=-E_-\\
&&[E_+,\;E_-]=-q^{-H}\mathcal{E}\qquad [\mathcal{E},\;E_{\pm}]=0=[\mathcal{E},\;H],
\end{eqnarray}
where $0<q<1$, was introduced by Kalnins {\it et  al} \cite{Kalnins}.

The elements $\mathcal{C}= qq^{-H}\mathcal{E}+(q-1)E_+E_-$ and $\mathcal{E}$ lie in the  center of this algebra. It admits a class of irreducible representations for $\mathcal{C}=l^2\mathbf{1}$ and $\mathcal{E}= l^2q^{\lambda-1}\mathbf{1}$, where $l$ and $\lambda$ are real numbers with  $l\neq0$. 
Setting $N=H$, $a=E_-$  and $a^\dagger=E_+$, we get
\begin{eqnarray}
&&[N,\;a]=-a,\quad [N,\;a^\dagger]= a^\dagger,\\
&& aa^\dagger-a^\dagger a =l^2q^{-N+\lambda-1}.
\end{eqnarray}
Hence, $F(N)= 1$ and $G(N)=l^2q^{-N+\lambda-1}$ leading to
\begin{eqnarray}
\varphi(n)= \sum_{k=0}^{n-1}l^2q^{-k+\lambda-1}= l^2q^\lambda\frac{1-q^{-n}}{q-1}=l^2q^{\lambda-n}[n]_q
\end{eqnarray}
Suppose $q<1$. Then, the series
\begin{eqnarray}
 \mathcal{N}_{l,\lambda}(x)&=&\sum_{n=0}^\infty\frac{q^{n(n+1)/2}x^n}{(l^2q^\lambda)^n[n]_q!}
=\sum_{n=0}^\infty\frac{q^{n(n-1)/2}}{(q;q)_n}\Big(\frac{(1-q)qx}{ l^2q^\lambda}\Big)^n
\cr&=& E_q\big((1-q)qx/(l^2q^\lambda)\big).
\end{eqnarray}
has a radius of convergence $R_{l,\lambda}=\infty$  \cite{Gasper}.  While in the case  $q>1$, the factors $\varphi(n)$ remain positive for every $n\geq0$, but the series $\mathcal{N}_{l,\lambda}(x)$ converges only for $|x|<\frac{l^2q^\lambda}{q-1}:=R_{l,\lambda}$.

We have the following result:
\begin{proposition}
The power-moment problem (\ref{PMoment}) has a solution given by
\begin{eqnarray}
 d\omega_{l,\lambda}(x)%&=&\frac{1}{2\pi}\frac{1-q}{l^2q^\lambda\ln{q^{-1}}}\frac{E_q((1-q)qx/(l^2q^\lambda))dx}{E_q((1-q)x/(l^2q^\lambda ))}\cr
&=&\frac{1}{2\pi}\frac{1-q}{l^2q^\lambda\ln{q^{-1}}}\frac{\mathcal{N}_{l,\lambda}(x)}{\mathcal{N}_{l,\lambda}(x/q)}dx, \;\;\mbox{for}\;\; 0<q<1,
\end{eqnarray}
and 
\begin{eqnarray}
 d\omega_{l,\lambda}(x)=\frac{1}{2\pi}\frac{d_q^{l,\lambda}x}{1-(q-1)x/(l^2q^\lambda)}, \quad 0<x<\frac{l^2q^\lambda}{q-1} \;\;\mbox{for}\;\;q>1.
\end{eqnarray}
\end{proposition}
{\bf Proof:} See proof of proposition \ref{prop123}

%In this case the moment problem (\ref{PMoment}) is then written as follows:
%\begin{eqnarray}
% \int_0^{+\infty}x^n\frac{2\pi\;d\omega_{l,\lambda}(x)}{\mathcal{N}_{l,\lambda}(x)}= (l^2q^\lambda)^n q^{-n(n+1)/2}[n]_q!.
%\end{eqnarray}
%The change of variable $y=\frac{x}{l^2q^\lambda}$ yields
%\begin{eqnarray}
% \int_0^{+\infty}y^n\frac{2\pi\;d\omega_{l,\lambda}(l^2q^\lambda y)}{E_q\big((1-q)qy\big)}
%=q^{-n(n+1)/2}[n]_q!. 
%\end{eqnarray}
%Atakishiyev and Atakishiyeva \cite{Atakishiyeva} have demonstrated that
%\begin{eqnarray}
% g_q(n)= \int_0^{+\infty}\frac{y^{n-1}dy}{E_q((1-q)y)}= \frac{\ln{q^{-1}}}{1-q}q^{-n(n-1)/2}[n-1]!_q
%\end{eqnarray}
%from which we deduce
%\begin{eqnarray*}
% d\omega_{l,\lambda}(l^2q^\lambda y)=\frac{1}{2\pi}\frac{1-q}{\ln{q^{-1}}}\frac{E_q((1-q)qy)}{E_q((1-q)y)}dy.
%\end{eqnarray*}
%Then the result.\hfill$\square$\\
Hence, the states
\begin{eqnarray}
 |z\rangle_{l,\lambda}= \mathcal{N}_{l,\lambda}^{-1/2}(|z|^2)\sum_{n=0}^\infty\frac{q^{n(n+1)/4}z^n}{\sqrt{(l^2q^\lambda)^n[n]_q!}}|n\rangle 
\end{eqnarray}
form a family of coherent states for $z\in\mathbf{D}_{l,\lambda}= \left\{z\in\mathbb{C};\;|z|^2<R_{l,\lambda}\right\}$.

\subsubsection{The Chung-Chung-Nam-Um oscillator algebra (1993)}
The generalized $(q,\alpha,\beta)$-algebra was introduced by Chung {\it et al} \cite{Chung} with the generators satisfying
\begin{eqnarray}
 &&[N,\;a]=-a,\qquad [N,\;a^\dagger]= a^\dagger,\\
&&aa^\dagger-qa^\dagger a= q^{\alpha N+\beta}
\end{eqnarray}
where $q\in\mathbb{R}_+^*$ and $\alpha$, $\beta$ are real parameters. One notices that $F(N)=q$ and $G(N)=q^{\alpha N+\beta}$. Therefore,
\begin{eqnarray}
 \varphi(n)= f(n)=: F_{\alpha,\beta}(n;q) %&=& q^{n-1}\sum_{j=0}^{n-1}\frac{q^{\alpha j+\beta}}{q^j}= q^{n-1+\beta}\sum_{j=0}^{n-1}q^{(\alpha-1)j}\cr
&=&\left\{\begin{array}{lcl}nq^{n-1+\beta} &\mbox{if} & \alpha=1\\ q^\beta\frac{q^n-q^{\alpha n}}{q-q^\alpha}& \mbox{if} & \alpha\neq 1.
          \end{array}\right.
\end{eqnarray}
% \begin{eqnarray}
%  \varphi(n)\equiv f(n)\equiv F_{\alpha,\beta}(n;q) &=& q^{n-1}\sum_{j=0}^{n-1}\frac{q^{\alpha j+\beta}}{q^j}= q^{n-1+\beta}\sum_{j=0}^{n-1}q^{(\alpha-1)j}\cr
% &=&\left\{\begin{array}{lcl}nq^{n-1+\beta} &\mbox{if} & \alpha=1\\ q^\beta\frac{q^n-q^{\alpha n}}{q-q^\alpha}& \mbox{if} & \alpha\neq 1.
%           \end{array}\right.
% \end{eqnarray}
This algebra  generalizes the algebras introduced by:
\begin{itemize}
\item Tamm-Dancoff, with $\alpha= 1,\;\beta=0$;
\item Arick-Coon-Kuryshkin, with $\alpha=\beta=0$; and
\item Biedenharn-Macfarlane, with $\alpha=-1,\;\beta=0$.
\end{itemize}
% 
% The Fock space of the Bose oscillator $\mathcal{F}_{\alpha,\beta}$ associated to this deformation is generated by the orthonormalized states
% \begin{eqnarray}
%  |n\rangle_{\alpha,\beta}= \frac{(a^\dagger)^n}{\sqrt{F_{\alpha,\beta}(n;q)!}}|0\rangle,\qquad n=0,\;1,\;2,\; ...
% \end{eqnarray}
% where
% \begin{eqnarray}
% F_{\alpha,\beta}(n;q)!= \left\{\begin{array}{lccl}\;1 & & \mbox{if} & n=0\\
% F_{\alpha,\beta}(n;q)F_{\alpha,\beta}(n-1;q)...F_{\alpha,\beta}(1;q)& &\mbox{if}& n\geq1\end{array}.\right.
% \end{eqnarray}

\subsubsection{The Borzov-Damasky-Yegorov oscillator algebra (1993)}
The generalized $\mathcal{W}_{\alpha,\beta}^\gamma(q)$-algebra was introduced by Borzov {\it et al} \cite{Borzov} in order to unify
a large class of konwn $q$-deformed oscillator algebras. The generators satisfy
\begin{eqnarray}
 &&[N,\;a]=-a,\qquad [N,\;a^\dagger]= a^\dagger,\\
&&aa^\dagger-q^\gamma a^\dagger a= q^{\alpha N+\beta}
\end{eqnarray}
where $q\in\mathbb{R}_+^*$ and $\alpha$, $\beta$, $\gamma$ are real parameters. 
Here $F(N)=q^\gamma$ and $G(N)=q^{\alpha N+\beta}$ leading to 
\begin{eqnarray}
 \varphi(n)= f(n)=F_{\alpha,\beta}^\gamma(n;q) 
% &=& q^{\gamma(n-1)}\sum_{j=0}^{n-1}\frac{q^{\alpha j+\beta}}{q^{\gamma j}}= q^{\gamma(n-1)+\beta}\sum_{j=0}^{n-\gamma}q^{(\alpha-\gamma)j}\cr
&=&\left\{\begin{array}{lcl}nq^{\gamma(n-1)+\beta} &\mbox{if} & \alpha=\gamma\\ 
q^\beta\frac{q^{\gamma n}-q^{\alpha n}}{q^\gamma-q^\alpha}& \mbox{if} & \alpha\neq \gamma.\end{array}\right.
\end{eqnarray}
% \begin{eqnarray}
%  \varphi(n)\equiv f(n)\equiv F_{\alpha,\beta}^\gamma(n;q) 
% &=& q^{\gamma(n-1)}\sum_{j=0}^{n-1}\frac{q^{\alpha j+\beta}}{q^{\gamma j}}= q^{\gamma(n-1)+\beta}\sum_{j=0}^{n-\gamma}q^{(\alpha-\gamma)j}\cr
% &=&\left\{\begin{array}{lcl}nq^{\gamma(n-1)+\beta} &\mbox{if} & \alpha=\gamma\\ 
% q^\beta\frac{q^{\gamma n}-q^{\alpha n}}{q^\gamma-q^\alpha}& \mbox{if} & \alpha\neq \gamma.\end{array}\right.
% \end{eqnarray}
% 
% The Fock space of the Bose oscillator $\mathcal{F}_{\alpha,\beta}^\gamma$ associated to this deformation is generated by the orthonormalized states
% \begin{eqnarray}
%  |n\rangle_{\alpha,\beta}^\gamma= \frac{(a^\dagger)^n}{\sqrt{F_{\alpha,\beta}^\gamma(n;q)!}}|0\rangle,\qquad n=0,\;1,\;2,\; ...
% \end{eqnarray}
% where
% \begin{eqnarray}
% F_{\alpha,\beta}^\gamma(n;q)!= \left\{\begin{array}{lccl}\;1 & & \mbox{if} & n=0\\
% F_{\alpha,\beta}^\gamma(n;q) F_{\alpha,\beta}^\gamma(n-1;q)...F_{\alpha,\beta}^\gamma(1;q)& &\mbox{if}& n\geq1\end{array}.\right.
% \end{eqnarray}

\subsubsection{The Brzezi\'nski - Egusquinza - Macfarlane oscillator algebra (1993)}
Brzezi\'nski {\it et al} \cite{Brzezenski} introduced this algebra as $q$-deformation  {\it \`a la}  Biedenharn - Macfarlane of the
Calogero-Vasiliev oscillator algebra. It is governed by the relations:
\begin{eqnarray}
 &&[N,\;a]=-a,\qquad [N,\;a^\dagger]= a^\dagger,\quad [N,\;K] = 0,\\
&& aK = -Ka,\quad a^\dagger K = -K a^\dagger,\quad K^2= \mathbf{1},\\
&& aa^\dagger - qa^\dagger a= q^{-N}(\mathbf{1}+ 2\alpha K),\label{Brzezinskyalg}
\end{eqnarray}
where $\alpha\in\mathbb{R}^*$, $q\in\mathbb{R}_+$ and $K= (-)^N$ is the Klein operator. 
Here $F(N)= q$ and $G(N)= q^{-N}( 1+ 2\alpha K)$ leading to
\begin{eqnarray}
 \varphi(n)=:f_\alpha(n)%&=& q^{n-1}\sum_{j=0}^{n-1}\frac{q^{-j}(1+2\alpha(-1)^j)}{q^j}\cr
&=&\frac{q^n-q^{-n}}{q-q^{-1}}+ 2\alpha\frac{q^n-(-1)^nq^{-n}}{q+q^{-1}}.
\end{eqnarray}
% 
% The Fock space of the Bose oscillator $\mathcal{F}_{\alpha}$ associated to this deformation is generated by the orthonormalized states
% \begin{eqnarray}
%  |n\rangle_{\alpha}= \frac{(a^\dagger)^n}{\sqrt{f_{\alpha}(n;q)!}}|0\rangle,\qquad n=0,\;1,\;2,\; ...
% \end{eqnarray}
% where
% \begin{eqnarray}
% f_{\alpha}(n;q)!= \left\{\begin{array}{lccl}\;1 & & \mbox{if} & n=0\\
% f_{\alpha}(n;q)f_{\alpha}(n-1;q)...f_{\alpha}(1;q)& &\mbox{if}& n\geq1\end{array}.\right.
% \end{eqnarray}

\subsubsection{The Quesne oscillator algebra (2002)}
The coherent states introduced by Quesne \cite{Quesne} may be associated with the $q$-deformed algebra  satisfying   the relations \cite{Hounkonnou&Ngompe07a}:
\begin{eqnarray}
 &&[N,\;a]=-a,\qquad [N,\;a^\dagger]= a^\dagger,\\
&&aa^\dagger-a^\dagger a= q^{-N-1}\;\mbox{ or }\; qaa^\dagger-a^\dagger a= \mathbf{1}\label{Quesnalg}
\end{eqnarray}
where $0<q<1$.

The first relation of (\ref{Quesnalg}) suggests $F(N)=1$ and $G(N)=q^{-N-1}$ leading to
\begin{eqnarray}
 \varphi(n)=  \frac{1-q^{-n}}{q-1}=q^{-n}[n]_q=:[n]_q^Q.
\end{eqnarray}
% The first relation of (\ref{Quesnalg}) suggests to choose $F(N)=1$ and $G(N)=q^{-N-1}$; therefore
% \begin{eqnarray}
%  \varphi(n)=\sum_{j=0}^{n-1}q^{-j-1}= q^{-1}\frac{1-q^{-n}}{1-q^{-1}}= \frac{1-q^{-n}}{q-1}\equiv[n]_q^Q
% =q^{-n}[n]_q
% \end{eqnarray}
%  
% The Fock space of the Bose oscillator $\mathcal{F}_Q$ associated to this deformation is generated by the orthonormalized states
% \begin{eqnarray}
%  |n\rangle_{{}_Q}= \frac{(a^\dagger)^n}{\sqrt{[n]_q^Q!}}|0\rangle,\qquad n=0,\;1,\;2,\; ...
% \end{eqnarray}

It is a particular case of the {\it Kalnins-Miller-Mukherjee} algebra developed above with $l=1$, $\lambda=0$.

\subsubsection{The $(q;\alpha,\beta,\gamma;\nu)$-Burban oscillator algebra (2007)}
In 2007, Burban \cite{Burban1} introduced the $(q;\alpha,\beta,\gamma;\nu)$-oscillator algebra whose generators satisfy the relations
\begin{eqnarray}
 &&[N,\;a]=-a,\qquad [N,\;a^\dagger]= a^\dagger,\quad [N,\;K] = 0,\\
&& aK = -Ka,\quad a^\dagger K = -K a^\dagger,\quad K^2= \mathbf{1},\\
&& aa^\dagger - q^\gamma a^\dagger a= q^{\alpha N+\beta}(\mathbf{1}+ 2\nu K),\label{Burban1alg}
\end{eqnarray}
where $\nu\in\mathbb{R}^*$, $\alpha, \beta\in\mathbb{R}$, $q\in\mathbb{R}_+$ and $K= (-)^N$ is the Klein operator.
Here $F(N)=q^\gamma$ and $G(N)=q^{\alpha N+\beta}(1+ 2\nu K)$ leading to
\begin{eqnarray}
 \varphi(n) 
&=& q^{\gamma(n-1)}\sum_{j=0}^{n-1}\frac{q^{\alpha j+\beta}(1+2\nu(-1)^j)}{q^{\gamma j}}\cr
&=&\left\{\begin{array}{lcl}q^{\gamma(n-1)+\beta}(n+\nu(1-(-1)^n)) &\mbox{if} & \alpha=\gamma\\ 
q^\beta\left(\frac{q^{\gamma n}-q^{\alpha n}}{q^\gamma-q^\alpha}+2\nu
\frac{q^{\gamma n}-(-1)^nq^{\alpha n}}{q^\gamma+q^\alpha}\right)& \mbox{if} & \alpha\neq \gamma.\end{array}\right.
\end{eqnarray}
% 
% \begin{eqnarray}
%  \varphi(n) 
% &=& q^{\gamma(n-1)}\sum_{j=0}^{n-1}\frac{q^{\alpha j+\beta}(1+2\nu(-1)^j)}{q^{\gamma j}}\cr
% &=&\left\{\begin{array}{lcl}q^{\gamma(n-1)+\beta}(n+\nu(1-(-1)^n)) &\mbox{if} & \alpha=\gamma\\ 
% q^\beta\left(\frac{q^{\gamma n}-q^{\alpha n}}{q^\gamma-q^\alpha}+2\nu
% \frac{q^{\gamma n}-(-1)^nq^{\alpha n}}{q^\gamma+q^\alpha}\right)& \mbox{if} & \alpha\neq \gamma.\end{array}\right.
% \end{eqnarray}
For  $\alpha= \gamma$  and $1+2\nu>0$
\begin{eqnarray}
f(n)= F_{\alpha,\beta;\nu}^\gamma(n;q)= q^{\gamma(n-1)+\beta}(n+\nu(1-(-1)^n)).
\end{eqnarray}
For each of  the following cases
\begin{enumerate}
 \item[1.]  ($q<1$, $\alpha<\gamma$ and $-1<2\nu<-\frac{q^\gamma+q^\alpha}{q^\gamma-q^\alpha}$);
 \item[2.]  ($q>1$, $\alpha>\gamma$ and $-1<2\nu<-\frac{q^\gamma+q^\alpha}{q^\gamma-q^\alpha}$);
 \item[3.]  ($q<1$, $\alpha>\gamma$ and $1+2\nu>0$);
 \item[4.]  $q>1$, $\alpha<\gamma$ and $1+2\nu>0$)
\end{enumerate}
we have
\begin{eqnarray}
 f(n)= F_{\alpha,\beta;\nu}^\gamma(n;q)= q^\beta\left(\frac{q^{\gamma n}-q^{\alpha n}}{q^\gamma-q^\alpha}+2\nu
\frac{q^{\gamma n}-(-1)^nq^{\alpha n}}{q^\gamma+q^\alpha}\right).
\end{eqnarray}
% 
% 
% However for
% \begin{itemize}
%  \item  ($q<1$, $\alpha<\gamma$ and $-1<2\nu<-\frac{q^\gamma+q^\alpha}{q^\gamma-q^\alpha}$) or
%  \item  ($q>1$, $\alpha>\gamma$ and $-1<2\nu<-\frac{q^\gamma+q^\alpha}{q^\gamma-q^\alpha}$) or
%  \item  ($q<1$, $\alpha>\gamma$ and $1+2\nu>0$) or 
%  \item  $q>1$, $\alpha<\gamma$ and $1+2\nu>0$),
% \end{itemize}
% \begin{eqnarray}
%  f(n)\equiv F_{\alpha,\beta;\nu}^\gamma(n;q)= q^\beta\left(\frac{q^{\gamma n}-q^{\alpha n}}{q^\gamma-q^\alpha}+2\nu
% \frac{q^{\gamma n}-(-1)^nq^{\alpha n}}{q^\gamma+q^\alpha}\right)
% \end{eqnarray}
% 
% The Fock space of the Bose oscillator $\mathcal{F}_{\alpha,\beta;\nu}^\gamma$ associated to this deformation is generated by 
% the orthonormalized states
% \begin{eqnarray}
%  |n\rangle_{\alpha,\beta;\nu}^\gamma= \frac{(a^\dagger)^n}{\sqrt{F_{\alpha,\beta;\nu}^\gamma(n;q)!}}|0\rangle,\qquad n=0,\;1,\;2,\; ...
% \end{eqnarray}
% where
% \begin{eqnarray}
% F_{\alpha,\beta;\nu}^\gamma(n;q)!= \left\{\begin{array}{lccl}\;1 & & \mbox{if} & n=0\\
% F_{\alpha,\beta;\nu}^\gamma(n;q) F_{\alpha,\beta;\nu}^\gamma(n-1;q)...F_{\alpha,\beta;\nu}^\gamma(1;q)& &\mbox{if}& n\geq1\end{array}.\right.
% \end{eqnarray}

\subsubsection{The $(p,q)$ and $(p,q;\mu,\nu, f)$-oscillator algebras (2007)}
In $2007$, our group \cite{Hounkonnou&Ngompe07a} introduced an   algebra  generalizing  the {\it Quesne}  oscillator algebra:
\begin{eqnarray}
&[N,\;a]=-a,\qquad& [N,\;a^\dagger]= a^\dagger,\\
& p^{-1}aa^\dagger - a^\dagger a = q^{-N-1},\;& qaa^\dagger - a^\dagger a = p^{N+1}\label{Hounkalg}
\end{eqnarray}
From   the first relation (\ref{Hounkalg}),  we get $F(N)=p$ and $G(N)=pq^{-N-1}$ and then
\begin{eqnarray}
 \varphi(n)=p^{n-1}\sum_{j=0}^{n-1}\frac{pq^{-j-1}}{p^j}= \frac{p^n-q^{-n}}{q-p^{-1}}=:[n]_{p,q}^Q.
\end{eqnarray}
The $(p,q;\mu,\nu, f)$-oscillator algebra is defined through the following commutation relations
\begin{eqnarray}
&&[N,\;a]=-a,\qquad [N,\;a^\dagger]= a^\dagger,\\
&&  \frac{p^{\mu-1}}{q^{\nu-1}}aa^\dagger - p^{-1}a^\dagger a = \left(\frac{q^{\nu}}{p^{\mu-1}}\right)^Nf(p,q)\label{Hounkalg2}
\end{eqnarray}
where $p,\;q,\;\mu,\;\nu$ are real numbers such that $0<pq<1$, $p^\mu< q^{\nu-1}$, $p > 1$ and  $f$ {\it a well behaved real and non-negative function} of deformation parameters $p$ and $q$, satisfying $\lim f(p,q)=1$ as $(p,p)\to(1,1)$.
 
Here, $\displaystyle F(N)=\frac{q^{\nu-1}}{p^\mu}$ and $\displaystyle G(N)= f(p,q) \frac{q^{\nu-1}}{p^{\mu-1}} \left(\frac{q^{\nu}}{p^{\mu-1}}\right)^N$, so that
\begin{eqnarray}
 \varphi(n)%&=&  \left(\frac{q^{\nu-1}}{p^{\mu}}\right)^{n-1}f(p,q) \frac{q^{\nu-1}}{p^{\mu-1}} \sum_{j=0}^{n-1}
% \left(\frac{q^{\nu}}{p^{\mu-1}}\right)^j\left(\frac{q^{\nu-1}}{p^{\mu}}\right)^{-j}
%\cr&=& f(p,q) \left(\frac{q^{\nu}}{p^{\mu}}\right)^{n}\frac{q^{-n}}{p^{-1}}\sum_{j=0}^{n-1}(pq)^j\cr
&=& f(p,q) \left(\frac{q^{\nu}}{p^{\mu}}\right)^{n}\frac{p^n-q^{-n}}{q-p^{-1}}=: [n]_{p,q,f}^{\mu,\nu}.
\end{eqnarray}
The series
\begin{eqnarray}
\mathcal{N}_{p,q,f}^{\mu,\nu}(x) = \sum_{n=0}^\infty\frac{x^n}{[n]_{p,q,f}^{\mu,\nu}!}
\end{eqnarray}
has a radius of convergence  $R= +\infty$.
It had also shown in \cite{Hounkonnou&Ngompe07a} that the moment problem (\ref{PMoment}) has, for $\mu=1$ and $\nu=0$, the following solution
\begin{eqnarray}
 d\omega_{p,q,f}^{1,0}(x)=\frac{1}{2\pi}\frac{p^{-1}-q}{f(p,q)\ln{(pq)^{-1}}}
\frac{\mathcal{N}_{p,q,f}^{1,0}(x)}{\mathcal{N}_{p,q,f}^{1,0}(x/(pq))}dx.
\end{eqnarray}
Finally, the states
\begin{eqnarray}
 |z\rangle_{p,q,f}^{1,0}= \left(\mathcal{N}_{p,q,f}^{1,0}(|z|^2) \right)^{-1/2}
\sum_{n=0}^\infty\frac{z^n}{\sqrt{[n]_{p,q,f}^{1,0}!}}|n\rangle,\;z\in\mathbb{C},
\end{eqnarray}
form a family of coherent states.

\subsubsection{Unified $(p, q; \alpha, \beta, \nu; \gamma)$-deformed oscillator algebra (2012)}
More recently, Balo\"itcha {\it et al} \cite{Baloitcha} introduced the unified $(p, q; \alpha, \beta, \nu; \gamma)$-deformed 
oscillator algebra whose generators satisfy:
\begin{eqnarray}
 &&[N,\;a]=-a,\qquad [N,\;a^\dagger]= a^\dagger,\quad [N,\;K] = 0,\\
&& aK = -Ka,\quad a^\dagger K = -K a^\dagger,\quad K^2= 1,\\
&& aa^\dagger - p^\nu a^\dagger a= ( 1+ 2\gamma K)q^{\alpha N+\beta},\label{Baloitchalg}
\end{eqnarray}
where, $\alpha, \beta,\gamma, \nu\in\mathbb{R}$, $p, q\in\mathbb{R}_+$ and $K= (-)^N$ is the Klein operator.

Here, $F(N)= p^\nu$ and $G(N)=  ( 1+ 2\gamma (-)^N)q^{\alpha N+\beta}$ and
\begin{eqnarray}
 \varphi(n)%&=& p^{\nu(n-1)}\sum_{j=0}^{n-1}\frac{( 1+ 2\gamma (-1)^j)q^{\alpha j+\beta}}{p^j}
%\cr&=&q^\beta p^{\nu(n-1)}\left(\sum_{j=0}^{n-1}({q^{\alpha}}/{p^{\nu }})^j+2\gamma\sum_{j=0}^{n-1}(-1)^j({q^{\alpha}}/{p^{\nu }})^j\right)\cr
&=&\left\{\begin{array}{lcl}
q^\beta\left(\frac{p^{\nu n}-q^{\alpha n}}{p^\nu-q^\alpha}+2\gamma\frac{p^{\nu n}-(-1)^nq^{\alpha n}}{p^\nu+q^\alpha}\right)
&\mbox{ if }& p^\nu\neq q^\alpha
\\ q^{\beta+\alpha(n-1)}\left(n+2\gamma\frac{1-(-)^n}{2}\right)&\mbox{ if }&p^\nu=q^\alpha.\end{array}\right..
\end{eqnarray}
For  $p^\nu= q^\alpha$ and $1+2\gamma>0$
\begin{eqnarray}
f(n)=: F_{\alpha,\beta;\gamma}^\nu(n;p,q)= q^{\beta+\alpha(n-1)}\left(n+2\gamma\frac{1-(-)^n}{2}\right).
\end{eqnarray}
For each of the following cases:
\begin{enumerate}
 \item[1.] $p^\nu>q^\alpha$ and $1+2\gamma>0$, and 
 \item[2.] $p^\nu<q^\alpha$ and $-1<2\gamma<-\frac{p^\nu+q^\alpha}{p^\nu-q^\alpha}$
\end{enumerate}
we have
\begin{eqnarray}
 f(n)=: F_{\alpha,\beta;\gamma}^\nu(n;p,q)= q^\beta\left(\frac{p^{\nu n}-q^{\alpha n}}{p^\nu-q^\alpha}
+2\gamma\frac{p^{\nu n}-(-1)^nq^{\alpha n}}{p^\nu+q^\alpha}\right).
\end{eqnarray}
 
%%%%%%%%%%%%%%%%%%%%%%%%%%%%%%%%%%%%%%%%%%%%%%%%%%%%%%%%%%%%%%

\section{${\cal R}(p,q)$-deformed quantum algebras: coherent states and
special functions}\label{chap3}
$\;$

% \begin{abstract}
 We provide with a generalization of well known $(p,q)$-deformed Heisenberg
algebras, called ${\cal R}(p,q)$-deformed quantum algebras,
and study the corresponding ${\cal R}(p, q)$-series. A general
formulation of the binomial theorem is given. Special functions are
obtained as limit cases. This work well prolongs a previous work by
Odzijewicz \cite{Odzijewicz98}.
%[{\it Commun. Math. Phys.} {\bf 192} (1998)]. 
 Known results in the literature are recovered. 
% \end{abstract}

% 
% This chapter is organized as follows. Section~\ref{Sect32} is devoted to the
% theoretical framework in which some hypothesis are fixed and
% coherent states map is defined as a complex analytic map from a
% complex disc into a complex separable Hilbert  space. The annihilation
% operator $A$ and its Hermitian conjugate, the creation operator
% $A^\dag$, generate a $C^*-$algebra. A monomorphism
% is embedding to compute the realization of this algebra on the set
% of holomorphic functions defined on the disc.
% In section~\ref{Sect33}, a realization of the annihilation operator is defined
% followed by the definition of coherent states map and the
% exponential function corresponding to a given meromorphic function.
% The focal point of this section is a generalization of the
% $q-$binomial and  $(p,q)-$binomial theorems (\cite{Burban&Klimyk,Jagannathan&Rao,Odzijewicz98}),
% which we call ${\cal R}(p, q)-$binomial theorem, preceded by no less
% two important lemmas.
% Section~\ref{Sect34} concerns with the study of the generalized algebra
% generated by the annihilation and creation operators.
% In Section~\ref{Sect35}, we define the ${\cal R}(p, q)-$coherent states in
% accordance with the conditions given in ~\cite{Klauder63a,Klauder63b}. In
% Section~\ref{Sect36}, the ${\cal R}(p, q)-$trigonometric and hyperbolic
% functions are defined.
% In Section~\ref{Sect36}, modified $(p,q)-$Bessel's functions are given. 
% %Then follow the concluding remarks in Section~\ref{Sect38}.

\subsection{Theoretical framework}\label{Sect32}
The development displayed in this section is essentially based on
the formalism elaborated by Odzijewicz ~\cite{Odzijewicz98} in a nice,
mathematically based work published in 1998, but unfortunately
hushed up in the recent literature on the topic. In the mentioned
work, this author investigated the quantum algebras generated by the
coherent state maps of the disc, leading to a generalized analysis
which includes standard analysis as well as q-analysis. He provided
with the meromorphic continuation of the generalized basic
hypergeometric series and constructed a reproducing measure, when
the series is treated as a reproducing kernel. Indeed, much to our
very great surprise, most  all the remarkable coherent state
generalizations, performed from the generalization of exponential
function by different authors, can be generated
from this more general theory.

Let ${\cal H}$ be an infinite dimensional separable Hilbert space
and $\{\vert n\rangle\}_{n=0}^{\infty}$ its canonical basis. Assume
that there exists a sequence $\{f_n\}_{n=0}^{\infty}$ in  $ {\cal
H}$ such that
\begin{equation}\label{n1}
f_n = c_n C \vert n\rangle
\end{equation}
where $C$ and its inverse $C^{-1}$ are bounded operators on
${\cal H}$, and $c_n$ ($n=0,1, 2, \cdots$) are real positive
numbers satisfying the conditions
\begin{eqnarray}{\label{cond}}
\sup_{n\in\mathbb{N}}{\frac{c_{n-1}}{c_n}}< +\infty\quad
\mbox{and}\quad R^{-1} = \lim_{n\to\infty} \sup{\sqrt[n]{c_n}}.
\end{eqnarray}
\begin{definition}~~\cite{Odzijewicz98}
A coherent states map is a complex analytic map
\begin{eqnarray}\label{comap}
&&K: \mathbb{D}_R \longrightarrow {\cal H}\setminus \{0\}
\nonumber \\
&&\qquad z\quad \hookrightarrow\quad K(z) =  \sum_{n = 0}^{\infty}
{f_n z^n}
\end{eqnarray}
where $\mathbb{D}_R$ =$\{z\in\mathbb{C}$ : $\vert z\vert < R \}$.
The states $K(z)$ are called coherent states and the operator $A$
admitting these states  as eigenstates with eigenvalues
$z\in\mathbb{D}_R$, i.e.
\begin{equation}\label{eigv}
AK(z) = zK(z),
\end{equation}
is said to be the annihilation operator.
\end{definition}
The relations (\ref{n1})-(\ref{eigv}) lead to
\begin{eqnarray}
C^{-1}A C\vert0\rangle =0 \quad \mbox{and} \quad && C^{-1}A C \vert
n\rangle = \frac{c_{n-1}}{c_n}\vert n-1\rangle \quad \forall n \geq
1.
\end{eqnarray}
Therefore, $\|C^{-1}AC\| < + \infty$  meaning that $ A$ is a bounded
operator. Its Hermitian conjugate, called the creation operator and
denoted by $A^{\dag}$, is also bounded. The algebra closure,  spanned
by the operators $\{A, A^\dag\}$, in their norm topology, gives the
so-called $C^*-$algebra ${\cal A}_K$.
\begin{proposition}
Let ${\cal O}(\mathbb{D}_R)$ be a set of holomorphic functions defined
on the disc $\mathbb{D}_R$.\\
The map $I_K : {\cal H}\longrightarrow {\cal O}(\mathbb{D}_R)$ such
as
\begin{equation}
\label{IK} I_K(v):=
< v \; | \; K(\cdot) >, \qquad  v \in {\cal H}
\end{equation}
is an antilinear monomorphism of complex vector spaces. Moreover, considering
the topology of simple convergence, $I_K\left({\cal H}\right)$
is dense on ${\cal O}(\mathbb{D}_R)$.
\end{proposition}
\begin{proof}
From (\ref{comap})
\begin{equation}
f_{N+1} =  \lim_{z\rightarrow 0}{\frac{1}{z^{N+1}}}
\big(K(z)-\sum_{k=0}^{N}{f_k z^k}\big) \quad  \texttt{and} \quad f_0
= K(0),
\end{equation}
it results, by induction, that $f_N$ belongs to the closure of the
linear space spanned by the elements of the subset $ \{ K(z),\;\;
|z| < \epsilon < R \} \subset \mathcal{H} $. Otherwise, the coherent
states  $K(z)$, $z\in\mathbb{D}_R$, form a linearly subset dense in
${\cal H}$. One can easily check that $I_K$ is antilinear, i.e. $  $,
\begin{eqnarray}
I_K( \alpha u+\beta v)= \bar{\alpha}I_K(u)+\bar{\beta}I_K(v),
\quad \forall \alpha , \beta \in \mathbb{C}  \;\; u, v \in \mathcal{H},
\end{eqnarray}
and
\begin{equation} z^n = I_K\left(c_n^{-1}(C^{\dag})^{-1}\vert
n\rangle\right)
\end{equation}
implying $z^n\in I_K\left({\cal H}\right),\quad \forall
n\geq 0.$ \cqfd
\end{proof}

Therefore, $I_K\left({\cal H}\right)$ inherits the Hilbert space
structure endowed with the scalar product
\begin{equation}
\left <I_K(v_1)\vert I_K(v_2)\right >:= \left <v_1\vert v_2\right >
\quad \forall  v_1, v_2 \in{\cal H}.
\end{equation}
Thus, the Hilbert space ${\cal H}$ can be identified with the dense
subspace $I_K\left({\cal H}\right)$ of ${\cal O}(\mathbb{D}_R)$.
Hence, $I_K\circ{\cal A}_K\circ I_K^{-1}$ is the analytic
representation of the algebra ${\cal A}_K$ as shown in the following
diagram:

\vspace{0.1cm}
\begin{eqnarray*}
\xymatrix{{\cal H}\ar[d]_{I_K} \ar[r]^X & {\cal H}\ar[d]^{I_K}\\
I_K\left({\cal H}\right)\ar[d]_{Id}  & I_K\left({\cal H}\right)\ar[d]^{Id}
& & \beta_X = I_K\circ X\circ I_K^{-1}\\
{\cal O}(\mathbb{D}_R) \ar[r]^{\beta_X} &  {\cal O}(\mathbb{D}_R)}
\end{eqnarray*}
\vspace{0.1cm} So, given an operator $X\in{\cal A}_K$, we find a
unique linear operator $\beta_X$ into ${\cal O}(\mathbb{D}_R)$ such
that $\beta_X = I_K\circ X\circ I_K^{-1}$. Consequently, the action of
the analytic representation $(I_K\circ A^\dag \circ I_K^{-1})$ of
the creation operator $A^\dag$  on $\varphi\in{\cal O}(\mathbb{D}_R)$ yields:
\begin{equation}
\left(I_K\circ A^\dag\circ I_K^{-1}\right)\varphi(z) = z \varphi(z),
\end{equation}
while that  of the analytic representation $\partial:= (I_K\circ
A\circ I_K^{-1})$ of the annihilation operator, the so-called
\textit{$K-$derivative} ~\cite{Odzijewicz98}, depends on the operator $C$ and
parameters $c_n$. On the basis elements $z^n$, it gives
\begin{equation}
\partial z^n =
\sum_{k=0}^{\infty}c_kc_n^{-1}\langle n\vert C^{-1}A^\dag C\vert k\rangle
z^k.
\end{equation}
Without loss of generality and as  a matter of convenience, we
restrict, in the sequel, the analysis to a unity operator  $C$,
(i.e. $C = \mathbb{I}$) leading to
\begin{equation}
{\label{coef1}} \partial z^n =
[n]z^{n-1}
\end{equation}
where
\begin{eqnarray} [n]=\left\{
\begin{array}{l}  0 \;\; {\rm if} \;\;n=0 \\
\left(\frac{c_{n-1}}{c_n}\right)^2 \;\; {\rm if} \;\ n\geq 1.
\end{array}\right.
\end{eqnarray}
Therefore,
\begin{equation}{\label{coef2}}
c_n = \frac{c_0}{\sqrt{[n]!}} \qquad \qquad   \qquad   \qquad
\end{equation}
with $[0]!= 1$ and $[n]! = [n]([n-1]!)$, and
\begin{eqnarray}{\label{coef3}}
 && AA^\dag \vert n\rangle = [n+1]\vert
n\rangle,
% \nonumber \\
%&&\mbox{and}
\quad A^\dag A\vert n\rangle = [n]\vert
n \rangle , \\
&&\nonumber \\
&& c_n \geq \frac{c_0}{\|A\|^n}\;\;\; \mbox{for}\;\;\; \|A\| \geq R.
\qquad
\end{eqnarray}
The sequence $\{[n]\}_{n \geq 0}$ converges and
\begin{equation}{\label{ray}}
R = \lim_{n
\to \infty}{\sqrt[n]{[n]}}. \qquad \qquad
\end{equation}
Similarly, the $K-$exponential function is defined by
\begin{equation}{\label{exp1}}
\mathrm{Exp}(\bar{v}, z) := \big< K(v)\vert K(z)\big>
\end{equation}
and satisfies the equation
\begin{equation}{\label{equadif1}}
\partial \mathrm{Exp}(\bar{v}, z) =
\bar{v}\mathrm{Exp}(\bar{v}, z). \;\;
\end{equation}
Setting $v = 1$,  one has
\begin{equation}
{\label{exp2}} \mathrm{Exp}(z):= \mathrm{Exp}(1, z)\qquad
\end{equation}
which satisfies the equation
\begin{equation}{\label{equadif2}}
\partial \mathrm{Exp}(z) = \mathrm{Exp}(z).\qquad
\end{equation}
Hence, provided the   $K-$derivative ($\partial$) realization, we
can determine the analytic representation, $I_K\circ {\cal A}_K\circ
I_K^{-1}$, of the $C^\star-$ algebra ${\cal A}_K$. Moreover through the
equalities (\ref{coef1}) and (\ref{coef2}), the coherent states map
$K: \mathbb{D}_R \longrightarrow {\cal H}\setminus \{0\}$ generate
the coherent states
\begin{equation}{\label{coh1}} K(z)
= \sum_{n = 0}^{\infty}c_nz^n\vert n\rangle = \sum_{n=0}^{\infty}
\frac{c_0}{\sqrt{[n]!}}z^n\vert n\rangle ,
\end{equation} and the
exponential functions (\ref{exp1}) and (\ref{exp2}) are reduced to
\begin{equation}{\label{exp3}} \mathrm{Exp}(\bar{v}, z) = \sum_{n=0}^{\infty}
\frac{c_0^2}{[n]!}(\bar{v}z)^n
\end{equation}
and, after normalization (i.e. $c_0 = 1$), to
\begin{equation}{\label{exp4}} \mathrm{Exp}(z) =
\sum_{n=0}^{\infty} \frac{z^2}{[n]!},
\end{equation} respectively.

In the next section, we aim at extending the above formalism to a
general $(p, q)$-analysis, constructing a $(\mathcal{ R}, p, q)-$
derivative, where $\mathcal{ R}$ is a meromorphic function defined
on $\mathbb{C}\times\mathbb{C}$ and can be, in particular cases, a
rational function.

\subsection{${\cal R}(p, q)-$basic hypergeometric series related to meromorphic functions}\label{Sect33}
 The problem we set here consists in  defining
the derivative which leads to a generalization of $(p, q)-$algebras
and $(p,q)-$basic hypergeometric series. Let $p$ and $q$ be two
positive real numbers  such that $0<q<p$. Consider
 a meromorphic function
${\cal R}$, defined on $\mathbb{C}\times\mathbb{C}$ by
\begin{equation}\label{Rxy}
{\cal R}(x,y) = \sum_{k, l = - L}^{\infty} r_{kl}x^ky^l
\end{equation}
with an eventual isolated singularity at the zero, where $r_{kl}$
are complex numbers, $L\in\mathbb{N}\cup\{0\}$, ${\cal R}(p^n, q^n)>
0$ $\forall n\in\mathbb{N}$, and ${\cal R}(1, 1) = 0$, and the
following linear operators  defined on ${\cal O}(\mathbb{D}_R)$
~\cite{Chakrabarti&Jagan,Jagannathan&Rao,Odzijewicz98} by:
\begin{eqnarray}\label{operat}
 Q&:& \varphi \longmapsto Q\varphi(z) = \varphi(qz)
\nonumber \\
 P&:& \varphi \longmapsto P\varphi(z) = \varphi(pz) \nonumber\\
{\label{deriva}} \partial_{p,q}&:&\varphi \longmapsto
\partial_{p,q}\varphi(z) = \frac{\varphi(pz) - \varphi(qz)}{z(p-q)}.
\end{eqnarray}
Then, we define  the analytic representation of the annihilation
operator $A$, called the ${\cal R}(p, q)-$ derivative, by
\begin{equation}{\label{deriva1}} \partial_{{\cal R}(p,q)} :=
\partial_{p,q}\frac{p - q}{P-Q}{\cal R}(P, Q) = \frac{p -
q}{pP-qQ}{\cal R}(pP, qQ)\partial_{p,q}.
\end{equation}
Using relations (\ref{coef1}), (\ref{coef2}), and
(\ref{coh1}), we define the ${\cal R}(p, q)$-factors (also called ${\cal R}(p, q)$-numbers) by ${\cal R}(p^n,q^n)$, $n=0,\; 1,\; 2,\; \cdots$ from which we deduce the ${\cal R}(p, q)$-factorials
\begin{eqnarray}
%&&\label{Rnumb} {\cal R}!(p^n,q^n)(p,q)}:= {\cal R}(p^n,q^n) \quad \mbox{for} \quad n\geq 0,\\
&&\label{Rfac} {\cal R}!(p^n,q^n):=
\left\{\begin{array}{lr} 1 \quad \mbox{for   } \quad n=0 \quad \\
{\cal R}(p,q)\cdots{\cal R}(p^n,q^n) \quad \mbox{for } \quad n\geq
1, \quad \end{array} \right.
\end{eqnarray} 
and inducing
the coefficients $c_n$ and the coherent states map $K$ in the form
\begin{eqnarray}\label{CfRpq}
&&c_n^2= \frac{c_0^2}{{\cal R}!(p^n,q^n)} \\ && \label{KRpq}
K_{{\cal R}(p,q)}(z) =
\sum_{n=0}^{\infty}\frac{c_0}{\sqrt{{\cal R}!(p^n,q^n)}}z^n\vert
n\rangle.
\end{eqnarray}
Besides, the relations (\ref{exp3}) and (\ref{exp4}) can be readily
generalized to take the form
\begin{eqnarray}
&&{\label{exp5}}
\mathrm{Exp}_{{\cal R}(p,q)}(\bar{v},z) =
\sum_{n=0}^{\infty}\frac{c_0^2}{{\cal R}!(p^n,q^n)}(\bar{v}z)^n \\
&&\mbox{and}\nonumber \\ &&{\label{exp6}} \mathrm{Exp}_{{\cal R}(p,q)}(z) =
\sum_{n=0}^{\infty}\frac{1}{{\cal R}!(p^n,q^n)}z^n,
\end{eqnarray}
respectively, with the virtue that $c_0^2 = \mathrm{Exp}_{{\cal R}(p,q)}(0)= 1$. Unless we say otherwise $R$ will represent the
radius of convergence of the series. The equations (\ref{equadif1})
and (\ref{equadif2}) remain valid and (\ref{exp6}) is a
solution of the ${\cal R}(p, q)-$difference equation
\begin{equation} {\label{diff2}} {\cal R}(P,Q)\mathrm{Exp}_{{\cal R}(p,q)}(z) =
z\mathrm{Exp}_{{\cal R}(p,q)}(z). \end{equation}

The following two statements
are essential to perform a generalization of the $(p,q)$-binomial
theorem.
\begin{lemma}
{\label{lem1}}Let
\begin{eqnarray}
&&F(z) = \frac{z}{z-{\cal R}(1,0)}, \nonumber
\\
&&G(P,Q) = \frac{p(Q-P){\cal R}(pP,qQ)+(pP-qQ){\cal R}(1,0)}{pQ{\cal
R}(pP,qQ)}
\end{eqnarray}
if ${\cal R}(1,0)\neq 0$, and
\begin{eqnarray}
&&F(z) =z, \nonumber
\\&&
G(P,Q) = \frac{qQ-pP}{pQ{\cal R}(pP,qQ)}  \qquad  \qquad  \qquad
\qquad \qquad
\end{eqnarray}
if ${\cal R}(1,0) = 0$.
Then the exponential function in (\ref{exp6}) satisfies
\begin{equation}
\mathrm{Exp}_{{\cal R}(p,q)}(z) = \left[1- F(z)G(P,Q)\right]\mathrm{Exp}_{{\cal
R}(p,q)}\left(\frac{q}{p}z\right).
\end{equation}

\end{lemma}

\begin{proof}

From definitions  (\ref{deriva}) and (\ref{deriva1}), we deduce
\begin{equation}{\label{opera}}
1 = \frac{Q}{P} + z
\frac{pP-qQ}{pQ{\cal R}(pP,qQ)}\frac{Q}{P}\partial_{{\cal R}(p,q)}.
\qquad\qquad
\end{equation}
The action of the operator  (\ref{opera}) on the exponential
function (\ref{exp6})
 gives
\begin{eqnarray*} \mathrm{Exp}_{{\cal R}(p,q)}(z) = \left[1-
z\frac{qQ-pP}{pQ{\cal R}(pP,qQ)}\right]\mathrm{Exp}_{{\cal R}(p,q)}
\left(\frac{q}{p}z\right) 
\end{eqnarray*}
which corresponds to the
case ${\cal R}(1,0) = 0$. Now, if ${\cal R}(1,0)\neq 0$, the
identity in (\ref{opera}) can  be rewritten
\begin{eqnarray}
 && 1-z\frac{1}{{\cal R}(1,0)}\partial_{{\cal R}(p,q)} =
\left[\frac{Q}{P}-z\frac{1}{{\cal R}(1,0)}\frac{Q}{P}\partial_{{\cal R}(p,q)}\right] \cr 
&&\qquad \qquad \qquad +z\left[ \frac{p(Q-P){\cal R}(pP,qQ)
+(pP-qQ){\cal R}(1,0)}{pQ{\cal R}(pP,qQ){\cal R}(1,0)}\frac{Q}{P}
\right]\partial_{{\cal R}(p,q)}\cr 
&&
\end{eqnarray}
which, acting on $\mathrm{Exp}_{{\cal R}(p,q)}(z)$, leads, after  a short
computation, to the result:
\begin{eqnarray*}
\mathrm{Exp}_{{\cal R}(p,q)}(z) &=&  \left[1- \frac{z}{z-{\cal R}(1,0)}
\frac{p(Q-P){\cal R}(pP,qQ) +(pP-qQ){\cal R}(1,0)}{pQ{\cal
R}(pP,qQ){\cal R}(1,0)}\right]\cr && \qquad \times  \mathrm{Exp}_{{\cal R}(p,q)} \left(\frac{q}{p}z\right).
\end{eqnarray*}
\cqfd
\end{proof}

\begin{lemma}{\label{lem2}}
Under assumptions of the Lemma~\ref{lem1},
\begin{equation}{\label{exp7}}
\mathrm{Exp}_{{\cal R}(p,q)}(z) =
\prod_{k=0}^{n-1}\left[1-F\left(\frac{q^k}{p^k}z\right)G(P,Q)\right]
\mathrm{Exp}_{{\cal R}(p,q)}\left(\frac{q^n}{p^n}z\right).
\end{equation}
\end{lemma}
\begin{proof}: It is immediate by induction using   Lemma~\ref{lem1}.
 \cqfd
\end{proof}
Finally, the $\mathcal {R}(p,q)$-binomial formula is given
through the following statement.
\begin{theorem}{\label{theo1}}
Let  Lemmas~\ref{lem1} and \ref{lem2} be satisfied. Then, a
generalization of the $(p,q)-$binomial theorem can be expressed as:
\begin{equation}
\sum_{n=0}^{\infty}\frac{1}{{\cal R}!(p^n,q^n)}(z)^n
=\prod_{n=0}^{\infty}\left[1-F\left(\frac{q^n}{p^n}z\right)G(P,Q)\right]\cdot 1.
\end{equation}
\end{theorem}

\begin{prooft}: The result follows from (~\ref{exp7}), tending $n$ to
$+\infty$.  \cqfd
\end{prooft}

Finally, let us consider the particular case  when ${\cal R}$ is a
rational function defined by
\begin{equation}\label{rRs}
{}_r{\cal R}_s(x,y) =
\frac{x^{1+s-r}(x-y)(c_1p^{-1}x-d_1q^{-1}y) \dots
(c_sp^{-1}x-d_sq^{-1}y)}{(-y)^{1+s-r}(a_1p^{-1}x-b_1q^{-1}y) \dots
(a_rp^{-1}x-b_rq^{-1}y)}
\end{equation}
where $ a_1, \dots a_r; b_1, \cdots b_r; c_1, \cdots, c_s; d_1,
\cdots, d_s$ are complex numbers, $r$ and  $s$ being non negative
integers. Then,  the ${}_r{\cal R}_s(p,q)-$numbers and factorials,
as well as the coefficients $c_n$ of the corresponding exponential
function are readily found to be
\begin{eqnarray}
&&\big[n\big]_{{}_r{\cal R}_s(p,q)}=
\frac{(p^n-q^n)(c_1p^{n-1}-d_1q^{n-1}) \dots
(c_sp^{n-1}-d_sq^{n-1})}{\big(-(q/p)^n\big)^{1+s-r}(a_1p^{n-1}-b_1q^{n-1})
\dots (a_rp^{n-1}-b_rq^{n-1})},
\\ &&\big[n\big]!_{{}_r{\cal R}_s(p,q)}=
\frac{\big((p,q),(c_1,d_1), \cdots,(c_s,d_s);(p,q)\big)_n}{\big((a_1,b_1),
\cdots, (a_r,b_r);(p,q)\big)_n \left[(-1)^n (q/p)^{(^n_2)}\right]^{1+s-r}},
\; n\geq1,\qquad
\\ && c_n^2 =
\frac{\big((a_1,b_1), \cdots, (a_r,b_r);(p,q)\big)_n}{\big((p,q),(c_1,d_1),
\cdots, (c_s,d_s);(p,q)\big)_n}\left[(-1)^n (q/p)^{(^n_2)}\right]^{1+s-r}
\end{eqnarray}
where
\begin{eqnarray}
&&\big((a_1,b_1), \cdots, (a_r,b_r);(p,q)\big)_n =
\big((a_1,b_1);(p,q)\big)_n \cdots, \big((a_r,b_r);(p,q)\big)_n, \nonumber\\
&&\big((a_i,b_i);(p,q)\big)_0 = 1 \qquad\mbox{and}\nonumber \\
&&\big((a_i,b_i);(p,q)\big)_n
=\prod_{k=0}^{n-1}(a_ip^k-b_iq^k)\quad\mbox{for}\quad n\geq 1.
\end{eqnarray}
Therefore,
\begin{eqnarray*}
 \mathrm{Exp}_{{}_r{\cal R}_s(p,q)}(z) =
\sum_{n=0}^{\infty}\frac{\big((a_1,b_1), \cdots,
(a_r,b_r);(p,q)\big)_n \quad z^n}{\big((p,q),(c_1,d_1), \cdots,
(c_s,d_s);(p,q)\big)_n}
\left[(-1)^n (q/p)^{(^n_2)}\right]^{1+s-r}
\end{eqnarray*}
proving that the ${}_r{\cal R}_s(p,q)-$exponential function
corresponds to the twin-basic hypergeometric series
 ${_r\Phi}_s$ ~\cite{Jagannathan&Rao,Chakrabarti&Jagan}:
\begin{equation}
\mathrm{Exp}_{{}_r{\cal R}_s(p,q)}(z)={_r\Phi}_s\big((a_1,b_1), \cdots, (a_r,b_r);(c_1,d_1), \cdots,
(c_s,d_s);(p,q);z\big).
\end{equation}
For $r = s+1$, the coherent states map $K_{{}_r{\cal R}_s(p,q)}$ and
the  exponential function $\mathrm{Exp}_{{}_r{\cal R}_s(p,q)}$ are defined on
the unit disc $\mathbb{D}_1$ while, for $r < s+1$, they
 are defined on the whole complex plane. The exponential function
 $\mathrm{Exp}_{{\cal R}(p,q)}$ thus appears as a
natural generalization of twin-basic (or $(p,q)-$) hypergeometric
series.

\subsection{${\cal R}(p, q)-$ deformed quantum algebras}\label{Sect34}
In this section, we deal with the study of the ${\cal R}(p, q)-$
deformed quantum algebra. Relations between the annihilation and
creation operators, and the operators $Q$ and $P$ as well as  the
algebra generated by the meromorphic function ${\cal R}$ are
obtained. Some relevant particular representations recovered  in
this framework are also investigated.

The use of the relation (\ref{coef3}) and the $({\cal R}(p,q)$-factors engenders
\begin{eqnarray*}
AA^\dag K(z) = \sum_{n=0}^{\infty}{\cal R}(p^{n+1},q^{n+1})
c_nz^n \vert n\rangle = {\cal R}(pP,qQ)K(z)
\end{eqnarray*}
giving
\begin{equation}\label{algRpq1}
AA^\dag = {\cal R}(pP,qQ).
\end{equation}
By analogy, one can show that
\begin{equation}\label{algRpq2}
A^\dag A = {\cal R}(P,Q).
\end{equation}

If one passes to analytic representation in which $Q$, $P$ and
$\partial_{p,q}$ are given by (\ref{operat}), and $A^\dag$ acts as
mutliplication by $z$, one obtains the relations:
\begin{equation}
A Q = q Q A ,\qquad A P = p P A,
\end{equation}
\begin{equation}
Q A^\dag = q A^\dag Q , \qquad P A^\dag = p A^\dag P.
\end{equation}
It is clear that
\begin{equation}\label{algRpq3}
Q P = P Q.
\end{equation}
So, the ${\cal R}(p, q)$-deformed quantum algebra is generated by
the operators $\{1$, $A$, $A^\dag$, $Q$, $P\}$ that verify the
 relations (\ref{algRpq1})-(\ref{algRpq3}).

Let us determine the analytic representation of the number
operator $N$. Taking into account $N\vert n\rangle = n\vert
n\rangle$ and  (\ref{IK}), one obtains
\begin{equation}
I_K\left(\frac{1}{c_n}N\vert
n\rangle\right)(z) = I_K\left(\frac{n}{c_n}\vert n\rangle\right)(z) = nz^n
= z\frac{dz^n}{dz}
\end{equation}
implying
\begin{equation}
\left(I_K\circ N\circ I_K^{-1}\right)(z^n) = z\frac{dz^n}{dz}.
\end{equation}
Therefore, $\forall f \in {\cal O}(\mathbb{D}_R)$ we infer
\begin{equation}
\left(I_K\circ N\circ I_K^{-1}\right)f(z) = z\frac{d}{dz}f(z)
\end{equation}
and
\begin{equation}
p^Nf(z) = Pf(z)= f(pz)\;\;\; \mbox{and}\;\;\; q^Nf(z) = Qf(z)= f(qz)
\end{equation}
for $f \in {\cal O}(\mathbb{D}_R)$.

To sum up, the ${\cal R}(p, q)-$deformed quantum algebra is then generated by
the set of operators $\{1, A, A^\dag, N\}$ and the commutation
relations
\begin{equation}\label{algN1}
\left[N, A\right] = - A \;\;\;\mbox{and}\;\;\; \left[N,
A^\dag\right] = A^\dag \end{equation}
and
\begin{equation}\label{algN2}
\left[A, A^\dag\right] = {\cal R}(p^{N+1},q^{N+1}) - {\cal
R}(p^N,q^N).
\end{equation}

\subsection*{Remarks}
Particulars cases are readily recovered.
\begin{enumerate}

\item[{\bf(i).}]{\bf Odziejewicz generalization} ~\cite{Odzijewicz98} is recovered for $0 < q < p = 1.$\\
{\it Consider the meromorphic function ${\cal R}$ defined on $\mathbb{C}$ by
\begin{equation}\label{Rz}
{\cal R}(z) = \sum_{k= - L}^{\infty} r_k z^k
\end{equation}
which may have an isolated singularity at the zero and such that
$L\in\mathbb{N}\cup\{0\}$,\\ ${\cal R}(q^n)> 0$ for $n>0$, ${\cal
R}(0)> 0$, and ${\cal R}(1) =  0$. The following results hold:
\begin{itemize}
\item
The ${\cal R}(q)$-derivative is given by
\begin{equation}\label{deR}
\partial_{{\cal R}(q)} = \partial_q\frac{1-q}{1-Q}{\cal R}(Q)
= \frac{1-q}{1-qQ}{\cal R}(qQ)\partial_q
\end{equation}
where $ Q\varphi(z) = \varphi(qz)$ and $ \partial_q = \frac{1-Q}{(1-q)z}$.
\item
The ${\cal R}(q)$-factors and ${\cal R}(q)$-factorials are given
 by
\begin{eqnarray} \label{RNbz}
&&{\cal R}(q^n)\quad
\mbox{i.e.}\quad \partial_{{\cal R}}z^n =  {\cal R}(q^n)z^{n-1}
\;\; \mbox{for}\;\; n \geq 0,
\\&& \label{Rfacz}
 {\cal R}!(q^n)= \left\{\begin{array}{lr} 1 \quad \mbox{for   }
\quad n=0 \quad \\ {\cal R}(q)\cdots{\cal R}(q^n) \quad \mbox{for }
\quad n\geq 1,\end{array} \right.
\end{eqnarray}
respectively.
\item
The coefficients $c_n$, the coherent states map $K_{{\cal R}}$ and
the exponential functions $\mathrm{Exp}_{{\cal R}}$ are given  by
\begin{eqnarray} \label{cn}
&&c_n^2 = \frac{c_0^2}{{\cal R}!(q^n)},\\
&&\label{KR} K_{{\cal R}}(z) =
\sum_{n=0}^{\infty}\frac{c_0}{\sqrt{{\cal R}!(q^n)}} z^nn\rangle,
\qquad \qquad \qquad \\
&& \label{ExpR1} \mathrm{Exp}_{{\cal R}}(\bar{v},z)=
\sum_{n=0}^{\infty}\frac{1}{{\cal R}!(q^n)}(\bar{v}z)^n\qquad \quad
\\&&\label{ExpR2} \mathrm{Exp}_{{\cal R}}(z)=
\sum_{n=0}^{\infty}\frac{1}{{\cal R}!(q^n)}z^n,\qquad \qquad
\qquad\qquad
\end{eqnarray}
respectively.
\item The latter exponential function is
solution of the $q-$difference equation
\begin{equation}\label{diffz}
[{\cal R}(Q)\mathrm{Exp}_{{\cal R}}](z) = z\mathrm{Exp}_{{\cal R}}(z),   \qquad
\qquad\qquad \end{equation} and satisfies the following generalization of
the $q-$binomial theorem
\begin{equation}\label{binRq} \mathrm{Exp}_{{\cal R}}(z)
= \sum_{n=0}^{\infty}\frac{1}{{\cal R}!(q^n)}z^{n} =
\prod_{n=0}^{\infty}[1 - F(zq^n)G(Q)].1
\end{equation}
where
\begin{eqnarray} && F(z) = \frac{z}{z-{\cal R}(0)}, \nonumber \\ && G(Q)
=\frac{(Q-1){\cal R}(qQ)+(1-qQ){\cal R}(0)}{Q{\cal R}(qQ)} \qquad
\qquad \qquad \end{eqnarray} for $L = 0 $, and \begin{eqnarray} && F(z) = z,
\nonumber \\ && G(Q) = \frac{qQ-1}{Q{\cal R}(qQ)} \qquad \qquad  \qquad
\qquad \qquad \qquad \qquad
\end{eqnarray} for $L> 0$.
\item The quantum algebra ${\cal A}_{K_{{\cal R}(q)}}$ generated by the set of operators

$\{1, A, A^\dag, Q\}$
 satisfies the relations:
\begin{eqnarray}\label{algRq1}
&&AA^\dag = {\cal R}(qQ), \\
&& \label{algRq2}A^\dag A =
{\cal R}(Q),
\end{eqnarray}
from which one obtains
\begin{equation}
\|A\|= \|A^\dag\|= \sqrt{\sup_{n\in\mathbb{N}}{{\cal R}(q^n)}},
\end{equation}
for $L = 0$, i.e. $A$ and $A^\dag$ are bounded. Due to
(\ref{algRq1}) and (\ref{algRq2}), $A$ and $A^\dag$ are unbounded if $L >0$.
\end{itemize}
}%\end{remark}
Indeed,  setting ${\cal R}(z) = {\cal R}(1,z)$, (\ref{Rz}) readily
follows  from (\ref{Rxy}). The derivative (\ref{deR}) is then
deduced from (\ref{deriva1}). Using % (\ref{Rnumb}) and 
(\ref{Rfac}),
one gets (\ref{RNbz}) and (\ref{Rfacz}). Coefficients (\ref{cn}) and
coherent states maps (\ref{KR}) are obtained from (\ref{CfRpq}) and
(\ref{KRpq}), while the exponential functions in (\ref{exp5}) and
(\ref{exp6}) are reduced to (\ref{ExpR1}) and (\ref{ExpR2}),
respectively. In the other hand, the $q-$difference equation
(\ref{diffz}) is directly obtained from (\ref{diff2}). In the same
way  and using Theorem~\ref{theo1}, we get the generalization of the
$q-$binomial theorem (\ref{binRq}). Finally, the relations
(\ref{algRq1}) and (\ref{algRq2}) between $A$, $A^\dag$, and $Q$
simply follow from (\ref{algRpq1}) and (\ref{algRpq2}).

\vspace{0.2cm}
\item[{\bf(ii).}] {\bf Jagannathan's generalization} ~\cite{Jagannathan&Rao}
The main results summarized as follows are particular cases of
Theorem~\ref{theo1}:
 {\it \begin{itemize}
\item The  $(p,q)-$binomial theorem is:
\begin{equation}\label{binpq}
{_1\Phi}_0((a,b);-;(p,q);z;) = \frac{((p,bz);(p,q))_\infty}{((p,az);(p,q))_\infty}.
\end{equation}
\item The  exponential functions, denoted by $e_{p,q}$ and
$E_{p,q}$, are
\begin{equation}\label{epq} e_{p,q}(z):= {_1\Phi}_0((1,0); -; (p,q);z)=
\sum_{n=0}^{\infty}\frac{p^{n(n-1)/2}}{((p,q);(p,q))_n}z^n
\end{equation}
 \begin{equation}\label{Epq} E_{p,q}(z):=  {_1\Phi}_0((0,1); -;(p,q);-z)
=\sum_{n=0}^{\infty}\frac{q^{n(n-1)/2}}{((p,q);(p,q))_n}z^n \end{equation}
 and \begin{equation}\label{Eepq} e_{p,q}(z)E_{p,q}(-z) =
1. \end{equation}
\end{itemize}
}

Indeed, if we consider \begin{equation} {_1{\cal R}_0}(x,y) =
\frac{x-y}{\frac{a}{p}x-\frac{b}{q}y} \end{equation} where $0<q<p$
and $a, b$ are complex numbers, then we get \begin{equation}
{_1{\cal R}_0}(p^n,q^n) = \frac{p^n-q^n}{ap^{n-1}-bq^{n-1}},
\end{equation} \begin{eqnarray} &&{_1{\cal R}_0}!(p^n,q^n) =
\frac{((p,q);(p,q))_n} {((a,b);(p,q))_n}, \\ && c_n^2 =
\frac{((a,b);(p,q))_n}{((p,q);(p,q))_n}
\end{eqnarray}
and
\begin{eqnarray*} \mathrm{Exp}_{{_1{\cal R}_0}(p,q)} &=&
\sum_{n=0}^{\infty}\frac{(a-b)(ap-bq)...(ap^{n-1}-bq^{n-1})}{(p-q)...(p^n-q^
n)}z^n\\ &=&
\sum_{n=0}^{\infty}\frac{((a,b);(p,q))_n}{((p,q);(p,q))_n}z^n,
\end{eqnarray*} where $((a,b);(p,q))_n =
(a-b)(ap-bq)...(ap^{n-1}-bq^{n-1})$, meaning that \begin{equation}
{\label{q1}} \mathrm{Exp}_{{_1{\cal R}_0}(p,q)} = {_1\Phi}_0((a,b);-;(p,q);z).
\end{equation} In the other hand, since ${_1{\cal R}_0}(1,0)=
\frac{p}{a}$, we get \begin{equation} F(z) =
\frac{az}{az-p}\;\;\;\mbox{and}\;\;\; G(P,Q) = \frac{a-b}{a}.
\end{equation} Then, the application of  Theorem~\ref{theo1} yields
\begin{eqnarray*} \mathrm{Exp}_{{_1{\cal R}_0}(p,q)} &=&
\prod_{n=0}^{\infty}\left[1-\frac{z(q/p)^n}{p-az(q/p)^n}(a-b)\right]\\
&=&
\prod_{n=0}^{\infty}\left[\frac{p-bz(q/p)^n}{p-az(q/p)^n}\right]\\
&=& \prod_{n=0}^{\infty}\left[\frac{pp^n-bzq^n}{pp^n-azq^n}\right].
\end{eqnarray*}
Thus, \begin{equation} {\label{q2}} \mathrm{Exp}_{{_1{\cal
R}_0}(p,q)} = \frac{((p,bz);(p,q))_\infty}{((p,az);(p,q))_\infty}.
\end{equation} So, (\ref{q1}) and (\ref{q2}) lead to (\ref{binpq}).
Finally, (\ref{epq}), (\ref{Epq}) and (\ref{Eepq}) are
straightforwardly obtained. It is worth noticing that one can also
consider the meromorphic functions
\begin{eqnarray} &&{_1{\cal R}_0}(x,y) = \frac{x-y}{\frac{x}{p}}\\
&&{_1{\cal R}_0}(x,y) = \frac{x-y}{-\frac{y}{q}}, \end{eqnarray} and
use Theorem~\ref{theo1} to immediately obtain (\ref{epq}),
(\ref{Epq}) and (\ref{Eepq}).
\end{enumerate}

%%%%%%%%%%%%%%%%%%%%%%%%%%%%%%%%%%%%%%%%%%%%%%%%%%%%%%%%%%%%%%%%%%%%%%%

\subsection{${\cal R}(p, q)-$coherent states}\label{Sect35}

This section aims at proving that the coherent states derived from
the coherent states map (\ref{KRpq}) satisfy the following conditions
~\cite{Ali&al,Klauder63a,Klauder63b}:
\begin{enumerate}
\item[(i).] normalizability (as any vector of Hilbert space)
\item[(ii).] continuity in the label z, and
\item[(iii).] existence of a resolution of identity with a positive definite weight
function, implying that the states form an overcomplete set.
\end{enumerate}

\subsubsection{Normalizability}
The coherent states defined as
\begin{equation}{\label{CSn}}
\vert z\rangle_{\mathcal{R}(p,q)} := \left(\mathrm{Exp}_{{\cal
R}(p,q)}(\vert z\vert^2)\right)^{-1/2}
\sum_{n=0}^{\infty}\frac{1}{\sqrt{{\cal R}!(p^n,q^n)}}z^n\vert
n\rangle.
\end{equation}
are normalized so that
\begin{equation}
{_{\mathcal{R}(p,q)}}\langle z\vert z\rangle_{\mathcal{R}(p,q)} = 1.
\end{equation}

\subsubsection{Continuity in z}
The coherent states $\vert z\rangle_{\mathcal{R}(p,q)}$ are
continuous in $z$. Indeed,
\begin{equation}
\left\| \vert z\rangle_{\mathcal{R}(p,q)} - \vert
z'\rangle_{\mathcal{R}(p,q)} \right\|^2 = 1 - 2
Re\left({_{\mathcal{R}(p,q)}}\langle z\vert z'\rangle_{\mathcal{R}(p,q)}\right) + 1
,
\end{equation}
where
\begin{eqnarray*}
{_{\mathcal{R}(p,q)}}\langle z\vert z'\rangle_{\mathcal{R}(p,q)}
   = \left(\mathrm{Exp}_{{\cal R}(p,q)}(\vert z\vert^2)
\mathrm{Exp}_{{\cal R}(p,q)}(\vert
z'\vert^2)\right)^{-1/2}\mathrm{Exp}_{{\cal R}(p,q)}(\bar{z},z').
\end{eqnarray*}
So,
\begin{equation}
\vert z - z'\vert \to   0   \quad \Longrightarrow \left\|\vert
z\rangle_{\mathcal{R}(p,q)} - \vert z'\rangle_{\mathcal{R}(p,q)}
\right\|^2 \to  0.
\end{equation}

\subsubsection{Resolution of identity}
Assume that there exists a positive measure $\mu$ on the disc
$\mathbb{D}_R$ for which the resolution of the identity
\begin{equation}
\mathbb{I} = \int_{\mathbb{D}_{R}}^{}{\frac{\vert K(z)\rangle\langle
K(z)\vert} {\mathrm{Exp}_{{\cal R}(p,q)}(\bar{z},z)}d\mu(\bar{z},z)}
\end{equation}
holds.
\begin{proposition}
$I_K({\cal H})$ is a subspace of the Hilbert space
$L^2(\mathbb{D}_R,\frac{1}{\mathrm{Exp}_{{\cal R}(p,q)}}d\mu)$.
Moreover, for $\varphi , \psi \in I_K({\cal H}) $
\begin{equation}\label{scaprod}
\langle \varphi \vert \psi \rangle =
\int_{\mathbb{D}_R}^{}{\overline{\varphi(z)} \psi(z) \frac{1}
{\mathrm{Exp}_{{\cal R}(p,q)}(\bar{z},z)}d\mu(\bar{z},z)},
\end{equation}
and
\begin{equation}
\varphi(z) = \int_{\mathbb{D}_R}^{}{\varphi(v)\mathrm{Exp}_{{\cal
R},p,q}(\bar{v},z) \frac{1} {\mathrm{Exp}_{{\cal R}(p,q)}(\bar{v},v)}d\mu(\bar{v},v)}.
\end{equation}
\end{proposition}

\begin{proof}:
Let $\varphi \in I_K({\cal H})$. There exists $\zeta\in{\cal H}$
such that $I_K(\zeta) = \varphi$. So,
\begin{eqnarray*}
\langle \varphi\vert\varphi\rangle &=&
\int_{\mathbb{D}_R}^{}{\overline{\varphi(z)} \varphi(z) \frac{1}
{\mathrm{Exp}_{{\cal R}(p,q)}(\bar{z},z)}d\mu(\bar{z},z)}\\
&=& \int_{\mathbb{D}_R}^{}{\overline{I_K(\zeta)(z)} I_K(\zeta)(z)
\frac{1}
{\mathrm{Exp}_{{\cal R}(p,q)}(\bar{z},z)}d\mu(\bar{z},z)}\\
&=& \int_{\mathbb{D}_R}^{}{\langle\zeta\vert K(z)\rangle\langle
K(z)\vert\zeta\rangle
\frac{1}{\mathrm{Exp}_{{\cal R}(p,q)}(\bar{z},z)}d\mu(\bar{z},z)}\\
&=&\langle\zeta\vert \int_{\mathbb{D}_R}^{}{\frac{\vert
K(z)\rangle\langle K(z)\vert}
{\mathrm{Exp}_{{\cal R}(p,q)}(\bar{z},z)}d\mu(\bar{z},z)}\vert\zeta\rangle\\
&=& \langle\zeta\vert\zeta\rangle = \left\|\zeta\right\|^2 < \infty.
\end{eqnarray*}
Therefore, $\varphi\in L^2(\mathbb{D}_R,
\frac{1}{\mathrm{Exp}_{{\cal R}(p,q)}}d\mu)$. 

Thus, since $ I_K({\cal
H})$ is a vector space as subspace of ${\cal O}(\mathbb{D}_R)$, we
conclude that $I_K({\cal H})$ is a subspace of the Hilbert space
$L^2(\mathbb{D}_R,\frac{1}{\mathrm{Exp}_{{\cal R}(p,q)}}d\mu)$.

Hence, the scalar product (\ref{scaprod}) is well defined. Moreover,
for $\varphi\in I_K({\cal H})$ and $z\in \mathbb{D}_R$, we get
\begin{eqnarray*}
\varphi(z) &=& \left(I_K\circ I_K^{-1}\circ\varphi\right)(z)=
\langle I_K{^{-1}}\circ\varphi\vert K(z)\rangle\\
&=& \langle I_K{^{-1}}\circ\varphi\vert
\int_{\mathbb{D}_{R}}^{}{\frac{\vert K(v)\rangle\langle K(v)\vert}
{\mathrm{Exp}_{{\cal R}(p,q)}(\bar{v},v)}d\mu(\bar{v},v)}\vert K(z)\rangle\\
&=& \int_{\mathbb{D}_{R}}^{}{\frac{\langle
I_K{^{-1}}\circ\varphi\vert K(v)\rangle \langle K(v)\vert
K(z)\rangle}
{\mathrm{Exp}_{{\cal R}(p,q)}(\bar{v},v)}d\mu(\bar{v},v)}\\
&=& \int_{\mathbb{D}_{R}}^{}{\varphi(v)\frac{\mathrm{Exp}_{{\cal R}(p,q)}(\bar{v},z)} {\mathrm{Exp}_{{\cal R}(p,q)}(\bar{v},v)}d\mu(\bar{v},v)}.
\end{eqnarray*} \qquad \cqfd
\end{proof}
\begin{proposition}\ For $n,m\in \mathbb{N}\cup\{0\}$
\begin{equation}\label{integra1}
\int_{\mathbb{D}_{R}}^{}{\frac{\bar{z}^nz^m} {\mathrm{Exp}_{{\cal R}(p,q)}(\bar{z},z)}d\mu(\bar{z},z)} = {\cal R}!(p^n,q^n)\;\delta_{m,n}.
\end{equation}
Passing to polar coordinate $z = \sqrt{x}e^{i\varphi}$, we get
\begin{equation}\label{integra2}
\int_{0}^{R^2}{x^n}\frac{d\nu(x)}{\mathrm{Exp}_{{\cal R}(p,q)}(x)} =
\frac{{\cal R}!(p^n,q^n)}{2\pi}, \quad \mbox{for} \qquad n = 0, 1,
2 \cdots.
\end{equation}
\end{proposition}

\begin{proof}: For $n\in \mathbb{N}\cup\{0\}$, $z^n\in I_K(\mathcal{H})$,
since $z^n = I_K(\sqrt{{\cal R}!(p^n,q^n)}\vert n\rangle)(z)$. So,
\begin{eqnarray*}
\langle z^n\vert z^m\rangle &=&
\int_{\mathbb{D}_R}^{}{\frac{\bar{z}^nz^m}
{\mathrm{Exp}_{{\cal R}(p,q)}(\bar{z},z)}d\mu(\bar{z},z)}\\
&=& \int_{\mathbb{D}_R}^{} {\frac{\overline{I_K(\sqrt{{\cal R}!(p^n,q^n)}\vert n\rangle)(z)} I_K(\sqrt{[m]!_{{\cal R}(p,q)}}\vert
m\rangle)(z)}
{\mathrm{Exp}_{{\cal R}(p,q)}(\bar{z},z)}d\mu(\bar{z},z)}\\
&=& \sqrt{{\cal R}!(p^n,q^n)}\sqrt{[m]!_{{\cal R},p,q}}\int_{\mathbb{D}_R}^{} {\frac{\langle K(z) \vert n
\rangle\langle m\vert K(z)\rangle}
{\mathrm{Exp}_{{\cal R}(p,q)}(\bar{z},z)}d\mu(\bar{z},z)}\\
&=& \sqrt{{\cal R}!(p^n,q^n)}\sqrt{[m]!_{{\cal R}(p,q)}}\langle
m\vert \int_{\mathbb{D}_R}^{} {\frac{\vert K(z)\rangle\langle K(z)
\vert}
{\mathrm{Exp}_{{\cal R}(p,q)}(\bar{z},z)}d\mu(\bar{z},z)}\vert n \rangle\\
&=& \sqrt{{\cal R}!(p^n,q^n)}\sqrt{[m]!_{{\cal R}(p,q)}}\langle
m\vert n \rangle
\end{eqnarray*}
from which (\ref{integra1}) follows. Setting $z=
\sqrt{x}e^{i\varphi}$, $ 0\leq x\leq R^2$ and $0\leq\varphi\leq
2\pi$, we get $d\mu(\bar{z},z)$ =
$d\mu(\sqrt{x}e^{-i\varphi},\sqrt{x}e^{i\varphi})$
 = $d\nu(x)d\varphi$ and
$\mathrm{Exp}_{{\cal R}(p,q)}(\bar{z},z) = \mathrm{Exp}_{{\cal
R},p,q}(x)$. Then, from (\ref{integra1}) and taking $m=n$ we have
\begin{eqnarray*}
{\cal R}!(p^n,q^n) =
\int_{0}^{2\pi}{d\varphi\int_{0}^{R^2}{x^n\frac{d\nu(x)}
{\mathrm{Exp}_{{\cal R}(p,q)}(x)}}} = 2\pi
\int_{0}^{R^2}{x^n\frac{d\nu(x)} {\mathrm{Exp}_{{\cal R}(p,q)}(x)}}.
\end{eqnarray*}
Dividing the left and right sides of the above equalities by $2\pi$ we
obtain (\ref{integra2}).\qquad \cqfd
\end{proof}

As in ~\cite{Odzijewicz98} the quantities (\ref{integra2}) may also be treated as
defining properties of the measure $\nu$. This is the famous moment
problem of finding $\nu$ from the knowledge of the fixed moments
${\cal R}!(p^n,q^n)$, $n = 0, 1, \cdots$.

It may be interesting to formulate the moment problem in a way more
adequate for $K-$analysis. So, the integration with respect to the
measure $\frac{d\nu}{\mathrm{Exp}_{{\cal R}(p,q)}}$ is replaced by
the $K -$ integration
\begin{equation}
{\cal I}x^n := \frac{1}{{\cal R}(p^{n+1},q^{n+1})}x^{n+1},
\end{equation}
with some unknown analytic weight function $\sigma$ such that its
Taylor expansion
\begin{equation}
\sigma(x) = \sum_{k=0}^{\infty}{a_kx^k}
\end{equation}
has $R$ as its convergence radius. The $K -$ integration is just the
right inverse of the $K -$ differentiation, i.e. ${\cal
I}.\partial_{{\cal R}(p,q)}= \mathbb{I}$. Then, instead of looking
for a measure $d\nu$ which satisfies (\ref{integra2}), the problem
is to find an analytic function $\sigma$ which satisfies the
moment conditions
\begin{equation}
 {\cal I}x^n\sigma(x)_{\vert_{x=R^2}} = \frac{{\cal R}!(p^n,q^n)}{2\pi}
\quad \mbox{for} \quad n = 0, 1, 2, \cdots,
\end{equation}\
i.e.
\begin{equation}
 \sum_{k=0}^{\infty}{\frac{R^{2(n+k+1)}}{{\cal R}(p^{n+k+1},q^{n+k+1})}a_k} = \frac{{\cal R}!(p^n,q^n)}{2\pi}
\quad \mbox{for} \quad n = 0, 1, 2, \cdots.
\end{equation}

%%%% 27/05/09%%%%%%%%%%%%%%%%%%%%%%%%%%%%%%%%%%%%%%%%%
\subsection*{Remarks}
Special cases can be recovered:
\begin{enumerate}
\item[\bf(i).] {\bf The $(p,q)-$ algebra of Chakrabarty and Jagannathan} ~\cite{Chakrabarti&Jagan}
{\it
\begin{eqnarray}\label{algCJ}
AA^\dag - qA^\dag A = q^{-N} &\qquad& AA^\dag - qA^\dag A = p^{N}
\cr \left[N,A\right] = - A &\qquad& \left[N,A^\dag\right]= A^\dag
\end{eqnarray}
affords the associated coherent states
\begin{equation}{\label{CSCJ}}
\vert z\rangle = \left[{\cal N}_{p,q}(\vert z\vert
^2)\right]^{-1/2}\sum^{\infty}_{n=0}\frac{z^n}{\sqrt{[n]_{p,q}!}}
\vert n\rangle
\end{equation}
where
\begin{equation}\label{expCJ}
{\cal N}_{p,q}(z) =\sum^{\infty}_{n=0}\frac{z^n}{[n]_{p,q}!},
\end{equation}
as well as the $(p,q)-$number and the $(p,q)-$factorial
given by
\begin{equation}\label{RnumCJ}
[n]_{p,q} = \frac{p^{-n}-q^n}{p^{-1}-q}
\end{equation}
and
\begin{equation}\label{cofaCJ}
[n]_{p,q}! =[1]_{p,q}[2]_{p,q} ...[n]_{p,q},
\end{equation}
respectively,
} taking $ {\cal R}(x,y) = \frac{1-
xy}{(p^{-1}-q)x}$.

Indeed, one obtains relations in (\ref{algCJ}) through relations
(\ref{algN1}) and (\ref{algN2}, whereas coherent states (\ref{CSCJ})
become particular cases of coherent states (\ref{CSn}) with the
function ${\cal N}$ in (\ref{expCJ}) as the analog of the
exponential function (\ref{exp6});  the $(p,q)-$number and
$(p,q)-$factorial are deduced from the ${\cal R}(p,q)$-factors and
(\ref{Rfac}), respectively.

\item[\bf(ii).] {\bf $(p,q)-$generalization of $q-$Quesne
algebra}~\cite{Hounkonnou&Ngompe07a,Quesne&al02}: {\it
\begin{eqnarray}\label{algHN1}
p^{-1}AA^\dag - A^\dag A = q^{-N-1} &\qquad& qAA^\dag - A^\dag A =
p^{N+1} \cr \left[N,A\right] = - A &\qquad& \left[N,A^\dag\right]=
A^\dag
\end{eqnarray}
generates the associated coherent states
\begin{equation}{\label{CSHN1}}
\vert z\rangle = \left[{\cal N}_{p,q}(\vert z\vert
^2)\right]^{-1/2}\sum^{\infty}_{n=0}\frac{z^n}{\sqrt{[n]^Q_{p,q}!}}
\vert n\rangle
\end{equation}
where
\begin{equation}\label{expHN1}
{\cal N}_{p,q}(z) =\sum^{\infty}_{n=0}\frac{z^n}{[n]^Q_{p,q}!}
\end{equation}
while the $(p,q)-$number is given by
\begin{equation}\label{RnumHN1}
[n]^Q_{p,q} = \frac{p^{-n}-q^n}{p^{-1}-q},
\end{equation}
}
setting $ {\cal R}(x,y) =
\frac{xy-1}{(q-p^{-1})y}$.

\item[\bf(iii).] {\bf $(p,q;\mu,\nu,f)-$deformed states of Hounkonnou and
Ngompe} ~\cite{Hounkonnou&Ngompe07a}:
{\it
\begin{equation}{\label{CSHN2}}
\vert z\rangle^{\mu,\nu}_{p,q,f} = \left[{\cal N}_{p,q}(\vert z\vert
^2)\right]^{-1/2}\sum^{\infty}_{n=0}\frac{z^n}{\sqrt{[n]^{\mu,\nu}_{p,q,f}!}}
\vert n\rangle
\end{equation}
where
\begin{equation}\label{expHN2}
{\cal N}^{\mu,\nu}_{p,q,f}(z)
=\sum^{\infty}_{n=0}\frac{z^n}{[n]^{\mu,\nu}_{p,q,f}!},
\end{equation}
with the $(p,q)-$number  given by
\begin{equation}\label{RnumHN2}
[n]^{\mu,\nu}_{p,q,f} = f(p,q)\frac{q^{\nu n}}{p^{\mu
n}}\frac{p^n-q^{-n}}{q-p^{-1}},
\end{equation}
such that $0 < pq < 1$, $p^\mu < q^{\nu - 1}$, $p > 1$, and $f$ a
well behaved real and non-negative function of deformation
parameters $p$ and $q$, satisfying
\begin{equation}
\lim_{(p,q) \to (1,1)}f(p,q) = 1,
\end{equation}
 becomes a particular case in the generalization provided in this work, setting $ {\cal R}(x,y) =
f(p,q)\frac{y^{\nu n}}{x^{\mu n}}\frac{x^n-y^{-n}}{q-p^{-1}}$,  $f$
being meromorphic}.

From the above mentioned $(p,q;\mu,\nu,f)-$deformed states,  other
deformed states known in the litterature can be easily recovered. See ~\cite{Hounkonnou&Ngompe07a} for more details.
\end{enumerate}

\subsection{$\mathcal{R}(p,q)$-trigonometric and hyperbolic functions}\label{Sect36}
From the expression (\ref{exp6}) of the exponential function, we  obtain
\begin{eqnarray}
\nonumber \mathrm{Exp}_{\mathcal{R}(p,q)}(iz) &=& \sum_{n=0}^{\infty}\frac{(iz)^n}{{\cal R}!(p^n,q^n)}\\
&=& \sum_{n=0}^{\infty}\frac{(-1)^nz^{2n}}{{\cal R}!(p^{2n},q^{2n})}
+
i\sum_{n=0}^{\infty}\frac{(-1)^nz^{2n+1}}{{\cal R}!(p^{2n+1},q^{2n+1})}
\end{eqnarray}
and
\begin{eqnarray}
\nonumber \mathrm{Exp}_{\mathcal{R}(p,q)}(-iz) &=& \sum_{n=0}^{\infty}\frac{(-iz)^n}{{\cal R}!(p^n,q^n)}\\
&=& \sum_{n=0}^{\infty}\frac{(-1)^nz^{2n}}{{\cal R}!(p^{2n},q^{2n})}
-
i\sum_{n=0}^{\infty}\frac{(-1)^nz^{2n+1}}{{\cal R}!(p^{2n+1},q^{2n+1})}.
\end{eqnarray}
We then define the ${\cal R}(p, q)-$cosine, sinus,
hyperbolic cosine and sinus functions by
\begin{eqnarray}
\cos_{\mathcal{R}(p,q)}(z) =
\sum_{n=0}^{\infty}\frac{(-1)^nz^{2n}}{{\cal R}!(p^{2n},q^{2n})},
\end{eqnarray}
\begin{eqnarray}
\sin_{\mathcal{R}(p,q)}(z) =
\sum_{n=0}^{\infty}\frac{(-1)^nz^{2n+1}}{{\cal R}!(p^{2n+1},q^{2n+1})},
\end{eqnarray}
\begin{eqnarray}
\cosh_{\mathcal{R}(p,q)}(z) =
\sum_{n=0}^{\infty}\frac{z^{2n}}{{\cal R}!(p^{2n},q^{2n})}
\end{eqnarray}
and
\begin{eqnarray}
 \sinh_{\mathcal{R}(p,q)}(z) = \sum_{n=0}^{\infty}\frac{z^{2n+1}}{{\cal R}!(p^{2n+1},q^{2n+1})},
\end{eqnarray}
respectively.
It is readily checked that
\begin{equation}
 \cos_{\mathcal{R}(p,q)}(z) =\frac{1}{2}
\left[\mathrm{Exp}_{\mathcal{R}(p,q)}(iz)+\mathrm{Exp}_{\mathcal{R}(p,q)}(-iz)\right],
\end{equation}
\begin{equation}
 i\sin_{\mathcal{R}(p,q)}(z) =\frac{1}{2}
\left[\mathrm{Exp}_{\mathcal{R}(p,q)}(iz)-\mathrm{Exp}_{\mathcal{R}(p,q)}(-iz)\right],
\end{equation}
\begin{equation}
\mathrm{Exp}_{\mathcal{R}(p,q)}(iz) = \cos_{\mathcal{R}(p,q)}(z) + i
\sin_{\mathcal{R}(p,q)}(z).
\end{equation}
In particular, the Euler formula is expressed as follows:
\begin{equation}
\mathrm{Exp}_{\mathcal{R}(p,q)}(i\theta) =
\cos_{\mathcal{R}(p,q)}(\theta) + i \sin_{\mathcal{R}(p,q)}(\theta).
\end{equation}
Besides, there follow the relations
\begin{equation}
\cosh_{\mathcal{R}(p,q)}(z) =\frac{1}{2}
\left[\mathrm{Exp}_{\mathcal{R}(p,q)}(z)+\mathrm{Exp}_{\mathcal{R}(p,q)}(-z)\right],
\end{equation}
\begin{equation}
\sinh_{\mathcal{R}(p,q)}(iz) =\frac{1}{2}
\left[\mathrm{Exp}_{\mathcal{R}(p,q)}(z)-\mathrm{Exp}_{\mathcal{R}(p,q)}(-z)\right],
\end{equation}
\begin{equation}
\mathrm{Exp}_{\mathcal{R}(p,q)}(z) =
\cosh_{\mathcal{R}(p,q)}(z)+ \sinh_{\mathcal{R}(p,q)}(z).
\end{equation}
The  derivatives are immediately expressible from
their definition:
\begin{equation}
\partial_{\mathcal{R}(p,q)}\sinh_{\mathcal{R}(p,q)}(az) = a\;
\cosh_{\mathcal{R}(p,q)}(az),
\end{equation}
\begin{equation}
\partial_{\mathcal{R}(p,q)}\cosh_{\mathcal{R}(p,q)}(az)= - a\;
\sinh_{\mathcal{R}(p,q)}(az),
\end{equation}
\begin{equation}
\partial_{\mathcal{R}(p,q)}\sin_{\mathcal{R}(p,q)}(az) = a\;
\cos_{\mathcal{R}(p,q)}(az) \quad \mbox{and}
\end{equation}
\begin{equation}
\partial_{\mathcal{R}(p,q)}\cos_{\mathcal{R}(p,q)}(az)= - a\;
\sin_{\mathcal{R}(p,q)}(az), \quad a\in\mathbb{C}.
\end{equation}
Therefore, the ${\mathcal{R}(p,q)}-$oscillator equation
\begin{equation}
\partial^2_{\mathcal{R}(p,q)}f(z) + \omega^2f(z)=0
\end{equation}
can be solved to give the solution
\begin{equation}
f(z) = C_1 \cos_{\mathcal{R}(p,q)}(\omega z) + C_2
\sin_{\mathcal{R}(p,q)}(\omega z).
\end{equation}

\subsection{Modified $(p,q)$-Bessel functions}\label{Sect37}
The $(p,q)-$analogues of the $q-$Bessel functions ~\cite{IsmailMoudard05,Jackson1905}
\begin{eqnarray}
&& J_s^{(1)}(z;q) = \frac{(q^{s+1};q)_{\infty}}{(q;q)_{\infty}}
\sum_{n=0}^{\infty}\frac{(-1)^n(\frac{z}{2})^{2n+s}}{(q,q^{s+1};q)_n}
\\&&
J_s^{(2)}(z;q) = \frac{(q^{s+1};q)_{\infty}}{(q;q)_{\infty}}
\sum_{n=0}^{\infty}\frac{q^{n(n-1)}(-1)^n(\frac{z}{2})^{2n+s}q^{n(s+1)}}{(q,
q^{s+1};q)_n}
\end{eqnarray}
can be defined  by
\begin{eqnarray}
&& J_s^{(1)}(z\vert p,q) = B(s+1\vert
p,q)
\sum_{n=0}^{\infty}\frac{p^{n(n-1)}(-1)^n(\frac{z}{2})^{2n+s}}{((p,q),(p^{s+
1},q^{s+1});(p,q))_n} \\
&& J_s^{(2)}(z\vert p,q) = B(s+1\vert p,q) \sum_{n=0}^{\infty}
\frac{q^{n(n-1)}(-1)^n(\frac{z}{2})^{2n+s}(\frac{q}{p})^{n(s+1)}}{((p,q),(p^
{s+1},q^{s+1});(p,q))_n},
\end{eqnarray}
where $0<q<p<1$, $z, s\in\mathbb{C}$ with $0<\vert z\vert<1$, and
\begin{equation}
B(s\vert p,q) = \frac{((p^s,q^s);(p,q))_{\infty}}{((p,q);(p,q))_{\infty}}.
\qquad\qquad\qquad\qquad
\end{equation}

\begin{remark}\
\begin{itemize}
\item $J^{(k)}_s(z;\frac{q}{p}) \neq J^{(k)}_s(z\vert p,q) $  for $ k
= 1, 2 $ since:
\begin{eqnarray*}
&& \frac{((p^s,q^s);(p,q))_n}{((p,q);(p,q))_n} =
\frac{p^{ns}((\frac{q}{p})^s;q/p)_n}{(\frac{q}{p};\frac{q}{p})_n}
\\&&
\sum_{n=0}^{\infty}\frac{(-1)^n(z/2)^{2n+s}}{(\frac{q}{p},(\frac{q}{p})^{s+1}
;\frac{q}{p})_n}\neq
\sum_{n=0}^{\infty}\frac{p^{n(n-1)}(-1)^n(z/2)^{2n+s}}{((p,q),(p^{s+1},q^{s+1})
;(p,q))_n} \\&&
\sum_{n=0}^{\infty}\frac{(\frac{q}{p})^{n(n-1)+s+1}(-1)^n(\frac{z}{2})^{2n+s}}
{(\frac{q}{p},(\frac{q}{p})^{s+1};\frac{q}{p})_n}\neq
\sum_{n=0}^{\infty}\frac{q^{n(n-1)}(-1)^n(\frac{z}{2})^{2n+s}}{((p,q),(q^{s+
1},q^{s+1});(p,q))_n}.
\end{eqnarray*}
\item For $p=1$,
\begin{equation}
J^{(k)}_s(z\vert p,q)= J^{(k)}_s(z;q) \quad \mbox{for}\quad  k = 1, 2.
\end{equation}
\end{itemize}
 So, the $(p,q)-$Bessel functions generalize the $q-$Bessel functions ~\cite{Koekoek}.
\end{remark}

\begin{proposition}:
The following relation between $J_s^{(1)}$ and $J_s^{(2)}$ holds:
\begin{eqnarray}
{\label{prop1}} J_s^{(2)}(z\vert p,q) =
\frac{((p^{s+2},\frac{z^2}{4});(p,q))_{\infty}}{((p^{s+2},0);(p,q))_{\infty}
} J_s^{(1)}(z\vert p,q).
\end{eqnarray}
\end{proposition}
\begin{proof}
\begin{eqnarray*}
&&\sum_{n=0}^{\infty}\frac{p^{n(n-1)}(-1)^n(\frac{z}{2})^{2n}}{((p,q),(p^{s+1
}, q^{s+1});(p,q))_n}
=\sum_{n=0}^{\infty}\frac{(-1)^n(\frac{z^2}{4p^{s+2}})^{n}}
{(\frac{q}{p},(\frac{q}{p})^{s+1};\frac{q}{p})_n} \\ && \qquad\qquad
={_2\phi_1}\left(0,0;\left(\frac{q}{p}\right)^{s+1};\frac{q}{p};-\frac{z^2}
{4p^{s+2}}\right) \\&& \qquad\qquad =
\frac{1}{(\frac{z^2}{4p^{s+2}};\frac{q}{p})_{\infty}}
{_0\phi_1}\left(-;\left(\frac{q}{p}\right)^{s+1};\frac{q}{p},-\frac{z^2q^{s+
1}}{4p^{2s+3}}\right) \\&& \qquad\qquad =
\lim_{n\to\infty}\left(\frac{1}{(\frac{z^2}{4p^{s+2}};\frac{q}{p})_n}\right
) \sum_{n=0}^{\infty}\frac{(-1)^n(\frac{z^2q^{s+1}}{4p^{2s+3}})^{n}[(-1)^n
(\frac{q}{p})^{n(n-1)/2}]^2}
{(\frac{q}{p},(\frac{q}{p})^{s+1};\frac{q}{p})_n} \\&&\qquad\qquad=
\frac{(p^{s+2},0);(p,q)_{\infty}}{(p^{s+2},\frac{z^2}{4});(p,q)_{\infty}}
\sum_{n=0}^{\infty}\frac{(-1)^nq^{n(n-1)}(\frac{z}{2})^{2n}(\frac{q^{s+1}}{p
^{s+1}})^{n}}{((p,q),(p^{s+1},q^{s+1});(p,q))_n}
\end{eqnarray*}
 using the following Heine's transformation ~\cite{Koekoek}:
\begin{equation}
{_2\phi_1}(a,b;c;q,z) =
\frac{1}{(z;q)_{\infty}}{_0\phi_1}(-;c,z;q,cz).
\end{equation}
Thus \begin{eqnarray} &&
\sum_{n=0}^{\infty}\frac{(-1)^nq^{n(n-1)}(\frac{z}{2})^{2n}(\frac{q^{s+1}}
{ p ^{s+1}})^{n}}{((p,q),(p^{s+1},q^{s+1});(p,q))_n}=  \nonumber \\
&&\qquad\qquad \qquad
\frac{(p^{s+2},\frac{z^2}{4});(p,q)_{\infty}}{(p^{s+2},0);(p,q)_{\infty}}
\sum_{n=0}^{\infty}\frac{p^{n(n-1)}(-1)^n(\frac{z}{2})^{2n}}{((p,q),(p^{s+1}
,q^{s+1});(p,q))_n}.
\end{eqnarray}
Multiplying both sides of this equality by $B(s+1\vert
p,q)(\frac{z}{2})^2$ leads to (\ref{prop1}).\cqfd
\end{proof}

\begin{proposition}
The following three-term recursion relation:
\begin{eqnarray}
&& \qquad
\left[(p^{\frac{s}{2}}P^{\frac{1}{2}}-q^{\frac{s}{2}}Q^{\frac{1}{2}})(p^{\frac{s}{2}+1}
P^{\frac{1}{2}}-q^{\frac{s}{2}+1}Q^{\frac{1}{2}})+\frac{z^2}{4}
\right]J_s^{(1)}(z\vert p,q)= \nonumber
\\&&
{\label{prop2}} \qquad
\frac{z}{2}(p^{\frac{s+3}{2}}-q^{\frac{s+3}{2}})(p^{\frac{s}{2}+1}P^{\frac{1
}{2}}-q^{\frac{s}{2}+1}Q^{\frac{1}{2}})\left[J_{s+1}^{(1)}(z\vert
p,q)+ J_{s-1}^{(1)}(z\vert p,q)\right]
\end{eqnarray}
is satisfied for the $(p,q)-$Bessel function $J_s^{(1)}(z\vert p,q)$.
\end{proposition}

\begin{proof}: We have
\begin{eqnarray}
&& \left(p^{\frac{s}{2}}P^{\frac{1}{2}}-q^{\frac{s}{2}}Q^{\frac{1}{2}}\right)
J ^ {(1)}_{s}(z\vert p,q)\qquad \nonumber \\ &&\qquad \qquad\qquad
=\frac{z}{2}B(s+1\vert p,q)
\sum_{n=0}^{\infty}\frac{p^{n(n-1)}(-1)^n(\frac{z}{2})^{2n+s}(p^{s+n}-q^{s+n
} )}{((p,q),(p^{s+1},q^{s+1});(p,q))_n}\qquad
\nonumber\\ &&\qquad \qquad\qquad
{\label{prop3}}
=\frac{z}{2}\left(\frac{z}{2}\right)J^{(1)}_{s-1}(z\vert p,q)
\end{eqnarray}
and
\begin{eqnarray}
&&\left(p^{\frac{s}{2}+1}P^{\frac{1}{2}}-q^{\frac{s}{2}+1}Q^{\frac{1}{2}}
\right)^{-1}J^{(1)}_{s}(z\vert p,q)=\frac{1}{2}B(s+1\vert p,q)
\nonumber\\ &&  \qquad \qquad\qquad
\times\sum_{n=0}^{\infty}\frac{p^{n(n-1)}(-1)^n(\frac{z}{2})^{2n+s}}{((p,q),
( p ^ { s +1 } ,q^{s+1});(p,q))_n(p^{s+n+1}-q^{s+n+1})^{-1}}
\nonumber\qquad \\ &&\qquad \qquad\qquad\qquad \qquad \qquad\qquad\qquad
{\label{prop4}} = \frac{z}{2}J^{(1)}_{s+1}(z\vert p,q).
\end{eqnarray}
Adding (\ref{prop3}) and (\ref{prop4}), we obtain (\ref{prop2}).
\cqfd
\end{proof}
%%%%%%%%%%%%%%%%%%%%%%%%%%%%%%%%%%%%%%%%%%%%%%%%%%%%%%%%%%%%%%%%%%%%%%%%%%%%

\section[${\cal R}(p,q)$-calculus]{${\cal R}(p,q)$-calculus: differentiation, integration and Hopf algebras}\label{chap4}

$\;$
% \begin{abstract}

In this section  we build a framework for ${\cal R}(p,q)$-deformed calculus, which provides a method of computation for a deformed 
${\cal R}(p,q)$-derivative,  generalizing known deformed derivatives of analytic function
defined on a complex disc as particular cases corresponding to  conveniently chosen  meromorphic functions.
Under prescribed conditions, we define the ${\cal R}(p,q)$-derivative.  
The main result resides in the proof that ${\cal R}(p,q)$-algebra is a Hopf algebra. 
Relevant examples are also given.
% \end{abstract}
% 
% The chapter is organized as follows. We first recall in Section \ref{2section2} the definition of the ${\cal R}(p,q)$-factors 
% and their associated quantum algebra.
% We introduce a new  algebra generated by four quantities provided some conditions are satisfied.
% The Section \ref{2section3} is devoted to the definition of the ${\cal R}(p,q)-$differential calculus  yielding
% the ${\cal R}(p,q)-$integration.   
% A  construction of Hopf algebra structure is given in Section \ref{2section4}, while in
% Section \ref{2section5} we show that some particular cases can be deduced from the constructed general formalism.
% Then follow some concluding remarks in Section \ref{2section6}.

\subsection{${\cal R}(p,q)-$factors and their associated quantum algebras}\label{2section2}

In the previous chapter (see also \cite{Bukweli&Hounkonnou12a,Bukweli&Hounkonnou12d,Hounkonnou&Bukweli10}) we have built the ${\cal R}(p,q)-$factors
which are a generalization of Heine $q-$factors (also called Heine $q-$number in physics literature)
\begin{eqnarray}
[n]_q=\frac{1-q^n}{1-q}, \qquad n= 0, 1, 2, \cdots
\end{eqnarray}
and Jagannathan-Srinivasa $(p,q)-$factors \cite{Jagannathan&Rao}
\begin{eqnarray}\label{2pqnbrs}
[n]_{p,q}= \frac{p^n-q^n}{p-q}, \qquad n= 0, 1, 2, \cdots
\end{eqnarray}
as follows.
Let $p$ and $q$ be two
positive real numbers  such that $0<q<p\leq1$. Consider, as in the previous chapter, a meromorphic function
${\cal R}$, defined on $\mathbb{C}\times\mathbb{C}$ by
\begin{equation}\label{2Rxy}
{\cal R}(x,y) = \sum_{k, l = - L}^{\infty} r_{kl}x^ky^l
\end{equation}
with an eventual isolated singularity at the zero, where $r_{kl}$
are complex numbers, $L\in\mathbb{N}\cup\{0\}$, ${\cal R}(p^n, q^n)>
0$ $\forall n\in\mathbb{N}$, and ${\cal R}(1, 1) = 0$ by definition.
Then, the  ${\cal R}(p,q)$-factors denoted by $\displaystyle {\cal R}(p^n,q^n)$, $n= 0, 1, 2, \cdots$
%\begin{eqnarray}
%[n]_{{\cal R}(p,q)}= {\cal R}(p^n,q^n),\qquad n= 0, 1, 2, \cdots
%\end{eqnarray}
are used to deduce the ${\cal R}(p,q)-$factorial
\begin{eqnarray}
{{\cal R}!(p^n,q^n)}= \left\{\begin{array}{l} 1 \quad \mbox{for   } \quad n=0 \quad \\
{\cal R}(p,q)\cdots{\cal R}(p^n,q^n) \quad \mbox{for } \quad n\geq
1, \quad \end{array} \right.
\end{eqnarray}
the  ${\cal R}(p,q)-$binomial coefficient
\begin{eqnarray}
&&\left[\begin{array}{c} m  \\ n\end{array} \right]_{{\cal R}(p,q)}=
\frac{{\cal R}!(p^m,q^m)}{{\cal R}!(p^n,q^n){\cal R}!(p^{m-n},q^{m-n})},\quad m, n= 0, 1, 2, \cdots,\quad m\geq n,\cr
&&
\end{eqnarray}
and the  ${\cal R}(p,q)$-exponential function
\begin{eqnarray}\label{2dino}
 \mathrm{Exp}_{{\cal R}(p,q)}(z) =
\sum_{n=0}^{\infty}\frac{1}{{{\cal R}!(p^n,q^n)}}z^n.
\end{eqnarray}
 Denote by $\mathbb{D}_R$ =$\{z\in\mathbb{C}$ : $\vert z\vert < R \}$ a complex disc  and by
${\cal O}(\mathbb{D}_R)$ the set of holomorphic functions defined
on $\mathbb{D}_R$, where $R$ is the radius of convergence of the series (\ref{2dino}). 
 
 We then define the following linear operators   on ${\cal O}(\mathbb{D}_R)$ by
(see \cite{Bukweli&Hounkonnou12a,Hounkonnou&Bukweli10} and references therein):
\begin{eqnarray}\label{2operat}
 &&\quad Q: \varphi \longmapsto Q\varphi(z) = \varphi(qz),
\nonumber \\
 &&\quad P: \varphi \longmapsto P\varphi(z) = \varphi(pz), \nonumber\\
&&{\label{2deriva}} \partial_{p,q}:\varphi \longmapsto
\partial_{p,q}\varphi(z) = \frac{\varphi(pz) - \varphi(qz)}{z(p-q)},
\end{eqnarray}
$\varphi\in{\cal O}(\mathbb{D}_R)$,
$0<q<p\leq 1$, and the ${\cal R}(p,q)$-derivative  by
\begin{equation}{\label{2deriva1}}
\partial_{{\cal R}(p,q)} := \partial_{p,q}\frac{p - q}{P-Q}{\cal R}(P, Q)
= \frac{p - q}{pP-qQ}{\cal R}(pP, qQ)\partial_{p,q}.
\end{equation}
Note that the  ${\cal R}(p,q)$-exponential function is invariant under the action
of the  ${\cal R}(p,q)$-derivative since
\begin{eqnarray}
\partial_{{\cal R}(p,q)}z^n=\left\{\begin{array}{l} 0 \quad \mbox{for   } \quad n=0 \quad \\
{\cal R}(p^n,q^n)z^{n-1} \quad \mbox{for } \quad n\geq
1. \quad \end{array} \right.
\end{eqnarray}
In \cite{Hounkonnou&Bukweli10}, we also studied
the ${\cal R}(p,q)-$deformed quantum algebra generated by the set of operators
$\{1, A, A^\dag, N\}$ and the commutation relations
\begin{equation}\label{2algN1}
\left[N, A\right] = - A \;\;\;\mbox{and}\;\;\; \left[N,
A^\dag\right] = A^\dag \end{equation}
with
\begin{equation}\label{2algN2}
A A^\dag= {\cal R}(p^{N+1},q^{N+1}),\quad \mbox{and}\quad A^\dag  A = {\cal
R}(p^N,q^N).
\end{equation}
 This algebra is defined  on ${\cal O}(\mathbb{D}_R)$ as:
\begin{eqnarray}
A^\dag := z,\qquad A:=\partial_{{\cal R}(p,q)}, \qquad N:= z\partial_z,
\end{eqnarray}
where $\partial_z:=\frac{\partial}{\partial z}$ is the usual derivative on $\mathbb{C}$.
Therefore, the following  holds:
\begin{proposition}
\begin{eqnarray}
P= p^{z\partial z},\qquad Q=q^{z\partial z}
\end{eqnarray}
and the algebra ${\cal A}_{{\cal R}(p,q)}$ generated by $\{1,\; z,\; z\partial_z,\;\partial_{{\cal R}(p,q)}\}$
satisfies the  relations:
\begin{eqnarray}\label{2ComuRel2}
&&z\;\partial_{{\cal R}(p,q)}= {\cal R}(P, Q),\qquad
\partial_{{\cal R}(p,q)}\;z ={\cal R}(pP, qQ),\cr
&&[z\partial_z,\;z]= z,\qquad\qquad [z\partial_z,\;\partial_{{\cal R}(p,q)}]= -\partial_{{\cal R}(p,q)}.
\end{eqnarray}
\end{proposition}
\begin{proposition}
If there exist two functions
$\Psi_1$ and $\Psi_2: $ $\mathbb{C}\times\mathbb{C}\longrightarrow \mathbb{C}$ such that
\begin{eqnarray}
&&\Psi_i(p,q)>0\quad\mbox{for   }\; i=1,\;2\\
&&\left[\begin{array}{c} n+1 \\ k\end{array} \right]_{{\cal R}(p,q)}= \Psi_1^k(p,q)\left[\begin{array}{c} n  \\ k\end{array} \right]_{{\cal R}(p,q)}
+ \Psi_2^{n+1-k}(p,q)\left[\begin{array}{c} n  \\ k-1\end{array} \right]_{{\cal R}(p,q)},
\end{eqnarray}
\begin{eqnarray}
&& ba= \Psi_1(p,q)ab,\; xy=\Psi_2(p,q)yx,\mbox{ and }
[i,\;j]=0 \mbox{ for } i{\in}\{a,b\}, j{\in}\{x,y\}\cr
&&
\end{eqnarray}
for quantities  $a$, $b$, $x$,  $y$, then 
\begin{eqnarray}\label{2genbin}
 (ax+by)^n= \sum_{k=0}^n \left[\begin{array}{c} n \\ k\end{array} \right]_{{\cal R}(p,q)}a^{n-k}b^ky^kx^{n-k}.
\end{eqnarray}

\end{proposition}
{\bf Proof:} By induction over $n$. Indeed, the equality (\ref{2genbin}) holds for $n=1$ since
\begin{eqnarray*}
 (ax+by)^1 &=& ax+by = \left[\begin{array}{c} 1  \\ 0\end{array} \right]_{{\cal R}(p,q)}a^1b^0y^0x^2 + \left[\begin{array}{c} 1  \\ 1\end{array} \right]_{{\cal R}(p,q)}a^0b^1y^1x^0 \cr&=&
 \sum_{k=0}^1\left[\begin{array}{c} 1  \\ k\end{array} \right]_{{\cal R}(p,q)}a^{1-k}b^ky^kx^{1-k}.
\end{eqnarray*}
Suppose that the equality (\ref{2genbin}) holds for $n\leq m$, this means in particular for $n=m$,
\begin{eqnarray}
 (ax+by)^m= \sum_{k=0}^m \left[\begin{array}{c} m \\ k\end{array} \right]_{{\cal R}(p,q)}a^{m-k}b^ky^kx^{m-k},
\end{eqnarray}
and let us prove that it remains valid for $n=m+1$. Indeed,
\begin{eqnarray*}
&& (ax+by)^{m+1}=(ax+by)^m(ax+by) \cr&&\quad= \sum_{k=0}^m \left[\begin{array}{c} m \\ k\end{array} \right]_{{\cal R}(p,q)}a^{m-k}b^ky^kx^{m-k}(ax+by)
 \cr&&\quad=\sum_{k=0}^m \left[\begin{array}{c} m \\ k\end{array} \right]_{{\cal R}(p,q)}a^{m-k}b^ky^kx^{m-k}ax+
\sum_{k=0}^m \left[\begin{array}{c} m \\ k\end{array} \right]_{{\cal R}(p,q)}a^{m-k}b^ky^kx^{m-k}by
 \cr&&\quad=\sum_{k=0}^m \left[\begin{array}{c} m \\ k\end{array} \right]_{{\cal R}(p,q)}\Psi_1^k(p,q)a^{m+1-k}b^ky^kx^{m+1-k}
 \cr&&\quad+\sum_{k=0}^m \left[\begin{array}{c} m \\ k\end{array} \right]_{{\cal R}(p,q)}\Psi_2^{m-k}(p,q)a^{m-k}b^{k+1}y^{k+1}x^{m-k} \cr&&\quad=a^{m+1}x^{m+1}+\sum_{k=1}^m \left[\begin{array}{c} m \\ k\end{array} \right]_{{\cal R}(p,q)}\Psi_1^k(p,q)a^{m+1-k}b^ky^kx^{m+1-k}
 \cr&&\quad+\sum_{k=0}^{m-1} \left[\begin{array}{c} m \\ k\end{array} \right]_{{\cal R}(p,q)}\Psi_2^{m-k}(p,q)a^{m-k}b^{k+1}y^{k+1}x^{m-k}+ b^{m+1}y^{m+1}
 \cr&&\quad=a^{m+1}x^{m+1}+\sum_{k=1}^m \Psi_1^k(p,q)\left[\begin{array}{c} m \\ k\end{array} \right]_{{\cal R}(p,q)}a^{m+1-k}b^ky^kx^{m+1-k}
 \cr&&\quad+\sum_{k=1}^{m} \Psi_2^{m+1-k}(p,q)\left[\begin{array}{c} m \\ k-1\end{array} \right]_{{\cal R}(p,q)}a^{m+1-k}b^{k}y^{k}x^{m+1-k}+ b^{m+1}y^{m+1}
 \cr&&\quad=a^{m+1}x^{m+1}+ b^{m+1}y^{m+1}
+\sum_{k=1}^{m}\left(\Psi_1^k(p,q)\left[\begin{array}{c} m \\ k\end{array} \right]_{{\cal R}(p,q)}\right. \cr&&\qquad\qquad+\left.
\Psi_2^{m+1-k}(p,q)\left[\begin{array}{c} m \\ k-1\end{array} \right]_{{\cal R}(p,q)}\right)a^{m+1-k}b^{k}y^{k}x^{m+1-k}
 \cr&&\quad=a^{m+1}x^{m+1}+ b^{m+1}y^{m+1}+\sum_{k=1}^{m}\left[\begin{array}{c} m \\ k\end{array} \right]_{{\cal R}(p,q)}a^{m+1-k}b^{k}y^{k}x^{m+1-k}.
\end{eqnarray*}
\hfill$\Box$

\subsection{${\cal R}(p,q)-$differential and integration calculi}\label{2section3}
\subsubsection{Differential calculus}
We define a linear operator $d_{{\cal R}(p,q)}$ on ${\cal A}_{{\cal R}(p,q)}$ by
\begin{equation}\label{2diff1}
d_{{\cal R}(p,q)}= (dz)\partial_{{\cal R}(p,q)}.
\end{equation}
It follows that
\begin{eqnarray}
&& d_{{\cal R}(p,q)}1= 0,\quad d_{{\cal R}(p,q)}z = (dz){\cal R}(p,q),
\quad d_{{\cal R}(p,q)}\partial_{{\cal R}(p,q)} = (dz)\partial_{{\cal R}(p,q)}^2\nonumber\\
&& d_{{\cal R}(p,q)}(z\partial_z)= (dz)(z\partial_z+1)\partial_{{\cal R}(p,q)}
\qquad\mbox{and}\quad d_{{\cal R}(p,q)}^2 = 0.
\end{eqnarray}
Hence, the set of  zero-forms $\Omega^0({\cal A}_{{\cal R}(p,q)})$ is naturally ${\cal A}_{{\cal R}(p,q)}$,
 while a one form $\omega$, element of $\Omega^1({\cal A}_{{\cal R}(p,q)})$, is given by
\begin{eqnarray}
\omega= (dz)\omega_0(z,z\partial_z,\partial_{{\cal R}(p,q)}),
\end{eqnarray}
where
$\omega_0(z,z\partial_z,\partial_{{\cal R}(p,q)})=
\sum_{i,j,k=0}^\infty \alpha_{ijk}(z)^i(z\partial_z)^j(\partial_{{\cal R}(p,q)})^k$
with $\alpha_{ijk}$ belonging to $\mathbb{C}$. Therefore, $d\omega=0$ for
$\omega\in\Omega^1({\cal A}_{{\cal R}(p,q)})$.

\begin{proposition}
For a nonnegative integer $n$, the following equalities hold:
\begin{eqnarray}\label{2diff20}
& & d_{{\cal R}(p,q)}(z^n)= (dz){\cal R}(p^n,q^n)z^{n-1},\cr
&& d_{{\cal R}(p,q)}(z\partial_z)^n= (dz)(z\partial_z+1)^n\partial_{{\cal R}(p,q)},\\
&& d_{{\cal R}(p,q)}(\partial^n_{{\cal R}(p,q)})= (dz)\partial^{n+1}_{{\cal R}(p,q)}.\nonumber
\end{eqnarray}
Moreover if $f\in{\cal O}(\mathbb{D}_R)$ then
\begin{eqnarray}\label{2diff2}
 d_{{\cal R}(p,q)}f(z)= (dz)\partial_{{\cal R}(p,q)} f(z).
\end{eqnarray}
\end{proposition}
{\bf Proof:} The equalities in (\ref{2diff20}) follow from the definition of the
${\cal R}(p,q)-$ derivative  (\ref{2deriva1}), the commutation relations (\ref{2ComuRel2}) and
the definition of the differential (\ref{2diff1}).
Then, (\ref{2diff2}) follows by definition (\ref{2deriva1}).\hfill$\Box$
%
%\begin{eqnarray*}
% d_{{\cal R}(p,q)}f(z)&=& d_{{\cal R}(p,q)}\sum_{n=0}^{\infty}a_nz^n = \sum_{n=0}^{\infty}a_nd_{{\cal R}(p,q)}z^n
%=\sum_{n=0}^{\infty}a_n(dz)\partial_{{\cal R}(p,q)}z^n
%\cr&=&(dz)\partial_{{\cal R}(p,q)}\sum_{n=0}^{\infty}a_nz^n
%=(dz)\partial_{{\cal R}(p,q)}f(z).
%\end{eqnarray*}

\begin{proposition}
The differential $d_{{\cal R}(p,q)}$ obeys the two following equivalent Leibniz rules
\begin{equation}\label{2leibniz1}
 d_{{\cal R}(p,q)}(fg)= (dz)\frac{p - q}{pP-qQ}{\cal R}(pP, qQ)
\left\{\partial_{p,q}(f))(Pg)+(Qf)(\partial_{p,q}(g))\right\},
\end{equation}
\begin{equation}\label{2leibniz2}
d_{{\cal R}(p,q)}(fg) = (dz)\frac{p - q}{pP-qQ}{\cal R}(pP, qQ)
\left\{(\partial_{p,q}(f))(Qg)+(Pf)(\partial_{p,q}(g))\right\}.
\end{equation}
for $f,g\in{\cal O}(\mathbb{D}_R)$.
\end{proposition}
{\bf Proof:}
This follows from
\begin{eqnarray*}
\cr &&\partial_{p,q}(fg)=(\partial_{p,q}(f))(Qg)+(Pf)(\partial_{p,q}(f)
=(\partial_{p,q}(f)(Pg)+(Qf)(\partial_{p,q}(g)).
\end{eqnarray*}
\hfill$\Box$

\subsubsection{${\cal R}(p,q)$-integration}
We define the operator  ${\cal I}_{{\cal R}(p,q)}$ over ${\cal O}(\mathbb{D}_R)$ as the inverse image of the ${\cal R}(p,q)$-derivative.
For elements $z^n$ of the basis of ${\cal O}(\mathbb{D}_R)$, ${\cal I}_{{\cal R}(p,q)}$ acts as follows:
\begin{eqnarray}
 {\cal I}_{{\cal R}(p,q)} z^n:=\left(\partial_{{\cal R}(p,q)}\right)^{-1}z^n
  = \frac{1}{{\cal R}(p^{n+1},q^{n+1})}z^{n+1}+ c,
\end{eqnarray}
where $n\geq 0 $ and $c$ is an integration constant.\\
Hence, if $f\in{\cal O}(\mathbb{D}_R)$ then
\begin{eqnarray}\label{2integra1}
 &&{\cal I}_{{\cal R}(p,q)}\;\partial_{{\cal R}(p,q)}f(z)= f(z)+c\;\; \mbox{  and   }\;\;
 \partial_{{\cal R}(p,q)}\;{\cal I}_{{\cal R}(p,q)}f(z)= f(z)+c',
\end{eqnarray}
where $c$ and $c'$ are integration constants.

Provided that ${\cal R}(P,Q)$ is invertible, one can define the ${\cal R}(p,q)$-integration
by the following formula
\begin{eqnarray}\label{2integra2}
 {\cal I}_{{\cal R}(p,q)}= {\cal R}^{-1}(P,Q)\;z,
\end{eqnarray}
with $c=c'=0$.

One can also derive the definite integrals:
\begin{eqnarray}
&&\int_\alpha^\beta f(z)d_{{\cal R}(p,q)}z={\cal I}_{{\cal R}(p,q)}(\beta)-{\cal I}_{{\cal R}(p,q)}f(\alpha),\quad
\alpha,\beta\in\mathbb{D}_R;\\
&&\int_\alpha^{+\infty} f(z)d_{{\cal R}(p,q)}z=\lim_{n\to\infty}\int_\alpha^{p^n/q^n} f(z)d_{{\cal R}(p,q)}z;\\
&&\int_{-\infty}^{+\infty} f(z)d_{{\cal R}(p,q)}z=\lim_{n\to\infty}\int_{-p^n/q^n}^{p^n/q^n} f(z)d_{{\cal R}(p,q)}z.
\end{eqnarray}
Moreover, the Eqs. (\ref{2leibniz1}) and (\ref{2leibniz2}) lead to the following formulae:
\begin{eqnarray}
&&{\cal I}_{{\cal R}(p,q)}\;\partial_{{\cal R}(p,q)}(f(z)g(z))= f(z)g(z)+c
\\&&\qquad\qquad\qquad\qquad\qquad= {\cal I}_{{\cal R}(p,q)}\left\{ \frac{p - q}{pP-qQ}{\cal R}(pP, qQ)
(\partial_{p,q}(f))(Pg)\right\}
\cr&&\qquad\qquad\qquad\qquad\qquad+{\cal I}_{{\cal R}(p,q)}\;\frac{p - q}{pP-qQ}{\cal R}(pP, qQ)
\left\{(Qf)(\partial_{p,q}(g))\right\}\nonumber
\end{eqnarray}
and
\begin{eqnarray}
&&{\cal I}_{{\cal R}(p,q)}\; \partial_{{\cal R}(p,q)}(f(z)g(z))= f(z)g(z)+c
\\&&\qquad\qquad\qquad\qquad\qquad={\cal I}_{{\cal R}(p,q)}\left\{\frac{p - q}{pP-qQ}{\cal R}(pP, qQ)
(\partial_{p,q}(f))(Qg)\right\}
\cr&&\qquad\qquad\qquad\qquad\qquad+{\cal I}_{{\cal R}(p,q)}\;\frac{p - q}{pP-qQ}{\cal R}(pP, qQ)
\left\{(Pf)(\partial_{p,q}(g))\right\},\nonumber
\end{eqnarray}
respectively. These relation can be viewed as formulae of integration by parts.

\subsection{${\cal R}(p,q)$-Hopf algebra}\label{2section4}
The aim of this section is to establish conditions for which the ${\cal R}(p,q)$-algebra carries a Hopf algebra structure. This is summarized in a theorem given below.
Remind first that the algebra ${\cal A}_{{\cal R}(p,q)}$  will be a Hopf algebra if it admits operations of homomorphisms of a co-product $\mathbf{\Delta}$, a counit $\mathbf{\epsilon}$ and an anti-homomorphism of an antipode $\mathbf{S}$ \cite{Abe77}:
\begin{eqnarray}
&&\mathbf{\Delta}:\xymatrix{{\cal A}_{{\cal R}(p,q)}\ar[r]&{\cal A}_{{\cal R}(p,q)}\otimes{\cal A}_{{\cal R}(p,q)}},\quad\mathbf{\Delta}(\Omega_1\Omega_2)=
\mathbf{\Delta}(\Omega_1)\mathbf{\Delta}(\Omega_2); \cr
&&\mathbf{\epsilon}:\xymatrix{{\cal A}_{{\cal R}(p,q)}\ar[r]&\mathbb{C}},\quad\mathbf{\epsilon}(\Omega_1\Omega_2)=
\mathbf{\epsilon}(\Omega_1)\mathbf{\epsilon}(\Omega_2); \\
&&\mathbf{S}:\xymatrix{{\cal A}_{{\cal R}(p,q)}\ar[r]&{\cal A}_{{\cal R}(p,q)}},\quad\mathbf{S}(\Omega_1\Omega_2)=\mathbf{S}(\Omega_2)\mathbf{S}(\Omega_1)\nonumber
\end{eqnarray}
 satisfying
\begin{eqnarray}
&&(\mbox{id}\otimes\mathbf{\Delta})\mathbf{\Delta}(\Omega)= (\mathbf{\Delta}\otimes\mbox{id})\mathbf{\Delta}(\Omega),\label{2Hopf1}\\\
&&(\mbox{id}\otimes\mathbf{\epsilon})\mathbf{\Delta}(\Omega)=\Omega= (\mathbf{\epsilon}\otimes\mbox{id})\mathbf{\Delta}(\Omega),\label{2Hopf2}\\
&&m((\mbox{id}\otimes\mathbf{S})\mathbf{\Delta}(\Omega))=m((\mathbf{S}\otimes\mbox{id})\mathbf{\Delta}(\Omega))
=\mathbf{\epsilon}(\Omega)\mathbf{1},\label{2Hopf3}
\end{eqnarray}
for all $\Omega, \Omega_1, \Omega_2\in{\cal A}_{{\cal R}(p,q)}$. To prove this it is sufficient to show that these relations are satisfied by the generators governing the considered algebra. See \cite{Chari&Pressley} and references therein.

Let the Leibniz rule be written as
\begin{eqnarray}
(\partial_{{\cal R}(p,q)}\;fg)(z)&=& (\partial_{{\cal R},p,}f(z))(\Psi(p^{z\partial_z},q^{z\partial_z})g(z) \\&&\quad +
(\tilde{\Psi}(p^{z\partial_z},q^{z\partial_z})f(z))\partial_{{\cal R}(p,q)}g(z),\nonumber
\end{eqnarray}
for $ f, g\in{\cal O}(\mathbb{D}_R)$, where $\Psi(.,.)$ and $\tilde{\Psi}(.,.)$ are meromorphic functions.
Let the coproduct $\mathbf{\Delta}$, the counit $\mathbf{\epsilon}$, and the antipode $\mathbf{S}$ be defined as follows:
\begin{eqnarray}
&&\mathbf{\Delta}(A)= \alpha A\otimes\Psi(p^{\alpha_1N},q^{\alpha_2N})+
 \tilde{\alpha}\tilde\Psi(p^{\tilde\alpha_1N},q^{\tilde\alpha_2N})\otimes A,\label{2RHop1}\\
&&\mathbf{\Delta}(A^\dag)= \beta A^\dagger\otimes\Psi(p^{\beta_1N},q^{\beta_2N})+ \tilde{\beta}\tilde\Psi(p^{\tilde\beta_1N},q^{\tilde\beta_2N})\otimes A^\dag,\label{2RHop2}\\
&&\mathbf{\Delta}(N)= N\otimes \mathbf{1} + \mathbf{1}\otimes N+ \tau \mathbf{1}\otimes\mathbf{1},\label{2RHop3} \\
&&\mathbf{\Delta(1)}=\mathbf{1}\otimes\mathbf{1},\label{2RHop4}\\
&&\mathbf{\epsilon}(A)= 0,\quad
\mathbf{\epsilon}(A^\dag)= 0,\quad
\mathbf{\epsilon}(N)= -\tau,\quad
\mathbf{\epsilon(1)}= 1\label{2RHop5}\\
&&\mathbf{S}(A)=- s_1A,\quad
\mathbf{S}(A^\dag)= -\tilde{s}_1A^\dag,\quad
\mathbf{S}(N)= -N -2\tau\mathbf{1},\quad
\mathbf{S(1)}= \mathbf{1}\label{2RHop6},
\end{eqnarray}
where  $\alpha_i$, $\tilde\alpha_i$, $\beta_i$, $\tilde\beta_i$ ($i=1,2$) and $s_1$, $\tilde{s}_1$  $\alpha$, $\tilde{\alpha}$,
 $\beta$, $\tilde{\beta}$ and $\tau$ are real constants such that the following equations hold:
\begin{eqnarray}
&& \alpha \Psi(p^{-\tau\alpha_1},q^{-\tau\alpha_2})=1, \qquad \tilde{\alpha}\tilde\Psi(p^{-\tau\tilde\alpha_1},q^{-\tau\tilde\alpha_2})=1,\label{2CondHop1}\\
&&b\Psi(p^{-\tau\beta_1},q^{-\tau\beta_2})=1,\qquad
\tilde{\beta}\tilde\Psi(p^{-\tau\tilde\beta_1},q^{-\tau\tilde\beta_2})=1,\label{2CondHop2}\\
&&\alpha\beta=1,\quad\quad \tilde{\alpha}\tilde{\beta}=1,\label{2CondHop3}\\
&&\mathbf{\Delta}(\Psi(p^{\alpha_1N},q^{\alpha_2N}))=
\alpha\Psi(p^{\alpha_1N},q^{\alpha_2N})\otimes\Psi(p^{\alpha_1N},q^{\alpha_2N}),\label{2CondHop4}\\
&&\mathbf{\Delta}(\tilde\Psi(p^{\tilde\alpha_1N},q^{\tilde\alpha_2N}))
= \tilde{\alpha} \tilde\Psi(p^{\tilde\alpha_1N},q^{\tilde\alpha_2N})\otimes\tilde\Psi(p^{\tilde\alpha_1N},q^{\tilde\alpha_2N})\label{2CondHop5}\\
&&\mathbf{\Delta}(\Psi(p^{\beta_1N},q^{\beta_2N}))=\beta\Psi(p^{\beta_1N},q^{\beta_2N})\otimes\Psi(p^{\beta_1N},q^{\beta_2N}),\label{2CondHop6}\\
&&\mathbf{\Delta}(\tilde\Psi(p^{\tilde\beta_1N},q^{\tilde\beta_2N}))
=\tilde{\beta}\tilde\Psi(p^{\tilde\beta_1N},q^{\tilde\beta_2N})\otimes\tilde\Psi(p^{\tilde\beta_1N},q^{\tilde\beta_2N}),
\label{2CondHop7}\\
&&s_1 \alpha\Psi(p^{\alpha_1(N+1)},q^{\alpha_2(N+1)})= \tilde{\alpha}\tilde\Psi(p^{-\tilde\alpha_1(N+2\tau)},q^{-\tilde\alpha_2(N+2\tau)}),\label{2CondHop8}\\
&&\alpha\Psi(p^{-\alpha_1(N+2\tau)},q^{-\alpha_2(N+2\tau)})= s_1\tilde{\alpha}\tilde\Psi(p^{\tilde\alpha_1(N-1)},q^{\tilde\alpha_2(N-1)}),\label{2CondHop9}\\
&&\beta\Psi(p^{-\beta_1(N+2\tau)},q^{-\beta_2(N+2\tau)})=
\tilde{s}_1\tilde{\beta}\tilde\Psi(p^{\tilde\beta_1(N+1)},q^{\tilde\beta_2(N+1)}),\label{2CondHop10}\\
&&\tilde{s}_1\beta\Psi(p^{-\beta_1(N-1)},q^{-\beta_2(N-1)})=
\tilde{\beta}\tilde\Psi(p^{\tilde\beta_1(N+2\tau)},q^{\tilde\beta_2(N+2\tau)}).\label{2CondHop11}.
\end{eqnarray}
Then the following main statement is true:
\begin{theorem}\label{2Hoptheo}
The more general deformed algebra generated by $\{\mathbf{1}, N, A, A^\dag\}$ satisfying:
\begin{eqnarray}
 [N,\;A]= -A,\quad [N,\;A^\dag]=A^\dag \quad [A,\;A^\dag]_\gamma= AA^\dag + \gamma A^\dag A,
\end{eqnarray}
where $\gamma$ is a real constant such that
\begin{eqnarray}
&&\tilde\Psi(p^{\tilde\beta_1N},q^{\tilde\beta_2N})\otimes\Psi(p^{\alpha_1N},q^{\alpha_2N})
=
\\&&\qquad\qquad\qquad\qquad \gamma\tilde\Psi(p^{\tilde\beta_1(N-1)},q^{\tilde\beta_2(N-1)})\otimes \Psi(p^{\alpha_1(N-1)},q^{\alpha_2(N-1)}),\label{2CondHop12}
\cr&&\cr
&&\tilde\Psi(p^{\tilde\alpha_1(N+1)},q^{\tilde\alpha_2(N+1)})\otimes\Psi(p^{\beta_1(N+1)},q^{\beta_2(N+1)})
=\\&&\qquad\qquad\qquad\qquad
\gamma \tilde\Psi(p^{\tilde\alpha_1N},q^{\tilde\alpha_2N})\otimes\Psi(p^{\beta_1N},q^{\beta_2N}),\label{2CondHop13}\qquad\nonumber
\end{eqnarray}
is a Hopf algebra.
\end{theorem}
\begin{flushleft}{\bf Proof:} Notice first that
\end{flushleft}
\begin{eqnarray}
&& A\theta^{\lambda N}= \theta^{\lambda(N+1)}A,\quad \theta^{\lambda(N)}A=A\theta^{\lambda(N-1)},\\
&&A^\dag\theta^{\lambda(N)}=A^\dag\theta^{\lambda(N-1)},\quad \theta^{\lambda(N)}A^\dag= A^\dag\theta^{\lambda(N+1)},\nonumber
\end{eqnarray}
where $\theta= p,\;q$, so that, for $\lambda = \alpha$,$\beta$ and $\tilde\lambda = \tilde\alpha$, $\tilde\beta.$
\begin{eqnarray}
&&A\Psi(p^{\lambda_1 N},q^{\lambda_2 N})= \Psi(p^{\lambda_1 (N+1)},q^{\lambda_2 (N+1)})A,\cr&&
 A\tilde\Psi(p^{\tilde\lambda_1 N},q^{\tilde\lambda_2 N})= \tilde\Psi(p^{\tilde\lambda_1 (N+1)},q^{\tilde\lambda_2 (N+1)})A,\cr&&
A^\dag\Psi(p^{\lambda_1 N},q^{\lambda_2 N})=\Psi(p^{\lambda_1 (N-1)},q^{\lambda_2 (N-1)})A^\dag,\\&&
A^\dag\tilde\Psi(p^{\tilde\lambda_1 N},q^{\tilde\lambda_2 N})=\tilde\Psi(p^{\tilde\lambda_1 (N-1)},q^{\tilde\lambda_2 (N-1)})A^\dag.\nonumber
\end{eqnarray}
Let us now prove that the above definitions of coproduct, counit and antipode satisfy the properties (\ref{2Hopf1})-(\ref{2Hopf3}) for $\Omega\in\{A, A^\dag, N, \mathbf{1}\}$. Indeed,
\\
$\maltese$ For $\Omega=N$ and using (\ref{2RHop3}) and (\ref{2RHop4}), we have
\begin{eqnarray*}
(\mbox{id}\otimes\mathbf{\Delta})\mathbf{\Delta}(N)&=& (\mbox{id}\otimes\mathbf{\Delta})(N\otimes \mathbf{1} + \mathbf{1}\otimes N+ \tau \mathbf{1}\otimes\mathbf{1})\\
&=&N\otimes\mathbf{\Delta(1)}+\mathbf{1}\otimes\mathbf{\Delta}(N)+\tau\mathbf{1}\otimes\mathbf{\Delta(1)}\\
&=&N\otimes \mathbf{1}\otimes \mathbf{1}+ \mathbf{1}\otimes N\otimes \mathbf{1}+\mathbf{1}\otimes\mathbf{1}\otimes N+2\tau \mathbf{1}\otimes \mathbf{1}\otimes \mathbf{1}\\
&=&\mathbf{\Delta}(N)\otimes\mathbf{1}+\mathbf{\Delta(1)}\otimes N+\tau\mathbf{\Delta(1)}\otimes\mathbf{1}\\
&=&(\mathbf{\Delta}\otimes\mbox{id})(N\otimes\mathbf{1}+\mathbf{1}\otimes N+\tau\mathbf{1}\otimes\mathbf{1})\\
&=&(\mathbf{\Delta}\otimes\mbox{id})\mathbf{\Delta}(N).
\end{eqnarray*}
So, (\ref{2Hopf1}) is satisfied. Also,
\begin{eqnarray*}
(\mbox{id}\otimes\mathbf{\epsilon})\mathbf{\Delta}(N)&=& (\mbox{id}\otimes\mathbf{\epsilon})(N\otimes \mathbf{1} +
\mathbf{1}\otimes N+ \tau \mathbf{1}\otimes\mathbf{1})
\cr&=&N\otimes\mathbf{\epsilon(1)}+\mathbf{1}\otimes\mathbf{\epsilon}(N)+\tau\mathbf{1}\otimes\mathbf{\epsilon(1)}\\
&=&N\otimes 1-\tau\mathbf{1}\otimes 1 +\tau\mathbf{1}\otimes 1
= N
\end{eqnarray*}
and
\begin{eqnarray*}
(\mathbf{\epsilon}\otimes\mbox{id})\mathbf{\Delta}(N)
&=&(\mathbf{\epsilon}\otimes\mbox{id})(N\otimes\mathbf{1}+\mathbf{1}\otimes N+\tau\mathbf{1}\otimes\mathbf{1})
\cr&=& \mathbf{\epsilon}(N)\otimes\mathbf{1}+\mathbf{\epsilon(1)}\otimes N+\tau\mathbf{\epsilon(1)}\otimes\mathbf{1}\\
&=&-\tau\otimes \mathbf{1}+ 1\otimes N\otimes \tau\otimes \mathbf{1}=N,
\end{eqnarray*}
where we use (\ref{2RHop3}) and the fact that $\mathbf{\epsilon}(N)=-\tau$. Hence, $N$ satisfies (\ref{2Hopf2}). Next,
\begin{eqnarray*}
m((\mbox{id}\otimes\mathbf{S})\mathbf{\Delta}(N))&=& m((\mbox{id}\otimes\mathbf{S})(N\otimes \mathbf{1} +
\mathbf{1}\otimes N+ \tau \mathbf{1}\otimes\mathbf{1}))\\
&=&m(N\otimes\mathbf{S(1)}+\mathbf{1}\otimes\mathbf{S}(N)+\tau\mathbf{1}\otimes\mathbf{S(1))}\\
&=&m(N\otimes \mathbf{1}-\mathbf{1}\otimes N - 2\tau\mathbf{1}\otimes\mathbf{1}+\tau\mathbf{1}\otimes\mathbf{1})\\
&=&-\tau m(\mathbf{1}\otimes\mathbf{1})=-\tau\mathbf{1} =\mathbb{\epsilon}(N)\mathbf{1},
\end{eqnarray*}
and similarly
\begin{eqnarray*}
m((\mathbf{S}\otimes\mbox{id})\mathbf{\Delta}(N))
=-\tau m(\mathbf{1}\otimes\mathbf{1})=-\tau\mathbf{1} =\mathbb{\epsilon}(N)\mathbf{1},
\end{eqnarray*}
where we use (\ref{2RHop3}) and the fact that $\mathbf{S}(N)= - N -2\tau\mathbf{1}$. Therefore  $N$ satisfies (\ref{2Hopf3}).

$\maltese$ For $\Omega=A$, we have
\begin{eqnarray*}
(\mbox{id}\otimes\mathbf{\Delta})\mathbf{\Delta}(A)
&=&(\mbox{id}\otimes\mathbf{\Delta})(\alpha A\otimes\Psi(p^{\alpha_1N},q^{\alpha_2N})+\tilde{\alpha}\tilde\Psi(p^{\tilde\alpha_1N},q^{\tilde\alpha_2N})\otimes A)\\
&=&\alpha A\otimes\mathbf{\Delta}(\Psi(p^{\alpha_1N},q^{\alpha_2N}))
+\tilde{\alpha}\tilde\Psi(p^{\tilde\alpha_1N},q^{\tilde\alpha_2N})\otimes\mathbf{\Delta}(A)\\
&=&\alpha A\otimes\mathbf{\Delta}(\Psi(p^{\alpha_1N},q^{\alpha_2N}))
\\&&+ \alpha\tilde{\alpha}\tilde\Psi(p^{\tilde\alpha_1N},q^{\tilde\alpha_2N})\otimes A\otimes\Psi(p^{\alpha_1N},q^{\alpha_2N})
\\&&+ \tilde{\alpha}^2 \tilde\Psi(p^{\tilde\alpha_1N},q^{\tilde\alpha_2N})\otimes\tilde\Psi(p^{\tilde\alpha_1N},q^{\tilde\alpha_2N})\otimes A\\
&=&\alpha^2 A\otimes\Psi(p^{\alpha_1N},q^{\alpha_2N})\otimes\Psi(p^{\alpha_1N},q^{\alpha_2N})\\
&&+\alpha\tilde{\alpha}\tilde\Psi(p^{\tilde\alpha_1N},q^{\tilde\alpha_2N})\otimes A\otimes\Psi(p^{\alpha_1N},q^{\alpha_2N})
\\&&+\tilde{\alpha}\mathbf{\Delta}(\tilde\Psi(p^{\tilde\alpha_1N},q^{\tilde\alpha_2N}))\otimes A\\
&=&\alpha\mathbf{\Delta}(A)\otimes\Psi(p^{\alpha_1N},q^{\alpha_2N})+
\tilde{\alpha}\mathbf{\Delta}(\tilde\Psi(p^{\tilde\alpha_1N},q^{\tilde\alpha_2N}))\otimes A\\
&=&(\mathbf{\Delta}\otimes\mbox{id})(\alpha A\otimes\Psi(p^{\alpha_1N},q^{\alpha_2N})+ \tilde{\alpha}\tilde\Psi(p^{\tilde\alpha_1N},q^{\tilde\alpha_2N})\otimes A)\\
&=&(\mathbf{\Delta}\otimes\mbox{id})\mathbf{\Delta}(A),
\end{eqnarray*}
where we use (\ref{2RHop1}) and the fact that $\alpha$, $\tilde\alpha$, $\alpha_i$ and $\tilde\alpha_i$ ($i=1,2$) satisfy equations
(\ref{2CondHop4}) and (\ref{2CondHop5}). Hence, (\ref{2Hopf1}) holds.

The property (\ref{2Hopf2}) also holds since
\begin{eqnarray*}
(\mathbf{\epsilon}\otimes\mbox{id})\mathbf{\Delta}(A)
&=& (\alpha \mathbf{\epsilon}\otimes\mbox{id})(A\otimes\Psi(p^{\alpha_1N},q^{\alpha_2N}) +\tilde{\alpha}\tilde\Psi(p^{\tilde\alpha_1N},q^{\tilde\alpha_2N})\otimes A)\\
&=& \alpha\mathbf{\epsilon}(A)\otimes\Psi(p^{\alpha_1N},q^{\alpha_2N})+ \tilde{\alpha}\mathbf{\epsilon}(\tilde\Psi(p^{\tilde\alpha_1N},q^{\tilde\alpha_2N}))\otimes A\\
&=& \tilde{\alpha}\tilde\Psi(p^{-\tilde\alpha_1\tau},q^{-\tilde\alpha_2\tau})A = A=\alpha\Psi(p^{-\alpha_1\tau},q^{-\alpha_2\tau})A\\
&=& \alpha A\otimes\mathbf{\epsilon}(\Psi(p^{\alpha_1N},q^{\alpha_2N}))+
 \tilde{\alpha}\tilde\Psi(p^{\tilde\alpha_1N},q^{\tilde\alpha_2N})\otimes\mathbf{\epsilon}(A)\\
&=&(\mbox{id}\otimes\mathbf{\epsilon})(\alpha A\otimes\Psi(p^{\alpha_1N},q^{\alpha_2N})
+\tilde{\alpha}\tilde\Psi(p^{\tilde\alpha_1N},q^{\tilde\alpha_2N})\otimes A)\\
&=&(\mbox{id}\otimes\mathbf{\epsilon})\mathbf{\Delta}(A),
\end{eqnarray*}
where the use of (\ref{2RHop1}), (\ref{2RHop5}) and (\ref{2CondHop1}) has been made.

$A$ satisfies also property (\ref{2Hopf3}) since
\begin{eqnarray*}
&&m((\mathbf{S}\otimes\mbox{id})\mathbf{\Delta}(A))
=m((\mathbf{S}\otimes\mbox{id})(\alpha A\otimes\Psi(p^{\alpha_1N},q^{\alpha_2N})
\cr&&\quad\quad\quad+ \tilde{\alpha}\tilde\Psi(p^{\tilde\alpha_1N},q^{\tilde\alpha_2N})\otimes A)\\
&&\quad=m(\alpha\mathbf{S}(A)\otimes\Psi(p^{\alpha_1N},q^{\alpha_2N})+ \tilde{\alpha}\mathbf{S}(\tilde\Psi(p^{\tilde\alpha_1N},q^{\tilde\alpha_2N}))\otimes A)\\
&&\quad= -s_1\alpha A\Psi(p^{\alpha_1N},q^{\alpha_2N})+\tilde{\alpha}\mathbf{S}(\tilde\Psi(p^{\tilde\alpha_1N},q^{\tilde\alpha_2N}))A\\
&&\quad=\left[-s_1 \alpha\Psi(p^{\alpha_1(N+1)},q^{\alpha_2(N+1)})+ \tilde{\alpha}\tilde\Psi(p^{-\tilde\alpha_1(N+2\tau)},q^{-\tilde\alpha_2(N+2\tau)})\right]A\\
&&\quad=0.A=\mathbf{\epsilon}(A)\mathbf{1}= A.0
\cr&&\quad=A\left[\alpha\Psi(p^{-\alpha_1(N+2\tau)},q^{-\alpha_2(N+2\tau)}) -s_1\tilde{\alpha}\tilde\Psi(p^{\tilde\alpha_1(N-1)},q^{\tilde\alpha_2(N-1)})\right]
\cr&&\quad=\alpha A\mathbf{S}(\Psi(p^{\alpha_1N},q^{\alpha_2N}))
-s_1\tilde{\alpha}\tilde\Psi(p^{\tilde\alpha_1N},q^{\tilde\alpha_2N})A\\
&&\quad= m(\alpha A\otimes\mathbf{S}(\Psi(p^{\alpha_1N},q^{\alpha_2N}))+ \tilde{a}\tilde\Psi(p^{\tilde\alpha_1N},q^{\tilde\alpha_2N})\otimes\mathbf{S}(A))\\
&&\quad=m((\mbox{id}\otimes\mathbf{S})(\alpha A\otimes\Psi(p^{\alpha_1N},q^{\alpha_2N})+ \tilde{\alpha}\tilde\Psi(p^{\tilde\alpha_1N},q^{\tilde\alpha_2N})\otimes A))\\
&&\quad=m((\mbox{id}\otimes\mathbf{S})\mathbf{\Delta}(A)),
\end{eqnarray*}
where we  use (\ref{2RHop1}), (\ref{2RHop6}) and the fact that $s_1$, $\alpha$, $\tilde\alpha$, $\alpha_i$ and $\tilde\alpha_i$ ($i=1,2$)
 satisfy equations (\ref{2CondHop8}) and (\ref{2CondHop9}).

$\maltese$ For $\Omega=A^\dag$, one can perform the same computations and use (\ref{2RHop2}), (\ref{2CondHop2}), (\ref{2RHop5}), (\ref{2RHop6}),
(\ref{2CondHop6}), (\ref{2CondHop7}), (\ref{2CondHop10}) and (\ref{2CondHop11}) to prove that (\ref{2Hopf1})-(\ref{2Hopf3}) also hold.

$\maltese$ Computing $\mathbf{\Delta}(A)\mathbf{\Delta}(A^\dag)$ and $\mathbf{\Delta}(A^\dag)\mathbf{\Delta}(A)$ we obtain:
\begin{eqnarray*}
 \mathbf{\Delta}(A)\mathbf{\Delta}(A^\dag)&=& (\alpha A\otimes\Psi(p^{\alpha_1N},q^{\alpha_2N})+\tilde{\alpha}\tilde\Psi(p^{\tilde\alpha_1N},q^{\tilde\alpha_2N})\otimes A)\\
&&\quad\times(\beta A^\dag\otimes\Psi(p^{\beta_1N},q^{\beta_2N})
+\tilde{\beta}\tilde\Psi(p^{\tilde\beta_1N},q^{\tilde\beta_2N})\otimes A^\dag)\\
&=& \alpha\beta AA^\dag\otimes\Psi(p^{\alpha_1N},q^{\alpha_2N})\Psi(p^{\beta_1N},q^{\beta_2N})
\cr&&\quad+\alpha\tilde{\beta}A\tilde\Psi(p^{\tilde\beta_1N},q^{\tilde\beta_2N})\otimes\Psi(p^{\alpha_1N},q^{\alpha_2N}) A^\dag\\
&&\quad+\tilde{\alpha}\beta\tilde\Psi(p^{\tilde\alpha_1N},q^{\tilde\alpha_2N})A^\dag\otimes A\Psi(p^{\beta_1N},q^{\beta_2N})
\cr&&\quad+\tilde{\alpha}\tilde{\beta}\tilde\Psi(p^{\tilde\alpha_1N},q^{\tilde\alpha_2N})\tilde\Psi(p^{\tilde\beta_1N},q^{\tilde\beta_2N})\otimes AA^\dag
\end{eqnarray*}
and
\begin{eqnarray*}
\mathbf{\Delta}(A^\dag)\mathbf{\Delta}(A)&=&(\beta A^\dag\otimes\Psi(p^{\beta_1N},q^{\beta_2N})+\tilde{\beta}\tilde\Psi(p^{\tilde\beta_1N},q^{\tilde\beta_2N})\otimes A^\dag)\\
&&\quad\times(\alpha A\otimes\Psi(p^{\alpha_1N},q^{\alpha_2N})+\tilde{\alpha}\tilde\Psi(p^{\tilde\alpha_1N},q^{\tilde\alpha_2N})\otimes A)\\
&=& \alpha\beta A^\dag A\otimes\Psi(p^{\beta_1N},q^{\beta_2N})\Psi(p^{\alpha_1N},q^{\alpha_2N})
\cr&&\quad+
\tilde{\alpha}\beta A^\dag\tilde\Psi(p^{\tilde\alpha_1N},q^{\tilde\alpha_2N})\otimes\Psi(p^{\beta_1N},q^{\beta_2N})A\\
&&\quad + \alpha\tilde{\beta}\tilde\Psi(p^{\tilde\beta_1N},q^{\tilde\beta_2N})A\otimes A^\dag\Psi(p^{\alpha_1N},q^{\alpha_2N})
\cr&&\quad+\tilde{\alpha}\tilde{\beta}\tilde\Psi(p^{\tilde\beta_1N},q^{\tilde\beta_2N})\tilde\Psi(p^{\tilde\alpha_1N},q^{\tilde\alpha_2N})\otimes A^\dag A
\end{eqnarray*}
respectively.
Thus,
\begin{eqnarray*}
\mathbf{\Delta}(A)\mathbf{\Delta}(A^\dag)&-&\gamma\mathbf{\Delta}(A^\dag)\mathbf{\Delta}(A)
\cr&&\quad=\alpha\beta(AA^\dag-\gamma A^\dag A)\otimes\Psi(p^{\alpha_1N},q^{\alpha_2N})\Psi(p^{\beta_1N},q^{\beta_2N})\\
&&\qquad+\tilde{\alpha}\tilde{\beta}\tilde\Psi(p^{\tilde\alpha_1N},q^{\tilde\alpha_2N})\tilde\Psi(p^{\tilde\beta_1N},q^{\tilde\beta_2N})\otimes (AA^\dag-\gamma A^\dag A)
\cr
&&\qquad +\alpha\tilde{\beta}(A\tilde\Psi(p^{\tilde\beta_1N},q^{\tilde\beta_2N})\otimes\Psi(p^{\alpha_1N},q^{\alpha_2N}) A^\dag
\cr&&\qquad-\gamma\tilde\Psi(p^{\tilde\beta_1N},q^{\tilde\beta_2N})A\otimes A^\dag\Psi(p^{\alpha_1N},q^{\alpha_2N}))\\
&&\qquad + \tilde{\alpha}\beta(\tilde\Psi(p^{\tilde\alpha_1N},q^{\tilde\alpha_2N})A^\dag\otimes A\Psi(p^{\beta_1N},q^{\beta_2N})
\cr&&\qquad-\gamma A^\dag\tilde\Psi(p^{\tilde\alpha_1N},q^{\tilde\alpha_2N})\otimes\Psi(p^{\beta_1N},q^{\beta_2N})A).
\end{eqnarray*}
Therefore, $[A,\;A^\dag]_\gamma= AA^\dag-\gamma A^\dag A$  implies
\begin{eqnarray}
\mathbf{\Delta}([A,\;A^\dag]_\gamma)&=&[A,\;A^\dag]_\gamma\otimes\Psi(p^{\alpha_1N},q^{\alpha_2N})\Psi(p^{\beta_1N},q^{\beta_2N})
\\&&+\tilde\Psi(p^{\tilde\alpha_1N},q^{\tilde\alpha_2N})\tilde\Psi(p^{\tilde\beta_1N},q^{\tilde\beta_2N})\otimes[A,\;A^\dag]_\gamma\nonumber
\end{eqnarray}
provided that
\begin{eqnarray*}
A\tilde\Psi(p^{\tilde\beta_1N},q^{\tilde\beta_2N})\otimes\Psi(p^{\alpha_1N},q^{\alpha_2N}) A^\dag
=\gamma\tilde\Psi(p^{\tilde\beta_1N},q^{\tilde\beta_2N})A\otimes A^\dag\Psi(p^{\alpha_1N},q^{\alpha_2N})
\end{eqnarray*}
\begin{eqnarray*}
\tilde\Psi(p^{\tilde\alpha_1N},q^{\tilde\alpha_2N})A^\dag\otimes A\Psi(p^{\beta_1N},q^{\beta_2N})
=\gamma A^\dag\tilde\Psi(p^{\tilde\alpha_1N},q^{\tilde\alpha_2N})\otimes\Psi(p^{\beta_1N},q^{\beta_2N})A,
\end{eqnarray*}
\begin{eqnarray*}
\alpha\beta=1,\quad\mbox{and}\quad \tilde{\alpha}\tilde{\beta}=
\end{eqnarray*}
or
\begin{eqnarray*}\label{2Hopeq5}
&&\tilde\Psi(p^{\tilde\beta_1N},q^{\tilde\beta_2N})\otimes\Psi(p^{\alpha_1N},q^{\alpha_2N})
=\gamma\tilde\Psi(p^{\tilde\beta_1(N-1)},q^{\tilde\beta_2(N-1)})\otimes \Psi(p^{\alpha_1(N-1)},q^{\alpha_2(N-1)}),
\cr
&&\tilde\Psi(p^{\tilde\alpha_1(N+1)},q^{\tilde\alpha_2(N+1)})\otimes\Psi(p^{\beta_1(N+1)},q^{\beta_2(N+1)})
=\gamma \tilde\Psi(p^{\tilde\alpha_1N},q^{\tilde\alpha_2N})\otimes\Psi(p^{\beta_1N},q^{\beta_2N}),
\end{eqnarray*}
\begin{eqnarray*}
\alpha\beta=1,\quad\mbox{and}\quad \tilde{\alpha}\tilde{\beta}=1,
\end{eqnarray*}
which are  (\ref{2CondHop12}), (\ref{2CondHop13}) and (\ref{2CondHop2}),  respectively.
\hfill$\Box$

\subsection{Relevant particular cases}\label{2section5}
Let us now apply the above general formalism to particular deformed algebras, well spred in the literature.

\subsubsection{Jagannathan-Srinivasa deformation}\label{2janasec}
\begin{itemize}
\item[{\bf A.}] Taking ${\cal R}(x,y) = \frac{x-y}{p-q}$, we obtain the Jagannathan-Srinivasa
$(p,q)$- factors and $(p,q)$-factorials
\begin{eqnarray*}
 [n]_{p,q}=\frac{p^n-q^n}{p-q},
\end{eqnarray*}
and
\begin{eqnarray}
[n]!_{p,q}= \left\{\begin{array}{lr} 1 \quad \mbox{for   } \quad n=0 \quad \\
\frac{((p,q);(p,q))_n}{(p-q)^n} \quad \mbox{for } \quad n\geq
1, \quad \end{array} \right.
\end{eqnarray}
respectively.

Referring the readers to \cite{Jagannathan&Rao} for details on  $(p,q)$-calculus, let us restrict the present description to  some new  relevant properties.
\begin{proposition}\label{2PropJan1}
If $n$ and $m$ are nonnegative integers, then
\begin{eqnarray}
[n]_{p,q}&=& \sum_{k=0}^{n-1}p^{n-1-k}q^k,
\cr[n+m]_{p,q}&=& q^m[n]_{p,q}+p^n[m]_{p,q}\cr&=& p^m[n]_{p,q}+q^n[m]_{p,q},
\cr [-m]_{p,q}&=& -q^{-m}p^{-m}[m]_{p,q},
\\\; [n-m]_{p,q}&=& q^{-m}[n]_{p,q}-q^{-m}p^{n-m}[m]_{p,q}\cr&=&p^{-m}[n]_{p,q}-q^{n-m}p^{-m}[m]_{p,q},\quad
\cr [n]_{p,q}&=& [2]_{p,q}[n-1]_{p,q}-pq[n-2]_{p,q}.\nonumber
\end{eqnarray}
\end{proposition}

\begin{proposition}\label{2PropJan2}
The $(p,q)-$binomial coefficients
\begin{equation}
 \left[\begin{array}{c} n \\ k \end{array}\right]_{p,q}=
\frac{((p,q);(p,q))_n}{((p,q);(p,q))_k((p,q);(p,q))_{n-k}},\quad 0\leq k\leq n;\;\; n\in\mathbb{N},
\end{equation}
where $((p,q);(p,q))_m = (p-q)(p^2-q^2)\cdots(p^m-q^m)$, $m\in\mathbb{N}$ satisfy the following identities
\begin{eqnarray}
 \left[\begin{array}{c} n \\ k \end{array}\right]_{p,q}&=& \left[\begin{array}{c} n \\ n-k \end{array}\right]_{p,q}\\&=&
p^{k(n-k)}\left[\begin{array}{c} n \\ k \end{array}\right]_{q/p}=
p^{k(n-k)}\left[\begin{array}{c} n \\ n-k \end{array}\right]_{q/p},\label{2Janeq1}
\cr \left[\begin{array}{c} n+1 \\ k \end{array}\right]_{p,q} &=& p^k\left[\begin{array}{c} n \\ k \end{array}\right]_{p,q}
+q^{n+1-k}\left[\begin{array}{c} n \\ k-1 \end{array}\right]_{p,q},\label{2Janeq2}
\\ \left[\begin{array}{c} n+1 \\ k \end{array}\right]_{p,q} &=&
p^{k}\left[\begin{array}{c} n \\ k \end{array}\right]_{p,q} +
p^{n+1-k}\left[\begin{array}{c} n \\ k-1 \end{array}\right]_{p,q}\\&&\quad -(p^n-q^n)
\left[\begin{array}{c} n-1 \\ k-1 \end{array}\right]_{p,q}\quad\label{2Janeq3}\nonumber
\end{eqnarray}
with
\begin{eqnarray*}
\left[\begin{array}{c} n \\ k \end{array}\right]_{q/p}= \frac{(q/p; q/p)_n}{(q/p; q/p)_k(q/p; q/p)_{n-k}},
\end{eqnarray*}
where $(q/p; q/p)_n = (1-q/p)(1-q^2/p^2)\cdots (1-q^n/p^n)$
and the $(p,q)$-shifted factorial
\begin{eqnarray*}
((a,b);(p,q))_n &\equiv& (a-b)(ap-bq)\cdots(ap^{n-1}-bq^{n-1})\cr
 &=& \sum_{k=0}^{n}\left[\begin{array}{c} n \\ k \end{array}\right]_{p,q}(-1)^kp^{(n-k)(n-k-1)/2}
q^{k(k-1)/2}a^{n-k}b^k.
\end{eqnarray*}
\end{proposition}
\begin{proposition}\label{2PropJan3}
If  the quantities $x$, $y$, $a$ and $b$ are  such that $xy= qyx$, $ba=pab$, $[i,\;j]=0$ for  $i\in\{a, b\}$ and $j\in\{x, y\}$,  and, moreover, $p$ and $q$ commute with each element of the set $\{a, b, x, y\}$,
then
\begin{eqnarray}
 (ax+by)^n= \sum_{k=0}^{n}\left[\begin{array}{c} n \\ k \end{array}\right]_{p,q}a^{n-k}b^ky^kx^{n-k}.
\end{eqnarray}
\end{proposition}
The latter result is a generalization of noncommutative form of the $q$-binomial theorem \cite{Gasper&Rahman90}, which
can be obtained setting $a$, $b$ and $p$ equal to $1$, i.e.
\begin{eqnarray}
 (x+y)^n= \sum_{k=0}^{n}\left[\begin{array}{c} n \\ k \end{array}\right]_{q}y^kx^{n-k},
\end{eqnarray}
where $$\left[\begin{array}{c} n \\ k \end{array}\right]_{q}= {(q;q)_n}/{(q;q)_k(q;q)_{n-k}},$$
with $(q;q)_n = (1-q)(1-q^2)\cdots(1-q^n)$.\\
{\bf Proof of Proposition~\ref{2PropJan3}} A proof has been proposed in \cite{Jagannathan&Rao}. Here we provide another one by induction
on $n$. Indeed, the result is true for $n=1$. Suppose it remains valid for all $n\leq m$ and prove that this
is also true for $n=m+1$:
\begin{eqnarray*}
&& (ax+by)^{m+1}= (ax+by)^m(ax+by)
\cr&&\qquad=
\sum_{k=0}^{m}\left[\begin{array}{c} m \\ k \end{array}\right]_{p,q}a^{m-k}b^ky^kx^{m-k}(ax+by)
\cr&&\qquad= \sum_{k=0}^{m}\left[\begin{array}{c} m \\ k \end{array}\right]_{p,q}p^ka^{m+1-k}b^ky^kx^{m+1-k}
\cr&&\qquad+ \sum_{k=0}^{m}\left[\begin{array}{c} m \\ k \end{array}\right]_{p,q}q^{m-k}a^{m-k}b^{k+1}y^{k+1}x^{m-k}
\cr&&\qquad= a^{m+1}x^{m+1}+\sum_{k=1}^{m}\left[\begin{array}{c} m \\ k \end{array}\right]_{p,q}p^ka^{m+1-k}b^ky^kx^{m+1-k}
\cr&&\qquad+\sum_{k=0}^{m-1}\left[\begin{array}{c} m \\ k \end{array}\right]_{p,q}q^{m-k}a^{m-k}b^{k+1}y^{k+1}x^{m-k}
+ b^{m+1}y^{m+1}
\cr&&\qquad= a^{m+1}x^{m+1}+b^{m+1}y^{m+1}
\cr&&\qquad+
\sum_{k=1}^{m}\left(p^k\left[\begin{array}{c} m \\ k \end{array}\right]_{p,q}
q^{m+1-k}\left[\begin{array}{c} m \\ k-1 \end{array}\right]_{p,q}\right)
a^{m+1-k}b^ky^kx^{m+1-k} 
\cr&&\qquad= a^{m+1}x^{m+1}+
\sum_{k=1}^{m}\left[\begin{array}{c} m+1 \\ k \end{array}\right]_{p,q}a^{m+1-k}b^ky^kx^{m+1-k}
\cr&&\qquad\quad+ b^{m+1}y^{m+1},
\cr&&\qquad=\sum_{k=0}^{m+1}\left[\begin{array}{c} m+1 \\ k \end{array}\right]_{p,q}a^{m+1-k}b^ky^kx^{m+1-k}
\end{eqnarray*}
where the use of (\ref{2Janeq2}) has been made. Hence the  result is true for all $n\in\mathbb{N}.$\hfill$\Box$\\
The ${\cal R}(p,q)-$derivative is thus reduced to the $(p,q)-$derivative \cite{Jagannathan&Rao}
\begin{eqnarray}
 \partial_{p,q}= \frac{1}{(p-q)z}(P-Q),
\end{eqnarray}
namely, for $f\in{\cal O}(\mathbb{D}_R),$
\begin{eqnarray}
\partial_{p,q}f(z)= \frac{f(pz)-f(qz)}{z(p-q)}.
\end{eqnarray}
The associated algebra ${\cal A}_{p,q}$, generated by $\{1,\; A,\; A^\dag,\;N\},$
satisfies the relations:
\begin{eqnarray}
&& A\;A^\dag- pA^\dag A= q^N, \qquad A\;A^\dag- qA^\dag A= p^N\\
&&[N,\;A^\dag]= A^\dag\quad\qquad [N,\;A]= -A,
\end{eqnarray}
and its realization on ${\cal O}(\mathbb{D}_R)$, engendered by
$\{1,\; z,\; z\partial_z,\;\partial_{p,q}\}$, satisfies the relations
\begin{eqnarray}\label{2ComuRel1}
&&z\;\partial_{p,q} -p\;\partial_{p,q}\;z = q^{z\partial_z}\quad
z\;\partial_{p,q} -q\;\partial_{p,q}\;z = p^{z\partial_z}\cr
&&[z\partial_z,\;z]= z\qquad\qquad [z\partial_z,\;\partial_{p,q}]= -\partial_{p,q}.
\end{eqnarray}
Therefore, the differential operator $d_{p,q}$ is then given by
\begin{eqnarray}
d_{p,q} = (dz)\frac{1}{(p-q)z}(P-Q)
\end{eqnarray}
with the following properties:
\begin{eqnarray}
&& d_{p,q}1= 0,\qquad d_{p,q}z = (dz),
\qquad d_{p,q}\partial_{p,q} = (dz)\partial_{p,q}^2\\
&& d_{p,q}(z\partial_z)= (dz)(z\partial_z+1)\partial_{p,q}
\qquad\mbox{and}\quad d_{p,q}^2 = 0.\nonumber
\end{eqnarray}
The differential of $f\in{\cal O}(\mathbb{D}_R)$ is then 
\begin{eqnarray}
 d_{p,q}f(z)= (dz)\frac{f(pz)-f(qz)}{(p-q)z}
\end{eqnarray}
affording the Leibniz rule
\begin{eqnarray}
 d_{p,q}(fg)(z)&=&(dz)\frac{f(pz)-f(qz)}{(p-q)z} g(qz) \\&&+ (dz)f(pz)\frac{g(pz)-g(qz)}{(p-q)z}\cr
&=& \{d_{p,q}f(z)\}\cdot g(qz) + f(pz)\cdot d_{p,q}g(z)\nonumber
\end{eqnarray}
or, equivalently,
\begin{eqnarray}
d_{p,q}(fg)(z)&=& (dz)\frac{f(pz)-f(qz)}{(p-q)z} g(pz)\\&& + (dz)f(qz)\frac{g(pz)-g(qz)}{(p-q)z}\cr
&=& \{d_{p,q}f(z)\}\cdot g(pz) + f(qz)\cdot d_{p,q}g(z).\nonumber
\end{eqnarray}
The  $(p,q)-$integration is obtained from (\ref{2integra2}) as follows:
\begin{eqnarray}
 {\cal I}_{p,q}f(z) &=& \frac{p-q}{P-Q}zf(z)
= (p-q)\sum_{\nu=0}^\infty\frac{Q^\nu}{P^{\nu+1}}zf(z)
\\&=& (p-q)z\sum_{\nu=0}^\infty f(zq^\nu/p^{\nu+1})q^\nu/p^{\nu+1}.\nonumber
\end{eqnarray}
Setting $p=1$, one recovers the $q-$derivative and $q-$integral of Jackson \cite{Koekoek}.
\item[{\bf B.}] Tacking ${\cal R}(x,y) = \frac{x-y}{\frac{a}{p}x-\frac{b}{q}y}$,
where $a, b\in\mathbb{C}$ with
$a\neq b$, the ${\cal R}(p,q)$-factors are given by
\begin{eqnarray}
{\cal R}(p^n,q^n)= [n]_{p,q}^{a,b}= \frac{p^n-q^n}{ap^{n-1}-bq^{n-1}},\quad n=0, 1, \cdots.
\end{eqnarray}
The ${\cal R}(p,q)$-factorials become
\begin{eqnarray}
[n]!_{p,q}^{a,b}= \left\{\begin{array}{lr} 1 \quad \mbox{for   } \quad n=0 \quad \\
\frac{((p,q);(p,q))_n}{((a,b);(p,q))_n} \quad \mbox{for } \quad n\geq
1, \quad \end{array} \right.
\end{eqnarray}
The derivative is now given by
\begin{eqnarray}
 \partial_{{\cal R}(p,q)}= \partial_{p,q}\frac{p-q}{P-Q}\frac{P-Q}{\frac{a}{p}P-\frac{b}{q}Q}
= \frac{1}{z}\frac{P-Q}{\frac{a}{p}P-\frac{b}{q}Q},
\end{eqnarray}
so that for $f\in{\cal O}(\mathbb{D}_R)$ we have
\begin{eqnarray}
\qquad\partial_{{\cal R}(p,q)}f(z)&=& \frac{1}{z}\frac{P-Q}{\frac{a}{p}P-\frac{b}{q}Q}f(z)\\
&=&\frac{1}{z}(P-Q)\frac{p}{aP}\sum_{\nu=0}^\infty(bp/aq)^\nu(Q/P)^\nu f(z)\cr
&=&\frac{p}{az}\sum_{\nu=0}^\infty(bp/aq)^\nu\left[(Q/P)^\nu-(Q/P)^{\nu+1}\right]f(z)\cr
&=& \frac{p}{az}
\sum_{\nu=0}^\infty(bp/aq)^\nu\left[f\left((q/p)^\nu z\right)-f\left((q/p)^{\nu+1} z\right)\right].\nonumber
\end{eqnarray}
Moreover,
\begin{eqnarray}
 {\cal I}_{{\cal R}(p,q)}= \frac{\frac{a}{p}P-\frac{b}{q}Q}{P-Q}z.
\end{eqnarray}
Applying this to $f\in{\cal O}(\mathbb{D}_R)$ we obtain
\begin{eqnarray}
\qquad {\cal I}_{{\cal R}(p,q)}f(z)&=& \frac{\frac{a}{p}P-\frac{b}{q}Q}{P-Q}zf(z)\\
&=& \left(\frac{a}{p}P-\frac{b}{q}Q\right)\frac{1}{P}\sum_{\nu=0}^\infty\left(Q/P\right)^\nu zf(z)\cr
&=& \left(\frac{a}{p}-\frac{b}{q}(Q/P)\right)\sum_{\nu=0}^\infty\left(Q/P\right)^\nu zf(z)\cr
&=& \sum_{\nu=0}^\infty\left[(a/p)\left(Q/P\right)^\nu-(b/q)\left(Q/P\right)^{\nu+1}\right] zf(z)
\cr&=& \sum_{\nu=0}^\infty\left[(a/p)(q/p)^\nu zf\left((q/q)^\nu z\right)\right.\cr&&\left.\qquad\qquad-(b/q)(q/p)^{\nu+1}zf\left((q/p)^{\nu+1}\right)\right]\cr
&=& (z/p)\sum_{\nu=0}^\infty(q/p)^\nu\left[af\left((q/q)^\nu z\right)-bf\left((q/p)^{\nu+1}\right)\right].\nonumber
\end{eqnarray}
\end{itemize}

Let us display the Hopf algebra structure of Jagannathan-Srinivasa algebra according to the theorem~\ref{2Hoptheo}.
To this end notice first that the Leibniz rule of the derivative is given by
\begin{eqnarray}
 \partial_{p,q}(fg)(z)= (\partial_{p,q}f(z))p^{z\partial_z}g(z)+(q^{z\partial_z}f(z))\partial{p,q}g(z)
\end{eqnarray}
from which we deduce $\Psi(x,y)= x$ and $\tilde\Psi(x,y)=y$. So, $\alpha_2=0$, $\tilde\alpha_1=0$, $\beta_2=0$ and $\tilde\beta_1=0$. Equations (\ref{2CondHop1})-(\ref{2CondHop3}) yield
\begin{eqnarray}
\alpha=p^{\alpha_1\tau},\quad \tilde{\alpha}=q^{\tilde\alpha_2\tau},\quad \beta=p^{\beta_1\tau},\quad \tilde{\beta}=q^{\tilde\beta_2\tau}\\
\tau(\alpha_1+\beta_1)=0\quad \mbox{and}\quad \tau(\tilde\alpha_2+\tilde\beta_2)=0,
\end{eqnarray}
while equations (\ref{2CondHop8})-(\ref{2CondHop11}) are reduced to
\begin{eqnarray*}
&&s_1p^{\alpha_1(\tau+1)}p^{\alpha_1N}=q^{-\tilde\alpha_2\tau}q^{-\tilde\alpha_2N},\qquad
p^{-\alpha_1\tau}p^{-\alpha_1N}=s_1q^{\tilde\alpha_2(\tau-1)}q^{\tilde\alpha_2N},\\
&&p^{-\beta_1\tau}p^{-\beta_1N}=\tilde{s}_1q^{\tilde\beta_2(\tau+1)}q^{\tilde\beta_2N},\qquad
\tilde{s}_1p^{\beta_1(\tau+1)}p^{-\beta_1N}=q^{3\tilde\beta_2\tau}q^{\tilde\beta_2N}.
\end{eqnarray*}
Therefore,
\begin{eqnarray}
\alpha_1=\tilde\alpha_2= \beta_1=\tilde\beta_2=0,\qquad \alpha=\tilde \alpha=\beta=\tilde \beta=s_1=\tilde s_1=1
\end{eqnarray}
and equations (\ref{2CondHop4})-(\ref{2CondHop7}) are satisfied and (\ref{2CondHop12})-(\ref{2CondHop13}) yield
$\gamma=1$. Thus, the coproduct, the counit and the antipode are given by
\begin{eqnarray}
&&\mathbf{\Delta}(A)= A\otimes1+ 1\otimes A,\\
&&\mathbf{\Delta}(A^\dag)= A^\dagger\otimes1+ 1\otimes A^\dag,\\
&&\mathbf{\Delta}(N)= N\otimes \mathbf{1} + \mathbf{1}\otimes N+ \tau \mathbf{1}\otimes\mathbf{1}, \\
&&\mathbf{\Delta(1)}=\mathbf{1}\otimes\mathbf{1},\\
&&\mathbf{\epsilon}(A)= 0,\quad
\mathbf{\epsilon}(A^\dag)= 0,\quad
\mathbf{\epsilon}(N)= -\tau,\quad
\mathbf{\epsilon(1)}= 1\\
&&\mathbf{S}(A)=- A,\quad
\mathbf{S}(A^\dag)= -A^\dag,\quad
\mathbf{S}(N)= -N -2\tau\mathbf{1},\quad
\mathbf{S(1)}= \mathbf{1}
\end{eqnarray}
respectively, where $\tau$ is a real number, usually set equal to  $0$.

\subsubsection{Chakrabarty and Jagannathan deformation}

The algebra of Chakrabarty and Jagannathan \cite{Chakrabarti&Jagan} is obtained by taking $\displaystyle{\cal R}(x,y)$  $= \frac{1-xy}{(p^{-1}-q)x}$.
Indeed, the ${\cal R}(p,q)$-factors and ${\cal R}(p,q)$-factorials are reduced to $(p^{-1},q)$-factors and
$(p^{-1},q)$-factorials, namely
\begin{eqnarray*}
 [n]_{p^{-1},q}=\frac{p^{-n}-q^n}{p^{-1}-q},
\end{eqnarray*}
and
\begin{eqnarray}
[n]!_{p^{-1},q}= \left\{\begin{array}{lr} 1 \quad \mbox{for   } \quad n=0 \quad \\
\frac{((p^{-1},q);(p^{-1},q))_n}{(p^{-1}-q)^n} \quad \mbox{for } \quad n\geq
1, \quad \end{array} \right.
\end{eqnarray}
respectively.
The properties of this deformation can be readily recovered from the previous  section \ref{2janasec}
by replacing the parameter $p$ by $p^{-1}$.\\
The ${\cal R}(p,q)-$derivative is also reduced to $(p^{-1},q)-$derivative. Indeed,
\begin{eqnarray}
 \partial_{{\cal R}(p,q)} &=& \partial_{p,q}\frac{p-q}{P-Q}\frac{1-PQ}{(p^{-1}-q)P}
\cr&=&\frac{1}{(p^{-1}-q)z}(P^{-1}-Q)=: \partial_{p^{-1},q}.
\end{eqnarray}
Therefore, for $f\in{\cal O}(\mathbb{D}_R)$
\begin{eqnarray}
\partial_{p^{-1},q}f(z)= \frac{f(p^{-1}z)-f(qz)}{z(p^{-1}-q)}
\end{eqnarray}
and the differential of $f\in{\cal O}(\mathbb{D}_R)$ is given by
\begin{eqnarray}
 d_{p^{-1},q}f(z)= (dz)\frac{f(p^{-1}z)-f(qz)}{z(p^{-1}-q)}.
\end{eqnarray}
Computing the Leibniz rule we get
\begin{eqnarray}
 d_{p^{-1},q}(fg)(z)&=&(dz)\frac{f(p^{-1}z)-f(qz)}{z(p^{-1}-q)} g(qz) \\
&&+ (dz)f(p^{-1}z)\frac{g(p^{-1}z)-g(qz)}{z(p^{-1}-q)}\cr
&=& \{d_{p^{-1},q}f(z)\}\cdot g(qz) + f(p^{-1}z)\cdot d_{p^{-1},q}g(z)\nonumber
\end{eqnarray}
or, equivalently,
\begin{eqnarray}
d_{p^{-1},q}(fg)(z)&=&(dz)\frac{f(p^{-1}z)-f(qz)}{z(p^{-1}-q)} g(p^{-1}z)
\\&& + (dz)f(qz)\frac{g(p^{-1}z)-g(qz)}{z(p^{-1}-q)}\cr
&=& \{d_{p^{-1},q}f(z)\}\cdot g(p^{-1}z) + f(qz)\cdot d_{p^{-1},q}g(z).\nonumber
\end{eqnarray}
We obtain from (\ref{2integra2}) the action of the 
$(p^{-1},q)$-integration on $f\in{\cal O}(\mathbb{D}_R)$ as follows:
\begin{eqnarray}
\quad {\cal I}_{p^{-1},q}f(z) &=& \frac{p^{-1}-q}{P^{-1}-Q}zf(z)
= (p^{-1}-q)\sum_{\nu=0}^\infty Q^\nu P^{\nu+1}zf(z)
\\&=& (1-pq)z\sum_{\nu=0}^\infty f(zq^\nu p^{\nu+1})(pq)^\nu.\nonumber
\end{eqnarray}

The same   Hopf algebra structure as that of Jagannathan-Srinivasa is also obtained for Chakrabarty and Jagannathan deformation.

\subsubsection{Generalized  $q$-Quesne deformation}
The generalized Quesne algebra  \cite{Hounkonnou&Ngompe07b, Quesne&al02}  is found by taking ${\cal R}(x,y)= \frac{xy-1}{(q-p^{-1})y}$.
Indeed, the $(p,q)$-Quesne factors and  factorials are given by
\begin{eqnarray*}
 [n]_{p,q}^Q=\frac{p^n-q^{-n}}{q-p^{-1}},
\end{eqnarray*}
and
\begin{eqnarray}
[n]_{p,q}^Q!= \left\{\begin{array}{lr} 1 \quad \mbox{for   } \quad n=0 \quad \\
\frac{((p,q^{-1});(p,q^{-1}))_n}{(q-p^{-1})^n} \quad \mbox{for } \quad n\geq
1, \quad \end{array} \right.
\end{eqnarray}
respectively.
There follow some relevant new properties:
\begin{proposition}\label{2PropoQes1}
 If $n$ and $m$ are nonnegative integers, then
\begin{eqnarray}
\;[-m]_{p,q}^Q&=& -p^{-m}q^m[m]_{p,q}^Q,\label{2Qeq1}\\
\;[n+m]_{p,q}^Q&=& q^{-m}[n]_{p,q}^Q+p^n[m]_{p,q}^Q= p^m[n]_{p,q}^Q+q^{-n}[m]_{p,q}^Q,\label{2Qeq2}\\
\quad\qquad [n-m]_{p,q}^Q&=& q^{m}[n]_{p,q}^Q-p^{n-m}q^m[m]_{p,q}^Q= p^{-m}[n]_{p,q}^Q+p^{-m}q^{m-n}[m]_{p,q}^Q,\label{2Qeq3}\cr
&&\\
\;[n]_{p,q}^Q &=& \frac{q-p^{-1}}{p-q^{-1}}[2]_{p,q}^Q[n-1]_{p,q}^Q-pq^{-1}[n-2]_{p,q}^Q.\label{2Qeq4}
\end{eqnarray}
\end{proposition}
{\bf Proof:} We obtain Eqs.(\ref{2Qeq1}) and (\ref{2Qeq2}) applying the relations
\begin{eqnarray*}
 p^{-m}-q^{m} = -p^{-m}q^{m}(p^m-q^{-m})
\end{eqnarray*}
and
\begin{eqnarray*}
p^{n+m}-q^{-n-m}&=& q^{-m}(p^{n}-q^{-n})+p^n(p^{m}-q^{-m})
\cr&=&p^m(p^{n}-q^{-n})+q^{-n}(p^{m}-q^{-m}),
\end{eqnarray*}
respectively.  Eq.(\ref{2Qeq3}) follows combining Eqs.(\ref{2Qeq1}) and (\ref{2Qeq2}). Note that
\begin{eqnarray}\label{2Qeq5}
\qquad[n]_{p,q^{-1}}&=& \frac{p^n-q^{-n}}{p-q^{-1}}= \frac{q-p^{-1}}{p-q^{-1}}\frac{p^n-q^{-n}}{q-p^{-1}}
= [n]_{p,q}^{Q},\quad n=1, 2, \cdots
\end{eqnarray}
which, combined with the following identity
\begin{eqnarray*}
 [n]_{p,q{-1}}&=& [2]_{p,q{-1}}[n-1]_{p,q{-1}}-pq^{-1}[n-2]_{p,q^{-1}},
\end{eqnarray*}
gives Eq.(\ref{2Qeq4}).\hfill$\Box$
\begin{proposition}\label{2PropoQes2}
The $(p,q)$-Quesne binomial coefficients
\begin{eqnarray}
 \left[\begin{array}{c} n \\ k \end{array}\right]_{p,q}^Q=
\frac{((p,q^{-1});(p,q^{-1}))_n}{((p,q^{-1});(p,q^{-1}))_k((p,q^{-1});(p,q^{-1}))_{n-k}},
\label{2Qeq6}
\end{eqnarray}
where $\quad 0\leq k\leq n,\;\; n\in\mathbb{N},$ satisfy the following properties:
\begin{eqnarray}
\left[\begin{array}{c} n \\ k \end{array}\right]_{p,q}^Q
&=& \left[\begin{array}{c} n \\ n-k \end{array}\right]_{p,q}^Q=
p^{k(n-k)}\left[\begin{array}{c} n \\ k \end{array}\right]_{1/qp}\\&=&
p^{k(n-k)}\left[\begin{array}{c} n \\ n-k \end{array}\right]_{1/qp},\label{2Qeq7}
\cr \left[\begin{array}{c} n+1 \\ k \end{array}\right]_{p,q}^Q &=&
p^k\left[\begin{array}{c} n \\ k \end{array}\right]_{p,q}^Q
+q^{-n-1+k}\left[\begin{array}{c} n \\ k-1 \end{array}\right]_{p,q}^Q,\label{2Qeq8}
\\ \left[\begin{array}{c} n+1 \\ k \end{array}\right]_{p,q}^Q &=&
p^{k}\left[\begin{array}{c} n \\ k \end{array}\right]_{p,q}^Q +
p^{n+1-k}\left[\begin{array}{c} n \\ k-1 \end{array}\right]_{p,q}^Q \cr&&-(p^n-q^{-n})
\left[\begin{array}{c} n-1 \\ k-1 \end{array}\right]_{p,q}^Q.\quad\label{2Qeq9}\nonumber
\end{eqnarray}
\end{proposition}
{\bf Proof:} This is direct using the Proposition~\ref{2PropJan1}  and
\begin{eqnarray}
 \left[\begin{array}{c} n \\ k \end{array}\right]_{p,q}^Q =
\left[\begin{array}{c} n \\ k \end{array}\right]_{p,q^{-1}}.
\end{eqnarray}\hfill$\Box$

\begin{proposition}\label{2PropoQes3}
If the quantities  $x$, $y$, $a$ and $b$ are    such that $xy= q^{-1}yx$, $ba=pab$,   $[i,\;j]=0$ for  $i\in\{a, b\}$ and $j\in\{x, y\}$,  and,  moreover,  $p$ and $q$ commute with  each element of the set  $\{a, b, x, y\}$,
then
\begin{eqnarray}
 (ax+by)^n= \sum_{k=0}^{n}\left[\begin{array}{c} n \\ k \end{array}\right]_{p,q}^Qa^{n-k}b^ky^kx^{n-k}.
\end{eqnarray}
\end{proposition}
{\bf Proof:} By induction on $n$. Indeed, the result is true for $n=1$.
Suppose it  remains valid for $n\leq m$ and prove that this is also true for $n=m+1:$
\begin{eqnarray*}
 (ax+by)^{m+1}&=& (ax+by)^m(ax+by)\cr&=&
\sum_{k=0}^{m}\left[\begin{array}{c} m \\ k \end{array}\right]_{p,q}^Qa^{m-k}b^ky^kx^{m-k}(ax+by)
\cr&=& \sum_{k=0}^{m}\left[\begin{array}{c} m \\ k \end{array}\right]_{p,q}^Qp^ka^{m+1-k}b^ky^kx^{m+1-k}
\cr&&+ \sum_{k=0}^{m}\left[\begin{array}{c} m \\ k \end{array}\right]_{p,q}^Qq^{-m+k}a^{m-k}b^{k+1}y^{k+1}x^{m-k}
\cr&=& a^{m+1}x^{m+1}
+\sum_{k=1}^{m}\left[\begin{array}{c} m \\ k \end{array}\right]_{p,q}^Qp^ka^{m+1-k}b^ky^kx^{m+1-k}
\cr&&+\sum_{k=0}^{m-1}\left[\begin{array}{c} m \\ k \end{array}\right]_{p,q}^Qq^{-m+k}a^{m-k}b^{k+1}y^{k+1}x^{m-k}
+ b^{m+1}y^{m+1}
\cr&=& a^{m+1}x^{m+1}+
\sum_{k=1}^{m}\left(p^k\left[\begin{array}{c} m \\ k \end{array}\right]_{p,q}^Q\right.
\cr&&\left.+
q^{-m-1+k}\left[\begin{array}{c} m \\ k-1 \end{array}\right]_{p,q}^Q\right)
a^{m+1-k}b^ky^kx^{m+1-k} +b^{m+1}y^{m+1}
\cr&=& a^{m+1}x^{m+1}+
\sum_{k=1}^{m}\left[\begin{array}{c} m+1 \\ k \end{array}\right]_{p,q}^Qa^{m+1-k}b^ky^kx^{m+1-k}\cr&&+ b^{m+1}y^{m+1}
\cr&=&\sum_{k=0}^{m+1}\left[\begin{array}{c} m+1 \\ k \end{array}\right]_{p,q}^Qa^{m+1-k}b^ky^kx^{m+1-k},
\end{eqnarray*}
where the use of (\ref{2Qeq8}) has been made. Therefore the result is true for all $n\in\mathbb{N}.$\hfill$\Box$\\
%%%%%%%%%%%%%%%%%%%%%%%%%%%%%%%%%%%%%%%
The $(p,q)-$Quesne derivative is given by
\begin{eqnarray}
 \partial_{p,q}^Q = \partial_{p,q}\frac{p-q}{P-Q}\frac{PQ-1}{(q-p^{-1})Q}
=\frac{1}{(q-p^{-1})z}(P-Q^{-1}).
\end{eqnarray}
Therefore, for $f\in{\cal O}(\mathbb{D}_R)$
\begin{eqnarray}
\partial_{p,q}^Qf(z)= \frac{f(pz)-f(q^{-1}z)}{z(q-p^{-1})}
\end{eqnarray}
and the differential is given by
\begin{eqnarray}
 d_{p,q}^Qf(z)= (dz)\frac{f(pz)-f(q^{-1}z)}{z(q-p^{-1})}
\end{eqnarray}
leading to the Leibniz rule
\begin{eqnarray}
 d_{p,q}^Q(fg)(z)&=&(dz)\frac{f(pz)-f(q^{-1}z)}{z(q-p^{-1})} g(q^{-1}z) \\
&&+ (dz)f(pz)\frac{g(pz)-g(q^{-1}z)}{z(q-p^{-1})}\cr
&=& \{d_{p,q}^Qf(z)\}\cdot g(q^{-1}z) + f(pz)\cdot d_{p,q}^Qg(z)\nonumber
\end{eqnarray}
or, equivalently,
\begin{eqnarray}
 d_{p,q}^Q(fg)(z)&=&(dz)\frac{f(pz)-f(q^{-1}z)}{z(q-p^{-1})} g(pz) \\
&&+ (dz)f(q^{-1}z)\frac{g(pz)-g(q^{-1}z)}{z(q-p^{-1})}\cr
&=& \{d_{p,q}^Qf(z)\}\cdot g(pz) + f(q^{-1}z)\cdot d_{p,q}^Qg(z).\nonumber
\end{eqnarray}
The action of the  $(p,q)-$Quesne integration on $f\in{\cal O}(\mathbb{D}_R)$
is obtained from (\ref{2integra2}) as follows:
\begin{eqnarray}
 {\cal I}_{p,q}^Qf(z) &=& \frac{q-p^{-1}}{P-Q^{-1}}zf(z)
= (p^{-1}-q)\sum_{\nu=0}^\infty P^\nu Q^{\nu+1}zf(z)
\\&=& (p^{-1}-q)z\sum_{\nu=0}^\infty f(zp^\nu q^{\nu+1})p^\nu q^{\nu+1}.\nonumber
\end{eqnarray}

The same structure of Hopf algebra as for Jagannathan-Srinivasa is also found for the generalized  $q-$Quesne  deformation.

\subsubsection{$(p,q;\mu,\nu,h)-$deformation}
The deformed Hounkonnou-Ngompe \cite{Hounkonnou&Ngompe07a} algebra is obtained by taking $${\cal R}(x,y)= h(p,q)\frac{y^\nu}{x^\mu}\frac{xy-1}{(q-p^{-1})y},$$
such that $0< pq < 1$, $p^\mu<q^{\nu-1}$, $p>1$, and $h(p,q)$ is a well behaved real and non-negative
function of deformation parameters $p$ and $q$ such that  $h(p,q)\to 1$ as $(p,q)\to (1,1).$
Here the ${\cal R}(p,q)-$factors become $(p,q;\mu,\nu,h)$-factors, namely
\begin{eqnarray}
 [n]^{\mu,\nu}_{p,q,h}= h(p,q)\frac{q^{\nu n}}{p^{\mu n}}\frac{p^n-q^{-n}}{q-p^{-1}}.
\end{eqnarray}
\begin{proposition}\label{2PropoHouk1}
The $(p,q;\mu,\nu,h)-$factors verify the following properties, for $m,n\in\mathbb{N}$:
\begin{eqnarray}
\;[-m]^{\mu,\nu}_{p,q,h}= -\frac{q^{-2\nu m+m}}{p^{-2\mu m+m}}[m]^{\mu,\nu}_{p,q,h},\label{2Heq1}
\end{eqnarray}
\begin{eqnarray}
\;[n+m]^{\mu,\nu}_{p,q,h}&=& \frac{q^{\nu m-m}}{p^{\mu m}}[n]^{\mu,\nu}_{p,q,h}
+\frac{q^{\nu n}}{p^{\mu n-n}}[m]^{\mu,\nu}_{p,q,h}
\\&=&\frac{q^{\nu m}}{p^{\mu m- m}}[n]^{\mu,\nu}_{p,q,h}
+\frac{q^{\nu n- n}}{p^{\mu n}}[m]^{\mu,\nu}_{p,q,h}\;,\label{2Heq2}\nonumber
\end{eqnarray}
\begin{eqnarray}
\quad[n-m]^{\mu,\nu}_{p,q,h}&=&\frac{q^{-\nu m+ m}}{p^{-\mu m}}[n]^{\mu,\nu}_{p,q,h}
-\frac{q^{\nu(n-2m)+m}}{p^{\mu(n-2m)-n+m}}[m]^{\mu,\nu}_{p,q,h}
\\&=& \frac{q^{-\nu m}}{p^{-\mu m+m}}[n]^{\mu,\nu}_{p,q,h}
-\frac{q^{\nu(n-2m)-n+m}}{p^{\mu(n-2m)+m}}[m]^{\mu,\nu}_{p,q,h},\label{2Heq3}\nonumber
\end{eqnarray}
\begin{eqnarray}
\quad [n]^{\mu,\nu}_{p,q,h}&=& \frac{q-p^{-1}}{p-q^{-1}}\frac{q^{-\nu}}{p^{-\mu}}\frac{1}{h(p,q)}[2]^{\mu,\nu}_{p,q,h}
[n-1]^{\mu,\nu}_{p,q,h}-\frac{q^{2\nu-1}}{p^{2\nu-1}}[n-2]^{\mu,\nu}_{p,q,h}.\label{2Heq4}
\end{eqnarray}
\end{proposition}
{\bf Proof:} This is direct using the Proposition~\ref{2PropoQes1} and the fact that
\begin{eqnarray}\label{2Houkeq}
 [n]^{\mu,\nu}_{p,q,h}= h(p,q)\frac{q^{\nu n}}{p^{\mu n}}[n]^Q_{p,q,h}.
\end{eqnarray}\hfill$\Box$
\begin{proposition}\label{2PropoHouk2}
The $(p,q,\mu,\nu,h)-$ binomial coefficients
\begin{eqnarray}
 \left[\begin{array}{c} n \\ k \end{array}\right]_{p,q,h}^{\mu,\nu}:=
\frac{[n]_{p,q,h}^{\mu,\nu}!}{[k]_{p,q,h}^{\mu,\nu}![n-k]_{p,q,h}^{\mu,\nu}!}=
\frac{q^{\nu k(n-k)}}{p^{\mu k(n-k)}}\left[\begin{array}{c} n \\ k \end{array}\right]_{p,q}^Q,
\label{2Heq5}
\end{eqnarray}
where $0\leq k\leq n,\;\; n\in\mathbb{N},$ satisfy the following properties:
\begin{eqnarray}
\left[\begin{array}{c} n \\ k \end{array}\right]_{p,q,h}^{\mu,\nu}
&=& \left[\begin{array}{c} n \\ n-k \end{array}\right]_{p,q,h}^{\mu,\nu},\label{2Heq6}
\\\qquad \left[\begin{array}{c} n+1 \\ k \end{array}\right]_{p,q,h}^{\mu,\nu} &=&
\frac{q^{\nu k}}{p^{(\mu-1)k}}\left[\begin{array}{c} n \\ k \end{array}\right]_{p,q,h}^{\mu,\nu}
+\frac{q^{(\nu-1)(n+1-k)}}{p^{\mu(n+1-k)}}\left[\begin{array}{c} n \\ k-1 \end{array}\right]_{p,q,h}^{\mu,\nu},
\label{2Heq7}
\\ \left[\begin{array}{c} n+1 \\ k \end{array}\right]_{p,q,h}^{\mu,\nu} &=&
\frac{q^{\nu k}}{p^{(\mu-1)k}}\left[\begin{array}{c} n \\ k \end{array}\right]_{p,q,h}^{\mu,\nu} +
\frac{q^{\nu(n+1-k)}}{p^{(\mu-1)(n+1-k)}}\left[\begin{array}{c} n \\ k-1 \end{array}\right]_{p,q,h}^{\mu,\nu}
\\&&\qquad-(p^n-q^{-n})
\frac{q^{\nu n}}{p^{\mu n}}\left[\begin{array}{c} n-1 \\ k-1 \end{array}\right]_{p,q,h}^{\mu,\nu}.\label{2Heq8}\nonumber
\end{eqnarray}
\end{proposition}
{\bf Proof:} This is direct using the Proposition~\ref{2PropoQes2} and the fact that
\begin{eqnarray}
 [n]_{p,q,h}^{\mu,\nu}!= h^n(p,q)\frac{q^{n(n+1)/2}}{p^{n(n+1)/2}}[n]_{p,q}^Q!,
\end{eqnarray}
where the use of Eq.(\ref{2Houkeq}) has been made.\hfill$\Box$
%%%%%%%%%%%%%%%%%%%%%%%%%%%%%%%%%%%%%%%%%%%%%%%%%%%%%%%%%
\begin{proposition}\label{2PropoHouk3}
If  the quantities $x$, $y$, $a$ and $b$ are   such that $\displaystyle xy= \frac{q^{\nu-1}}{p^{\mu}}yx$,
$\displaystyle ba=\frac{q^\nu}{p^{\mu-1}}ab$, $[i,\;j]=0$ for  $i\in\{a, b\}$ and $j\in\{x, y\}$,  and, moreover,  $p$ and $q$ commute  with each element of the set  $\{a, b, x, y\}$,  then
\begin{eqnarray}
 (ax+by)^n= \sum_{k=0}^{n}\left[\begin{array}{c} n \\ k \end{array}\right]_{p,q,h}^{\mu,\nu}a^{n-k}b^ky^kx^{n-k}.
\end{eqnarray}
\end{proposition}
{\bf Proof:} By induction on $n$. Indeed, the result is true for $n=1$.
Suppose it  remains valid for $n\leq m$ and prove that this is also true for $n=m+1:$
\begin{eqnarray*}
 (ax+by)^{m+1}&=& (ax+by)^m(ax+by)\cr&=&
\sum_{k=0}^{m}\left[\begin{array}{c} m \\ k \end{array}\right]_{p,q,h}^{\mu,\nu}a^{m-k}b^ky^kx^{m-k}(ax+by)
\cr&=& \sum_{k=0}^{m}\left[\begin{array}{c} m \\ k \end{array}\right]_{p,q,h}^{\mu,\nu}
\frac{q^{\nu k}}{q^{(\mu-1)k}}a^{m+1-k}b^ky^kx^{m+1-k}
\cr&&\quad+ \sum_{k=0}^{m}\left[\begin{array}{c} m \\ k \end{array}\right]_{p,q,h}^{\mu,\nu}
\frac{q^{(\nu-1)(m-k)}}{p^{\mu(m-k)}}a^{m-k}b^{k+1}y^{k+1}x^{m-k}
\cr&=& a^{m+1}x^{m+1}+
\sum_{k=1}^{m}\left[\begin{array}{c} m \\ k \end{array}\right]_{p,q,h}^{\mu,\nu}
\frac{q^{\nu k}}{q^{(\mu-1)k}}a^{m+1-k}b^ky^kx^{m+1-k}
\cr&&\quad+
\sum_{k=0}^{m-1}\left[\begin{array}{c} m \\ k \end{array}\right]_{p,q,h}^{\mu,\nu}
\frac{q^{(\nu-1)(m-k)}}{p^{\mu(m-k)}}a^{m-k}b^{k+1}y^{k+1}x^{m-k}
\cr&&\quad+ b^{m+1}y^{m+1}
\cr&=& a^{m+1}x^{m+1}+ b^{m+1}y^{m+1}+
\sum_{k=1}^{m}\left(\frac{q^{\nu k}}{q^{(\mu-1)k}}\left[\begin{array}{c} m \\ k \end{array}\right]_{p,q,h}^{\mu,\nu}\right.\cr&&\quad\left.+
\frac{q^{(\nu-1)(m+1-k)}}{p^{\mu(m+1-k)}}\left[\begin{array}{c} m \\ k-1 \end{array}\right]_{p,q,h}^{\mu,\nu}\right)
a^{m+1-k}b^ky^kx^{m+1-k}
\cr&=& a^{m+1}x^{m+1}+
\sum_{k=1}^{m}\left[\begin{array}{c} m+1 \\ k \end{array}\right]_{p,q,h}^{\mu,\nu}
a^{m+1-k}b^ky^kx^{m+1-k}
\cr&&\quad+ b^{m+1}y^{m+1}
\cr&=&\sum_{k=0}^{m+1}\left[\begin{array}{c} m+1 \\ k \end{array}\right]_{p,q,h}^{\mu,\nu}a^{m+1-k}b^ky^kx^{m+1-k},
\end{eqnarray*}
where the use of (\ref{2Heq7}) has been made. Therefore, the result is true for all $n\in\mathbb{N}.$\hfill$\Box$\\
%%%%%%%%%%%%%%%%%%%%%%%%%%%%%%%%%%%%%%%%%%%%%%%%%%%%%%%%%
The ${\cal R}(p,q)$-derivative is then reduced to the $(p,q;\mu,\nu,h)-$derivative, given by
\begin{eqnarray}
 \partial_{{\cal R}(p,q)} &=& \partial_{p,q}\frac{p-q}{P-Q}h(p,q)\frac{Q^\nu}{P^\mu} \frac{PQ-1}{(q-p^{-1})Q}\\
&=&\frac{h(p,q)}{(q-p^{-1})z}\frac{Q^\nu}{P^\mu}(P-Q^{-1})\equiv\partial_{p,q,h}^{\mu,\nu}.\nonumber
\end{eqnarray}
Therefore the $(p,q;\mu,\nu,h)-$derivative and the $(p,q;\mu,\nu,h)-$differential of $f\in{\cal O}(\mathbb{D}_R)$
 are given by
\begin{eqnarray}
\partial_{p,q,h}^{\mu,\nu}= h(p,q)\frac{f(zq^\nu/p^{\mu-1})-f(zq^{\nu-1}/p^\mu)}{z(q-p^{-1})}
\end{eqnarray}
and
\begin{eqnarray}
 d_{p,q,h}^{\mu,\nu}f(z)= (dz)h(p,q)\frac{f(zq^\nu/p^{\mu-1})-f(zq^{\nu-1}/p^\mu)}{z(q-p^{-1})}
\end{eqnarray}
respectively, with the Leibniz rule 
\begin{eqnarray}
\qquad\quad d_{p,q,h}^{\mu,\nu}(fg)(z)&=&(dz)h(p,q)\frac{f(zq^\nu/p^{\mu-1})-f(zq^{\nu-1}/p^\mu)}{z(q-p^{-1})} g(zq^{\nu-1}/p^\mu)
\cr &+& (dz)f(zq^\nu/p^{\mu-1})h(p,q)\frac{g(zq^\nu/p^{\mu-1})-g(zq^{\nu-1}/p^\mu)}{z(q-p^{-1})}\cr
&=& \{d_{p,q,h}^{\mu,\nu}f(z)\}\cdot g(zq^{\nu-1}/p^\mu) + f(zq^\nu/p^{\mu-1})\cdot d_{p,q,h}^{\mu,\nu}g(z)\nonumber
\end{eqnarray}
which is equivalent to
\begin{eqnarray}
\qquad\quad d_{p,q,h}^{\mu,\nu}(fg)(z)&=&(dz)h(p,q)\frac{f(zq^\nu/p^{\mu-1})-f(zq^{\nu-1}/p^\mu)}{z(q-p^{-1})} g(zq^\nu/p^{\mu-1})
\cr &+& (dz)f(zq^{\nu-1}/p^\mu)h(p,q)\frac{g(zq^\nu/p^{\mu-1})-g(zq^{\nu-1}/p^\mu)}{z(q-p^{-1})}\cr
&=& \{d_{p,q,h}^{\mu,\nu}f(z)\}\cdot g(zq^\nu/p^{\mu-1}) + f(zq^{\nu-1}/p^\mu)\cdot d_{p,q,h}^{\mu,\nu}g(z).\nonumber
\end{eqnarray}
 From (\ref{2integra2}) we obtain the action of the
$(p,q,\mu,\nu,h)$-integration  on $f\in{\cal O}(\mathbb{D}_R)$ as follows:
\begin{eqnarray}
\quad {\cal I}_{p,q,h}^{\mu,\nu}f(z) &=& \frac{q-p^{-1}}{h(p,q)}\frac{P^\mu/Q^{\nu}}{P-Q^{-1}}zf(z)=\frac{p^{-1}-q}{h(p,q)}\frac{P^\mu/Q^{\nu-1}}{1-PQ}zf(z)\\
&=& \frac{p^{-1}-q}{h(p,q)}\frac{P^\mu}{Q^{\nu-1}}\sum_{j=0}^\infty P^jQ^jzf(z)\cr
&=& \frac{p^{-1}-q}{h(p,q)}\sum_{j=0}^\infty P^{j+\mu}Q^{j+1-\nu} zf(z)\cr
&=& \frac{z(p^{-1}-q)}{h(p,q)}\frac{p^\mu}{q^{\nu-1}}\sum_{j=0}^\infty f(zp^{j+\mu}q^{j+1-\nu}).\nonumber
\end{eqnarray}

From the derivative Leibniz rule
\begin{eqnarray}
\partial^{\mu,\nu}_{p,q,h}(f(z)g(z))&=& \left(\partial^{\mu,\nu}_{p,q,h}f(z)\right)
 \frac{q^{\nu z\partial_z}}{p^{(\mu-1) z\partial_z}}g(z)\\&&+ \left(\frac{q^{(\nu-1) z\partial_z}}{p^{\mu z\partial_z}} f(z)\right) \partial^{\mu,\nu}_{p,q,h}(g(z)),\nonumber
\end{eqnarray}
one deduces $\Psi(x,y)= x^{-(\mu-1)}y^{\nu}$ and $\tilde\Psi(x,y)=y^{-\mu}x^{\nu-1}$.
Hence, (\ref{2CondHop1})-(\ref{2CondHop3}) yield
\begin{eqnarray}
&& \alpha= p^{-\alpha_1\tau(\mu-1)}q^{\alpha_2\tau\nu}, \;
\tilde\alpha= p^{-\tilde\alpha_1\tau\mu}q^{\tilde\alpha_2\tau(\nu-1)},\\
&&\beta= p^{-\beta_1\tau(\mu-1)}q^{\beta_2\tau\nu}, \;
\tilde\beta= p^{-\tilde\beta_1\tau\mu}q^{\tilde\beta_2\tau(\nu-1)},\\
&&(\alpha_1+\beta_1)\tau=0, \;\; (\alpha_2+\beta_2)\tau=0,\;\;
(\tilde\alpha_1+\tilde\beta_1)\tau=0, \;\; (\tilde\alpha_2+\tilde\beta_2)\tau=0.
\end{eqnarray}
Equations (\ref{2CondHop4})-(\ref{2CondHop7}) are of course  satisfied and (\ref{2CondHop12})-(\ref{2CondHop13}) give
\begin{eqnarray}
&& \gamma = p^{-[\tilde\beta_1\mu+\alpha_1(\mu-1)]}q^{\tilde\beta_2(\nu-1)+\alpha_2\nu}\;\;\mbox{and}\;\;
\gamma = p^{-[\tilde\alpha_1\mu+\beta_1(\mu-1)]}q^{\tilde\alpha_2(\nu-1)+\beta_2\nu}
\end{eqnarray}
From equations (\ref{2CondHop8})-(\ref{2CondHop11}) we deduce
\begin{eqnarray}
&&s_1=p^{\alpha_1(\mu-1)}q^{-\alpha_2\nu},\; s_1=p^{\beta_1(\mu-1)}q^{-\beta_2\nu}\\
&&\tilde\alpha_1=-\alpha_1\frac{\mu-1}{\mu},\; \tilde\alpha_2=-\alpha_2\frac{\nu}{\nu-1},\;
\tilde\beta_1=-\beta_1\frac{\mu-1}{\mu},\; \tilde\beta_2=-\beta_2\frac{\nu}{\nu-1}
\end{eqnarray}
so that
\begin{eqnarray}
\tilde \alpha= \alpha^{-1},\quad \tilde \beta= \beta^{-1},\quad \gamma=1.
\end{eqnarray}
Setting
\begin{eqnarray}
\kappa_\alpha= \frac{q^{\alpha_2\nu}}{p^{\alpha_1(\mu-1)}}\;\; \mbox{and}\;\;
\kappa_\beta= \frac{q^{\beta_2\nu}}{p^{\beta_1(\mu-1)}},\;\mbox{ with }\; \alpha_1,\alpha_2,\beta_1,\beta_1\in\mathbb{R},
\end{eqnarray}
we get
\begin{eqnarray}
&&\Psi(p^{\alpha_1N},q^{\alpha_2N})=\kappa_\alpha^N,\;
\tilde\Psi(p^{\tilde\alpha_1N},q^{\tilde\alpha_2N})=\kappa_\alpha^{-N},\\
&&\Psi(p^{\beta_1N},q^{\beta_2N})=\kappa_\beta^N,\;
\tilde\Psi(p^{\tilde\beta_1N},q^{\tilde\beta_2N})=\kappa_\beta^{-N}.
\end{eqnarray}
The remaining conditions are $(\alpha_1+\beta_1)\tau=0$ and $(\alpha_2+\beta_2)\tau=0$ which hold if $\tau=0$ or
$\beta_1=-\alpha_1$ and $\beta_2=-\alpha_2$.

Suppose  $\tau=0$. Then $\alpha=\tilde \alpha= \beta=\tilde\beta=1$.
Then, the coproduct, the counit, the antipode are defined as follows:
\begin{eqnarray}
&&\mathbf{\Delta}(A)= A\otimes\kappa_\alpha^N +\kappa_\alpha^{-N}\otimes A,\\
&&\mathbf{\Delta}(A^\dag)= A^\dagger\otimes\kappa_\beta^N+ \kappa_\beta^{-N}\otimes A^\dag,\\
&&\mathbf{\Delta}(N)= N\otimes \mathbf{1} + \mathbf{1}\otimes N, \\
&&\mathbf{\Delta(1)}=\mathbf{1}\otimes\mathbf{1},\\
&&\mathbf{\epsilon}(A)= 0,\quad
\mathbf{\epsilon}(A^\dag)= 0,\quad
\mathbf{\epsilon}(N)= 0,\quad
\mathbf{\epsilon(1)}= 1,\\
&&\mathbf{S}(A)=- \kappa_\alpha^{-1}A,\quad
\mathbf{S}(A^\dag)= -\kappa_\beta^{-1}A^\dag,\quad
\mathbf{S}(N)= -N,\quad
\mathbf{S(1)}= \mathbf{1}.
\end{eqnarray}

Suppose now $\tau\neq0$ that means $\beta_1=-\alpha_1$ and $\beta_2=-\alpha_2$. So,
\begin{eqnarray*}
\alpha=\kappa_\alpha^{\tau}, \beta=\kappa_\alpha^{-\tau}, s_1= \kappa_\alpha^{-1}, \tilde s_1=\kappa_\alpha.
\end{eqnarray*}
Thus, the coproduct, the counit, the antipode are defined as follows:
\begin{eqnarray}
&&\mathbf{\Delta}(A)= A\otimes\kappa_\alpha^{N+\tau} +\kappa_\alpha^{-N-\tau}\otimes A,\\
&&\mathbf{\Delta}(A^\dag)= A^\dagger\otimes\kappa_\alpha^{-N-\tau}+ \kappa_\alpha^{N+\tau}\otimes A^\dag,\\
&&\mathbf{\Delta}(N)= N\otimes \mathbf{1} + \mathbf{1}\otimes N +\tau\mathbf{1}\otimes\mathbf{1}, \\
&&\mathbf{\Delta(1)}=\mathbf{1}\otimes\mathbf{1},\\
&&\mathbf{\epsilon}(A)= 0,\quad
\mathbf{\epsilon}(A^\dag)= 0,\quad
\mathbf{\epsilon}(N)= -\tau,\quad
\mathbf{\epsilon(1)}= 1,\\
&&\mathbf{S}(A)=- \kappa_\alpha^{-1}A,\;\;
\mathbf{S}(A^\dag)= -\kappa_\alpha A^\dag,\;\;
\mathbf{S}(N)= -N-2\tau\mathbf{1},\;\;
\mathbf{S(1)}= \mathbf{1}.
\end{eqnarray}

%%%%%%%%%%%%%%%%%%%%%%%%%%%%%%%%%%%%%%%%%%%%%%%%%%%%%%%%%%%%%%%%%%%%%%%%%%%%%%%

\section{${\cal R}(p,q)$-deformed Rogers-Szeg\"o polynomials: associated quantum algebras, 
deformed Hermite polynomials and relevant properties}\label{chap5}
% \author{J D Bukweli Kyemba and M N Hounkonnou}
% \address{International Chair of Mathematical Physics
% and Applications (ICMPA-UNESCO Chair), University of
% Abomey-Calavi, 072 B.P.: 50 Cotonou, Republic of Benin}
% \eads{\mailto{norbert.hounkonnou@cipma.uac.bj},
% \mailto{hounkonnou@yahoo.fr}\\
% \mailto{desbuk@gmail.com}}
$\;$

% \begin{abstract}
This section addresses a new characterization of  ${\cal R}(p,q)-$deformed Rogers-Szeg\"o
polynomials by providing their three-term recursion relation and the associated quantum algebra built with  corresponding 
creation and annihilation operators. The whole construction is performed in  
 a unified way, generalizing all known relevant results which are straightforwardly derived as particular cases. Continuous 
${\cal R}(p,q)-$deformed  Hermite polynomials and their recursion
relation are also deduced. Novel relations are provided and discussed.
% \end{abstract}

% \pacs{02.30.Gp, 02.20.Uw, 02.30.Tb
% %02.40.Gh, 11.10.Nx.,03.65.-w, 03.65.Ca, 03.65.Ge, 02.30.Sa}
% }
%  \submitto{\JPA}
%  \today

% \maketitle

%\tableofcontents

 The present
investigation aims at
giving a new realization of the previous
generalized deformed quantum algebras and an  explicit definition of
 the ${\cal R}(p,q)-$Rogers-Szeg\"o polynomials, together with
their three-term recursion relation and  the deformed difference
equation giving rise to the creation and annihilation operators.

\subsection{${\cal R}(p,q)$-Rogers-Szeg\"o polynomials and related quantum algebras} \label{Sect43}
This section aims at providing  realizations of $({\cal
R},p,q)-$deformed quantum algebras induced by $({\cal
R},p,q)-$Rogers-Szeg\"o polynomials. We first define the latter
and their three-term
recursion relation, and then following the procedure elaborated in
\cite{Galetti, Jagannathan&Sridhar10}, we prove that every sequence of these
polynomials forms a basis for the corresponding deformed quantum
algebra.

Indeed, Galetti in \cite{Galetti}, upon recalling the technique of
construction of raising and lowering operators which satisfy an
algebra akin to the usual harmonic oscillator algebra, by using
the three-term recursion relation and the differentiation
expression of Hermite polynomials, has shown that a similar
procedure can be carried out to construct a
$q$-deformed harmonic oscillator algebra, with the help of relations
controlling the   Rogers-Szeg\"o polynomials. Following this author,
Jagannathan and Sridhar in \cite{Jagannathan&Sridhar10} adapted the same approach
 to construct a Bargman-Fock
realization of the harmonic oscillator as well as realizations of
$q$- and  $(p,q)$-
deformed harmonic oscillators based on  Rogers-Szeg\"o polynomials.

As matter of clarity, this section is stratified as follows. We first develop  the synoptic schemes of
known different generalizations and then  display the formalism of
${\cal R}(p,q)$-Rogers-Szeg\"o polynomials.

\subsubsection{Hermite polynomials and harmonic oscillator approach}
The Hermite polynomials are defined as orthogonal polynomials  satisfying the
three-term recursion relation
\begin{eqnarray}\label{4Hermttr}
\mathbb{H}_{n+1}(z)=2z\mathbb{H}_n(z)-2n\mathbb{H}_{n-1}(z)
\end{eqnarray}
and the differentiation relation
\begin{eqnarray}\label{4Hermdiff}
\frac{d}{dz}\mathbb{H}_n(z)= 2n\mathbb{H}_{n-1}(z).
\end{eqnarray}
Inserting Eq. (\ref{4Hermdiff}) in Eq. (\ref{4Hermttr}), one gets
\begin{eqnarray}
\mathbb{H}_{n+1}(z)= \left(2z-\frac{d}{dz}\right)\mathbb{H}_n(z)
\end{eqnarray}
which includes the introduction of a raising operator (see
\cite{Galetti} and references therein), defined as
\begin{eqnarray}
\hat{a}_+= 2z-\frac{d}{dz}
\end{eqnarray}
such that the set of Hermite polynomials can be generated by the
application of this operator to the first polynomial
$\mathbb{H}_0(z)=1$, i.e.,
\begin{eqnarray}
\mathbb{H}_n(z)= \hat{a}_+^n \mathbb{H}_0(z).
\end{eqnarray}
From Eq.(\ref{4Hermdiff}), one defines the lowering operator $\hat{a}_-$ as
\begin{eqnarray}
\hat{a}_-\mathbb{H}_n(z)= \frac{1}{2}\frac{d}{dz}\mathbb{H}_n(z)=n\mathbb{H}_{n-1}(z).
\end{eqnarray}
Furthermore one constructs a number operator in the form
\begin{eqnarray}
\hat{n}= \hat{a}_+\hat{a}_-.
\end{eqnarray}
One can readily check that these operators satisfy the canonical commutation relations
\begin{eqnarray}\label{4HarmAlg}
[ \hat{a}_-,\;\hat{a}_+]=1,\qquad [ \hat{n},\;\hat{a}_-]=-\hat{a}_- ,\qquad [ \hat{n},\;\hat{a}_+]=\hat{a}_+,
\end{eqnarray}
although the operators $\hat{a}_-$  and $\hat{a}_+$ are not the
usual creation and annihilation  operators associated
with the quantum mechanics harmonic oscillator. Thus, we see that
one can obtain  raising,  lowering and  number operators from
the two basic relations satisfied by the Hermite polynomials,
i. e. the three-term recursion relation and the differentiation
relation, respectively, so that they satisfy the well known
commutation relations.

On the other hand, if one considers the usual Hilbert space
spanned by the vectors $|n\rangle$, generated from the vacuum
$|0\rangle$ by the raising operator $\hat{a}_+$, then together
with the lowering operator $\hat{a}_-$, the following relations
hold
\begin{eqnarray}
&&\hat{a}_-\hat{a}_+-\hat{a}_+\hat{a}_+= 1,\cr
&&\langle 0|0\rangle =1,\cr
&&|n\rangle = \hat{a}_+^n|0\rangle,\cr
&&\hat{a}_-|0\rangle =0.
\end{eqnarray}
In particular,  the next expressions, established using  the
previous equations, are in order:
\begin{eqnarray}
&&\hat{a}_+|n\rangle = |n+1\rangle,\cr
&&\hat{a}_-|n\rangle = |n-1\rangle,\cr
&&\langle m|n\rangle =n!\delta_{mn}.
\end{eqnarray}

Now, on the other hand,  examining  the procedure given in \cite{Jagannathan&Sridhar10}, the authors considered the  sequence of polynomials
\begin{eqnarray}
\psi_n(z)=\frac{1}{\sqrt{n!}}{\bf h}_n(z),
\end{eqnarray}
where
\begin{eqnarray}
{\bf h}_n(z)=(1+z)^n =\sum_{k=0}^{n}\left(\begin{array}{c}n\\k\end{array}\right) z^k,
\end{eqnarray}
obeying the  relations
\begin{eqnarray}
\frac{d}{dz}\psi_n(z)&=&\sqrt{n}\psi_{n-1}(z),\\
(1+z)\psi_n(z) &=&\sqrt{n+1} \psi_{n+1}(z),\label{4Hermttr1}\\
(1+z)\frac{d}{dz}\psi_n(z)&=&n\psi_{n}(z),\label{4Hermdiff1}\\
\frac{d}{dz}\left((1+z)\psi_n(z)\right)&=&(n+1)\psi_n(z).
\end{eqnarray}
Here equations (\ref{4Hermttr1}) and (\ref{4Hermdiff1}) are the recursion relation
and the differential equation for polynomials $\psi_n(z),$ respectively.
By analogy to the work done by Galleti, Jagannathan and Sridhar proposed the following relations:
\begin{eqnarray}
\hat{a}_+=(1+z),\qquad \hat{a}_-=\frac{d}{dz}, \qquad \hat{n}= (1+z)\frac{d}{dz},
\end{eqnarray}
for creation (or raising), annihilation (or lowering) and number
operators, respectively, and found that  the set
$\{\psi_n(z)\;|n=0, 1, 2, \cdots\}$ forms a basis for the
Bargman-Fock realization of the harmonic oscillator
(\ref{4HarmAlg}).

\subsubsection{Rogers-Szeg\"o polynomials and $q$-deformed harmonic oscillator}\label{4qHosc}
Here in analogous way as  Jagannathan and
Sridhar \cite{Jagannathan&Sridhar10}, we perform a  construction of the creation, annihilation and
number operators from the three-term recurrence relation and the
$q-$difference equation founding the Rogers-Szeg\"o
polynomials. This procedure  a little differs from
 that used by Galetti \cite{Galetti} to obtain raising, lowering and number operators.\\
The Rogers-Szeg\"o polynomials are defined as
\begin{eqnarray}
H_n(z;q)=\sum_{k=0}^{n}\left[\begin{array}{c}n\\k\end{array}\right]_q z^k,\quad n=0,1,2\cdots
\end{eqnarray}
and satisfy a three-term recursion relation
\begin{eqnarray}\label{4qRogersttr}
H_{n+1}(z;q)=(1+z)H_n(z;q)-z(1-q^n)H_{n-1}(z;q)
\end{eqnarray}
as well as the $q$-difference equation
\begin{eqnarray}
\partial_qH_n(z;q)= [n]_qH_{n-1}(z;q).
\end{eqnarray}
In the limit case $q\to 1$, the Rogers-Szeg\"o polynomial of degree $n$ ($n= 0, 1, 2, \cdots$) well converges to
\begin{eqnarray*}
{\bf h}_n(z) =\sum_{k=0}^{n}\left(\begin{array}{c}n\\k\end{array}\right) z^k
\end{eqnarray*}
as required.
Defining
\begin{eqnarray}\label{4qRogersnorm}
\psi_n(z;q)=\frac{1}{\sqrt{[n]!_q}}H_n(z)
=\frac{1}{\sqrt{[n]!_q}}\sum_{k=0}^{n}\left[\begin{array}{c}n\\k\end{array}\right]_q z^k,\quad n=0,1,2\cdots,
\end{eqnarray}
one can straightforwardly infer  that
\begin{eqnarray}\label{4qRogersdiff}
\partial_q\psi_n(z;q)= \sqrt{[n]_q}\psi_{n-1}(z;q)
\end{eqnarray}
with the property that for $n=0, 1, 2, \cdots$
\begin{eqnarray}
\partial_q^{n+1}\psi_n(z;q) = 0\quad\mbox{and}\quad \partial_q^{m}\psi_n(z;q) \neq 0 \quad\mbox{for any  }
m<n+1.
\end{eqnarray}
It follows from Eqs. (\ref{4qRogersttr}) and (\ref{4qRogersnorm})
that the polynomials $\{\psi_n(z;q)\;|\; n= 0, 1, 2, \cdots\}$
satisfy the following three-term recursion relation
\begin{eqnarray}\label{4qRogersttr1}
 \sqrt{[n+1]_q}\psi_{n+1}(z;q)=(1+z)\psi_n(z;q)-z(1-q) \sqrt{[n]_q}\psi_{n-1}(z;q)
\end{eqnarray}
and the  $q-$difference equation
\begin{eqnarray}\label{4qRogersttr2}
\left((1+z)-(1-q)z\;\partial_q\right)\psi_n(z;q)=  \sqrt{[n+1]_q}\psi_{n+1}(z;q)
\end{eqnarray}
obtained from Eq.(\ref{4qRogersdiff}).
Hence, it is natural to formally define the number operator $N$  as
\begin{eqnarray}\label{4qRogersnb}
N\psi_n(z;q)= n\psi_{n}(z;q)
\end{eqnarray}
determined for the creation and  annihilation operators expressed as
\begin{eqnarray}
A^\dag = 1+z-(1-q)z\;\partial_q \quad\mbox{and}\quad A = \partial_q
\end{eqnarray}
respectively.
Indeed, the proofs of the following  relations are immediate:
\begin{eqnarray}
N\psi_n(z;q)&=& n\psi_{n}(z;q),\\
A^\dag\psi_n(z;q)&=&  \sqrt{[n+1]_q}\psi_{n+1}(z;q),\\
A\psi_n(z;q)&=&  \sqrt{[n]_q}\psi_{n-1}(z;q),\\
A^\dag A\psi_n(z;q)&=&  [n]_q\psi_{n}(z;q) = [N]_q\psi_{n}(z;q),\\
AA^\dag\psi_n(z;q)&=&  [n+1]_q\psi_{n}(z;q) = [N + 1]_q\psi_{n}(z;q).
\end{eqnarray}
Therefore, one concludes that the set of polynomials $\{\psi_n(z;q)\;|n = 0, 1, 2,
\cdots \}$ provides a basis for a realization of the $q$-deformed
harmonic oscillator algebra given by
\begin{eqnarray}\label{4qHarm}
AA^\dag - qA^\dag A=1,\qquad [ N, \;A]=-A ,\qquad [N,\;A^\dag]=A^\dag
\end{eqnarray}

%%%%%%%%%%%GENERALIZATION%%%%%%%%%%%%%%%SECTION III
\subsubsection{${\cal R}(p,q)-$generalized Rogers-Szeg\"o polynomials and quantum algebras}
We can now supply the general procedure for constructing
the recursion relation for the ${\cal R}(p,q)-$Rogers-Szeg\"o
polynomials and the related ${\cal R}(p,q)$-difference
equation that allow to define the
creation, annihilation and number operators for a given $({\cal
R},p,q)-$deformed quantum algebra. This is summarized as follows.
\begin{theorem}\label{4theoBuk}
If $\phi_i(x,y)$ ($i=1,2,3$) are functions satisfying:
\begin{eqnarray}
&&\phi_i(p,q)\neq 0\qquad \mbox{for   } i=1, 2, 3,\label{4RpqRog0}\\
&&\phi_i(P,Q)z^k= \phi_i^k(p,q)z^k \quad \mbox{for }\; z\in\mathbb{C},\;k= 0, 1, 2,
\cdots\quad i=1, 2\label{4RpqRog1}
\end{eqnarray}
and if, moreover, the following relation between ${\cal R}(p,q)-$binomial coefficients holds
\begin{eqnarray}\label{4RpqRog2}
&& \left[\begin{array}{c} n+1 \\ k \end{array}\right]_{{\cal R}(p,q)} =
\phi_1^k(p,q)\left[\begin{array}{c} n \\ k \end{array}\right]_{{\cal R}(p,q)} +
\phi_2^{n+1-k}(p,q)\left[\begin{array}{c} n \\ k-1 \end{array}\right]_{{\cal R}(p,q)}
\cr&&\qquad\qquad\qquad\qquad\qquad\qquad\qquad
-\phi_3(p,q){\cal R}(p^n,q^n)\left[\begin{array}{c} n-1 \\ k-1 \end{array}\right]_{{\cal R}(p,q)}
\end{eqnarray}
for $1\leq k\leq n$, then the ${\cal R}(p,q)-$Rogers-Szeg\"o polynomials defined as
\begin{eqnarray}
H_n(z;{\cal R}(p,q))=\sum_{k=0}^{n}\left[\begin{array}{c}n\\k\end{array}\right]_{{\cal R}(p,q)} z^k,\quad n=0,1,2\cdots
\end{eqnarray}
satisfy the three-term recursion relation
\begin{eqnarray}\label{4RpqRogersttr}
H_{n+1}(z;{\cal R}(p,q))&=& H_n\left(\phi_1(p,q)z:{\cal R}(p,q)\right)
\cr&&\qquad+ z\phi_2^n(p,q)H_n\left(z\phi_2^{-1}(p,q);{\cal R}(p,q)\right)
\cr&&\qquad-z\phi_3(p,q){\cal R}(p^n,q^n)H_{n-1}\left(z;{\cal R}(p,q)\right)
\end{eqnarray}
and ${\cal R}(p,q)-$difference equation
\begin{eqnarray}\label{4RpqRogerdiff}
\partial_{{\cal R}(p,q)}H_n(z;{\cal R}(p,q))= {\cal R}(p^n,q^n)H_{n-1}(z;{\cal R}(p,q)).
\end{eqnarray}
\end{theorem}
{\bf Proof:} Multiplying the two sides of the relation (\ref{4RpqRog2}) by $z^k$ and adding for $k= 1$ to $n$ we get
\begin{eqnarray}
&&\sum_{k=1}^n\left[\begin{array}{c} n+1 \\ k \end{array}\right]_{{\cal R}(p,q)}z^k =
\sum_{k=1}^n\phi_1^k(p,q)\left[\begin{array}{c} n \\ k \end{array}\right]_{{\cal R}(p,q)}z^k
\cr&&\qquad\qquad\qquad\qquad\qquad +
\sum_{k=1}^n\phi_2^{n+1-k}(p,q)\left[\begin{array}{c} n \\ k-1 \end{array}\right]_{{\cal R}(p,q)}z^k
\cr&&\qquad\qquad\qquad\qquad\qquad
-\phi_3(p,q){\cal R}(p^n,q^n)\sum_{k=1}^n\left[\begin{array}{c} n-1 \\ k-1 \end{array}\right]_{{\cal R}(p,q)}z^k.
\end{eqnarray}
After a short computation and using the condition (\ref{4RpqRog2})
we get Eq.(\ref{4RpqRogersttr}). Then there immediately results the  proof of
Eq.(\ref{4RpqRogerdiff}).
\hfill $\square$\\
Setting
\begin{equation}
 \psi_{n}(z;{\cal R}(p,q))= \frac{1}{\sqrt{{\cal R}!(p^n,q^n)}}H_n(z;{\cal R}(p,q)),
\end{equation}
and using the equations (\ref{4RpqRogersttr}) and
(\ref{4RpqRogerdiff}) yield the three-term recursion relation
\begin{eqnarray}\label{4RpqRogersttr1}
&&\left(\phi_1(P,Q)+ z\phi_2^n(p,q)\phi_2^{-1}(P,Q)
-z\phi_3(p,q)\partial_{{\cal R}(p,q)}\right)\psi_n(z;{\cal R}(p,q))=\cr
&&\qquad\qquad\qquad\sqrt{{\cal R}(p^{n+1},q^{n+1})}\;\;\psi_{n+1}(z;{\cal R}(p,q))
\end{eqnarray}
and ${\cal R}(p,q)-$difference equation
\begin{eqnarray}\label{4RpqRogerdiff1}
\partial_{{\cal R}(p,q)}\psi_n(z;{\cal R}(p,q))= \sqrt{{\cal R}(p^n,q^n)}\;\;\psi_{n-1}(z;{\cal R}(p,q))
\end{eqnarray}
for the polynomials $\psi_n(z;{\cal R}(p,q))$ with the virtue that  for $n=0, 1, 2, \cdots$
\begin{eqnarray}
\partial_{{\cal R}(p,q)}^{n+1}\psi_n(z;{\cal R}(p,q)) = 0\;\mbox{and}\;\;
\partial_{{\cal R}(p,q)}^{m}\psi_n(z;{\cal R}(p,q)) \neq 0 \;\;\;\mbox{for }m<n+1.
\end{eqnarray}
Now, formally defining  the number operator $N$ as
\begin{eqnarray}
 N\psi_n(z;{\cal R}(p,q))= n\psi_n(z;{\cal R}(p,q)),
\end{eqnarray}
and the raising  and lowering  operators by
\begin{eqnarray}
&&A^\dag =\left(\phi_1(P,Q)+ z\phi_2^N(p,q)\phi_2^{-1}(P,Q)
-z\phi_3(p,q)\partial_{{\cal R}(p,q)}\right)
%\end{eqnarray}
 \;\mbox{and}\;
%\begin{eqnarray}
 \cr &&A = \partial_{{\cal R}(p,q)},
\end{eqnarray}
respectively, the set of  polynomials $\{\psi_n(z;{\cal
R},p,q)\;|\;n= 0, 1, 2, \cdots \}$ provides a basis for a
realization of ${\cal R}(p,q)-$deformed quantum algebra ${\cal
A}_{{\cal R}(p,q)}$ satisfying the commutation relations
(\ref{2algN1}). Provided the above formulated theorem, we can now show how the  realizations in
terms of Rogers-Szeg\"o polynomials  can be derived for different known deformations simply by
determining the functions $\phi_i$ ($i=1, 2, 3$) that satisfy the
relations (\ref{4RpqRog0})-(\ref{4RpqRog2}).
%%%%%%%%%%%%%%%%%%%%%%%%%%%%%%%%%%%%%%%%%%%%%%%%

\subsection{Continuous ${\cal R}(p,q)-$Hermite polynomials}\label{Sect44}
We exploit here the  peculiar relation established in the   theory of $q-$deformation between Rogers-Szeg\"o polynomials and Hermite
polynomials \cite{IsmailMoudard05,Jagannathan&Rao,Klimyk&Schmudgen,Koekoek}  and given by
\begin{eqnarray}
\mathbb{H}_n(\cos\theta;q)= e^{in\;\theta}H_n(e^{-2i\;\theta};q)=
\sum_{k=0}^n\left[\begin{array}{c}n\\k\end{array}\right]_q\;e^{i(n-2k)\theta}, n=0, 1, 2, \cdots,
\end{eqnarray}
where $\mathbb{H}_n$ and  $H_n$ stand for the Hermite and
Rogers-Szeg\"o polynomials, respectively. Is also of interest the
property that all the $q-$Hermite polynomials  can  be explicitly
recovered from the initial one $\mathbb{H}_{0}(\cos\theta;q)= 1$, using
the three-term recurrence relation
\begin{eqnarray}\label{4qHttr}
 \mathbb{H}_{n+1}(\cos\theta;q)= 2\cos\theta\mathbb{H}_{n}(\cos\theta;q)-(1-q^{n})\mathbb{H}_{n-1}(\cos\theta;q)
\end{eqnarray}
with $\mathbb{H}_{-1}(\cos\theta;q)=0$.

%\subsection{Generalization}
In the same way we define the ${\cal R}(p,q)$-Hermite polynomials through the ${\cal R}(p,q)$-Rogers-Szeg\"o
polynomials as
\begin{eqnarray}
\mathbb{H}_n(\cos\theta;{\cal R}(p,q))= e^{in\;\theta}H_n(e^{-2i\;\theta};{\cal R}(p,q)),\quad n=0, 1, 2, \cdots.
\end{eqnarray}
Then the next statement is true.
\begin{proposition}\label{4propBuk}
Under the hypotheses of the theorem \ref{4theoBuk}, the continuous ${\cal R}(p,q)$-Hermite polynomials satisfy
the following three-term recursion relation
\begin{eqnarray}
\mathbb{H}_{n+1}(\cos\theta;{\cal R}(p,q))&=&
e^{i\;\theta}\phi_1^{\frac{n}{2}}(p,q)\phi_1(P,Q)\mathbb{H}_{n}(\cos\theta;{\cal R}(p,q))
\cr&&\quad + e^{-i\;\theta}\phi_2^{\frac{n}{2}}(p,q)\phi_2^{-1}(P,Q) \mathbb{H}_{n}(\cos\theta;{\cal R}(p,q))
\cr&&\quad-\phi_3(p,q){\cal R}(p^n,q^n)\mathbb{H}_{n-1}(\cos\theta;{\cal R}(p,q)).\qquad
\end{eqnarray}
\end{proposition}

{\bf Proof:} Multiplying the two sides of the three-term
recursion relation (\ref{4RpqRogersttr}) by $e^{i(n+1)\theta},$ we
obtain, for $z= e^{-2i\theta}$,
\begin{eqnarray*}%\label{4RpqRogersttr}
e^{i(n+1)\theta}H_{n+1}(e^{-2i\theta};{\cal R}(p,q))&=&
e^{i(n+1)\theta}H_n\left(\phi_1(p,q)e^{-2i\theta};{\cal R}(p,q)\right)
\cr&+& e^{i(n-1)\theta}\phi_2^n(p,q)H_n\left(\phi_2^{-1}(p,q)e^{-2i\theta};{\cal R}(p,q)\right)
\cr&-&e^{i(n-1)\theta}\phi_3(p,q){\cal R}(p^n,q^n)H_{n-1}\left(e^{-2i\theta};{\cal R}(p,q)\right)
\cr&=&e^{i\theta}e^{in\theta}\phi_1(P,Q)H_n\left(e^{-2i\theta};{\cal R}(p,q)\right)
\cr&+& e^{-i\theta}\phi_2^n(p,q)e^{in\theta} \phi_2^{-1}(P,Q)H_n\left(e^{-2i\theta};{\cal R}(p,q)\right)
\cr&-& \phi_3(p,q){\cal R}(p^n,q^n)e^{i(n-1)\theta}H_{n-1}\left(e^{-2i\theta};{\cal R}(p,q)\right).
\end{eqnarray*}
The required result follows  from the use of the equalities
\begin{eqnarray}
e^{in\theta}\phi_1(P,Q)H_n\left(e^{-2i\theta};{\cal R}(p,q)\right)=
\phi_1^{\frac{n}{2}}(p,q)\phi_1(P,Q)e^{in\theta}H_n\left(e^{-2i\theta};{\cal R}(p,q)\right), 
\end{eqnarray}
\begin{eqnarray}
e^{in\theta}\phi_2^{-1}(P,Q)H_n\left(e^{-2i\theta};{\cal R}(p,q)\right)=
% \cr&&\qquad\qquad\qquad\qquad\qquad
\phi_2^{-\frac{n}{2}}(P,Q)\phi_2^{-1}(P,Q)e^{in\theta}
H_n\left(e^{-2i\theta};{\cal R}(p,q)\right)\cr
&&
\end{eqnarray}
with
\begin{eqnarray}
 \phi_j(P,Q)e^{-2ik\theta} = \phi_j^k(p,q)e^{-2ik\theta}, \; j=1, 2,\; k= 0, 1, 2, \cdots.
\end{eqnarray}
\hfill $\square$\\

%%%%%%%%%%%%%%%%%%% PARTICULAR CASES %%%%%%%%%%%%%%%%%%

\subsection{Relevant particular cases}\label{Sect45}
The following pertinent cases deserve to be raised, as their derivation from the previous general theory
appeals concrete expressions for the deformed function ${\cal R}(p,q).$

\subsubsection{${\cal R}(x,y) = \frac{x-y}{p-q}$}\label{4janasec}

In this case, the ${\cal R}(p,q)-$factors are simply given by
\begin{eqnarray*}
 [n]_{p,q}={\cal R}(p^n,q^n)=\frac{p^n-q^n}{p-q},\quad n= 0,\; 1,\; 2,\; \cdots
\end{eqnarray*}
 with the ${\cal R}(p,q)-$factorials defined by
\begin{eqnarray}
[n]!_{p,q}= \left\{\begin{array}{lr} 1 \quad \mbox{for   } \quad n=0 \quad \\
\displaystyle\prod_{k=1}^n\frac{p^k-q^k}{p-q}=\frac{((p,q);(p,q))_n}{(p-q)^n} \quad \mbox{for } \quad n\geq
1. \quad \end{array} \right.
\end{eqnarray}
They correspond to the Jagannathan-Srinivasa $(p,q)$-numbers  and $(p,q)$-factorials \cite{Jagannathan&Rao, Jagannathan&Sridhar10}.

There result the following  relevant properties.
\begin{proposition}\label{4PropJan1}
If $n$ and $m$ are nonnegative integers, then
\begin{eqnarray}
[n]_{p,q}&=& \sum_{k=0}^{n-1}p^{n-1-k}q^k,
\cr [n+m]_{p,q}&=& q^m[n]_{p,q}+p^n[m]_{p,q}= p^m[n]_{p,q}+q^n[m]_{p,q},
\cr [-m]_{p,q}&=& -q^{-m}p^{-m}[m]_{p,q},
\cr [n-m]_{p,q}&=& q^{-m}[n]_{p,q}-q^{-m}p^{n-m}[m]_{p,q}= p^{-m}[n]_{p,q}-q^{n-m}p^{-m}[m]_{p,q},\quad
\cr [n]_{p,q}&=& [2]_{p,q}[n-1]_{p,q}-pq[n-2]_{p,q}.
\end{eqnarray}
\end{proposition}

\begin{proposition}\label{4PropJan2}
The $(p,q)-$binomial coefficients
\begin{eqnarray}
 \left[\begin{array}{c} n \\ k \end{array}\right]_{p,q}\equiv\frac{[n]!_{p,q}}{[k]!_{p,q}[n-k]!_{p,q}}=
\frac{((p,q);(p,q))_n}{((p,q);(p,q))_k((p,q);(p,q))_{n-k}}, 
\end{eqnarray}
where $0\leq k\leq n,\;\; n\in\mathbb{N},$ and $((p,q);(p,q))_m = (p-q)(p^2-q^2)\cdots(p^m-q^m)$,
$m\in\mathbb{N},$ satisfy the following identities:
\begin{eqnarray}
 \left[\begin{array}{c} n \\ k \end{array}\right]_{p,q}&=& \left[\begin{array}{c} n \\ n-k \end{array}\right]_{p,q}=
p^{k(n-k)}\left[\begin{array}{c} n \\ k \end{array}\right]_{q/p}=
p^{k(n-k)}\left[\begin{array}{c} n \\ n-k \end{array}\right]_{q/p},\label{4Janeq1}
\\ \left[\begin{array}{c} n+1 \\ k \end{array}\right]_{p,q} &=& p^k\left[\begin{array}{c} n \\ k \end{array}\right]_{p,q}
+q^{n+1-k}\left[\begin{array}{c} n \\ k-1 \end{array}\right]_{p,q},\label{4Janeq2}
\\ \left[\begin{array}{c} n+1 \\ k \end{array}\right]_{p,q} &=&
p^{k}\left[\begin{array}{c} n \\ k \end{array}\right]_{p,q} +
p^{n+1-k}\left[\begin{array}{c} n \\ k-1 \end{array}\right]_{p,q} -(p^n-q^n)
\left[\begin{array}{c} n-1 \\ k-1 \end{array}\right]_{p,q}\quad\label{4Janeq3}
\end{eqnarray}
with
\begin{eqnarray*}
\left[\begin{array}{c} n \\ k \end{array}\right]_{q/p}= \frac{(q/p; q/p)_n}{(q/p; q/p)_k(q/p; q/p)_{n-k}},
\end{eqnarray*}
where $(q/p; q/p)_n = (1-q/p)(1-q^2/p^2)\cdots (1-q^n/p^n)$
and the $(p,q)-$shifted factorial
\begin{eqnarray}
((a,b);(p,q))_n &\equiv& (a-b)(ap-bq)\cdots(ap^{n-1}-bq^{n-1})\cr
 &=& \sum_{k=0}^{n}\left[\begin{array}{c} n \\ k \end{array}\right]_{p,q}(-1)^kp^{(n-k)(n-k-1)/2}
q^{k(k-1)/2}a^{n-k}b^k.
\end{eqnarray}
\end{proposition}

The algebra ${\cal A}_{p,q}$, generated by $\{1,\; A,\; A^\dag,\;N\}$, associated with
 $(p,q) -$ Janagathan - Srinivasa deformation, satisfies the following commutation relations \cite{Jagannathan&Rao, Jagannathan&Sridhar10}:
\begin{eqnarray}\label{4pqHosc}
 A\;A^\dag- pA^\dag A= q^N, \quad &&A\;A^\dag- qA^\dag A= p^N\cr
[N,\;A^\dag]= A^\dag,\quad\qquad\quad&& [N,\;A]= -A.\label{4JanSrialg}
\end{eqnarray}

The $(p,q)$-Rogers-Szeg\"o polynomials  studied in \cite{Jagannathan&Sridhar10}
appear as a particular case obtained by choosing
$\phi_1(x,y)=\phi_2(x,y)=\phi(x,y)= x$ and $\phi_3(x,y)= x-y$. Indeed,\\
$\phi(p,q)= p\neq 0$, $ \phi_3(p,q)=p-q\neq 0$, $\phi(P,Q)z^k=
\phi_1^k(p,q)z^k$ and Eq.(\ref{4Janeq3}) shows that
$$\left[\begin{array}{c} n+1 \\ k \end{array}\right]_{p,q} =
p^{k}\left[\begin{array}{c} n \\ k \end{array}\right]_{p,q} +
p^{n+1-k}\left[\begin{array}{c} n \\ k-1 \end{array}\right]_{p,q} -(p-q)[n]_{p,q}
\left[\begin{array}{c} n-1 \\ k-1 \end{array}\right]_{p,q}.$$
Hence, the hypotheses of the above theorem are satisfied and, therefore,
the $(p,q)$-Rogers-Szeg\"o polynomials
\begin{eqnarray}
 H_n(z;p,q)=\sum_{k=0}^{n}\left[\begin{array}{c}n\\k\end{array}\right]_{p,q}z^k\quad n= 0, 1, 2, \cdots
\end{eqnarray}
satisfy the three-term recursion relation
\begin{eqnarray}\label{4JannRogersttr}
H_{n+1}(z;p,q)&=& H_n(pz;p,q) + z p^nH_n(p^{-1}z;p,q)\cr&&-z(p^n-q^n)H_{n-1}(z;p,q)
\end{eqnarray}
and $(p,q)-$difference equation
\begin{eqnarray}\label{4JanRogerdiff}
\partial_{p,q}H_n(z;p,q)= [n]_{p,q}H_{n-1}(z;p,q).
\end{eqnarray}
Finally, the set of polynomials
\begin{eqnarray}
 \psi_n(z;p,q)= \frac{1}{\sqrt{[n]!_{p,q}}}H_n(z;p,q),\quad n=0, 1, 2, \cdots
\end{eqnarray}
forms a basis for a realization of the $(p,q)-$deformed harmonic
oscillator and quantum algebra ${\cal A}_{p,q}$ satisfying
the commutation relations (\ref{4JanSrialg}) with the number
operator $N$  defined  as
\begin{eqnarray}
 N\psi_n(z;p,q)= n\psi_n(z;p,q),
\end{eqnarray}
relating  the annihilation and  creation operators  given by
\begin{eqnarray}
 A= \partial_{p,q}\quad\mbox{and}\quad A^\dag= P+ zp^NP^{-1}-z(p-q)\partial_{p,q}
\end{eqnarray}
respectively. Naturally, setting $p=1$ one recovers the results of the
subsection~\ref{4qHosc}.

The continuous $(p,q)-$Hermite polynomials have been already suggested in
\cite{Jagannathan&Rao} without any further details. In the above achieved
generalization, these polynomials are given by
\begin{eqnarray}
 \mathbb{H}_n(\cos\theta;p,q)&=& e^{in\theta}H_n(e^{-2i\theta};p,q)\cr&=&
\sum_{k=0}^n\left[\begin{array}{c}n\\k\end{array}\right]_{p,q}\;e^{i(n-2k)\theta},\quad n=0, 1, 2, \cdots.
\end{eqnarray}
Since for the $(p,q)-$deformation  $\phi_1(x,y)=\phi_2(x,y)= x$ and
$\phi_3(x,y)=x-y$, from the Proposition \ref{4propBuk} we deduce
that the corresponding sequence of continuous $(p,q)-$polynomials
satisfies the three-term recursion relation
\begin{eqnarray}\label{4pqHttr}
 \mathbb{H}_{n+1}(\cos\theta;p,q)&=& p^{\frac{n}{2}}(e^{i\theta}P+e^{-i\theta}P^{-1})\mathbb{H}_n(\cos\theta;p,q)
\cr&&-(p^n-q^n)\mathbb{H}_{n-1}(\cos\theta;p,q),
\end{eqnarray}
with $Pe^{i\theta}=p^{-1/2}e^{i\theta}$. This relation turns to be the well-known
three-term recursion relation (\ref{4qHttr}) for continuous $q-$Hermite polynomials in the limit $p\to 1$.
As matter of illustration, let us  explicitly compute the first three  polynomials  using the relation (\ref{4pqHttr}), with
$\mathbb{H}_{-1}(\cos\theta;p,q)=0$ and $\mathbb{H}_0(\cos\theta;p,q)=1$:\
\begin{eqnarray*}
\mathbb{H}_1(\cos\theta;p,q)&=&p^0(e^{i\theta}P+e^{-i\theta}P^{-1})1-(p^0-q^0)0 = e^{i\theta}+e^{-i\theta}=2\cos\theta\cr
&=&\left[\begin{array}{c}1\\0\end{array}\right]_{p,q}e^{i\theta}+\left[\begin{array}{c}1\\1\end{array}\right]_{p,q}e^{-i\theta}.\cr
\mathbb{H}_2(\cos\theta;p,q)&=& p^{\frac{1}{2}}(e^{i\theta}P+e^{-i\theta}P^{-1})(e^{i\theta}+e^{-i\theta})
-(p-q)1  \cr&=& e^{2i\theta}+e^{-2i\theta}+p+q= 2\cos2\theta +p+q
 \cr&=&\left[\begin{array}{c}2\\0\end{array}\right]_{p,q}e^{2i\theta}+\left[\begin{array}{c}2\\1\end{array}\right]_{p,q}e^{0i\theta}
+\left[\begin{array}{c}2\\2\end{array}\right]_{p,q}e^{-i\theta}.
\end{eqnarray*}
\begin{eqnarray*}
&&\mathbb{H}_3(\cos\theta;p,q)= p(e^{i\theta}P+e^{-i\theta}P^{-1})(e^{2i\theta}+e^{-2i\theta}+p+q)
\cr&&\qquad\qquad\qquad\qquad\qquad\qquad\quad-(p^2-q^2)(e^{i\theta}+e^{-i\theta})
\cr&&\qquad\qquad= e^{3i\theta}+e^{-3i\theta}+(p^2+pq+q^2)(e^{i\theta}+e^{-i\theta})
\cr&&\qquad\qquad= 2\cos3\theta+2(p^2+pq+q^2)\cos\theta\cr
&&\qquad\qquad=\left[\begin{array}{c}3\\0\end{array}\right]_{p,q}e^{3i\theta}+\left[\begin{array}{c}3\\1\end{array}\right]_{p,q}e^{i\theta}
+\left[\begin{array}{c}3\\2\end{array}\right]_{p,q}e^{-i\theta}+\left[\begin{array}{c}3\\3\end{array}\right]_{p,q}e^{-3i\theta}.
%\\
%\cdots\cdot\qquad\quad\quad&&
\end{eqnarray*}

\subsubsection{${\cal R}(x,y)= \frac{1-xy}{(p^{-1}-q)x}$}

The ${\cal R}(p,q)$-factors and ${\cal R}(p,q)-$factorials are reduced to $(p^{-1},q)$-numbers and
$(p^{-1},q)$-factorials, namely,
\begin{eqnarray*}
 [n]_{p^{-1},q}=\frac{p^{-n}-q^n}{p^{-1}-q},
\end{eqnarray*}
and
\begin{eqnarray}
[n]!_{p^{-1},q}= \left\{\begin{array}{lr} 1 \quad \mbox{for   } \quad n=0 \quad \\
\frac{((p^{-1},q);(p^{-1},q))_n}{(p^{-1}-q)^n} \quad \mbox{for } \quad n\geq
1, \quad \end{array} \right.
\end{eqnarray}
respectively, which exactly reproduce the $(p,q)$-numbers and
$(p,q)-$factorials introduced by Chakrabarty and Jagannathan \cite{Chakrabarti&Jagan}.

The other properties can be recovered similarly to those of section \ref{4janasec}
replacing the parameter $p$ by $p^{-1}$.

The ${\cal R}(p,q)-$derivative is also reduced to $(p^{-1},q)-$derivative. Indeed,
\begin{eqnarray}
 \partial_{{\cal R}(p,q)} &=& \partial_{p,q}\frac{p-q}{P-Q}\frac{1-PQ}{(p^{-1}-q)P}
\cr&=&\frac{1}{(p^{-1}-q)z}(P^{-1}-Q)\equiv \partial_{p^{-1},q}
\end{eqnarray}
obtained by a simple replacement of the dilatation operator   $P$ by $P^{-1}$.

The algebra ${\cal A}_{p^{-1},q}$, generated by $\{1,\; A,\; A^\dag,\;N\}$, associated with
 $(p,q) -$ Chakrabarty and Jagannathan deformation satisfies the following commutation relations:
\begin{eqnarray}
 A\;A^\dag- p^{-1}A^\dag A= q^N, \quad &&A\;A^\dag- qA^\dag A= p^{-N}\cr
[N,\;A^\dag]= A^\dag\qquad\qquad\quad&& [N,\;A]= -A.\label{4ChakJagalg}
\end{eqnarray}

Hence, the
$(p^{-1},q)-$Rogers-Szeg\"o polynomials
\begin{eqnarray}
 H_n(z;p^{-1},q)=\sum_{k=0}^{n}\left[\begin{array}{c}n\\k\end{array}\right]_{p^{-1},q}z^k\quad n= 0, 1, 2, \cdots
\end{eqnarray}
obey the three-term recursion relation
\begin{eqnarray}\label{4ChakJagRogersttr}
H_{n+1}(z;p^{-1},q)&=& H_n(p^{-1}z;p^{-1},q) + z p^{-n}H_n(pz;p^{-1},q)
\cr&&\qquad-z(p^{-n}-q^n)H_{n-1}(z;p^{-1},q)
\end{eqnarray}
and $(p^{-1},q)-$difference equation
\begin{eqnarray}\label{4RChakJagRogerdiff}
\partial_{p^{-1},q}H_n(z;p,q)= [n]_{p^{-1},q}H_{n-1}(z;p,q).
\end{eqnarray}
Finally, the set of polynomials
\begin{eqnarray}
 \psi_n(z;p^{-1},q)= \frac{1}{\sqrt{[n]!_{p^{-1},q}}}H_n(z;p^{-1},q),\quad n=0, 1, 2, \cdots
\end{eqnarray}
forms a basis for a realization of the $(p^{-1},q)-$deformed harmonic oscillator and quantum algebra ${\cal A}_{p^{-1},q}$
 generating
the commutation relations (\ref{4ChakJagalg}) with the number
operator $N$ formally defined  as
\begin{eqnarray}
 N\psi_n(z;p^{-1},q)= n\psi_n(z;p^{-1},q),
\end{eqnarray}
and the annihilation and  creation operators  given by
\begin{eqnarray}
 A= \partial_{p^{-1},q}\quad\mbox{and}\quad A^\dag= P^{-1}+ zp^{-N}P-z(p^{-1}-q)\partial_{p^{-1},q},
\end{eqnarray}
respectively. Naturally, setting $p=1$ permits to recover the results of the subsection~\ref{4qHosc}.

\subsubsection{${\cal R}(x,y)= \frac{xy-1}{(q-p^{-1})y}$}

In this case, the ${\cal R}(p,q)-$factors and ${\cal R}(p,q)-$factorials are reduced to 
\begin{eqnarray*}
 [n]_{p,q}^Q=\frac{p^n-q^{-n}}{q-p^{-1}},
\end{eqnarray*}
and
\begin{eqnarray}
[n]!_{p,q}^Q= \left\{\begin{array}{lr} 1 \quad \mbox{for   } \quad n=0 \quad \\
\frac{((p,q^{-1});(p,q^{-1}))_n}{(q-p^{-1})^n} \quad \mbox{for } \quad n\geq
1, \quad \end{array} \right.
\end{eqnarray}
introduced in our previous work \cite{Hounkonnou&Ngompe07b},  generalizing the $q-$Quesne algebra \cite{Quesne&al02}.

Then follow some remarkable properties:
\begin{proposition}\label{4PropoQes1}
 If $n$ and $m$ are nonnegative integers, then
\begin{eqnarray}
\;[-m]_{p,q}^Q&=& -p^{-m}q^m[m]_{p,q}^Q,\label{4Qeq1}\\
\;[n+m]_{p,q}^Q&=& q^{-m}[n]_{p,q}^Q+p^n[m]_{p,q}^Q= p^m[n]_{p,q}^Q+q^{-n}[m]_{p,q}^Q,\label{4Qeq2}\\
\;[n-m]_{p,q}^Q&=& q^{m}[n]_{p,q}^Q-p^{n-m}q^m[m]_{p,q}^Q= p^{-m}[n]_{p,q}^Q+p^{-m}q^{m-n}[m]_{p,q}^Q,\label{4Qeq3}\\
\;[n]_{p,q}^Q &=& \frac{q-p^{-1}}{p-q^{-1}}[2]_{p,q}^Q[n-1]_{p,q}^Q-pq^{-1}[n-2]_{p,q}^Q.\label{4Qeq4}
\end{eqnarray}
\end{proposition}
{\bf Proof:}  Eqs.(\ref{4Qeq1}) and (\ref{4Qeq2}) are immediate by the application of the relations
$ p^{-m}-q^{m} = -p^{-m}q^{m}(p^m-q^{-m})$ and
$p^{n+m}-q^{-n-m}= q^{-m}(p^{n}-q^{-n})+p^n(p^{m}-q^{-m})=p^m(p^{n}-q^{-n})+q^{-n}(p^{m}-q^{-m}),$
respectively, while Eq.(\ref{4Qeq3}) results from the combination of Eqs.(\ref{4Qeq1}) and (\ref{4Qeq2}). Finally, the relation
\begin{eqnarray}\label{4Qeq5}
[n]_{p,q^{-1}}&=& \frac{p^n-q^{-n}}{p-q^{-1}}= \frac{q-p^{-1}}{p-q^{-1}}\frac{p^n-q^{-n}}{q-p^{-1}}
= \frac{q-p^{-1}}{p-q^{-1}}[n]_{p,q}^{Q},\; n=1, 2, \cdots
\end{eqnarray}
cumulatively taken with the  identity
\begin{eqnarray*}
 [n]_{p,q{-1}}&=& [2]_{p,q{-1}}[n-1]_{p,q{-1}}-pq^{-1}[n-2]_{p,q^{-1}}
\end{eqnarray*}
gives Eq.(\ref{4Qeq4}).\hfill$\Box$

\begin{proposition}\label{4PropoQes2}
The $(p,q)-$Quesne binomial coefficients
\begin{eqnarray}
 \left[\begin{array}{c} n \\ k \end{array}\right]_{p,q}^Q=
\frac{((p,q^{-1});(p,q^{-1}))_n}{((p,q^{-1});(p,q^{-1}))_k((p,q^{-1});(p,q^{-1}))_{n-k}}, \label{4Qeq6}
\end{eqnarray}
where $\quad 0\leq k\leq n;\;\; n\in\mathbb{N},$ satisfy the following properties
\begin{eqnarray}
\left[\begin{array}{c} n \\ k \end{array}\right]_{p,q}^Q
= \left[\begin{array}{c} n \\ n-k \end{array}\right]_{p,q}^Q=
p^{k(n-k)}\left[\begin{array}{c} n \\ k \end{array}\right]_{1/qp}=
p^{k(n-k)}\left[\begin{array}{c} n \\ n-k \end{array}\right]_{1/qp},\label{4Qeq7}
\end{eqnarray}
\begin{eqnarray}
\left[\begin{array}{c} n+1 \\ k \end{array}\right]_{p,q}^Q =
p^k\left[\begin{array}{c} n \\ k \end{array}\right]_{p,q}^Q
+q^{-n-1+k}\left[\begin{array}{c} n \\ k-1 \end{array}\right]_{p,q}^Q,\label{4Qeq8}
\end{eqnarray}
\begin{eqnarray}
 \left[\begin{array}{c} n+1 \\ k \end{array}\right]_{p,q}^Q &=&
p^{k}\left[\begin{array}{c} n \\ k \end{array}\right]_{p,q}^Q +
p^{n+1-k}\left[\begin{array}{c} n \\ k-1 \end{array}\right]_{p,q}^Q 
\cr&&\qquad\quad-(p^n-q^{-n})
\left[\begin{array}{c} n-1 \\ k-1 \end{array}\right]_{p,q}^Q.\quad\label{4Qeq9}
\end{eqnarray}
\end{proposition}
{\bf Proof:} It is straightforward, using the Proposition~\ref{4PropJan1}  and
\begin{eqnarray}
 \left[\begin{array}{c} n \\ k \end{array}\right]_{p,q}^Q =
\left[\begin{array}{c} n \\ k \end{array}\right]_{p,q^{-1}}.
\end{eqnarray}
\hfill$\Box$\\
Finally, the algebra ${\cal A}_{p,q}^Q$, generated by $\{1,\; A,\; A^\dag,\;N\}$, associated with
 $(p,q) -$ Quesne deformation satisfies the following commutation relations:
\begin{eqnarray}
 p^{-1}A\;A^\dag- A^\dag A= q^{-N-1}, \quad&& qA\;A^\dag- A^\dag A= p^{N+1}\cr
[N,\;A^\dag]= A^\dag,\qquad\qquad\qquad&& [N,\;A]= -A.\label{4Qalg}
\end{eqnarray}

The ($p,q)-$Rogers-Szeg\"o polynomials corresponding to the Quesne
deformation \cite{Hounkonnou&Ngompe07b} are deduced from our generalization by choosing
$\phi_1(x,y)=\phi_2(x,y)=\phi(x,y)= x$ and $\phi_3(x,y)= y-x^{-1}$. Indeed, it is worthy of attention that we get in this case
$\phi(p,q)= p\neq 0$, $ \phi_3(p,q)=q-p^{-1}\neq 0$, $\phi(P,Q)z^k= \phi_1^k(p,q)z^k$ and from Eq.(\ref{4Qeq9})
$$\left[\begin{array}{c} n+1 \\ k \end{array}\right]_{p,q}^Q =
p^{k}\left[\begin{array}{c} n \\ k \end{array}\right]_{p,q}^Q +
p^{n+1-k}\left[\begin{array}{c} n \\ k-1 \end{array}\right]_{p,q}^Q -(q-p^{-1})[n]_{p,q}
\left[\begin{array}{c} n-1 \\ k-1 \end{array}\right]_{p,q}^Q.$$
Hence, the hypotheses of the theorem are satisfied and, therefore,
the $(p,q) -$ Rogers - Szeg\"o polynomials
\begin{eqnarray}
 H_n^Q(z;p,q)=\sum_{k=0}^{n}\left[\begin{array}{c}n\\k\end{array}\right]_{p,q}^Qz^k,\quad n= 0, 1, 2, \cdots
\end{eqnarray}
satisfy the three-term recursion relation
\begin{eqnarray}\label{4QRogersttr}
H_{n+1}^Q(z;p,q)= H_n^Q(pz;p,q) &+& z p^nH_n^Q(p^{-1}z;p,q)
\cr&-&z(p^n-q^{-n})H_{n-1}^Q(z;p,q)
\end{eqnarray}
and the $(p,q)-$difference equation
\begin{eqnarray}\label{4QRogerdiff}
\partial_{p,q}^QH_n^Q(z;p,q)= [n]_{p,q}^QH_{n-1}^Q(z;p,q).
\end{eqnarray}
Thus, the set of polynomials
\begin{eqnarray}
 \psi_n^Q(z;p,q)= \frac{1}{\sqrt{[n]!_{p,q}^Q}}H_n^Q(z;p,q),\quad n=0, 1, 2, \cdots
\end{eqnarray}
forms a basis for a realization of the $(p,q)-$ Quesne deformed
harmonic oscillator and quantum algebra ${\cal A}_{p,q}^Q$
engendering the commutation relations (\ref{4Qalg}) with the number
operator $N$  formally defined  as
\begin{eqnarray}
 N\psi_n^Q(z;p,q)= n\psi_n^Q(z;p,q),
\end{eqnarray}
and the annihilation and  creation operators  given by
\begin{eqnarray}
 A= \partial_{p,q}^Q\quad\mbox{and}\quad A^\dag= P+ zp^NP^{-1}-z(q-p^{-1})\partial_{p,q},
\end{eqnarray}
respectively. Naturally, setting $p=1$ gives the Rogers-Szeg\"o
polynomials associated with the $q-$Quesne deformation \cite{Quesne&al02}.

The continuous $(p,q)-$Hermite polynomials corresponding to the $(p, q)-$ generalization of  Quesne deformation \cite{Hounkonnou&Ngompe07b} can be defined
as follows:
\begin{eqnarray}
 \mathbb{H}_n^Q(\cos\theta;p,q)&=& e^{in\theta}H_n^Q(e^{-2i\theta};p,q)
\cr&=&
\sum_{k=0}^n\left[\begin{array}{c}n\\k\end{array}\right]_{p,q}^Q\;e^{i(n-2k)\theta},\quad n=0, 1, 2, \cdots.
\end{eqnarray}
Since for the $(p, q)$-generalization of  Quesne deformation \cite{Hounkonnou&Ngompe07b} $\phi_1(x,y)=\phi_2(x,y)= x$ and
$\phi_3(x,y)=y-x^{-1}$, from the Proposition \ref{4propBuk} we
deduce that the corresponding sequence of continuous
$(p,q)-$Hermite polynomials satisfies the three-term recurrence
relation
\begin{eqnarray}\label{4QHttr}
 \mathbb{H}_{n+1}^Q(\cos\theta;p,q)&=&p^{\frac{n}{2}}(e^{i\theta}P+e^{-i\theta}P^{-1})\mathbb{H}_n^Q(\cos\theta;p,q)
\cr&&\qquad-
(p^n-q^{-n})\mathbb{H}_{n-1}^Q(\cos\theta;p,q).
\end{eqnarray}

%%%%%%%%%%%%%%%%%%%%%%%%%%%%%%%%%%%%%%%%%%%%%%%%%%%

\subsubsection{$\displaystyle{\cal R}(x,y)= h(p,q)y^\nu/x^\mu\left[\frac{xy-1}{(q-p^{-1})y}\right]$}

Here
$0< pq < 1$ , $p^\mu<q^{\nu-1}$, $p>1$ ,  $h$ is a well behaved
real and non-negative function of deformation parameters $p$ and
$q$ such that  $h(p,q)\to 1$ as $(p,q)\to (1,1).$

The ${\cal R}(p,q)-$factors become $(p,q;\mu,\nu,h)$-numbers introduced in our previous work \cite{Hounkonnou&Ngompe07a} and defined by
\begin{eqnarray}
 [n]^{\mu,\nu}_{p,q,h}= h(p,q)\frac{q^{\nu n}}{p^{\mu n}}\frac{p^n-q^{-n}}{q-p^{-1}}.
\end{eqnarray}
\begin{proposition}\label{4PropoHouk1}
The $(p,q;\mu,\nu,h)-$numbers verify the following properties, for $m,n\in\mathbb{N}$:
\begin{eqnarray}
\;[-m]^{\mu,\nu}_{p,q,h}= -\frac{q^{-2\nu m+m}}{p^{-2\mu m+m}}[m]^{\mu,\nu}_{p,q,h},\label{4Heq1}
\end{eqnarray}
\begin{eqnarray}
\;[n+m]^{\mu,\nu}_{p,q,h}&=& \frac{q^{\nu m-m}}{p^{\mu m}}[n]^{\mu,\nu}_{p,q,h}
+\frac{q^{\nu n}}{p^{\mu n-n}}[m]^{\mu,\nu}_{p,q,h}
\cr&=&\frac{q^{\nu m}}{p^{\mu m- m}}[n]^{\mu,\nu}_{p,q,h}
+\frac{q^{\nu n- n}}{p^{\mu n}}[m]^{\mu,\nu}_{p,q,h}\;,\label{4Heq2}
\end{eqnarray}
\begin{eqnarray}
\;[n-m]^{\mu,\nu}_{p,q,h}&=&\frac{q^{-\nu m+ m}}{p^{-\mu m}}[n]^{\mu,\nu}_{p,q,h}
-\frac{q^{\nu(n-2m)+m}}{p^{\mu(n-2m)-n+m}}[m]^{\mu,\nu}_{p,q,h}
\cr&=& \frac{q^{-\nu m}}{p^{-\mu m+m}}[n]^{\mu,\nu}_{p,q,h}
-\frac{q^{\nu(n-2m)-n+m}}{p^{\mu(n-2m)+m}}[m]^{\mu,\nu}_{p,q,h},\label{4Heq3}
\end{eqnarray}
\begin{eqnarray}
\; [n]^{\mu,\nu}_{p,q,h}= \frac{q-p^{-1}}{p-q^{-1}}\frac{q^{-\nu}}{p^{-\mu}}\frac{1}{h(p,q)}[2]^{\mu,\nu}_{p,q,h}
[n-1]^{\mu,\nu}_{p,q,h}-\frac{q^{2\nu-1}}{p^{2\nu-1}}[n-2]^{\mu,\nu}_{p,q,h}.\label{4Heq4}
\end{eqnarray}
\end{proposition}
{\bf Proof:} It is direct using the Proposition~\ref{4PropoQes1} and the fact that
\begin{eqnarray}\label{4Houkeq}
 [n]^{\mu,\nu}_{p,q,h}= h(p,q)\frac{q^{\nu n}}{p^{\mu n}}[n]^Q_{p,q}.
\end{eqnarray}
\hfill$\Box$
\begin{proposition}\label{4PropoHouk2}
The $(p,q,\mu,\nu,h)-$ binomial coefficients
\begin{eqnarray}
 \left[\begin{array}{c} n \\ k \end{array}\right]_{p,q,h}^{\mu,\nu}:=
\frac{[n]!_{p,q,h}^{\mu,\nu}}{[k]!_{p,q,h}^{\mu,\nu}[n-k]!_{p,q,h}^{\mu,\nu}}=
\frac{q^{\nu k(n-k)}}{p^{\mu k(n-k)}}\left[\begin{array}{c} n \\ k \end{array}\right]_{p,q}^Q,
\label{4Heq5}
\end{eqnarray}
where $\;0\leq k\leq n;\;\; n\in\mathbb{N},$ satisfy the following properties
\begin{eqnarray}
\left[\begin{array}{c} n \\ k \end{array}\right]_{p,q,h}^{\mu,\nu}
= \left[\begin{array}{c} n \\ n-k \end{array}\right]_{p,q,h}^{\mu,\nu},\label{4Heq6}
\end{eqnarray}
\begin{eqnarray}
 \left[\begin{array}{c} n+1 \\ k \end{array}\right]_{p,q,h}^{\mu,\nu} =
\frac{q^{\nu k}}{p^{(\mu-1)k}}\left[\begin{array}{c} n \\ k \end{array}\right]_{p,q,h}^{\mu,\nu}
+\frac{q^{(\nu-1)(n+1-k)}}{p^{\mu(n+1-k)}}\left[\begin{array}{c} n \\ k-1 \end{array}\right]_{p,q,h}^{\mu,\nu},
\label{4Heq7}
\end{eqnarray}
\begin{eqnarray}
 \left[\begin{array}{c} n+1 \\ k \end{array}\right]_{p,q,h}^{\mu,\nu} &=&
\frac{q^{\nu k}}{p^{(\mu-1)k}}\left[\begin{array}{c} n \\ k \end{array}\right]_{p,q,h}^{\mu,\nu} +
\frac{q^{\nu(n+1-k)}}{p^{(\mu-1)(n+1-k)}}\left[\begin{array}{c} n \\ k-1 \end{array}\right]_{p,q,h}^{\mu,\nu}
\cr&&\qquad\qquad\quad-(p^n-q^{-n})
\frac{q^{\nu n}}{p^{\mu n}}\left[\begin{array}{c} n-1 \\ k-1 \end{array}\right]_{p,q,h}^{\mu,\nu}.\quad\label{4Heq8}
\end{eqnarray}
\end{proposition}
{\bf Proof:} There follow from the Proposition~\ref{4PropoQes2} and the fact that
\begin{eqnarray}
 [n]!_{p,q,h}^{\mu,\nu}= h^n(p,q)\frac{q^{n(n+1)/2}}{p^{n(n+1)/2}}[n]!_{p,q}^Q,
\end{eqnarray}
where  use of Eq.(\ref{4Houkeq}) has been made.\hfill$\Box$\\
%%%%%%
The algebra ${\cal A}_{p,q,h}^{\mu,\nu}$, generated by
$\{1,\; A,\; A^\dag,\;N\}$, associated with
 $(p,q,\mu,\nu,h)$-deformation, satisfies the following commutation relations:
\begin{eqnarray}
&& p^{-1}A\;A^\dag- \frac{q^\nu}{p^\mu} A^\dag A= h(p,q)\left(\frac{q^{\nu-1}}{p^\mu}\right)^{N+1},
\cr&&
qA\;A^\dag- \frac{q^\nu}{p^\mu}A^\dag A= h(p,q)\left(\frac{q^\nu}{p^{\mu-1}}\right)^{N+1}\qquad\cr
&&[N,\;A^\dag]= A^\dag,\qquad\qquad [N,\;A]= -A.\label{4Hkalg}
\end{eqnarray}

The $(p,q,\mu,\nu,h)$-Rogers-Szeg\"o \cite{Hounkonnou&Ngompe07a} polynomials are deduced from
the above general construction  by setting $\phi_1(x,y)= x^{1-\mu}y^\nu$,
$\phi_2(x,y) = x^{-\mu}y^{\nu-1}$ and $\phi_3(x,y)=
\frac{y-x^{-1}}{h(p,q)}$. Indeed, $\phi_i(p,q)\neq 0$ for $i=1, 2,
3$; $\phi_i(P,Q)z^k = \phi_i(p,q)^kz^k$ for $i= 1, 2$ and  the
property (\ref{4Heq8}) furnishes
\begin{eqnarray*}
\left[\begin{array}{c} n+1 \\ k \end{array}\right]_{p,q,h}^{\mu,\nu} &=&
\frac{q^{\nu k}}{p^{(\mu-1)k}}\left[\begin{array}{c} n \\ k \end{array}\right]_{p,q,h}^{\mu,\nu} +
\frac{q^{\nu(n+1-k)}}{p^{(\mu-1)(n+1-k)}}\left[\begin{array}{c} n \\ k-1 \end{array}\right]_{p,q,h}^{\mu,\nu}
\cr&&\qquad-
\frac{q-p^{-1}}{h(p,q)}[n]^{\mu,\nu}_{p,q,h}\left[\begin{array}{c} n-1 \\ k-1 \end{array}\right]_{p,q,h}^{\mu,\nu}.
\end{eqnarray*}
Therefore, the $(p,q,\mu,\nu,h)$-Rogers-Szeg\"o polynomials are defined as follows:
\begin{eqnarray}
H_n(z;p,q,\mu,\nu,h)=\sum_{k=0}^{n}\left[\begin{array}{c}n\\k\end{array}\right]_{p,q,h}^{\mu,\nu} z^k,\quad n=0,1,2\cdots
\end{eqnarray}
with the three-term recursion relation
\begin{eqnarray}\label{4HkRogersttr}
H_{n+1}(z;p,q,\mu,\nu,h)&=& H_n\left(\frac{q^{\nu}}{p^{\mu-1}}z:p,q,\mu,\nu,h\right)
\cr&&\qquad+ z\frac{q^{(\nu-1)n}}{p^{\mu n}}H_n\left(\frac{p^{\nu}}{q^{\nu-1}}z;p,q,\mu,\nu,h\right)
\cr&&\qquad-z\frac{q^{\nu n}}{p^{\mu n}}(p^n-q^{-n})H_{n-1}(z;p,q,\mu,\nu,h)
\end{eqnarray}
and $(p,q,\mu,\nu,h)-$difference equation
\begin{eqnarray}\label{4HkRogerdiff}
\partial_{p,q,h}^{\mu,\nu}H_n(z;p,q,\mu,\nu,h)= [n]_{p,q,h}^{\mu,\nu}H_{n-1}(z;p,q,\mu,\nu,h).
\end{eqnarray}
Hence, the set of polynomials
\begin{eqnarray}
 \psi_n(z;p,q,\mu,\nu,h)= \frac{1}{\sqrt{[n]!_{p,q,h}^{\mu,\nu}}}H_n(z;p,q,\mu,\nu,h),\quad n=0, 1, 2, \cdots
\end{eqnarray}
forms a basis for a realization of the $(p,q,\mu,\nu,h)-$deformed algebra ${\cal A}_{p,q,\mu,\nu,h}$ satisfying the commutation relations (\ref{4Hkalg}) with the number
operator $N$ formally defined  as
\begin{eqnarray}
 N\psi_n^Q(z;p,q,\mu,\nu,h)= n\psi_n(z;p,q,\mu,\nu,h),
\end{eqnarray}
together with the annihilation and the creation operators  given by
\begin{eqnarray}
 A= \partial_{p,q,h}^{\mu,\nu}\;\mbox{ and }\;A^\dag= \frac{Q^{\nu}}{P^{\mu-1}}
+ z\left(\frac{q^{\nu-1}}{p^{\mu}}\right)^N\frac{P^{\mu}}{Q^{\nu-1}}-z\frac{(q-p^{-1})}{h(p,q)}\partial_{p,q,h}^{\mu,\nu},
\end{eqnarray}
respectively.

The continuous $(p,q,\mu,\nu,h)-$Hermite polynomials \cite{Hounkonnou&Ngompe07a} can be now deduced as:
\begin{eqnarray}
 \mathbb{H}_n(\cos\theta;p,q,\mu,\nu,h)&=& e^{in\theta}H_n(e^{-2i\theta};p,q,\mu,\nu,h)\cr&=&
\sum_{k=0}^n\left[\begin{array}{c}n\\k\end{array}\right]_{p,q,h}^{\mu,\nu}\;e^{i(n-2k)\theta},\quad n=0, 1, 2, \cdots.
\end{eqnarray}
Since for the $(p,q,\mu,\nu,h)-$deformation
$\phi_1(x,y)= x^{1-\mu}y^\nu$, $\phi_2(x,y) = x^{-\mu}y^{\nu-1}$ and $\phi_3(x,y)= \frac{y-x^{-1}}{h(p,q)}$,
from the Proposition \ref{4propBuk}  the corresponding sequence of continuous
$(p,q,\mu,\nu,h)-$Hermite polynomials satisfies the three-term recursion relation
\begin{eqnarray}\label{4HkHttr}
 \mathbb{H}_{n+1}(\cos\theta;p,q,\mu,\nu,h)&=&
\frac{q^{\nu\frac{n}{2}}}{p^{(\mu-1)\frac{n}{2}}}
\frac{Q^{\nu}}{P^{\mu-1}}\mathbb{H}_n(\cos\theta;p,q,\mu,\nu,h)\cr&&
+\frac{q^{(\nu-1)\frac{n}{2}}}{p^{\mu\frac{n}{2}}}
\frac{Q^{-(\nu-1)}}{P^{-\mu}}
\mathbb{H}_n(\cos\theta;p,q,\mu,\nu,h)\cr&&
-(p^n-q^{-n})\frac{q^{\nu n}}{p^{\mu n}}\;\mathbb{H}_{n-1}(\cos\theta;p,q,\mu,\nu,h).
\end{eqnarray}

%%%%%%%%%%%%%%%%%%%%%%%%%%%%%%%%%%%%%%%%%%%%%%%%%%%%%%%%%%%%%%%%%%%%%%%%%
\section{Concluding remark}
We have first deformed the Heisenberg algebra   with the set of parameters
 $\{q, l,\lambda\}$  to generate  a new family of generalized 
coherent states respecting  the Klauder criteria. In this framework, 
 the matrix elements of  relevant operators have been  exactly computed and investigated from functional analysis point of view.
 Then, relevant statistical properties have been examined. Besides,
 a proof on the sub-Poissonian character of the statistics of  the main deformed states 
has been provided. This property   has been finally used to   determine   the induced generalized metric, characterizing 
the geometry of the considered system.

% \ref{chap2} % chapter 2

% \section{Concluding remark}\label{Sec2.4}
Next, a unified method of defining structure functions  from commutation relations of deformed single-mode oscillator algebras 
has been  presented. 
A natural approach to building coherent states associated to  deformed algebras has been then deduced.
Known deformed algebras have been  given as illustration and such mathematical
properties  as continuity in the label, normalizability and resolution of the identity of their corresponding coherent states have been  discussed.

% \ref{chap3} % chapter 3

Besides, we have generalized a class of  two - parameter
deformed Heisenberg algebras related to meromorphic functions. %defined on
%$\mathbb{C}\times\mathbb{C}$
There have been probed  relevant
families of coherent states maps and their corresponding hypergeometric
series.
The latter constitutes a generalization of known hypergemotric
series.  Moreover, a  ${\cal R}(p, q)$-binomial theorem, generalizing the
$(p,q)$-binomial theorem given in ~\cite{Jagannathan&Rao} has been deduced.
We have also defined the ${\cal R}(p, q)$-trigonometric, hyperbolic and $(p,q)$-Bessel functions, including
their main relevant properties.

% \ref{chap4} % chapter 4

% \section{Concluding remarks}\label{2section6}
Then, we have provided a new noncommutative algebra related to the ${\cal R}(p,q)$-deformation and
 shown that the notions of differentiation  and integration can be extended to it, thus  generalizing
 well known $q$ or/and $(p,q)$-differential and integration calculi \cite{Burban&Klimyk, Dobrogo06, Koekoek}.
Besides, we have performed a general procedure of constructing the Hopf algebra structure compatible with the
${\cal R}(p,q)$-algebra. As illustration, relevant examples have been given.

% \ref{chap5} % chapter 5

Finally, we have defined and discussed  a general formalism for constructing ${\cal R}(p,q)-$ deformed Rogers-Szeg\"o
polynomials.   The displayed approach not only provides novel relations, but also generalizes  well known standard and deformed
 Rogers-Szeg\"o polynomials. A full characterization of the latter, including the data on the 
three-term recursion relations and difference equation, has been provided.
 We have succeeded in elaborating  a new realization of  ${\cal R}(p,q)-$deformed quantum algebra generalizing the construction of 
$q-$deformed harmonic oscillator creation and annihilation
operators  performed in  \cite{Galetti,Jagannathan&Sridhar10}. 
The  continuous ${\cal R}(p,q)-$Hermite polynomials have been also investigated in detail and relevant particular cases and examples have been exhibited.

\section*{Acknowledgements}
This work is partially supported by the Abdus Salam International
Centre for Theoretical Physics (ICTP, Trieste, Italy) through the
Office of External Activities (OEA) - \mbox{Prj-15}. The ICMPA
is in partnership with
the Daniel Iagolnitzer Foundation (DIF), France.


\begin{thebibliography}{200}
\bibitem{Abe77} E. Abe, {\it Hopf algebras} (Cambridge University Press, Cambridge, 1977).

\bibitem{Akhiezer} N.I. Akhiezer, {\it The Classical Moment Problem and Some Related Questions in Analysis} (Olivier and Boyd, London, 1965).

\bibitem{Ali&al99} S. T. Ali, J.-P. Antoine, J.-P. Gazeau, {\it Coherent States, Wavelets and their Generalizations} (Springer-Verlag, New-York, 1999).

\bibitem{Ali&al} S. T. Ali, J.-P. Antoine, J.-P. Gazeau and U.A. Mueller, Coherent states and their generalizations:
A mathematical overview, Rev. Math. Phys.{\bf7} (1995), 1013-1104.

\bibitem{Aragone&al76} C. Aragone, E. Chalbaud, S. Salamo, Intelligent spin states, J. Math. Phys. {\bf 17} (1976), 1963-1971.

\bibitem{Aragone&al74} C. Aragone, G. Guerri, S. Salamo, and J.L. Tani, Intelligent spin states, J. Phys. A: Math. Nucl. Gen. {\bf 7} (1974), L149-L151.

\bibitem{Arik&Coon} M. Arik and D.D. Coon, Hilbert spaces of analytic functions and generated coherent states,
J. Math. Phys. {\bf 17} (1976), 424-427.

\bibitem{Atakishiyeva} N.M. Atakishiyev, M.K. Atakishiyeva, A $q$-analogue of the Euler gamma integral, Theoret. Math. Phys. {\bf 129} (2001), 1325-1334.

\bibitem{Baloitcha} E. Balo\"itcha, M.N. Hounkonnou and E.B. Ngompe Nkouankam, Unified $(p,q; \alpha, \beta,\nu; \gamma)$-deformation: Irreducible representations and induced deformed harmonic oscillator, J. Math. Phys. {\bf 53} (2012), 013504-013514.

\bibitem{Barut} A. O. Barut and L. Girardello, New coherent states associated with noncompact groups, Commun. Math. Phys. {\bf 21} (1971), 41-55.

\bibitem{Berezanskii} Ju. M. Berezanski\'{i}, {\it Expansions in Eigenfunctions of Selfadjoint Operators}, (Amer. Math. Soc., Providence, Rhode Island, 1968).

\bibitem{Biedenharn} L. C. Biedenharn,
The quantum group $SU_q (2)$ and a $q$-analogue of the boson operators, J. Phys. A {\bf 22} (1989), L873-L878.

\bibitem{Borzov} V.V. Borzov, E.V. Damaskinsky and S.V. Yegorov, Somme remarks on the representations of the generalized deformed algebra, {\it arXiv:q-alg/9509022}.

\bibitem{Brzezenski} T. Brzezi\'nski, J.L. Egusquinza, A.J. Macfarlane, Generalised Harmonic Oscillator Systems and Their Fock Space Description, Phys. Lett. B {\bf 311} (1993),  202-206.

\bibitem{Bukweli&Hounkonnou12a} J.D. Bukweli Kyemba and M.N. Hounkonnou, Characterization of $({\cal R},p,q)$-deformed Rogers-Szeg\"o polynomials: associated quantum algebras, deformed Hermite polynomials and relevant properties, J. Phys. A: Math. Theor. {\bf 45} (2012), 225204.

\bibitem{Bukweli&Hounkonnou12b} J.D. Bukweli Kyemba and M.N. Hounkonnou, On generalized oscillator algebras and their associated coherent states, submitted.

\bibitem{Bukweli&Hounkonnou12c} J.D. Bukweli Kyemba and M.N. Hounkonnou, $(q;l,\lambda)$-deformed Heisenberg algebra: coherent states, their statistics and geometry, Afr. Diaspora J. Math. {\bf 14} (2012), 38-56.

\bibitem{Bukweli&Hounkonnou12d} J.D. Bukweli Kyemba and M.N. Hounkonnou, ${\cal R}(p,q)$-calculus: differentiation, integration and Hopf algebras, submitted.

\bibitem{Burban1} I.M. Burban, Unified $(q; \alpha, \beta, \gamma; \nu)$-deformation of one-parametric $q$-deformed oscillator algebras, Phys. Let. A {\bf 42} (2009), 056201, {\it arxiv.org/pdf/0806.0613}

\bibitem{Burban93} I. M. Burban, Two-parameter deformation of oscillator algebra, Phys. Lett.  B {\bf 319} (1993), 485-489.

\bibitem{Burban2007} I. M. Burban, { On $(p, q;\alpha, \beta, l)$-deformed oscillator and its generalized quantum Heisenberg-Weyl algebra}, { Phys. Lett.  A} {\bf 366} (2007), 308-314.

\bibitem{Burban&Klimyk} I. M. Burban and A. U.  Klimyk, {$P,Q-$differentiation, $P,Q-$integration, and $P,Q-$ hypergeometric
functions related to quantum groups.} {Integral Transforms and Special Functions.} {\bf 2}(1994), 15-36.

\bibitem{Wess&al99} B. L. Cerchiai, R. Hinterding, J. Madore and J. Wess, The geometry of $q$-deformed phase space, Eur. Phys. J. C {\bf 8} (1999), 533-546.

\bibitem{Chakrabarti&Jagan} R. Chakrabarti and R. Jagannathan, {A $(p, q)-$oscillator realisation of two-parameter quantum algebras.} { J. Phys. \rm A: Math. Gen.} {\bf 24} (1991), L711-L718.

\bibitem{Chari&Pressley} V. Chari and A. Pressley, {\it A guide to Quantum groups} (Cambridge, Cambridge University Press, 1994).

\bibitem{Chung} W.-S. Chung, K.-S Chung, S.-T. Nam and C.-I. Um, Generalized deformed algebra, Phys.Lett.A, {\bf 183} (1993), 363-370.

\bibitem{Dancoff} S.M. Dancoff, Non-adiabatic meson theory of nuclear forces, Phys. Rev. {\bf 78} (1950), 382-385.

\bibitem{Daoud} M. Daoud, Photon-added coherent states for exactly solvable Hamiltonians, Phys. Letters A {\bf 305} (2002), 135-143.

\bibitem{Drinfeld85} V. G. Drinfel'd, Hopf Algebras and  the quantum Yang-Baxter equation, Dokl. Akad. Nauk SSSR, {\bf 283} (1985), 1060-1064.

\bibitem{Drinfeld} V.G. Drinfel'd, Quantum Group in  Proc. Int. Congr. Math., Berkeley (1986), pp. 798-820.

\bibitem{Dobrogo06} A. Dobrogowska and A. Odzijewicz, Second order $q-$difference equations solvable by factorization method, J. Comput. Appl. Math. {\bf 193} (2006), 319-346.

\bibitem{ElBaz} M. El Baz and  Y. Hassouni, Deformed exterior algebra, Quons and their coherent states, Int. Jour. Mod. Phys. {\bf 18} (2003), {\it preprint: IC/2002/106}.

%\bibitem{Faddeev}  L. D. Faddeev, {\it Integrable Models in $(1+1)$-Dimensional Quantum Field Theory} (Lectures in Les Houches Lectures 1982, Amsterdam, Elsevier, 1984).

\bibitem{Faddeev&Takhtajan79} L. D. Faddeev and L. A. Takhtajan, The Quantum Method of the Inverse Problem and the Heisenberg $XYZ$ model, Russ. Math. Surveys {\bf34} (1979), 11-68.

%\bibitem{Faddeev&Takhtajan} L. D. Faddeev and L. A. Takhtajan, {\it Hamiltonian methods in the theory of soliton} (Springer, Berlin-Heidelberg-NewYork, 1987)

\bibitem{Faddeev&Yu} L. D. Faddeev and N. Yu. Reshetkhin, Hamiltonian structures for integrable models field theory, Teor. Mat. Fiz., {\bf 56} (1983), 323-343.

\bibitem{Faddeev&al90} L.D. Faddeev, N.Yu. Reshitikhin and L.A. Takhtajan, Quantization of Lie groups and Lie algebras, Leningrad Math J. {\bf 1} (1990), 193-225.

\bibitem{Feinsilver1} Ph. Feinsilver, Commutators, anti-commutators and Eulerian calculus, Rocky Mountain J. Math., {\bf 12} (1982), 171-183.

\bibitem{Feinsilver2} Ph. Feinsilver, Discrete analogues of the Heisenberg-Weyl algebra, Mn. Math., {\bf 104} (1987), 89-108.

\bibitem{Floreanini&al93a} R. Floreanini, L. Lapointe and L. Vinet, A note on $(p,q)-$oscillators and bibasic hypergeometric functions, { J. Phys. \rm A: Math. Gen.} {\bf 26} (1993), L611-L614.

\bibitem{Floreanini&al93b}  R. Floreanini, L. Lapointe and L. Vinet, A quantum algebra approach to basic multivariable special
functions, { J. Phys. \rm A: Math. Gen.} {\bf 27} (1994), 6781-6797.

\bibitem{Galetti} D. Galetti, A realization of the q-deformed harmonic oscillator: Rogers-Szeg\"o and Stieltjes-Wigert polynomials,
{  Braz. J. Phys. } {\bf 33} (2003), 148-157.

\bibitem{Gasper&Rahman90} G. Gasper and M. Rahman, {\it Basic Hypergeometric Series} (Cambridge, Cambridge University Press, 1990).
 
\bibitem{Gasper} G. Gasper and M. Rahman, {\it Basic Hypergeometric Series} (Cambridge University Press, Cambridge, 2004).

\bibitem{Gazeau09} J.-P. Gazeau, {\it Coherent States in Quantum Physics} (WILEY-VCH Verlag GmbH \& Co. KGaA, Weinheim, 2009).

\bibitem{Gazeau&Klauder} J.-P. Gazeau and J.R. Klauder, {\it Generalized Coherent States for Arbitrary Quantum Systems} (Academic Press, New-York, 2000).

\bibitem{Gelfand&al94} M.I. Gelfand, M. I. Graev and V. S. Retakh, {$(r,s)-$Hypergeometric functions of one variable}
{ Russian Acad. Sci. Dokl. Math} {\bf 48} (1994), 591-596.

% \bibitem{Gelfand&al98} I. M. Gelfand,  M. I. Graev  and V. S. Retakh, { General gamma functions, exponentials and hypergeometric functions}, Russian Math. Surveys {\bf 53} (1998).

\bibitem{Gilmore} R. Gilmore, Geometry of symmetrized states  { Ann. Phys. (NY)} {\bf 74} (1972), pp 391-463.

\bibitem{Gilmore74a} R. Gilmore, On properties of coherent states, Rev. Mex. Fis. {\bf 23} (1974), 143-187.

\bibitem{Glauber1} R.J. Glauber, Photon correlations, Phys. Rev. Letters {\bf 10} (1963), 84-86.

\bibitem{Glauber2} R.J. Glauber, Coherent and incoherent states of the radiation field, Phys. Rev. {\bf 131} (1963), 2766-2788. 

\bibitem{Horzela&Szafraniec12} A. Horzela and F.H. Szafraniec, A measure-free approach to coherent states, J. Phys. A Math. Theor. {\bf 45} (1012), 244018.

\bibitem{Hounkonnou&Bukweli10} M.N. Hounkonnou  and J.D. Bukweli Kyemba, Generalized $({\cal R},p,q)$-deformed Heisenberg algebras: coherent states and special functions,  {\it J. Math. Phys.} {\bf 51} (2010), 063518.

\bibitem{Hounkonnou&Ngompe07a} M. N. Hounkonnou and E. B. Ngompe Nkouankam, {New $(p, q, \mu, \nu, f)$-deformed states.} { J. Phys. A: Math. Theor.} {\bf 40} (2007), 12113-12130.

\bibitem{Hounkonnou&Ngompe07b} M .N Hounkonnou and E. B. Ngompe Nkouankam,
{On $(p, q, \mu, \nu, \phi_1 , \phi_2 )$ generalized oscillator algebra and related bibasic hypergeometric functions},
{ J. Phys. A: Math. Theor.} {\bf 40} (2007), 8835-8843.

\bibitem{Hounkonnou&Sodoga} M. N. Hounkonnou and K. Sodoga, Generalized coherent states for associated hypergeometric-type functions, J. Phys. A: Math. Gen. {\bf 38} (2005), 7851-7857.

\bibitem{IsmailMoudard05} E. H. Ismail Moudard, \emph{\it Classical and Quantum Orthogonal Polynomials in one Variable.}
{ Enc. Math Appl. Vol. 98} (Cambridge, Cambridge University Press, 2005).

\bibitem{Jackson1905} F.H. Jackson, The application of basic numbers to Bessel's and Legendre's functions II in
{ Proc. London Math. Soc.(2)}, {\bf 3} (1905), pp  192-220.

\bibitem{Jagannathan&Rao} R. Jagannathan and K. Srinivasa Rao, Two-parameter quantum algebras, twin-basic numbers, and associated
generalized hypergeometric series, {\it arXiv:math/0602613}.

\bibitem{Jagannathan&Sridhar10} R. Jagannathan and R. Sridhar, {$(p,q)$-Rogers-Szeg\"o polynomials and the $(p,q)$-oscillator}, {\it  arXiv:1005.4309v1 [math.QA]}

\bibitem{Jannussis} A.D. Jannussis, L.C. Papaloucas  and P.D. Siafarikas, Eigenfunctions and Eigenvalues of the $Q$-Differential Operators, Hadronic J. {\bf 3} (1980), 1622-1632.

\bibitem{Jannussis&al81} A.D. Jannussis, G. Bbodimas, D. Sourlas and V. Zisis, Remarks on the $q$-quantization, Lett. Nuovo Cim. {\bf 30} ( 1981), 123-127. 

\bibitem{Jannussis&al83a} A. Jannussis, G. Bbodimas, D. Sourlas, L.C. Papaloucas and P.D. Siafaricas, $Q$-algebras without interaction, Lett. Nuovo Cim. {\bf 37} ( 1983), 119-123.  


\bibitem{Jannussis&al83b} A. Jannussis, G. Bbodimas, D. Sourlas, L. Papaloucas, P.D. Siafaricas and K. Vlachos, Some properties of $Q$-analysis and applications to noncanonical mechanics, Hadronic J. {\bf 6} ( 1983), 1653-1686.  

\bibitem{Jimbo85} M. Jimbo, A $q$-difference analogue of  $U(g)$ and Yang-Baxter equation, Lett. Math. Phys. {\bf 10} (1985), 63-69.

\bibitem{Jimbo} M. Jimbo, Quantum $R$-matrix for the generalized Toda system, Comm. Math. Phys. {\bf 102} (1986), 537-547 .

\bibitem{Junker&Roy} G. Junker and P. Roy, Conditionally exactly solvable problems and non-linear algebras, Phys. Lett. A {\bf 232} (1997), 155-161.

\bibitem{Kalnins}  G. Kalnins, W. Miller and S. Mukhejee, Models of $q$-algebra representations: matrix elements of the $q$-oscillator algebra, J. Math. Phys. {\bf 34} (1993), 5333-5356
 
\bibitem{Kassel} C. Kassel, {\it Quantum groups} (Springer-Verlag, New York, 1995).

% \bibitem{Katriel&Kibler92} J. Katriel and M. Kibler, { Normal ordering for deformed boson operators and operatorvalued deformed Stirling numbers}, J. Phys. A: Math. Gen. {\bf 25} (1992).
 
\bibitem{Kirillov} A.A. Kirillov, {\it Elements of the theory of representation} (Springer, Berlin, Heidelberg, New York, 1976).

\bibitem{Klauder63a} R. J. Klauder, {Continuous-representation theory I. Postulates of
continuous representation theory} { J. Math. Phys. } {\bf 4} (1963), 1055-1058.

\bibitem{Klauder63b} R. J. Klauder, {Continuous-representation theory II. Generalized
relation between quantum and classical dynamics} { J. Math. Phys. } {\bf 4} (1963), 1058-1073.

\bibitem{Klauder&al} J.R. Klauder, K.A. Penson and J.-M. Sixderniers,
Constructing coherent states through solutions of Stieljes and Hausdorff moment problems, Phys. Rev. A {\bf 64} (2001), 013817.

\bibitem{Klauder&Skagerstam} J.R. Klauder and B.S. Skagerstam, {\it Coherent states: Applications in Physics and Mathematical Physics}
 (Singapore: World Scientific, 1985)

\bibitem{Klimyk&Schmudgen} A. Klimyk and K. Schm\"udgen, \emph{\it Quantum Groups and their Representation}, (Berlin Heidelberg, Springer-Verlag, Berlin Heidelberg, 1997).

\bibitem{Koekoek} R. Koekoek and R. F. Swarttouw, {\it The Askey-scheme of orthogonal polynomials and its $q-$analogue}, (TUDelft Report No. 98-17, 1998).

\bibitem{Kostant} B. Kostant, {\it Group Representation in Mathematics and Physics}, ed. by V. Bargamann, Lecture Notes Phys., Vol 6 (Springer, Berlin, Heidelberg, New York, 1970).

\bibitem{Kosinski} K. Kosi\'nski, M. Majewski and P. Ma\'slanka, Representations of generalized oscillator algebra, 
{\it arXiv:q-alg/9501012v1}

\bibitem{Kulish&Reshetikhin83} P.P. Kulish and N.Y. Reshetikhin, Quantum linear problem for the sine-Gordon equation and higher representations, J. Sov. Math. {\bf 23} (1983), 2435-2445.

\bibitem{Kuryshkin80} V.V. Kuryshkin, Annales de la Fondation Louis de Broglie {\bf 5} (1980), p. 111.

\bibitem{Macfarlane} A. J. Macfarlane,
On $q$-analogues of the quantum harmonic oscillator and quantum group $SU(2)_q$, J. Phys. A {\bf 22} (1989), 4581-4588.

\bibitem{Manko&al} V.I. Man'ko, G. Marmo, E.C.G. Sudarshan and F. Zaccaria, $f$-oscillators and nonlinear coherent states,  Phys. Scr. {\bf 55} (1997), 528-541.

\bibitem{deMatos&Vogel} R. L. de Matos Filho and W. Vogel, Nonlinear coherent states, Phys. Rev. A {\bf 54} (1996), 4560-4563.  

\bibitem{Maximov&Odzijewicz} V. Maximov and A. Odzijewicz, {The $q-$deformation of quantum mechanics of one degree of freedom.}
{ J. Math. Phys.} {\bf 36} (1995), 1681-1689.

\bibitem{Meljanac} S. Meljanac, M. Milekovi, and S. Pallua, Unified view of deformed single-mode oscillator algebras, 
Phys. Lett. B {\bf 328} (1994), 55-59.

\bibitem{Nieto&Simmons78} M.M. Nieto and L.M. Simmmons, Jr., Coherent states for general potentials, Phys. Rev. Lett. {\bf 41} (1978), 207-210.
 
\bibitem{Nieto&Simmons79} M.M. Nieto and L.M. Simmmons, Jr., Coherent states for general potentials: I, II and III, Phys. Rev. D {\bf 20} (1979), 1321-1350.

\bibitem{Nieto&al81} M.M. Nieto, L.M. Simmmons, Jr., and V.P. Gutschick, Coherent states for general potentials: IV, Phys. Rev. D {\bf 20} (1981), 391-402.

\bibitem{Odzijewicz98} A. Odzijewicz, {Quatum algebras and $q$-special functions
 related to coherent states maps of the disc.} { Commun. Math. Phys.} {\bf 192} (1998), 183-215.

\bibitem{Pasquier} V. Pasquier and H. Saleur, Common structures between finite systems and conformal field theories through quantum groups, Nucl. Phys. B {\bf 330}, (1990),523-556.

\bibitem{Perelomov72} A. M. Perelomov, Coherent states for arbitrary Lie group, Commun. Math. Phys. {\bf 26} (1972), 222-236.

\bibitem{Perelomov} A.M. Perelomov, {\it Generalized Coherent States and their Applications} (Springer-Verlag, Berlin Heidelberg, 1986).

\bibitem{Polyanin} A.D. Polyanin and A. V. Manzhirov, {\it Handbook of integral equations} (CRC Press, Boca Raton, 1998).

\bibitem{Popov} D. Popov, Gazeau-Klauder quasi-coherent states for the Morse oscillator, Phys. Letters A {\bf 316} (2003), 369-381. 

\bibitem{Prudnikov} A.P. Prudnikov, Yu A. Brychkov and O.I. Marichev, {\it Integrals and Series}, vol 3 (Gordon and
Breach, New York 1990).

\bibitem{Quesne} C. Quesne, New $q$-deformed coherent states with an explicitly known resolution of unity, J. Phys. A: Math.
Gen. {\bf 35} (2002), 9213-9226. 

\bibitem{Quesne&al02} C. Quesne, K.A. Penson and V.M. Tkachuk, Maths-type $q$-deformed coherent states for $q > 1$,
{ Phys. Lett.  A} {\bf  313} (2003),29-36.  {\it arXiv:quant-ph/0303120v2}.

\bibitem {Sahai97} V. Sahai, $q-$Representations of Lie algebras $\mathcal{G}(a,b)$, J. Indian Math. Soc. {\bf 63} (1997), 83-98.

\bibitem{Schrodinger} E. Schr\"odinger, Der stetige \"Ubergang von der Mikro- zur Makromechanik, Naturwissensehaften, {\bf 14 } (1926), 664-666.

\bibitem{Tarmakin} J.A. Shohat and J.D. Tamarkin, {\it The Problem of Moments}, (APS, New York, 1943).

\bibitem{Sklyanin83} E. K. Sklyanin, Some Algebraic Structures Connected with the Yang-Baxter Equation. Representations of Quantum Algebras, Funct. Anal. Appl., {\bf 17} (1983, 34-48).

\bibitem{Sneddon} I.N. Sneddon, The Use of Integral Transforms (McGraw-Hill, New York, 1974).

\bibitem{Snyder} H. S. Snyder, Quantized space-time, Phys. Rev. {\bf 71} (1947), 38-41

\bibitem{Sudarshan} E.C. G. Sudarshan, Equivalence of semiclassical and quantum mechanical descriptions of statistical light beams,
 Phys. Rev. Letters {\bf 10} (1963), 277-279.

\bibitem{Sun&Fu} C.-P. Sun, H.-C. Fu, The $q$-deformed boson realisation of the quantum group $SU(n)_q$ and its representations, J. Phys. A: Math. Gen. {\bf 22} (1989),  L983-L986.

\bibitem{Tamm} I.E. Tamm, The relativistic interaction of elementary particles, J. Phys. UdSSR {\bf 9} (1945), 449-465.

\bibitem{Vasiliev} M.A. Vasiliev, High spin algebras and quantization on the sphere and hyperboloid, Int. J. Mod. Phys. A {\bf 6} (1991), 1115.

\bibitem{Vega} H. J. de Vega, Yang-Baxter Algebras, Integrable Theories and Quantum Groups, Int. J. Mod. Phys. A {\bf 4}, (1989), 2371.

\bibitem{Wess97} J. Wess, $q$-deformed phase space and its lattice structure, Int. J. Mod. Phys. A {\bf 12} (1997), 4997-5006.

\bibitem{Witten} E. Witten,  Gauge theories, vertex models, and quantum groups, Nucl. Phys. B {\bf 330} (1990), 285-346.

\bibitem{Yan} H. Yan, $q$-deformed oscillator algebra as a quantum group , J. Phys. A: Math. Gen. {\bf 23} (1990), L1155-L1160.

\bibitem{Li&Sheng} L. You-quan and S. Zheng-mao, A deformation of quantum mechanics, J. Phys. A: Math. Gen. {\bf 25} (1992), 6779-6788.

\bibitem{Yu&al72} S. Yu, M. Baker and D. D. Coon, First and Second Factorization in a Dual Multiparticle Theory with Nonlinear Trajectories, Phys. Rev. D {\bf5}(1972),3108-3121.

% \bibitem{Tarmakin} J.D. Tamarkin and J.A. Shohat, {\it The Problem of Moments}, (APS, New York, 1943).

\bibitem{Zhang&al90} W.M. Zhang, D.H. Feng and R.G. Gilmore, Coherent states: theory and some applications, Rev. Mod. Phys. {\bf 62} (1990), 867-927.

\end{thebibliography}
\end{document}